\documentclass[runningheads]{llncs}
\pdfoutput=1
\usepackage{amsmath,amssymb,amsfonts}
\usepackage{xcolor}
\def\BibTeX{{\rm B\kern-.05em{\sc i\kern-.025em b}\kern-.08em
    T\kern-.1667em\lower.7ex\hbox{E}\kern-.125emX}}
\usepackage{bbm}
\usepackage{stmaryrd}
\usepackage{latexsym}
\usepackage{wrapfig}
\usepackage{booktabs}
\usepackage{here}
\usepackage{thmtools}
\usepackage{thm-restate}

\usepackage{algorithm}
\usepackage{algorithmic}

\makeatletter
\let\MYcaption\@makecaption
\makeatother
\usepackage{subcaption}
\usepackage{graphicx}

\allowdisplaybreaks[1]

\newcommand{\eqdef}{\ensuremath{\stackrel{\mathrm{def}}{=}}}
\newcommand{\randassign}{\ensuremath{\stackrel{\mathrm{\$}}{\leftarrow}}}

\newcommand{\argmin}{\operatornamewithlimits{argmin}}
\newcommand{\supp}{\mathsf{supp}}

\newcommand{\diam}{\mathsf{diam}}
\newcommand{\Dists}{\mathbb{D}} % distributions
\newcommand{\Prob}{\mathrm{Pr}}

\newcommand{\expect}{\operatornamewithlimits{\mathbb{E}}}

\newcommand{\lift}[1]{{#1}^{\#}}

\newcommand{\liftinf}[1]{{#1}_{\Winf}^{\#}}

\newcommand{\natspos}{\mathbb{N}^{>0}}
\newcommand{\reals}{\mathbb{R}}
\newcommand{\realspos}{\mathbb{R}^{>0}}
\newcommand{\realsnng}{\mathbb{R}^{\ge0}}

\newcommand{\cali}{\mathcal{I}}

\newcommand{\calx}{\mathcal{X}}
\newcommand{\caly}{\mathcal{Y}}
\newcommand{\calz}{\mathcal{Z}}

\newcommand{\DP}[0]{\textsf{DP}}
\newcommand{\XDP}[0]{\textsf{XDP}}
\newcommand{\pDP}[0]{\textsf{PDP}}

\newcommand{\DistP}[0]{\textsf{DistP}}
\newcommand{\XDistP}[0]{\textsf{XDistP}}
\newcommand{\pDistP}[0]{\textsf{PDistP}}

\newcommand{\ut}[0]{\mathit{u}}
\newcommand{\util}[2]{\ut(#1, #2)}

\newcommand{\utmetric}[0]{\mathit{d}}

\newcommand{\sens}[0]{\mathit{\varDelta}}
\newcommand{\sensPhi}[0]{\sens_{\varPhi,\utmetric}}
\newcommand{\sensfunc}[0]{\mathit{W}}
\newcommand{\sensinf}[0]{\sensfunc_{\infty,\utmetric}}
\newcommand{\sensemd}[0]{\sensfunc_{1,\utmetric}}

\newcommand{\Dinf}[0]{\mathit{D}_{\infty}}
\newcommand{\maxdiverge}[2]{\Dinf(#1 \parallel #2)}
\newcommand{\appmaxdiverge}[2]{\Dinf^{\delta}(#1 \parallel #2)}

\newcommand{\aprmaxdiverge}[3]{\Dinf^{#1}(#2 \parallel #3)}

\newcommand{\Winf}{\mathit{W}_{\infty}}

\newcommand{\Winfu}{\mathit{W}_{\infty,\utmetric}}

\newcommand{\cp}[2]{\mathsf{cp}(#1, #2)}

\newcommand{\GammaInf}[0]{\mathit{\Gamma_{\!{\rm \infty,\utmetric}}}}

\newcommand{\alg}{\mathit{A}}

\newcommand{\TP}[3]{\mathit{Q}^{\sf tp}_{#1,#2,#3}}
\newcommand{\TPM}{\TP{k}{\nu}{\alg}}
\newcommand{\LambdaBeta}{\Lambda_{\beta,\eta,\alg}}

\newcommand{\Seq}{\mathbin{\odot}}
\newcommand{\liftSeq}{\mathbin{\bullet}}

\newcommand{\PL}{\rm{PL}}

\newcommand{\TM}{\rm{TM}}
\newcommand{\RL}{{\rm RL}}
\newcommand{\attr}{t}
\newcommand{\male}{{\it male}}
\newcommand{\female}{{\it female}}
\newcommand{\north}{{\it north}}
\newcommand{\south}{{\it south}}
\newcommand{\social}{{\it social}}
\newcommand{\lesssocial}{{\it less-social}}
\newcommand{\workplace}{{\it workplace}}
\newcommand{\nonworkplace}{{\it non-workplace}}
\newcommand{\home}{{\it home}}
\newcommand{\outside}{{\it out}}

\newif\ifcommentson\commentsontrue

\newif\ifconferenceon\conferenceonfalse
\ifconferenceon
\newcommand{\arxiv}[1]{}
\newcommand{\conference}[1]{#1}
\newcommand{\conferenceShort}[1]{}
\newcommand{\journal}[1]{}
\else
\newcommand{\arxiv}[1]{#1}
\newcommand{\conference}[1]{}
\newcommand{\conferenceShort}[1]{}
\newcommand{\journal}[1]{}
\fi

\newcommand{\fdiv}[1]{}

\ifcommentson
\newcommand{\commentsize}[0]{.85\textwidth}
\newcommand{\commentTM}[1]{\begin{center} \parbox{\commentsize}{\textbf{\textcolor{black}{Comment T.}} \textcolor{red}{#1 }}\end{center}}
\newcommand{\commentYK}[1]{\begin{center} \parbox{\commentsize}{\textbf{\textcolor{black}{Comment Y.}} \textcolor{red}{#1} }\end{center}}
\newcommand{\replyTM}[1]{\begin{center} \parbox{\commentsize}{\textbf{Reply T.} \textcolor{blue}{#1} }\end{center}}
\newcommand{\replyYK}[1]{\begin{center} \parbox{\commentsize}{\textbf{Reply Y.} \textcolor{blue}{#1} }\end{center}}
\newcommand{\commentT}[1]{\marginpar{\footnotesize \color{red} {\bf T:} \textsf{\scriptsize #1}}}
\newcommand{\commentY}[1]{\marginpar{\footnotesize \color{red} {\bf Y:} \textsf{\scriptsize #1}}}
\newcommand{\replyT}[1]{\marginpar{\footnotesize \color{red} {\bf T:} \textsf{\scriptsize #1}}}
\newcommand{\replyY}[1]{\marginpar{\footnotesize \color{red} {\bf Y:} \textsf{\scriptsize #1}}}
\else
\newcommand{\commentTM}[1]{}
\newcommand{\commentYK}[1]{}
\newcommand{\replyTM}[1]{}
\newcommand{\replyYK}[1]{}
\newcommand{\commentT}[1]{}
\newcommand{\commentY}[1]{}
\newcommand{\replyT}[1]{}
\newcommand{\replyY}[1]{}
\fi

\newcommand{\colorR}[1]{\textcolor{red}{#1}}

\newcommand{\pagelimitmarker}[1]{~\\ {\colorR{\ifthenelse{\thepage>#1}{\Huge Exceeding the page limit}{\huge Within the page limit}}}~\\ {\huge{\colorR{~~Page Limit\,\,\,\,\, = #1}}}~\\ {\huge{\colorR{~~Current Page = $\thepage$}}}}

\begin{document}
\title{Local Obfuscation Mechanisms for Hiding Probability Distributions
\thanks{This work was partially supported by JSPS KAKENHI Grant JP17K12667, JP19H04113, and Inria LOGIS project.}
}
\author{Yusuke Kawamoto\inst{1}\orcidID{0000-0002-2151-9560}
\and \\
Takao Murakami\inst{2}\orcidID{0000-0002-5110-1261}
}
\authorrunning{Y. Kawamoto et al.}
\institute{AIST, Tsukuba, Japan \\
\conference{\email{yusuke.kawamoto.aist@gmail.com}}
\and
AIST, Tokyo, Japan
}
\maketitle
\begin{abstract}
We introduce a formal model for the information leakage of probability distributions and define a notion called distribution privacy as the local differential privacy for probability distributions. Roughly, the distribution privacy of a local obfuscation mechanism means that the attacker cannot significantly gain any information on the distribution of the mechanism's input by observing its output. Then we show that existing local mechanisms can hide input distributions in terms of distribution privacy, while deteriorating the utility by adding too much noise. For example, we prove that the Laplace mechanism needs to add a large amount of noise proportionally to the infinite Wasserstein distance between the two distributions we want to make indistinguishable. To improve the tradeoff between distribution privacy and utility, we introduce a local obfuscation mechanism, called a tupling mechanism, that adds random dummy data to the output. Then we apply this mechanism to the protection of user attributes in location based services. By experiments, we demonstrate that the tupling mechanism outperforms popular local mechanisms in terms of attribute obfuscation and service quality.

\keywords{
local differential privacy \and
obfuscation mechanism \and
location privacy \and
attribute privacy \and
Wasserstein metric \and
compositionality
}
\end{abstract}

\section{Introduction}
\label{sec:intro}

\emph{Differential privacy}~\cite{Dwork:06:ICALP} is a quantitative notion of privacy that has been applied to a wide range of areas, including databases, geo-locations, and social network.
The protection of differential privacy can be achieved by 
adding controlled noise to given data that we wish to hide or obfuscate.
In particular, 
a number of recent studies have proposed 
\emph{local obfuscation mechanisms}% 
~\cite{Duchi:13:FOCS,Andres:13:CCS,Erlingsson_CCS14}, 
namely, randomized algorithms that perturb each single ``point'' data (e.g., a geo-location point) 
by adding certain probabilistic noise before sending it out to a data collector.
However, the obfuscation of a probability distribution of points (e.g., a distribution of locations of users at home/outside home) still remains to be investigated in terms of differential privacy.

For example, a location-based service (LBS) 
provider 
collects each user's geo-location data to provide a service (e.g., navigation or
point-of-interest search), and
has been widely studied in terms of the privacy of user locations.
As shown in~\cite{Andres:13:CCS,Bordenabe:14:CCS}, users can hide their accurate locations by sending to the LBS provider only approximate location information calculated by an obfuscation mechanism.

Nevertheless, 
a user's location information can be used for an attacker to infer the user's attributes (e.g., age, gender, social status, and residence area) or activities (e.g., working, sleeping, and shopping)~\cite{Liao:07:IJRR,Zheng:09:LBSN,Matsuo:07:IJCAI,Yang:19:TKDE}.
For example, when an attacker knows the distribution of residence locations, he may detect whether given users are \emph{at home} or \emph{outside home} after observing their obfuscated locations.
For another example, an attacker may learn whether users are \emph{rich} or \emph{poor} by observing their obfuscated behaviors.
These attributes can be used by robbers hence should be protected from them.
Privacy issues of such attribute inference are also known in other applications, including recommender systems~\cite{Otterbacher:10:CIKM,Weinsberg:12:RecSys} and online social networks~\cite{Gong:18:TOPS,Mislove:10:WSDM}.
However, to our knowledge, no literature has addressed the protection of attributes in terms of local differential privacy.

To illustrate the privacy of attributes in an LBS, let us consider a running example where users try to prevent an attacker from inferring whether they are at home or not.
Let $\lambda_{\home}$ and $\lambda_{\outside}$ be the probability distributions of locations of the users at home and outside home, respectively.
Then the privacy of this attribute means that the attacker cannot learn from an obfuscated location whether the actual location follows the distribution $\lambda_{\home}$ or $\lambda_{\outside}$.

This can be formalized using differential privacy.
For each $\attr\in\{ \home, \allowbreak \outside \}$,\,
we denote by $p( y \,|\, \lambda_{\attr})$ the probability of observing an obfuscated location $y$ when an actual location is distributed over $\lambda_{\attr}$.
Then the privacy of $\attr$ is defined by:
\begin{align*}
\frac{ p( y \,|\, \lambda_{\home}) }{ p( y \,|\, \lambda_{\outside})~ }
\le e^{\varepsilon}
{,}
\end{align*}
which represents that the attacker cannot distinguish whether the users follow the distribution $\lambda_{\home}$ or $\lambda_{\outside}$ with degree of $\varepsilon$.

To generalize this,
we define a notion, called \emph{distribution privacy} (\DistP{}), 
as the differential privacy for probability distributions.
Roughly,
we say that a mechanism $\alg$ provides \DistP{} w.r.t. $\lambda_{\home}$ and $\lambda_{\outside}$ if
no attacker can detect whether the actual location (input to $\alg$) is sampled from $\lambda_{\home}$ or $\lambda_{\outside}$
after he observed an obfuscated location $y$ (output by $\alg$)%
\footnote{In our setting, the attacker observes only a sampled output of $\alg$, and not the exact histogram of $\alg$'s output distribution.
See Section~\ref{sub:histogram} for more details.}.
Here we note that each user applies the mechanism $\alg$ locally by herself, hence can customize the amount of noise added to $y$ according to the attributes she wants to hide.

Although existing local differential privacy mechanisms are designed to protect point data, 
they also hide the distribution that the point data follow.
However, we demonstrate that they need to add a large amount of noise to obfuscate distributions, and thus deteriorate the utility of the mechanisms.

To achieve both high utility and strong privacy of attributes, we introduce a mechanism, called the \emph{tupling mechanism}, that not only perturbs an actual input, but also adds random dummy data to the output.
Then we prove that this mechanism provides \DistP{}.
Since the random dummy data obfuscate the shape of the distribution, users can instead reduce the amount of noise added to the actual input, hence they get better utility (e.g., quality of a POI service).

This implies that \DistP{} is a relaxation of differential privacy that guarantees the privacy of attributes while achieving higher utility by weakening the differentially private protection of point data.
For example, suppose that users do not mind revealing their actual locations outside home, but want to hide (e.g., from robbers) the fact that they are outside home. 
When the users employ the tupling mechanism, they output both their (slightly perturbed) actual locations and random dummy locations.
Since their outputs include their (roughly) actual locations, they obtain high utility (e.g., learning shops near their locations), while their actual location points are protected only weakly by differential privacy.
However, their attributes \emph{at home}/\emph{outside home} are hidden among the dummy locations, hence protected by \DistP{}.
By experiments, we demonstrate that the tupling mechanism 
is useful to protect the privacy of attributes, and outperforms popular existing mechanisms (the randomized response~\cite{Kairouz:16:ICML}, 
the planar Laplace~\cite{Andres:13:CCS} and Gaussian mechanisms) in terms of \DistP{} and service quality.

~\\
\textbf{Our contributions.}~~
The main contributions of this work are given as follows:
\begin{itemize}
\item We propose a formal model for the privacy of probability distributions in terms of differential privacy.
Specifically, we define the notion of distribution privacy (\DistP{}) to represent that  the attacker cannot significantly gain information on the distribution of a mechanism's input by observing its output.

\item We provide theoretical foundation of \DistP{}, including its useful properties (e.g., compositionality) and its interpretation (e.g., in terms of Bayes factor).

\item 
We quantify the effect of distribution obfuscation by existing local mechanisms.
In particular, we show that (extended) differential privacy mechanisms are able to make any two distributions less distinguishable, while deteriorating the utility by adding too much noise to protect all point data.
\item 
For instance, we prove that extended differential privacy mechanisms (e.g., the Laplace mechanism) need to add a large amount of noise proportionally to the $\infty$-Wasserstein distance $\sensinf(\lambda_0, \lambda_1)$ between the two distributions $\lambda_0$ and $\lambda_1$ that we want to make indistinguishable.
\item
We show that \DistP{} is a useful relaxation of differential privacy when users want to hide their attributes, but not necessarily to protect all point data.
\fdiv{
We show that every $f$-divergence privacy (e.g. KL-privacy) mechanism obfuscates distributions
proportionally to the Earth mover's distance~$\sensemd$.
}

\item To improve the tradeoff between \DistP{} and utility, we introduce the \emph{tupling mechanism}, which locally adds random dummies to the output.
Then we show that this mechanism provides \DistP{} and hight utility for users.

\item We apply local mechanisms to the obfuscation of attributes in location based services (LBSs).
Then we show that the tupling mechanism outperforms popular existing mechanisms in terms of \DistP{} and service quality.
\end{itemize}
\arxiv{All proofs of technical results can be found in Appendix.}
\conference{All proofs of technical results can be found in~\cite{arxiv}.}
\conferenceShort{All proofs of technical results can be found in~\cite{arxiv}.}

\section{Preliminaries}
\label{sec:preliminaries}

In this section we recall some notions of privacy and metrics used in this paper.
Let $\natspos$ be the set of positive integers,
and $\realspos$ (resp. $\realsnng$) be the set of positive (resp. non-negative) real numbers.
Let $[0, 1]$ be the set of non-negative real numbers not grater than $1$.
Let $\varepsilon, \varepsilon_0, \varepsilon_1 \in \realsnng$ and $\delta, \delta_0, \delta_1 \in [0, 1]$.

\subsection{Notations for Probability Distributions}
\label{subsec:notations-pd}

\begin{wrapfigure}[11]{r}{0.23\linewidth}
\vspace{-7ex}
\centering
\begin{picture}(65,105)
 \linethickness{0.7pt}
 \put( -8, 87){$\lambda_0$}
 \put( -8, 52){$\gamma$}
 \put(  2,  10){\line( 1, 0){66}}
 \put(  2, 75){\line( 1, 0){66}}
 \linethickness{0.6pt}
 \put( 10, 75){\framebox(10,10){\scriptsize $0.2$}}
 \put( 13, 67){\scriptsize $1$}
 \put( 30, 75){\framebox(10,25){\scriptsize $0.5$}}
 \put( 33, 67){\scriptsize $2$}
 \put( 50, 75){\framebox(10,15){\scriptsize $0.3$}}
 \put( 53, 67){\scriptsize $3$}

 \put( -8, 20){$\lambda_1$}
 \put( 10, 10){\framebox(10,15){\scriptsize $0.3$}}
 \put( 13,  2){\scriptsize $1$}
 \put( 30, 10){\framebox(10,10){\scriptsize $0.2$}}
 \put( 33,  2){\scriptsize $2$}
 \put( 50, 10){\framebox(10,25){\scriptsize $0.5$}}
 \put( 53,  2){\scriptsize $3$}

 \linethickness{0.8pt}
 \put( 23, 58){\scriptsize $0.1$}
 \put( 37, 58){\scriptsize $0.2$}
 \put( 27, 56){\vector(-1,-2){8}}
 \put( 43, 56){\vector( 1,-2){8}}

 \put(  3, 51){\scriptsize $0.2$}
 \put( 25, 45){\scriptsize $0.2$}
 \put( 58, 51){\scriptsize $0.3$}
 \put( 15, 56){\vector( 0,-1){17}}
 \put( 35, 56){\vector( 0,-1){17}}
 \put( 55, 56){\vector( 0,-1){17}}
\end{picture}
\caption{Coupling $\gamma$.}
\label{fig:eg:coupling-def}
\end{wrapfigure}
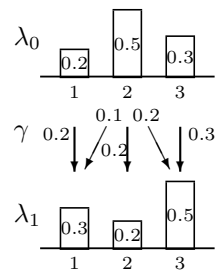

We denote by $\Dists\calx$ the \emph{set of all probability distributions} over a set~$\calx$, and by $|\calx|$ the number of elements in a finite set~$\calx$.

Given a finite set $\calx$ and a distribution $\lambda\in\Dists\calx$, the probability of drawing a value $x$ from $\lambda$ is denoted by $\lambda[x]$.
For a finite subset $\calx'\subseteq\calx$ we define $\lambda[\calx']$ by: $\lambda[\calx'] = \sum_{x'\in\calx'} \lambda[x']$.
For a distribution $\lambda$ over a finite set $\calx$, its \emph{support} $\supp(\lambda)$ is defined by $\supp(\lambda) = \{ x \in \calx \colon \lambda[x] > 0 \}$.
Given a $\lambda\in\Dists\calx$ and a $f:\calx\rightarrow\reals$, the expected value of $f$ over $\lambda$ is:
$\expect_{x\sim\lambda}[f(x)] \eqdef \sum_{x\in\calx} \lambda[x] f(x)$.

For a randomized algorithm $\alg:\calx\rightarrow\Dists\caly$ and a set $R\subseteq\caly$ we denote by $\alg(x)[R]$ the probability that given input $x$, $\alg$ outputs one of the elements of $R$.
Given a randomized algorithm $\alg:\calx\rightarrow\Dists\caly$ and a distribution $\lambda$ over $\calx$, we define $\lift{\alg}(\lambda)$ as the distribution of the output of $\alg$.
Formally, for a finite set $\calx$, the \emph{lifting} of $\alg$ w.r.t. $\calx$ is the function $\lift{\alg}: \Dists\calx\rightarrow\Dists\caly$ 
such that
$
\lift{\alg}(\lambda)[R] \eqdef
\sum_{x\in\calx}\lambda[x] \alg(x)[R]
$.

\fdiv{
\subsection{Divergence}
\label{subsec:divergence}
We first recall the notion of (approximate) max divergence, which is used to define differential privacy.
\begin{definition}[Max divergence]\label{def:max-divergence}\rm
For $\delta\in[0,1]$, the \emph{$\delta$-approximate max divergence} between $\mu, \mu'\in\Dists\caly$ is:
\[
\appmaxdiverge{\mu}{\mu'} = 
\max_{\substack{R\subseteq\supp(\mu) \\ \mu[R]\ge\delta}} \ln \frac{ \mu[R] - \delta }{ \mu'[R] }
{.}
\]
The \emph{max divergence} is defined by: 
$\maxdiverge{\mu\!}{\!\mu'} = \aprmaxdiverge{0}{\mu\!}{\!\mu'}$.
\end{definition}
}

\subsection{Differential Privacy (\DP{})}
\label{subsec:DP}

\emph{Differential privacy} \cite{Dwork:06:ICALP} captures the idea that given two ``adjacent'' inputs $x$ and $x'$ (from a set $\calx$ of data with an adjacency relation $\varPhi$), a randomized algorithm $\alg$ cannot distinguish $x$ from $x'$ (with degree of $\varepsilon$ and up to exceptions $\delta$).

\begin{definition}[Differential privacy] \label{def:DP} \rm
Let $e$ be the base of natural logarithm.
A randomized algorithm $\alg: \calx \rightarrow \Dists\caly$ provides \emph{$(\varepsilon,\delta)$-differential privacy (\DP{})} w.r.t. an adjacency relation $\varPhi\subseteq\calx\times\calx$ if for any $(x, x')\in\varPhi$ and any $R\subseteq\caly$,
\[
\Prob[ \alg(x)\in R ] \leq e^\varepsilon \,\Prob[ \alg(x')\in R ] + \delta
\]
where the probability is taken over the random choices in $\alg$.
\end{definition}

\journal{
Clearly the protection of differential privacy is stronger for smaller $\varepsilon$ and $\delta$.
}

\fdiv{
It is known that the above definition is equivalent to the following one using max-divergence.
\begin{restatable}{prop}{sDP}
\label{prop:sDP}
A randomized algorithm $\alg: \calx \rightarrow \Dists\caly$ provides $(\varepsilon,\delta)$-\DP{} w.r.t. $\varPhi\subseteq\calx\times\calx$ iff for any $(x, x')\in\varPhi$,\,
$\appmaxdiverge{\alg(x)}{\alg(x')} \le \varepsilon$ and
$\appmaxdiverge{\alg(x')}{\alg(x)} \le~\varepsilon$.
\end{restatable}
}
\fdiv{
Note that the sequential composition of differentially private mechanisms is differentially private:
For any $n$ independent randomized algorithms $\alg_1, \alg_2, \ldots, \alg_n$,
if each $\alg_i$ provides $(\varepsilon_i, \delta_i)$-\DP{}, 
the sequential composition $\alg_n \circ \alg_{n-1} \circ \ldots \circ \alg_1$ provides $(\sum_{i=1}^{n}\varepsilon_i, \sum_{i=1}^{n} \delta_i)$-\DP{}.
}

\subsection{Differential Privacy Mechanisms and Sensitivity}
\label{subsec:mechanisms}

Differential privacy can be achieved by a \emph{privacy mechanism}, namely a randomized algorithm that adds probabilistic noise to a given input that we want to protect.
The amount of noise added by some popular mechanisms (e.g., the exponential mechanism) depends on a \emph{utility function} $\ut:\calx\times\caly\rightarrow\reals$ that maps a pair of input and output to a utility score.
More precisely, the noise is added according to the ``sensitivity'' of $\ut$, which we define as follows.

\begin{definition}[Utility distance]\label{def:util-distance}\rm
The \emph{utility distance} w.r.t a utility function $\ut:(\calx\times\caly)\rightarrow\reals$ is the function $\utmetric$ given by:
$\utmetric(x,x') \eqdef \max_{y\in\caly} \bigl| \util{x}{y} - \util{x'}{y} \bigr|$.
\end{definition}
Note that $\utmetric$ is a pseudometric.
Hereafter we assume that for all $x,y$, $\ut(x,y)=0$ is logically equivalent to $x=y$.
Then the utility distance $\utmetric$ is a metric.

\begin{definition}[Sensitivity w.r.t. an adjacency relation]\label{def:sens-phi}\rm
The \emph{sensitivity} of a utility function $\ut$
w.r.t. an adjacency relation $\varPhi \subseteq \calx\times\calx$ is defined as:
\[
\sensPhi \eqdef
\max_{\substack{(x,x')\in\varPhi}} \utmetric(x,x') =
\max_{\substack{(x,x')\in\varPhi}} \max_{y\in\caly} \bigl| \util{x}{y} - \util{x'}{y} \bigr|
{.}
\]
\end{definition}
\journal{
For example, the exponential mechanism~\cite{McSherry:07:FOCS}
perturbs the output according to the sensitivity $\sensPhi$ and provides $\varepsilon$-\DP{} w.r.t. $\varPhi$.
}

\subsection{Extended Differential Privacy (\XDP{})}
\label{sub:XDP}

We review the notion of \emph{extended differential privacy}~\cite{Chatzikokolakis:13:PETS}, which relaxes \DP{} by incorporating a metric $d$.
Intuitively, this notion guarantees that when two inputs $x$ and $x'$ are closer in terms of $d$, the output distributions are less distinguishable.

\begin{definition}[Extended differential privacy]\label{def:XDP-simple}\rm
For a metric $d: \calx\times\calx\rightarrow\reals$, we say that a randomized algorithm $\alg: \calx \rightarrow \Dists\caly$ provides \emph{$(\varepsilon,\delta,d)$-extended differential privacy (\XDP{})} if for all $x, x'\in\calx$ and for any $R\subseteq\caly$,
\begin{align*}
\Prob[ \alg(x)\in R ] \leq e^{\varepsilon d(x,x')} \,\Prob[ \alg(x')\in R ] + \delta
{.}
\end{align*}
\journal{where the probability is taken over the random choices in $\alg$.}
\end{definition}

\subsection{Wasserstein Metric}
\label{sub:wasserstein-metric}

We recall the notion of probability coupling as follows.

\begin{definition}[Coupling]\label{def:coupling}\rm
Given $\lambda_0\in\Dists\calx_0$ and $\lambda_1\in\Dists\calx_1$, a \emph{coupling} of $\lambda_0$ and $\lambda_1$ is a $\gamma\in\Dists(\calx_0\times \calx_1)$ such that $\lambda_0$ and $\lambda_1$ are $\gamma$'s marginal distributions, i.e.,
for each $x_0\in \calx_0$,
$\lambda_0[x_0] =\!\sum_{x'_1\in \calx_1}\!\gamma[x_0, x'_1]$ and
for each $x_1\in \calx_1$,
$\lambda_1[x_1] =\!\sum_{x'_0\in \calx_0}\!\gamma[x'_0, x_1]$.
We denote by $\cp{\lambda_0}{\lambda_1}$ the set of all couplings of $\lambda_0$ and $\lambda_1$.
\end{definition}

\begin{example}[Coupling as transformation of distributions]\label{eg:coupling-def}
Let us consider two distributions $\lambda_0$ and $\lambda_1$ shown in Fig.~\ref{fig:eg:coupling-def}.
A coupling $\gamma$ of $\lambda_0$ and $\lambda_1$ shows a way of transforming $\lambda_0$ to $\lambda_1$.
For example,
$\gamma[2, 1] = 0.1$ moves from $\lambda_0[2]$ to $\lambda_1[1]$.
\end{example}

We then recall the $\infty$-Wasserstein metric~\cite{Vaserstein:69:PPI} 
between two distributions.

\begin{definition}[$\infty$-Wasserstein metric]\label{def:p-Wasserstein-metric}\rm
Let $\utmetric$ be a metric over $\calx$.
The \emph{$\infty$-Wasserstein metric} $\Winfu$ w.r.t. $\utmetric$ is defined by:
for any $\lambda_0, \lambda_1\in\Dists\calx$, 
\[
\Winfu(\lambda_0, \lambda_1) = 
\min_{\gamma\in \cp{\lambda_0}{\lambda_1}}\hspace{-3ex}
\max_{\hspace{3ex}(x_0, x_1)\in\supp(\gamma)}\hspace{-2ex}
\utmetric(x_0, x_1)
{.}
\]
\end{definition}

The $\infty$-Wasserstein metric $\Winfu(\lambda_0, \lambda_1)$ represents the minimum largest move between points in a transportation from $\lambda_0$ to $\lambda_1$.
Specifically, in a transportation $\gamma$,\, $\max_{(x_0, x_1)\in\supp(\gamma)}\allowbreak\utmetric(x_0, x_1)$ represents the largest move from a point in $\lambda_0$ to another in $\lambda_1$.
For instance, in the coupling $\gamma$ in Example~\ref{eg:coupling-def}, the largest move is $1$ (from $\lambda_0[2]$ to $\lambda_1[1]$, and from $\lambda_0[2]$ to $\lambda_1[3]$).
Such a largest move is minimized by a coupling that achieves the $\infty$-Wasserstein metric.
\fdiv{
\[
\Winfu(\lambda_0, \lambda_1) = 
\min_{\gamma\in \cp{\lambda_0}{\lambda_1}}\hspace{-3ex}
\max_{\hspace{3ex}(x_0, x_1)\in\supp(\gamma)}\hspace{-2ex}
\utmetric(x_0, x_1)
{.}
\]
}
We denote by $\GammaInf$ the set of all couplings that achieve the $\infty$-Wasserstein metric.

Finally, we recall the notion of the lifting of relations.

\begin{definition}[Lifting of relations]\label{def:lifting-relations}\rm
Given a relation $\varPhi\subseteq\calx\times\calx$, the \emph{lifting} of $\varPhi$ is the maximum relation $\lift{\varPhi}\subseteq \Dists\calx\times\Dists\calx$ such that for any $(\lambda_0, \lambda_1)\in\lift{\varPhi}$, there exists a coupling $\gamma\in\cp{\lambda_0}{\lambda_1}$ satisfying $\supp(\gamma)\subseteq\varPhi$.
\end{definition}

Note that by Definition~\ref{def:coupling}, the coupling $\gamma$ is a probability distribution over $\varPhi$ whose marginal distributions are $\lambda_0$ and~$\lambda_1$.
If $\varPhi = \calx\times\calx$, then $\lift{\varPhi} = \Dists\calx\times\Dists\calx$.

\journal{
\begin{example}[Lifted relation]\label{eg:lifted-relation}
Let us consider the distributions $\lambda_0$ and $\lambda_1$ in Fig.~\ref{fig:eg:coupling-def} in Example~\ref{eg:coupling-def}.
We define an adjacency relation
$\varPhi = \{ (1,1), (2,1), \allowbreak (2,2), \allowbreak (2,3), (3,3) \}$.
Then the coupling $\gamma$ satisfies $\supp(\gamma) = \varPhi$.
By Definition~\ref{def:lifting-relations}, we obtain $(\lambda_0, \lambda_1)\in\lift{\varPhi}$,
i.e., $\lambda_0$ is adjacent to $\lambda_1$ by the relation $\lift{\varPhi}$.
This means that we can construct $\lambda_1$ from $\lambda_0$ (based on $\gamma$) by moving some mass from $\lambda_0[x_0]$ to $\lambda_1[x_1]$ only for each $(x_0, x_1)\in\varPhi$ (i.e., only when $x_0$ is adjacent to $x_1$ by $\varPhi$).
\end{example}
}

\section{Privacy Notions for Probability Distributions}
\label{sec:dp-for-dist}

In this section we introduce a formal model for the privacy of user attributes,
which is motivated in Section~\ref{sec:intro}.

\subsection{Modeling the Privacy of User Attributes in Terms of \DP{}}
\label{sub:model:attribute-privacy}

As a running example, we consider an LBS (location based service) in which each user queries an LBS provider for a list of shops nearby.
To hide a user's exact location $x$ from the provider, the user applies a randomized algorithm $\alg: \calx\rightarrow\Dists\caly$, called a \emph{local obfuscation mechanism}, to her location $x$, and obtains an approximate information $y$ with the probability $\alg(x)[y]$.

To illustrate the privacy of attributes, let us consider an example in which users try to prevent an attacker from inferring whether they are $\male$ or $\female$ by obfuscating their own exact locations using a mechanism $\alg$.
For each $\attr\in\{ \male, \female \}$, let $\lambda_{\attr}\in\Dists\calx$ be the prior distribution of the location of the users who have the attribute $\attr$.
Intuitively, $\lambda_{\male}$ (resp. $\lambda_{\female}$) represents an attacker's belief on the location of the male (resp. female) users before the attacker observes an output of the mechanism $\alg$.
Then the privacy of $\attr$ can be modeled as a property that the attacker has no idea on whether the actual location $x$ follows the distribution $\lambda_{\male}$ or $\lambda_{\female}$ after observing an output $y$ of~$\alg$.

This can be formalized in terms of $\varepsilon$-local \DP{}.
For each $\attr\in\{ \male, \allowbreak \female \}$,\,
we denote by $p( y \,|\, \lambda_{\attr})$ the probability of observing an obfuscated location $y$ when an actual location $x$ is distributed over $\lambda_{\attr}$,
i.e.,
$p( y \,|\, \lambda_{\attr}) = 
 \sum_{x\in\calx} \lambda_{\attr}[x] \alg(x)[y]$.
Then we can define the privacy of $\attr$ by:
\begin{align*}
\textstyle
\frac{ p( y \,|\, \lambda_{\male}) }{ p( y \,|\, \lambda_{\female}) }
\le e^{\varepsilon}
{.}
\end{align*}
\journal{
which represents that the attacker cannot distinguish whether the users follow $\lambda_{\male}$ or $\lambda_{\female}$ (with degree of $\varepsilon$).
}

\subsection{Distribution Privacy and Extended Distribution Privacy}
\label{subsec:intro:DistP}

We generalize the privacy of attributes (in Section~\ref{sub:model:attribute-privacy}) and define the notion of \emph{distribution privacy} (\DistP{}) as the differential privacy where the input is a probability distribution of data rather than a value of data.
This notion models a level of obfuscation that hides which distribution a data value is drawn from.
Intuitively, we say a randomized algorithm $\alg$ provides \DistP{} if, by observing an output of $\alg$, we cannot detect from which distribution an input to $\alg$ is generated.

\begin{definition}[Distribution privacy]\label{def:max-DistP}\rm
Let $\varepsilon\in\realsnng$ and $\delta\in[0,1]$.
We say that a randomized algorithm $\alg:\calx\rightarrow\Dists\caly$ provides \emph{$(\varepsilon,\delta)$-distribution privacy (\DistP{}) w.r.t.} 
an adjacency relation $\varPsi\subseteq\Dists\calx\times\Dists\calx$ 
if its lifting $\lift{\alg}:\Dists\calx\rightarrow\Dists\caly$ provides $(\varepsilon,\delta)$-\DP{} w.r.t. $\varPsi$, 
i.e., for all pairs $(\lambda, \lambda')\in\varPsi$ and all $R \subseteq\caly$,\, we have:
\begin{align*}
\lift{\alg}(\lambda)[R] \leq
e^{\varepsilon}\cdot\lift{\alg}(\lambda')[R] + \delta
{.}
\end{align*}
We say $\alg$ provides \emph{$(\varepsilon,\delta)$-\DistP{} w.r.t.}\,$\Lambda\subseteq\Dists\calx$ if it provides $(\varepsilon,\delta)$-\DistP{} w.r.t.~$\Lambda^2$.
\end{definition}

For example, the privacy of a user attribute $\attr\in\{ \male{}, \allowbreak \female{} \}$ described in Section~\ref{sub:model:attribute-privacy} can be formalized as $(\varepsilon, 0)$-\DistP{} w.r.t. $\{\lambda_{\male}, \lambda_{\female}\}$.

Mathematically, \DistP{} is not a new notion but the \DP{} for distributions.
To contrast with \DistP{}, we refer to the \DP{} for data values as \emph{point privacy}. %(\PointP{}).

Next we introduce an extended form of distribution privacy to a metric.
Intuitively, extended distribution privacy guarantees that when two input distributions are closer,
then the output distributions must be less distinguishable.

\begin{definition}[Extended distribution privacy]\label{def:max-XDP-dist}\rm
Let 
$d: (\Dists\calx\times\Dists\calx)\rightarrow\reals$ be a utility distance, and $\varPsi\subseteq\Dists\calx\times\Dists\calx$.
We say that a mechanism $\alg:\calx\rightarrow\Dists\caly$ provides \emph{$(\varepsilon,d,\delta)$-extended distribution privacy (\XDistP{}) w.r.t.} $\varPsi$ if the lifting $\lift{\alg}$ provides $(\varepsilon,d,\delta)$-\XDP{} w.r.t.~$\varPsi$,
i.e., for all $(\lambda, \lambda')\in\varPsi$ and 
all $R\subseteq\caly$,
we have:
\begin{align*}
\lift{\alg}(\lambda)[R] \leq
e^{\varepsilon d(\lambda,\lambda')}\cdot
\lift{\alg}(\lambda')[R]+ \delta
{.}
\end{align*}
\end{definition}

\subsection{Interpretation by Bayes Factor}
\label{subsec:semantics}

The interpretation of \DP{} has been explored in 
previous work~\cite{Dwork:06:TCC,Chatzikokolakis:13:PETS} 
using the notion of \emph{Bayes factor}.
Similarly, the meaning of \DistP{} can also be explained in terms of Bayes factor, which compares the attacker's prior and posterior beliefs. 

Assume that an attacker has some belief on the input distribution before observing the output values of an obfuscater~$\alg$.
We denote by $p(\lambda)$ the prior probability that a distribution $\lambda$ is chosen 
as the input distribution.
By observing an output $y$ of $\alg$, the attacker updates his belief on the input distribution.
We denote by $p(\lambda | y)$ the posterior probability of $\lambda$ being chosen, given an output $y$.

For two distributions $\lambda_0, \lambda_1$, the Bayes factor $K(\lambda_0, \allowbreak \lambda_1, y)$ is defined as the ratio of the two posteriors  divided by that of the two priors:
$K(\lambda_0, \lambda_1, y) =
 \frac{p(\lambda_0|y)}{p(\lambda_1|y)} \big/ 
 \frac{p(\lambda_0)}{p(\lambda_1)}$.
If the Bayes factor is far from~$1$ the attacker significantly updates his belief on the distribution by observing a perturbed output $y$ of $\alg$.

Assume that $\alg$ provides $(\varepsilon,0)$-\DistP{}.
By Bayes' theorem, we obtain:
\begin{align*}
K(\lambda_0, \lambda_1, y) =
\textstyle
\frac{p(\lambda_0|y)}{p(\lambda_1|y)}\cdot
\frac{p(\lambda_1)}{p(\lambda_0)}
=
\frac{p(y|\lambda_0)}{p(y|\lambda_1)}
=
\frac{\lift{\alg}(\lambda_0)[y]}{\lift{\alg}(\lambda_1)[y]}
\le e^\varepsilon
{.}
\end{align*}
Intuitively, if the attacker believes that $\lambda_0$ is $k$ times more likely than $\lambda_1$ before the observation, then he believes that $\lambda_0$ is $k\cdot e^\varepsilon$ times more likely than $\lambda_1$ after the observation.
This means that for a small value of $\varepsilon$, \DistP{} guarantees that the attacker does not gain information on the distribution by observing $y$.

In the case of \XDistP{}, the Bayes factor $K(\lambda_0, \lambda_1, y)$ is bounded above by $e^{\varepsilon d(\lambda_0, \lambda_1)}$.
Hence the attacker gains more information for a larger distance $d(\lambda_0, \lambda_1)$.

\subsection{Privacy Guarantee for Attackers with Close Beliefs}
\label{subsec:close-belief}

In the previous sections, we assume that we know the distance between two actual input distributions, and can determine the amount of noise required for distribution obfuscation.
However, an attacker may have different beliefs on the distributions that are closer to the actual ones, 
e.g., more accurate distributions obtained by more observations and specific situations (e.g., daytime/nighttime).

To see this, for each $\lambda\in\Dists\calx$, let $\tilde{\lambda}$ be an attacker's belief on $\lambda$.
We say that an attacker has \emph{$(c, \utmetric)$-close beliefs} if each distribution $\lambda$ satisfies $\utmetric(\lambda, \tilde{\lambda}) \le c$.
Then extended distribution privacy in the presence of an attacker is given by:
\begin{restatable}[\XDistP{} with close beliefs]{prop}{beliefsXDistP}
\label{prop:beliefsXDistP}
Let $\alg: \calx \rightarrow \Dists\caly$ provide $(\varepsilon, \utmetric, 0)$-\XDistP{} w.r.t. some $\varPsi \subseteq \calx\times\calx$.
If an attacker has $(c, \utmetric)$-close beliefs, then for all $(\lambda_0, \lambda_1)\in\varPsi$ and all $R \subseteq\caly$, we have
$
\lift{\alg}(\tilde{\lambda_0})[R]
\le e^{\varepsilon \left( \utmetric(\lambda_0, \lambda_1) + 2c \right)}
\cdot\lift{\alg}(\tilde{\lambda_1})[R]
{.}
$
\end{restatable}

When the attacker's beliefs are closer to ours, then $c$ is smaller, hence a stronger distribution privacy is guaranteed.
\arxiv{See Appendix~\ref{sub:proofs:DistP:beliefs} for \DistP{}.}
\conference{See~\cite{arxiv} for a proposition with \DistP{}.}
Note that assuming some attacker's beliefs are inevitable also in many previous studies, 
e.g., when we want to protect the privacy of correlated data \cite{Kifer:11:SIGMOD,Kifer:12:PODS,Song:17:SIGMOD}.

\subsection{Difference from the Histogram Privacy}
\label{sub:histogram}

Finally, we present a brief remark on the difference between \DistP{} and the \emph{differential privacy of histogram publication} (e.g.,~\cite{Xu:13:VLDB}).
Roughly, a histogram publication mechanism is a \emph{central} mechanism that aims at \emph{hiding a single record} $x\in\calx$ and outputs an obfuscated histogram, e.g., a distribution $\mu\in\Dists\caly$, 
whereas a \DistP{} mechanism is a \emph{local} mechanism that aims at \emph{hiding an input distribution} $\lambda\in\Dists\calx$ and outputs a single perturbed value $y\in\caly$.

Note that neither of these implies the other.
The $\varepsilon$-\DP{} of a histogram publication mechanism means that 
for any two adjacent inputs $x, x' \in\calx$ and any histogram $\mu \in \Dists\caly$,
$
\frac{p(\mu | x)}{p(\mu | x')} \le e^\varepsilon.
$
However, this does not derive $\varepsilon$-\DistP{}, i.e., 
for any adjacent input distributions $\lambda, \lambda' \in\Dists\calx$ and any output $y \in \caly$,
$
\frac{p(y | \lambda)}{p(y | \lambda')} \le e^\varepsilon
$.

\begin{table*}[t]
  \centering
  %\vspace{-1.5ex}
  \caption{Summary of basic properties of \DistP{}.}
  \label{table:main:basic-properties}
  \scalebox{1}[1]{
  \footnotesize
  \renewcommand{\arraystretch}{1.05}
  \begin{tabular}{lll}
    \hline
    Sequential composition $\Seq$ ~~~~~&
      $\alg_b$ is $(\varepsilon_b,\delta_b)$-\DistP{} \\
    &
      $\Rightarrow$
      $\alg_1 \Seq \alg_0$ is $(\varepsilon_0+\varepsilon_1,\, (\delta_0+\delta_1)\cdot|\varPhi|)$-\DistP{} \\
    \hline
    Sequential composition $\liftSeq$ &
      $\alg_b$ is $(\varepsilon_b,\delta_b)$-\DistP{} \\
    &
      $\Rightarrow$
      $\alg_1 \liftSeq \alg_0$ is $(\varepsilon_0+\varepsilon_1, \delta_0+\delta_1)$-\DistP{} \\
    \hline
    Post-processing &
      $\alg_0$ is $(\varepsilon,\delta)$-\DistP{}
      $\Rightarrow$
      $\alg_1\circ\alg_0$ is $(\varepsilon,\delta)$-\DistP{} \\
    \hline
    Pre-processing (by $c$-stable $T$) &
      $\alg$ is $(\varepsilon,\delta)$-\DistP{}
      $\Rightarrow$
      $\alg\circ T$ is $(c\,\varepsilon,\delta)$-\DistP{} \\
    \hline
  \end{tabular}
  }
\end{table*}

\section{Basic Properties of Distribution Privacy}
\label{sub:properties-DistP}

In Table~\ref{table:main:basic-properties}, we show basic properties of \DistP{}.%
\arxiv{(See Appendices~\ref{subsec:compositionality} and~\ref{subsec:post-pre-processing} for the details.)}
\conference{(See the arXiv version~\cite{arxiv} for the full table with \XDistP{} and their detailed proofs.)}
\conferenceShort{(See the arXiv version~\cite{arxiv} for the full table with \XDistP{} and their detailed proofs.)}

The composition $A_1 \Seq A_0$ means that an identical input $x$ is given to two \DistP{} mechanisms $A_0$ and $A_1$,
whereas the composition $A_1 \liftSeq A_0$ means that independent inputs $x_b$ are provided to mechanisms $\alg_b$~\cite{Kawamoto:17:LMCS}.
The compositionality can be used to quantify the attribute privacy against an attacker who obtains multiple released data each obfuscated for the purpose of protecting a different attribute.
For example, let $\varPsi = \{ (\lambda_{\male},\, \lambda_{\female}),\, (\lambda_{\home},\, \lambda_{\outside}) \}$,
and $\alg_0$ (resp. $\alg_1$) be a mechanism providing $\varepsilon_0$-\DistP{} (resp. $\varepsilon_1$-\DistP{}) w.r.t. $\varPsi$.
When $\alg_0$ (resp. $\alg_1$) obfuscates a location $x_0$ for the sake of protecting \male{}/\female{} (resp. \home{}/\outside{}), then
both \male{}/\female{} and \home{}/\outside{} are protected with $(\varepsilon_0+\varepsilon_1)$-\DistP{}.

As for pre-processing, the stability notion is different from that for \DP{}:
\begin{definition}[Stability]\label{def:c-stable}\rm
Let $c\in\natspos$, 
$\varPsi\subseteq\Dists\calx\times\Dists\calx$,
and $\sensfunc$ be a metric over $\Dists\calx$.
A transformation $T:\Dists\calx\rightarrow\Dists\calx$ is \emph{$(c, \varPsi)$-stable} if 
for any $(\lambda_0,\lambda_1)\in\varPsi$,
$T(\lambda_0)$ can be reached from $T(\lambda_1)$ at most $c$-steps over $\varPsi$.\,
Analogously, $T:\Dists\calx\rightarrow\Dists\calx$ is \emph{$(c,\sensfunc)$-stable} if for any $\lambda_0,\lambda_1\in\Dists\calx$, 
$\sensfunc(T(\lambda_0),T(\lambda_1)) \le c \sensfunc(\lambda_0,\lambda_1)$.
\end{definition}

We present relationships among privacy notions 
\arxiv{in Appendices~\ref{sub:proofs:DistP-implies-DP} and~\ref{sub:proofs-PDistP-DistP}
}%
\conference{in~\cite{arxiv}.}
An important property is that when the relation $\varPsi\subseteq\Dists\calx\times\Dists\calx$ includes pairs of 
point distributions (i.e., distributions having single points with probability $1$), 
$\DistP{}$ (resp. $\XDistP{}$) implies $\DP{}$ (resp. $\XDP{}$).
In contrast, if $\varPsi$ does not include pairs of 
point distributions, 
\DistP{} (resp. $\XDistP{}$) may not imply \DP{} (resp. $\XDP{}$),
as in Section~\ref{sec:DistP-by-dummies}.

\section{Distribution Obfuscation by Point Obfuscation}
\label{sec:standard-DistP}

In this section we present how the point obfuscation mechanisms (including \DP{} and \XDP{} mechanisms) contribute to the obfuscation of probability distributions.
\arxiv{(See Appendix~\ref{sub:DistP-by-PointP} for the proofs.)}

\subsection{Distribution Obfuscation by \DP{} Mechanisms}
\label{subsec:DistP-by-DP}

We first show every \DP{} mechanism provides \DistP{}.
(See Definition~\ref{def:lifting-relations} for $\lift{\varPhi}$.)

\begin{restatable}[$(\varepsilon, \delta)$-\DP{} $\Rightarrow$ $(\varepsilon,\, \delta\cdot|\varPhi|)$-\DistP]{thm}{DPimpliesMaxDistP}
\label{thm:DPimpliesMaxDistP}
Let 
$\varPhi \subseteq \calx\times\calx$.
If $\alg:\calx\rightarrow\Dists\caly$ provides $(\varepsilon, \delta)$-\DP{} w.r.t. $\varPhi$, then it provides $(\varepsilon,\, \delta\cdot|\varPhi|)$-\DistP{} w.r.t. $\lift{\varPhi}$.
\end{restatable}

This means that the mechanism $\alg$ makes any pair $(\lambda_0, \lambda_1)\in\lift{\varPhi}$ indistinguishable up to the threshold $\varepsilon$ and with exceptions $\delta\cdot|\varPhi|$.
Intuitively, when $\lambda_0$ and $\lambda_1$ are adjacent w.r.t. the relation $\lift{\varPhi}$, we can construct $\lambda_1$ from $\lambda_0$ only by moving mass from $\lambda_0[x_0]$ to $\lambda_1[x_1]$ where $(x_0, x_1)\in\varPhi$ (i.e., $x_0$ is adjacent to $x_1$).

\begin{example}[Randomized response]
By Theorem~\ref{thm:DPimpliesMaxDistP}, the $(\varepsilon, 0)$-\DP{} 
randomized response~\cite{Kairouz:16:ICML} 
and RAPPOR~\cite{Erlingsson_CCS14} provide $(\varepsilon, 0)$-\DistP{}.
When we use these mechanisms, the estimation of the input distribution is harder for a smaller $\varepsilon$.
However, these \DP{} mechanisms tend to have small utility, because they add much noise to hide not only the input distributions, but everything about inputs.
\end{example}

\subsection{Distribution Obfuscation by \XDP{} Mechanisms}
\label{subsec:DistP-by-XDP}

Compared to \DP{} mechanisms, \XDP{} mechanisms are known to provide better utility.
Alvim et al.~\cite{Alvim:18:CSF} show the planar Laplace mechanism~\cite{Andres:13:CCS} adds less noise than the randomized response, since \XDP{} hides only closer locations.
However, we show \XDP{} mechanisms still need to add much noise proportionally to the $\infty$-Wasserstein distance between the distributions we want make indistinguishable.

\subsubsection{The $\infty$-Wasserstein distance $\sensinf$ as Utility Distance}
We first observe how much $\varepsilon'$ is sufficient for an $\varepsilon'$-\XDP{} mechanism (e.g., the Laplace mechanism) to make two distribution $\lambda_0$ and $\lambda_1$ indistinguishable in terms of $\varepsilon$-\DistP{}.

Suppose that $\lambda_0$ and $\lambda_1$ are point distributions such that $\lambda_0[x_0] = \lambda_1[x_1] = 1$ for some $x_0,x_1\in\calx$.
Then an $\varepsilon'$-\XDP{} mechanism $\alg$ satisfies:
\begin{align*}
\maxdiverge{\lift{\alg}(\lambda_0)\!}{\!\lift{\alg}(\lambda_1)}
&=
\maxdiverge{\alg(x_0)\!}{\!\alg(x_1)}
\le \varepsilon' \utmetric(x_0,x_1)
{.}
\end{align*}
In order for $\alg$ to provide $\varepsilon$-\DistP{},\, $\varepsilon'$ should be defined as $\frac{\varepsilon}{\utmetric(x_0,x_1)}$.
That is, the noise added by $\alg$ should be proportional to the distance between $x_0$ and $x_1$.

To extend this to arbitrary distributions, we need to define a utility metric between distributions.
A natural possible definition would be the largest distance between values of $\lambda_0$ and $\lambda_1$, i.e., the \emph{diameter} over the supports defined by:
\[
\diam(\lambda_0, \lambda_1) = 
\max_{x_0\in\supp(\lambda_0), x_1\in\supp(\lambda_1)} \utmetric(x_0, x_1)
{.}
\]
However, when there is an outlier in $\lambda_0$ or $\lambda_1$ that is far from other values in the supports,
then the diameter $\diam(\lambda_0, \lambda_1)$ is large.
Hence the mechanisms that add noise proportionally to the diameter would lose utility too much.

To have better utility, we employ the $\infty$-Wasserstein metric $\sensinf$.
The idea is that given two distributions $\lambda_0$ and $\lambda_1$ over $\calx$, we consider the cost of a transportation of weights from $\lambda_0$ to $\lambda_1$.
The transportation is formalized as a coupling $\gamma$ of $\lambda_0$ and $\lambda_1$ (see Definition~\ref{def:coupling}), and the cost of the largest move is 
$
\displaystyle\sens_{\supp(\gamma),\utmetric} = \hspace{-1ex}
 \max_{(x_0, x_1)\in\supp(\gamma)}\hspace{-1ex} \utmetric(x_0,x_1),
$
i.e., the sensitivity w.r.t. the adjacency relation $\supp(\gamma)\subseteq\calx\times\calx$ (Definition~\ref{def:sens-phi}).
The minimum cost of the largest move is given by the $\infty$-Wasserstein metric: 
$
\sensinf(\lambda_0, \lambda_1)=
\displaystyle\min_{\gamma\in\cp{\lambda_0}{\lambda_1}} \sens_{\supp(\gamma),\utmetric}
{.}
$

\subsubsection{\XDP{} implies \XDistP{}}
We show every \XDP{} mechanism provides \XDistP{} with the metric $\sensinf$.
To formalize this, we define a lifted relation $\liftinf{\varPhi}$ as the maximum relation over $\Dists\calx$ s.t. for any $(\lambda_0, \lambda_1)\in\liftinf{\varPhi}$, there is a coupling $\gamma\in\cp{\lambda_0}{\lambda_1}$ satisfying $\supp(\gamma)\subseteq\varPhi$ and $\gamma\in\GammaInf(\lambda_0, \lambda_1)$.
Then
$\liftinf{\varPhi}\subseteq\lift{\varPhi}$ holds.

\begin{restatable}[$(\varepsilon, \utmetric, \delta)$-\XDP{}$\Rightarrow$$(\varepsilon, \sensinf, \delta\!\cdot\!|\varPhi|)$-\XDistP]{thm}{maxDistPFromDP}
\label{thm:max-DistPFromDP}
If $\alg\!:\!\calx\rightarrow\Dists\caly$ provides \allowbreak $(\varepsilon, \utmetric, \delta)$-\XDP{} w.r.t. $\varPhi \subseteq \calx\times\calx$, 
it provides $(\varepsilon, \sensinf, \allowbreak {\delta\!\cdot\!|\varPhi|})$-\XDistP{} 
w.r.t.~$\liftinf{\varPhi}$.
\end{restatable}

By Theorem~\ref{thm:max-DistPFromDP}, 
when $\delta > 0$, the noise required for obfuscation is proportional to $|\varPhi|$, which is at most the domain size squared $|\calx|^2$.
This implies that for a larger domain $\calx$, the Gaussian mechanism is not suited for distribution obfuscation.
We will demonstrate this by experiments in Section~\ref{sub:exp-compare}.

In contrast, the Laplace/exponential mechanisms provide $(\varepsilon, \sensinf, 0)$-\DistP{}.
Since $\sensinf(\lambda_0, \lambda_1) \le \diam(\lambda_0, \lambda_1)$, the noise added proportionally to $\sensinf$ can be smaller than $\diam$.
This implies that obfuscating a distribution requires less noise than obfuscating a set of data.
However, the required noise can still be very large when we want to make two distant distributions indistinguishable.

\section{Distribution Obfuscation by Random Dummies}
\label{sec:DistP-by-dummies}

In this section we introduce a local mechanism called a \emph{tupling mechanism}
to improve the tradeoff between \DistP{} and utility, as motivated in Section~\ref{sec:intro}.
\journal{
and show that this mechanism provides \DistP{} (hence the privacy of attributes) and better utility for users by restricting the \DP{} protection of point data.
Although many studies have been using dummies to hide users' locations~ \cite{Bindschaedler:16:SP,Chow:09:WPES,Kido:05:ICDE}, they do not aim at protecting  users' attributes or activities in a \DP{} manner.
}

\subsection{Tupling Mechanism}
\label{sub:tupling-mechanism}

We first define the \emph{tupling mechanism} as a local mechanism that obfuscates a given input $x$ by using a point perturbation mechanism $\alg$ (not necessarily in terms of \DP{} or \XDP{}), and that also adds $k$ random dummies $r_1, r_2, \ldots, r_k$ to the output to obfuscate the input distribution
(Algorithm~\ref{alg:tupling}).
\begin{algorithm}[t]
\caption{Tupling mechanism $\TPM$}\label{alg:tupling}
\begin{footnotesize}
\begin{algorithmic}
\REQUIRE{$x$: input, $k$: \#dummies, $\nu$: distribution of dummies, $\alg$: randomized algorithm}
\ENSURE{$y = (r_1, \ldots, r_{i}, s, r_{i+1}, \ldots, r_k)$: the output value of the tupling mechanism}
\STATE $s \randassign \alg(x)$\,;~ \COMMENT{Draw an obfuscated value $s$ of an input $x$} \\
\STATE $r_1, r_2, \ldots, r_k \randassign \nu$\,;~ \COMMENT{Draw $k$ dummies from a given distribution $\nu$} \\
\STATE $i \randassign \{1,2,\ldots,k+1\}$\,;\,~\COMMENT{Draw $i$ to decide the order of the outputs} \\
\RETURN{$(r_1, \ldots, r_{i}, s, r_{i+1}, \ldots, r_k)$\,;}
\end{algorithmic}
\end{footnotesize}
\end{algorithm}
The probability that given an input $x$, the mechanism $\TPM$ outputs $\bar{y}$ is given by $\TPM(x)[\bar{y}]$.

\subsection{Privacy of the Tupling Mechanism}
\label{sub:tupling-privacy}

Next we show that the tupling mechanism provides \DistP{} w.r.t. the following class of distributions.
Given $\beta, \eta\in[0, 1]$ and $\alg: \calx\rightarrow\Dists\caly$, we define $\LambdaBeta$ by:
\[
\LambdaBeta =
\bigl\{
\lambda\in\Dists\calx \mid
\Pr\bigl[ y \randassign \caly: \lift{\alg}(\lambda)[y] \le \beta \bigr]
\ge 1 - \eta
\bigr\}
{.}
\]
For instance, a distribution $\lambda$ satisfying $\max_{x} \lambda[x] \le \beta$ belongs to $\Lambda_{\beta,0,\alg}$.
\arxiv{(See Proposition~\ref{prop:dist-prob-bounds} in Appendix~\ref{sub:tupling:properties}.)}
\conference{}

\begin{restatable}[\DistP{} of the tupling mechanism]{thm}{TuplingDP}
\label{thm:TuplingDP}
Let $k\in\natspos$, $\nu$ be the uniform distribution over $\caly$,
$\alg:\calx\rightarrow\Dists\caly$,
and $\beta, \eta\in[0, 1]$.
Given an $0 < \alpha < \frac{k}{|\caly|}$, let 
$\varepsilon_{\alpha} = \ln{\textstyle\frac{ k + (\alpha + \beta)\cdot|\caly| }{ k - \alpha\cdot|\caly| }}$
and $\delta_{\alpha} = 2 e^{-\frac{2\alpha^2}{k\beta^2}} + \eta$.
Then 
the $(k,\nu,\alg)$-tupling mechanism provides $(\varepsilon_{\!\alpha}, \delta_{\!\alpha})$-\!\DistP{}
 w.r.t. $\LambdaBeta^2$.
\end{restatable}

This claim states that just adding random dummies achieves \DistP{} without any assumption on $\alg$ (e.g., $\alg$ does not have to provide \DP{}).
For a smaller range size $|\caly|$ and a larger number $k$ of dummies, we obtain a stronger \DistP{}.

Note that the distributions protected by $\TPM$ belong to the set $\LambdaBeta$.
\begin{itemize}
\item When $\beta = 1$, $\LambdaBeta$ is the set of all distributions (i.e., $\Lambda_{1,\eta,\alg} = \Dists\calx$) while $\varepsilon_{\alpha}$ and $\delta_{\alpha}$ tend to be large.
\item For a smaller $\beta$, the set $\LambdaBeta$ is smaller while $\varepsilon_{\alpha}$ and $\delta_{\alpha}$ are smaller;
that is, the mechanism provides a stronger \DistP{} for a smaller set of distributions.
\item If $\alg$ provides $\varepsilon_{\alg}$-\DP{},\, $\LambdaBeta$ goes to $\Dists\calx$ for $\varepsilon_{\alg} \rightarrow 0$.
More generally, $\LambdaBeta$ is larger when 
the maximum output probability $\max_{y} \lift{\alg}(\lambda)[y]$ is smaller.
\end{itemize}
In practice, even when $\varepsilon_{\alg}$ is relatively large, a small number of dummies enables us to provide a strong \DistP{}, 
as shown by experiments in Section~\ref{sec:location:dist-obf}.

We note that Theorem~\ref{thm:TuplingDP} may not imply \DP{} of the tupling mechanism, depending on $\alg$.
For example, suppose that $\alg$ is the identity function.
For small $\varepsilon_{\alpha}$ and $\delta_{\alpha}$, we have $\beta \ll 1$, 
hence no point distribution $\lambda$ (where $\lambda[x] = 1$ for some $x$) belongs to $\LambdaBeta$, namely, the tupling mechanism does not provide $(\varepsilon_{\alpha}, \delta_{\alpha})$-\DP{}.

\subsection{Service Quality Loss and Cost of the Tupling Mechanism}
\label{sub:tupling-utility}

When a mechanism outputs a value $y$ closer to the original input $x$, she obtains a larger \emph{utility}, or equivalently, a smaller \emph{service quality loss} $\utmetric(x, y)$.
For example, in an LBS (location based service), if a user located at $x$ submits an obfuscated location $y$, the LBS provider returns the shops near $y$, hence the service quality loss can be expressed as the Euclidean distance $\utmetric(x, y) \eqdef \| x - y \|$.

Since each output of the tupling mechanism consists of $k+1$ elements, the quality loss of submitting a tuple $\bar{y} = (y_1, y_2, \ldots , y_{k+1})$ amounts to 
$\utmetric(x, \bar{y}) \mathbin{:=} \min_{i} \allowbreak \utmetric(x, y_i)$.
Then the expected quality loss of the mechanism is defined as follows.

\begin{definition}[Expected quality loss of the tupling mechanism]\label{def:tupling-utility}\rm
For a $\lambda\in\Dists\calx$ and a metric $\utmetric: \calx\times\caly\rightarrow\reals$,
the \emph{expected quality loss} of $\TPM$ is:
\begin{align*}
L\bigl(\TPM\bigr) =\textstyle
\!\sum_{x\in\calx}
\sum_{\bar{y}\in\caly^{k+1}} 
\lambda[x]\,
\TPM(x)[\bar{y}]\, \min_{i} \utmetric(x, y_i)
{.}
\end{align*}
\end{definition}

For a larger number $k$ of random dummies, $\min_{i} \allowbreak \utmetric(x, y_i)$ is smaller on average, hence $L\bigl(\TPM\bigr)$ is also smaller.
Furthermore, thanks to the distribution obfuscation by random dummies, we can instead reduce the perturbation noise added to the actual input $x$ to obtain the same level of \DistP{}.
Therefore, the service quality is much higher than existing mechanisms, as shown in Section~\ref{sec:location:dist-obf}.

\journal{
On the other hand, the traffic cost of applying the tupling mechanism is larger for a lager number $k$ of dummy data.
We show by experiments the trade-offs among the privacy, utility, and cost of the tupling mechanism in Sections~\ref{sec:location:dist-obf}.
}

\subsection{Improving the Worst-Case Quality Loss}
\label{sub:RL-mechanism}

As a point obfuscation mechanism $\alg$ used in the tupling mechanism $\TPM{}$, we define the \emph{restricted Laplace (RL) mechanism} below.
Intuitively, $(\varepsilon_{\!\alg}, r)$-\RL{} mechanism adds $\varepsilon_{\!\alg}$-\XDP{} Laplace noise only within a radius $r$ of the original location $x$.
This ensures that the worst-case quality loss of the tupling mechanisms is bounded above by the radius $r$, whereas the standard Laplace mechanism reports a location $y$ that is arbitrarily distant from $x$ with a small probability.

\begin{definition}[RL mechanism]\label{def:restricted-Laplace}\rm
Let $\caly_{x,r} = \{ y'\in\caly \,|\, \utmetric(x, y') \le r \}$.
We define \emph{$(\varepsilon_{\!\alg},\allowbreak r)$-restricted Laplace (RL)} mechanism as 
the $\alg:\calx\rightarrow\Dists\caly$ defined by:
$\alg(x)[y] = \frac{ e^{-\varepsilon \utmetric(x, y)} }{ \sum_{y'\in\caly_{x,r}} e^{-\varepsilon \utmetric(x, y')} }$ 
if $y \in \caly_{x,r}$, and 
$\alg(x)[y] = 0$ otherwise.
\end{definition}

Since the support of $\alg$ is limited to $\caly_{x,r}$,
$\alg$ provides better service quality but does not provide \DP{}.
Nevertheless, as shown in Theorem~\ref{thm:TuplingDP},\, $\TPM{}$ provides \DistP{}, due to dummies in $\caly\setminus\caly_{x,r}$.
This implies that \DistP{} is a relaxation of \DP{} that guarantees the privacy of attributes while achieving higher utility by weakening the \DP{} protection of point data.
In other words, \DistP{} mechanisms are useful when users want both to keep high utility and to protect the attribute privacy more strongly than what a \DP{} mechanism can guarantee
(e.g., when users do not mind revealing their actual locations outside home, but want to hide from robbers the fact that they are outside home, as motivated in Section~\ref{sec:intro}).

\section{Application to Attribute Privacy in LBSs}
\label{sec:location:dist-obf}

In this section we apply local mechanisms to the protection of the attribute privacy in location based services (LBSs)
where each user submits her own location $x$ to an LBS provider to obtain information relevant to $x$ (e.g., shops near $x$).
\journal{
The disclosure of user locations may allow an attacker to infer about the users' attributes and activities~\cite{Liao:07:IJRR,Zheng:09:LBSN,Matsuo:07:IJCAI,Yang:19:TKDE}.
To mitigate such information leaks, each user can locally apply a mechanism $Q$ to her location $x$ to report an obfuscated output $y$ to the LBS provider.
Then we show that the tupling mechanism provides much stronger attribute privacy and better service quality than existing popular mechanisms.
}

\begin{figure*}[t]
\centering
\begin{subfigure}[t]{0.24\textwidth}
  \mbox{\raisebox{-5pt}{\includegraphics[ width=1.05\textwidth]{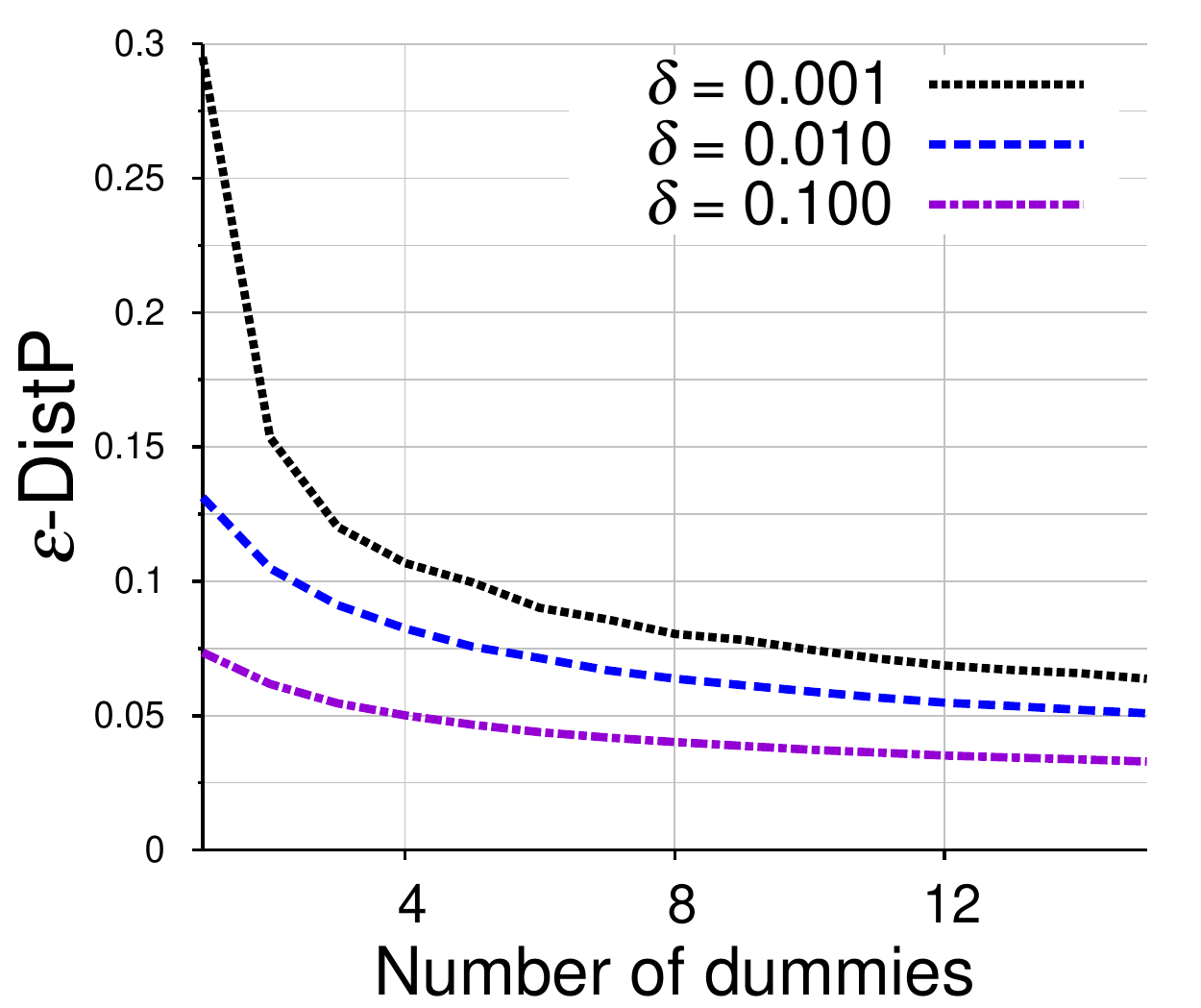}}}
\caption{\#dummies and $\varepsilon$-\DistP{} (when using $(100,\allowbreak 0.020)$-\RL{} mechanism).\label{fig:dummies:DistP}}
\end{subfigure}\hspace{0.4ex}\hfill
\begin{subfigure}[t]{0.24\textwidth}
  \mbox{\raisebox{-5pt}{\includegraphics[ width=1.05\textwidth]{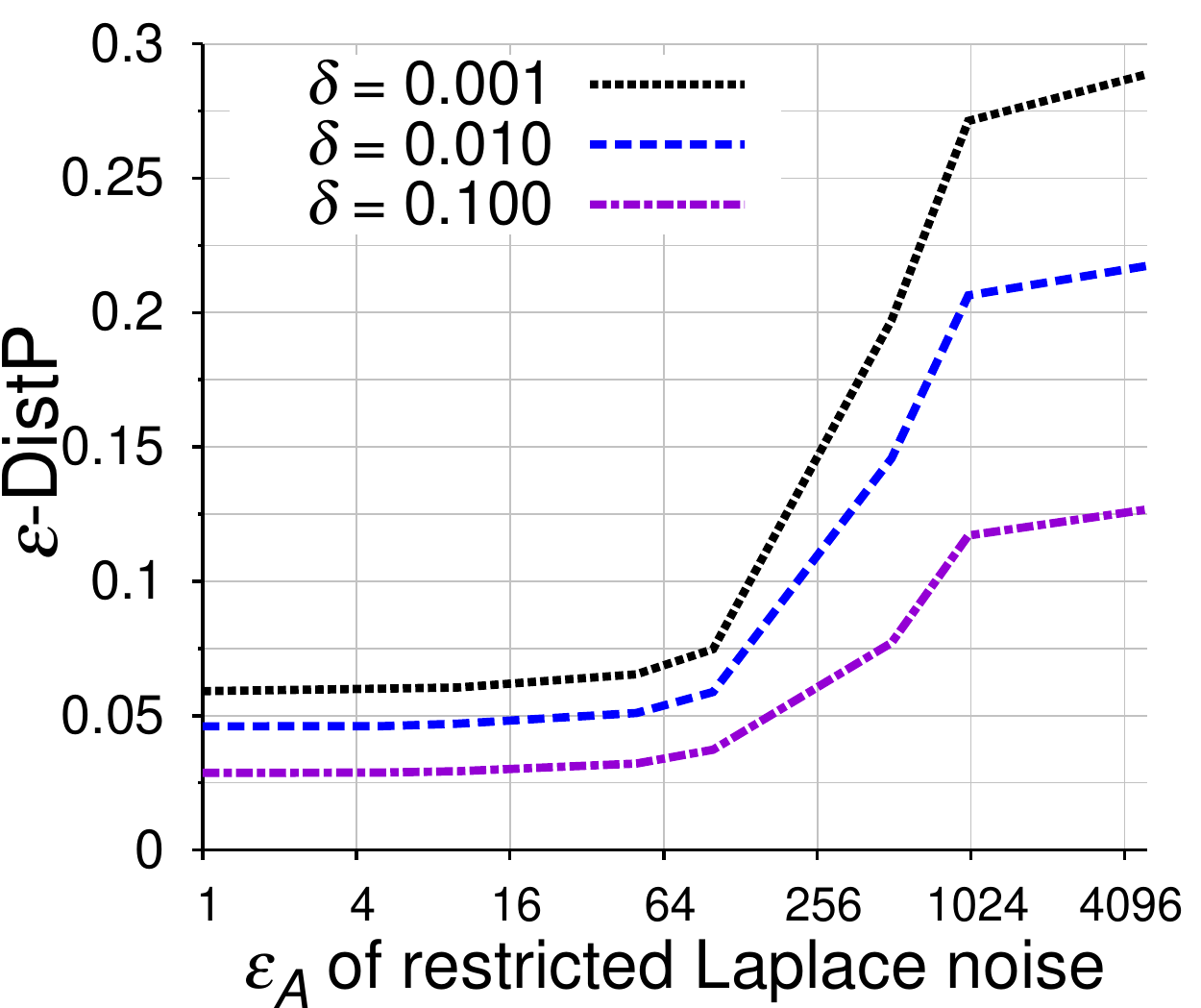}}}
\caption{$\varepsilon_\alg$ of $(\varepsilon_\alg, \allowbreak 0.020)$-\RL{} mechanism and $\varepsilon$-\DistP{} (with $10$ dummies).\label{fig:DP-noises:DistP}}
\end{subfigure}\hspace{0.4ex}\hfill
\begin{subfigure}[t]{0.24\textwidth}
  \mbox{\raisebox{-5pt}{\includegraphics[ width=1.05\textwidth]{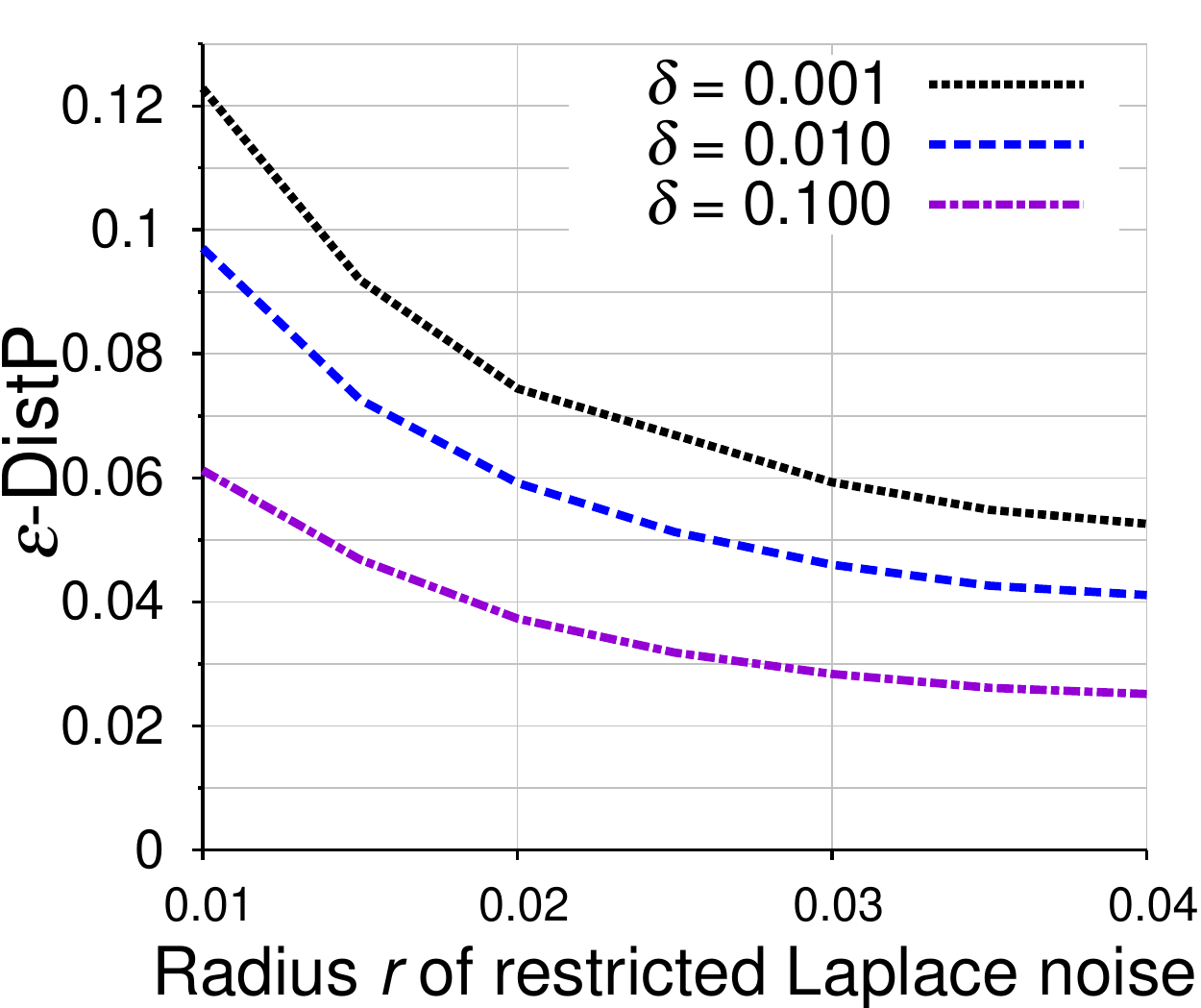}}}
\caption{A radius $r$ of $(100, r)$-\RL{} mechanism and $\varepsilon$-\DistP{} (with $10$ dummies).\label{fig:radius:DistP}}
\end{subfigure}\hspace{0.4ex}\hfill
\begin{subfigure}[t]{0.24\textwidth}
  \mbox{\raisebox{-5pt}{\includegraphics[ width=1.05\textwidth]{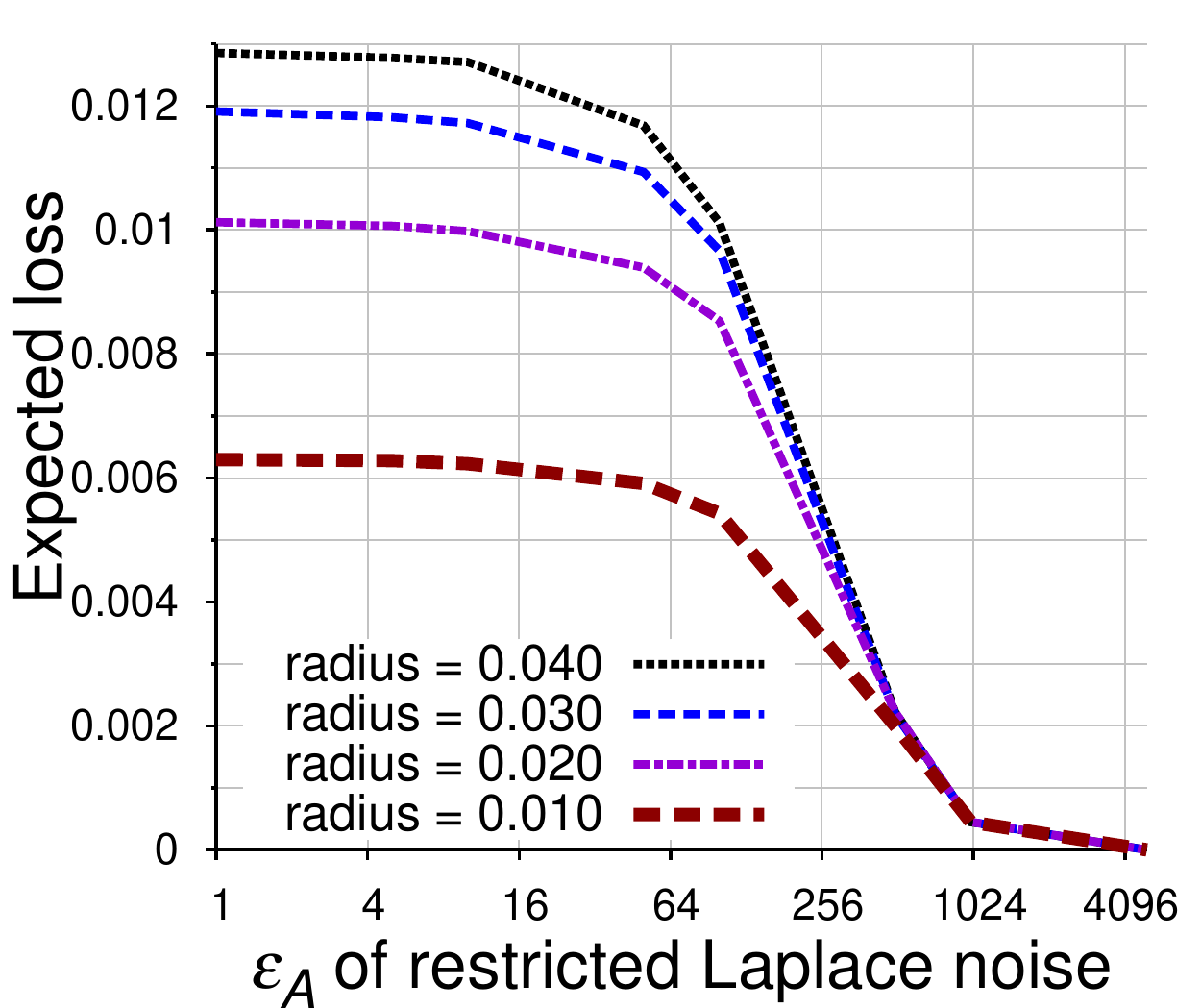}}}
\caption{$\varepsilon_\alg$~of $(\varepsilon_\alg, r)$- \allowbreak\RL{} mechanism and the expected loss (with $5$ dummies).\label{fig:tupling:loss}}
\end{subfigure}
\caption{Empirical \DistP{} and quality loss of $\TPM{}$ for the attribute \male{}/\female{}.\label{fig:tupling:privacy}}
\end{figure*}

\subsection{Experimental Setup}
\label{sub:exp-setup}

We perform experiments on location privacy in Manhattan by using the Foursquare dataset (Global-scale Check-in Dataset)~\cite{Yang_TIST15}.
We first divide Manhattan into $11 \times 10$ regions with $1.0 \mathrm{km}$ intervals.
To provide more useful information to users in crowded regions, we further re-divide these regions to $276$ regions by recursively partitioning each crowded region into four until each resulting region has roughly similar population density.%
\footnote{This partition may be useful to achieve smaller values $(\varepsilon, \delta)$ of \DistP{}, because $\beta$ tends to be smaller when the population density is closer to the uniform distribution.}
Let $\caly$ be the set of those $276$ regions, and $\calx$ be the set of the $228$ regions inside the central $10 \mathrm{km} \times 9 \mathrm{km}$ area in $\caly$.

As an obfuscation mechanism $Q$, we use the tupling mechanism $\TPM{}$ that uses an $(\varepsilon_{\!\alg}, r)$-RL mechanism $\alg$ and the uniform distribution $\nu$ over $\caly$ to generate dummy locations.
Note that $\nu$ is close to the population density distribution over $\caly$,
because each region in $\caly$ is constructed to have roughly similar population density.
In the definitions of the RL mechanism and the quality loss, we use the Euclidean distance $\| \cdot \|$ between the central points of the regions.

In the experiments, we measure the privacy of user attributes, formalized as \DistP{}.
For example, let us consider the attribute $\male/\female$.
For each $\attr\in\{ \male, \female \}$, let $\lambda_{\attr}\in\Dists\calx$ be the prior distribution of the location of the users having the attribute $\attr$.
Then, $\lambda_{\male}$ (resp. $\lambda_{\female}$) represents an attacker's belief on the location of the male (resp. female) users.
We define these as the empirical distributions that the attacker can calculate from 
the above Foursquare dataset.

\subsection{Evaluation of the Tupling Mechanism}
\label{sub:eval-tupling}

\subsubsection{Distribution privacy}
We demonstrate by experiments that the \male{} users cannot be recognized as which of \male{} or \female{} in terms of \DistP{}.
In Fig.~\ref{fig:tupling:privacy}, we show the experimental results on the \DistP{} of the tupling mechanism $\TPM{}$.
For a larger number $k$ of dummy locations, we have a stronger \DistP{} (Fig.~\ref{fig:dummies:DistP}).
For a larger $\varepsilon_{\!\alg}$,\, $(\varepsilon_{\!\alg}, 0.020)$-\RL{} mechanism $\alg$ adds less noise, hence the tupling mechanism provides a weaker \DistP{} (Fig.~\ref{fig:DP-noises:DistP})%
\footnote{In Fig.~\ref{fig:DP-noises:DistP}, for $\varepsilon_{\!\alg} \rightarrow 0$, $\varepsilon$ does not converge to $0$, since the radius $r = 0.020$ of \RL{} does not cover the whole $\caly$. However, if $r \ge \max_{x,y} \| x - y \|$, $\varepsilon$ converges to $0$.}.
For a larger radius $r$, the \RL{} mechanism $\alg$ spreads the original distribution $\lambda_{\male}$ and thus provides a strong \DistP{} (Fig.~\ref{fig:radius:DistP}).
We also show the relationship between $k$ and \DistP{} in the eastern/western Tokyo and London, which have different levels of privacy (Fig.~\ref{fig:DistP-Cities:MF}).

These results imply that if we add more dummies, we can decrease the noise level/radius of $\alg$ to have better utility, while keeping 
the same level $\varepsilon$ of \DistP{}.
Conversely, if $\alg$ adds more noise, we can decrease the number $k$ of dummies.

\subsubsection{Expected quality loss}
In Fig.~\ref{fig:tupling:loss}, we show the experimental results on the expected quality loss of the tupling mechanism.
For a larger $\varepsilon_{\!\alg}$, $\alg$ adds less noise, hence the loss is smaller.
We confirm that 
for more dummy data, the expected quality loss is smaller.
Unlike the planar Laplace mechanism ($\PL{}$),\, $\alg$ ensures that the worst quality loss is bounded above by the radius $r$.
Furthermore, for a smaller radius $r$, the expected loss is also smaller as shown in Fig.~\ref{fig:tupling:loss}.

\journal{
In addition, when $r$ is small, the output of $\alg$ is outside the domain (i.e., Manhattan) with a far smaller probability than $\PL{}$.
Hence if $\caly$ is large enough, then the \RL{} mechanism does not require additional computation (e.g. remapping technique~\cite{Chatzikokolakis17PETS}) to handle the output outside the domain.
}

\begin{figure}[t]
\begin{minipage}[t]{0.29\textwidth}
\centering
\mbox{\raisebox{-5pt}{\includegraphics[ width=1.05\textwidth]{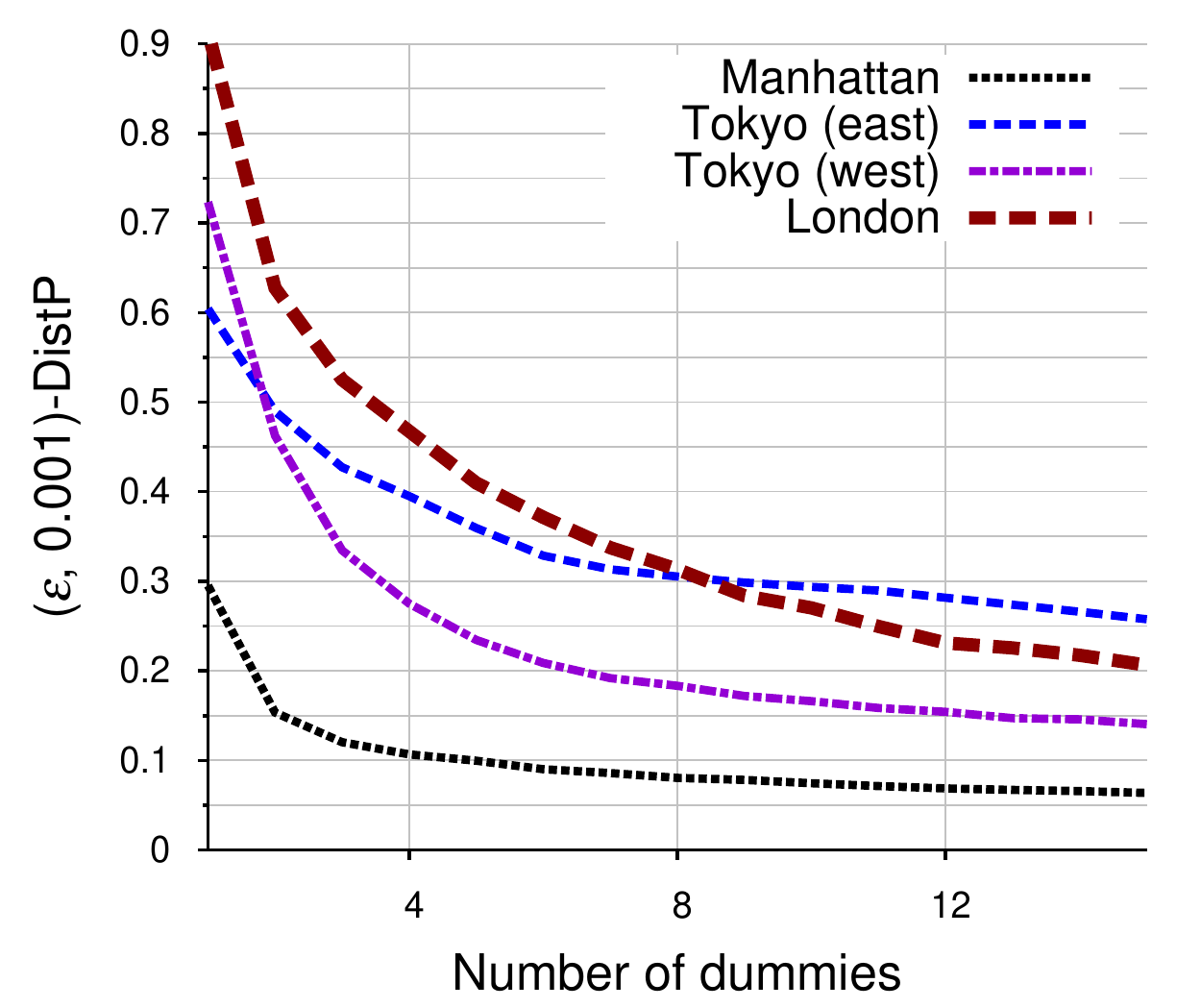}}}
\caption{$k$ and \DistP{} for \male{}/\female{} in different cities.
\label{fig:DistP-Cities:MF}}
\end{minipage}
\hfill
\begin{minipage}[t]{0.29\textwidth}
\centering
 \mbox{\raisebox{-5pt}{\includegraphics[ width=1.05\textwidth]{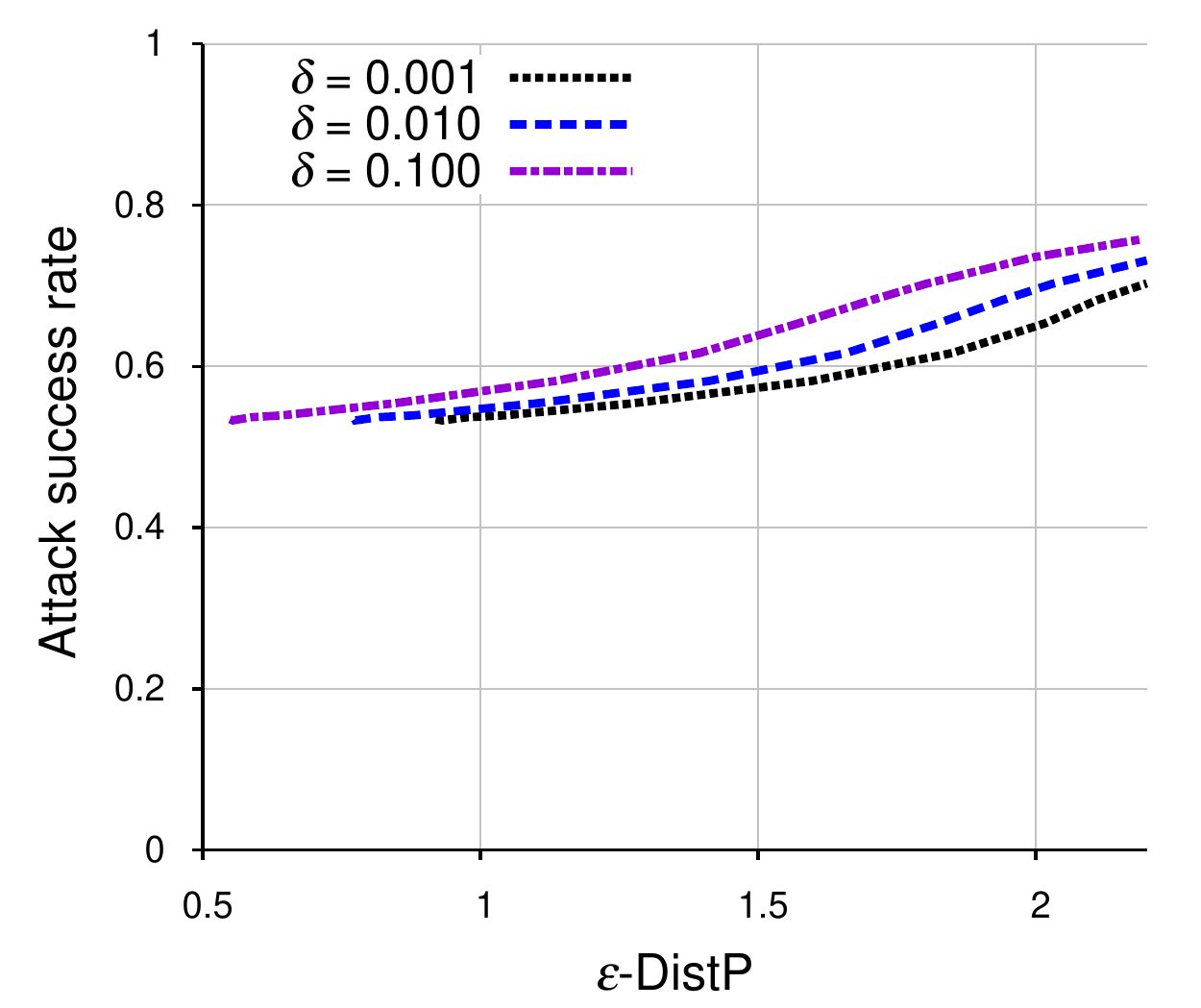}}}
\caption{\DistP{} and ASR 
of the tupling ($k = 10$, $r=0.020$).
\label{fig:DistP-ASR}}
\end{minipage}
\hfill
\begin{minipage}[t]{0.39\textwidth}
\centering
  \mbox{\raisebox{-5pt}{\includegraphics[ width=0.79\textwidth]{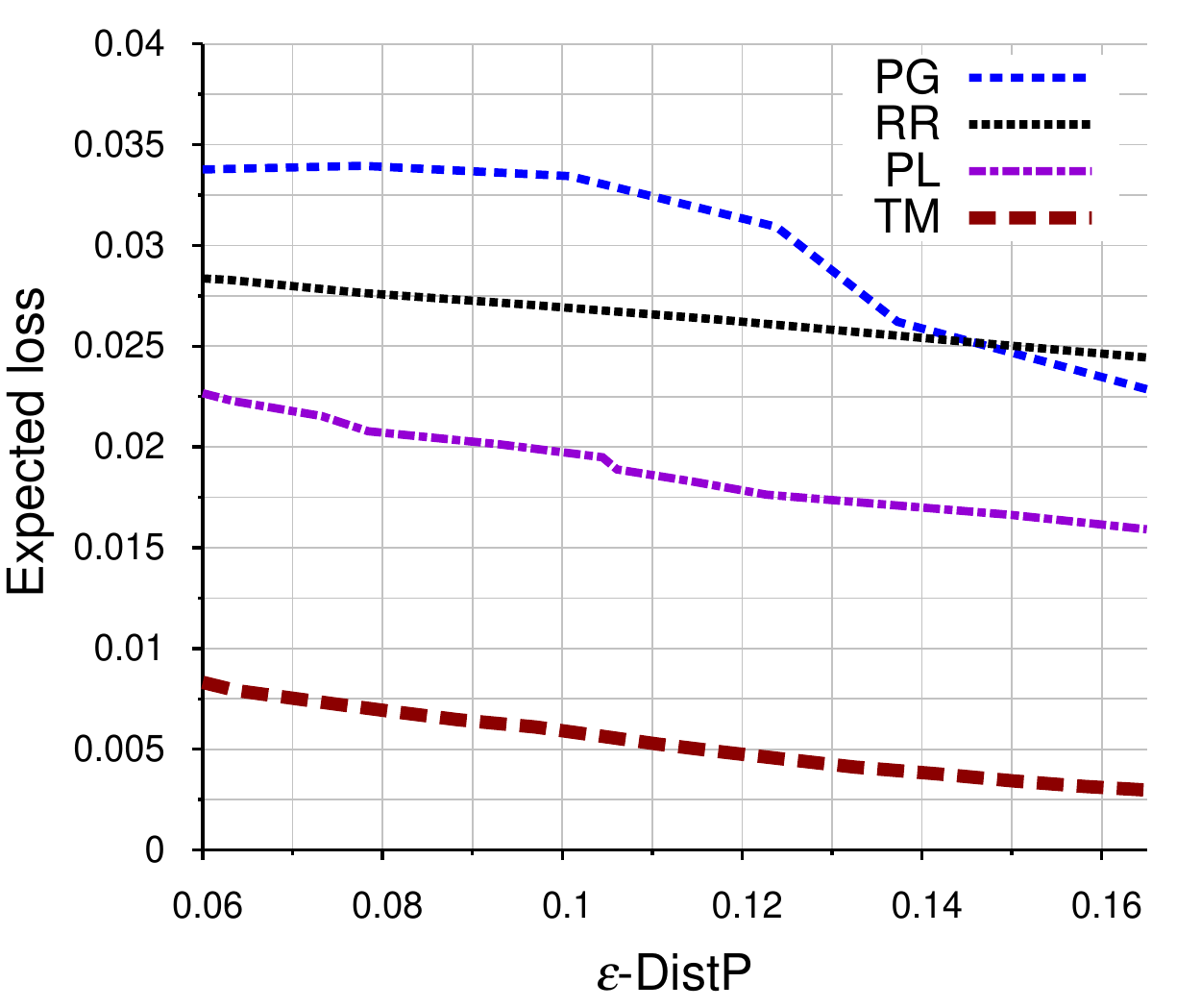}}}
\caption{$(\varepsilon, \allowbreak 0.001)$-\DistP{} and expected loss for \male{}/\female{} and TM using $k = 10$, $r = 0.020$. \label{fig:compare-mechanisms}}
\end{minipage}
\end{figure}

\subsection{Appropriate Parameters}
\label{sub:exp-parameters}

We define the \emph{attack success rate (ASR)} as the ratio that the attacker succeeds to infer a user has an attribute when~she does actually.
We use an inference algorithm based on the Bayes decision rule \cite{Duda_Wiley00} to minimize the identification error probability when the estimated posterior probability is accurate~\cite{Duda_Wiley00}.

In Fig.~\ref{fig:DistP-ASR}, we show the relationships between \DistP{} and ASR in Manhattan for the attribute \home{}, meaning the users located at their home.
In theory, $\rm{ASR} = 0.5$ represents the attacker learns nothing about the attribute, whereas the empirical ASR in our experiments fluctuates around $0.5$.
This seems to be caused by the fact that the dataset and the number of locations are finite.
From Fig.~\ref{fig:DistP-ASR}, we conclude that $\varepsilon = 1$ is an appropriate parameter for $(\varepsilon, 0.001)$-\DistP{} to achieve $\rm{ASR} = 0.5$ in our setting, and we confirm this for other attributes.
However, we note that this is an empirical criterion possibly depending on our setting, and the choice of $\varepsilon$ for \DistP{} can be as controversial as that for \DP{} and should also be investigated using approaches for \DP{} (e.g.,~\cite{Hsu:14:CSF}) in future work.

\arxiv{
In Appendices~\ref{sec:experiment-attributes:Manhattan},~\ref{sec:experiment-attributes:hours}, and~\ref{sec:experiment-attributes:cities}, we show that different levels of \DistP{} are achieved for different attributes in different cities.
In particular, we present appropriate parameters for the tupling mechanism for various attributes in a couple of cities.
}

\subsection{Comparison of Obfuscation Mechanisms}
\label{sub:exp-compare}

We demonstrate that the tupling mechanism (TM) outperforms the popular  mechanisms: the randomized response (RR), the planar Laplace (PL), and the planar Gaussian (PG).
In Fig.~\ref{fig:compare-mechanisms} we compare these concerning the relationship between $\varepsilon$-\DistP{} and expected quality loss.
Since PG always has some $\delta$, it provides a weaker \DistP{} than PL for the same quality loss.
We also confirm that PL has smaller loss than RR, since it adds noise proportionally to the distance.
\arxiv{Comparison using other attributes can be found in Appendix~\ref{sec:experiment-attributes:Manhattan}.}

Finally, we briefly discuss the computational cost of the tupling mechanism $\TPM{}$, compared to $\PL{}$.
In the implementation, for a larger domain $\calx$,\, $\PL{}$ deals with a larger size $|\calx|\times|\caly|$ of the mechanism's matrix, since it outputs each region with a non-zero probability.
In contrast, since the \RL{} mechanism $\alg$ used in $\TPM{}$ maps each location $x$ to a region within a radius $r$ of $x$, the size of $\alg$'s matrix is $|\calx|\times|\caly_{x,r}|$, requiring much smaller memory space than $\PL{}$.

Furthermore, the users of $\TM{}$ can simply ignore the responses to dummy queries, whereas the users of $\PL{}$ need to select relevant POIs (point of interests) from a large radius of $x$, which could cost computationally for many POIs.
Therefore, 
$\TM{}$ is more suited to be used in mobile environments than $\PL{}$.

\section{Related Work}
\label{sec:related}
\paragraph{\textbf{Differential privacy.}}
Since the seminal work of Dwork~\cite{Dwork:06:ICALP} on 
\DP{},
a number of its variants have been studied 
to provide different privacy guarantees; 
e.g., 
$f$-divergence privacy~\cite{Barthe:13:ICALP}, 
$d$-privacy~\cite{Chatzikokolakis:13:PETS}, 
Pufferfish privacy~\cite{Kifer:12:PODS}, 
local \DP{}~\cite{Duchi:13:FOCS}, 
and
utility-optimized local \DP{}~\cite{Murakami:19:USENIX}.
All of these are intended to 
protect the input data 
rather than the input distributions.
Note that distribution\emph{al} privacy~\cite{Blum:13:JACM} is different from \DistP{} and does not aim at protecting the privacy of distributions.

To our knowledge, this is the first work that investigates the \emph{differential privacy} of probability distributions lying behind the input. 
However, a few studies have proposed related notions.
Jelasity et al.~\cite{Jelasity:IHMMSec:14} propose \emph{distributional differential privacy} w.r.t. parameters $\theta$ and $\theta'$ of two distributions, which aims at protecting the privacy of the distribution parameters but is defined in a Bayesian style (unlike \DP{} and \DistP{}) to satisfy that for any output sequence~$y$,\, $p(\theta | y) \le e^{\varepsilon} p(\theta' | y)$.
After a preliminary version of this paper appeared in arXiv~\cite{arxiv}, a notion generalizing \DistP{}, called \emph{profile based privacy}, is proposed in~\cite{Geumlek:19:ISIT}.

Some studies are technically related to our work.
Song \textit{et al.}~\cite{Song:17:SIGMOD} propose the Wasserstein mechanism to provide Pufferfish privacy, which protects correlated inputs. 
Fernandes \textit{et al.}~\cite{Fernandes:19:POST} introduce Earth mover's privacy,
which is technically different from \DistP{} in that their mechanism obfuscates a vector (a bag-of-words)
instead of a distribution, and perturbs each element of the vector.
Sei \textit{et al.}~\cite{Sei:17:ITFS} propose a variant of the randomized response 
to protect individual data and provide high utility of database.
However, we emphasize again that 
our work differs from these studies in that 
we aim at protecting 
input distributions. 
\paragraph{\textbf{Location privacy.}}
Location privacy has been widely studied in the literature, and its survey can be found in%
~\cite{Chatzikokolakis:17:FTPS}. 
A number of location obfuscation methods have been proposed so far, 
and they can be broadly divided into the following four types: 
perturbation (adding noise)%
~\cite{Shokri:12:CCS,Andres:13:CCS,Bordenabe:14:CCS}, 
location generalization (merging regions)%
~\cite{Shokri:11:SP,Xue:09:LoCa}, and 
location hiding (deleting)%
~\cite{Shokri:11:SP,Hoh:07:CCS}, and 
adding dummy locations%
~\cite{Bindschaedler:16:SP,Chow:09:WPES,Kido:05:ICDE}.
Location obfuscation based on \DP{} (or its variant) have also been widely studied, 
and they can be categorized into the ones in the centralized model%
~\cite{Machanavajjhala:08:ICDE,Ho:11:GIS} 
and the ones in the local model \cite{Andres:13:CCS,Bordenabe:14:CCS}. 
However, these methods aim at protecting locations, 
and 
neither at protecting users' attributes (e.g., age, gender) 
nor activities (e.g., working, shopping) in a \DP{} manner. 
Despite the fact that users' attributes and activities can be inferred from their locations 
\cite{Liao:07:IJRR,Zheng:09:LBSN,Matsuo:07:IJCAI}, 
to our knowledge,
no studies have proposed obfuscation mechanisms to provide 
rigorous \DP{} guarantee for such attributes and activities.

\section{Conclusion}
\label{sec:conclude}

We have proposed a formal model for the privacy of probability distributions 
and introduced the notion of distribution privacy (\DistP{}).
Then we have shown that existing local mechanisms deteriorate the utility by adding too much noise to provide \DistP{}.
To improve the tradeoff between \DistP{} and utility, we have introduced the tupling mechanism and applied it to the protection of user attributes in LBSs.
Then we have demonstrated that the tupling mechanism outperforms popular local mechanisms in terms of attribute obfuscation and service quality.

\arxiv{
As future work, 
we will improve the theoretical bound on $\varepsilon$ for the tupling mechanism.
We also plan to 
design optimal mechanisms that minimize the quality loss and computational costs, while providing \DistP{}, in various applications.
}

\paragraph*{Acknowledgment}
We thank the reviewers, Catuscia Palamidessi, Gilles Barthe, and Frank D. Valencia for their helpful comments on preliminary drafts.

\bibliographystyle{IEEEtran}
\bibliography{short}

\appendix
\conference{

\section{Experimental Results}
\label{sec:appendix:esorics}

In this section we present some of the experimental results on the following four attributes. See~\cite{arxiv} for further experimental results.
\begin{itemize}
\item 
\social{}/\lesssocial{} represent whether a user's social status~\cite{Cheng:11:ICWSM} (the number of followers divided by the number of followings) is greater than $5$ or not. 
\item 
\workplace{}/\nonworkplace{} represent whether a user is at office or not.
This attribute can be thought as sensitive when it implies users are unemployed.
\item 
\home{}/\outside{} represent whether a user is at home or not.
\item 
\north{}/\south{} represent whether a user's home is located in the northern or southern Manhattan.
This attribute needs to be protected from stalkers.
\end{itemize}

First, we compare different obfuscation mechanisms for various attributes in Figs.~\ref{fig:compare-mechanisms},~\ref{fig:compare:delta0.001:SI:esorics}, and~\ref{fig:compare:delta0.001:HO:esorics}.
We also compare different time periods: 00h-05h, 06h-11h, 12h-17h, 18h-23h in Manhattan in 
Fig.~\ref{fig:times+MH+MF:tupling:privacy:esorics}.

Next, we compare the experimental results on five cities: Manhattan, eastern Tokyo, western Tokyo, London, and Paris.
In Table~\ref{tab:example-parameters:esorics} we show examples of parameters that achieve the same levels of \DistP{} in different cities.
More detailed can be found 
in Fig.~\ref{fig:cities+MF:tupling:privacy:esorics} (\male{}/\female{}).

Finally, we compare theoretical/empirical values of $\varepsilon$-\DistP{} as follows.
In Table~\ref{table:bounds:theory-experiments:esorics}, we show the theoretical values of $\varepsilon$ calculated by Theorem~\ref{thm:TuplingDP} for $\delta = 0.001, \allowbreak 0.01, \allowbreak 0.1$.
Compared to experiments, those values can only give loose upper bounds on $\varepsilon$, because of the concentration inequality used to derive Theorem~\ref{thm:TuplingDP}.

\begin{table}
\centering
\caption{The number $k$ of dummies required for achieving \DistP{} in different cities (MH = Manhattan, TKE = Tokyo (east), TKW = Tokyo (west), LD = London, PR = Paris) when $\varepsilon_{\alg} = 100$ and $r = 0.020$. 
Note that the data of Paris for \male{}/\female{} are excluded because of the insufficient sample size.
\label{tab:example-parameters:esorics}}
\vspace{1mm}
\scalebox{0.8}[0.8]{ 
\makeatother\begin{tabular}{l c c c c c c c}
\cmidrule[0.1em]{2-6}
& {\bf MH} & {\bf TKE} & {\bf TKW} & {\bf LD} & {\bf PR} \\[0.1ex]
\midrule
$(0.25, 0.001)$-\DistP{} for {\bf male\,/\,female} & 2 &  $>$20 & 5 & 10 & ---  \\[0.1ex]
$(0.50, 0.001)$-\DistP{} for {\bf social\,/\,less social}\! & 2  &  3 & $>$20 & 2 & 3 \\[0.1ex]
$(1.00, 0.001)$-\DistP{} for {\bf work\,/\,non-work}\! & 2  &  2  & $>$20 & 1 & 2  \\[0.1ex]
$(1.50, 0.001)$-\DistP{} for {\bf home\,/\,outside} & 3  &  5 & $>$20 & $>$20 & 4  \\[0.1ex]
\bottomrule
\end{tabular}
}
\end{table}
\begin{table}
\centering
\caption{Theoretical/empirical $\varepsilon$-\DistP{} of $\TPM{}$ ($k = 10$, $\varepsilon_{\!\alg} = 10$, $r = 0.020$). \label{table:bounds:theory-experiments:esorics}}
\vspace{1mm}
\renewcommand{\arraystretch}{1}
\scalebox{0.8}[0.8]{ 
\begin{tabular}{lcccc}
\hline
& $\delta = 0.001$ & $\delta = 0.01$ & $\delta = 0.1$ 
\\ \hline
Theoretical bounds &
2.170 & 1.625 & 1.140 
\\ \hline
Empirical values &
0.04450 & 0.03534 & 0.02295 
\\ \hline
\end{tabular}
}
\vspace{-1mm}
\end{table}

\begin{figure}[H]
\begin{tabular}{cc}
\begin{minipage}{1.0\hsize}
\centering
\vspace{-3.5mm}
\begin{subfigure}[t]{0.48\textwidth}\centering
  \mbox{\raisebox{-10pt}{\includegraphics[ width=0.5\textwidth]{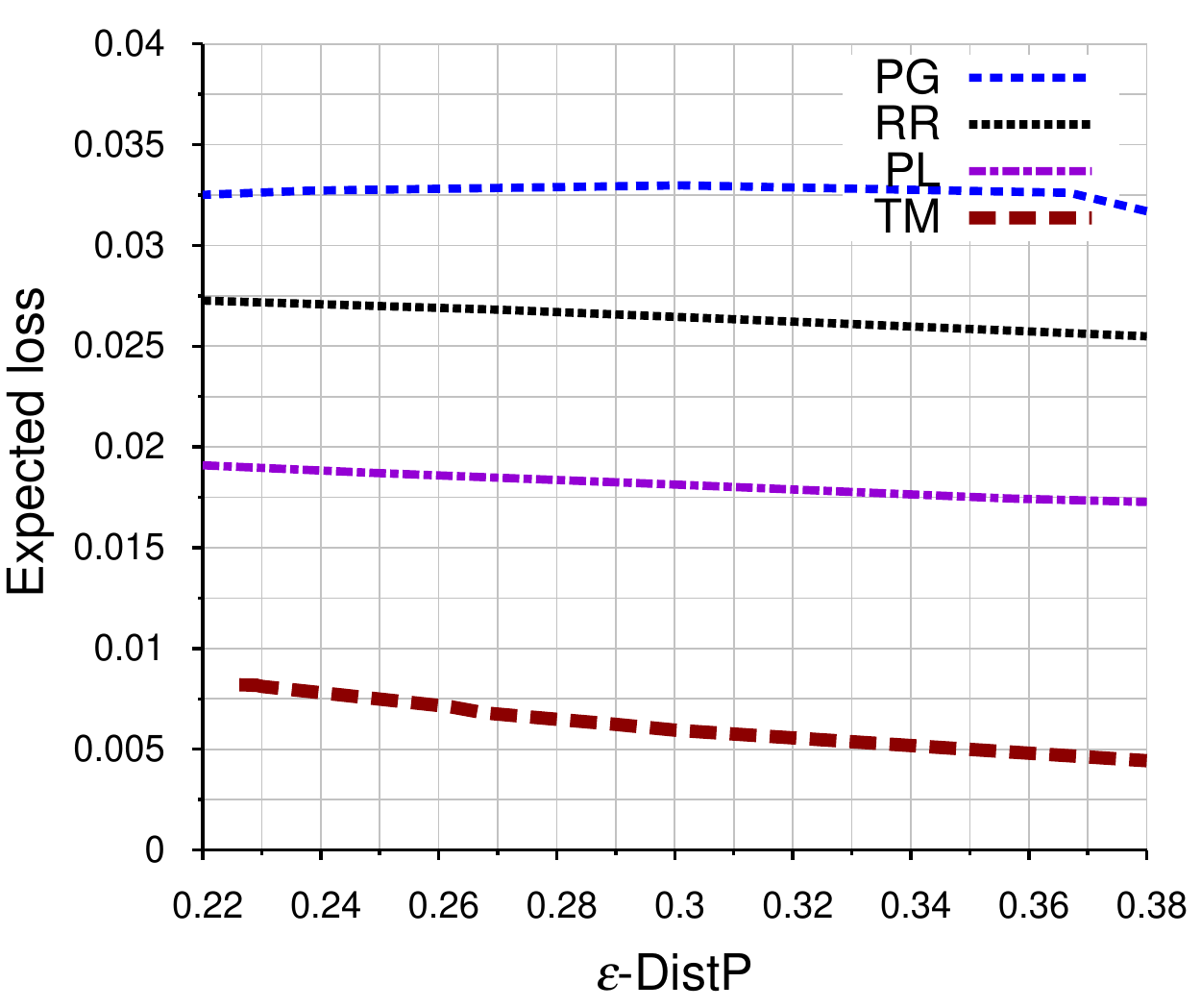}}}
\caption{$(\varepsilon, \allowbreak 0.001)$-\DistP{} and expected loss for \social{}/\lesssocial{} in Manhattan.
\label{fig:compare:delta0.001:SI:esorics}}
\end{subfigure}\hspace{0.3ex}\hfill
\begin{subfigure}[t]{0.48\textwidth}\centering
  \mbox{\raisebox{-10pt}{\includegraphics[ width=0.5\textwidth]{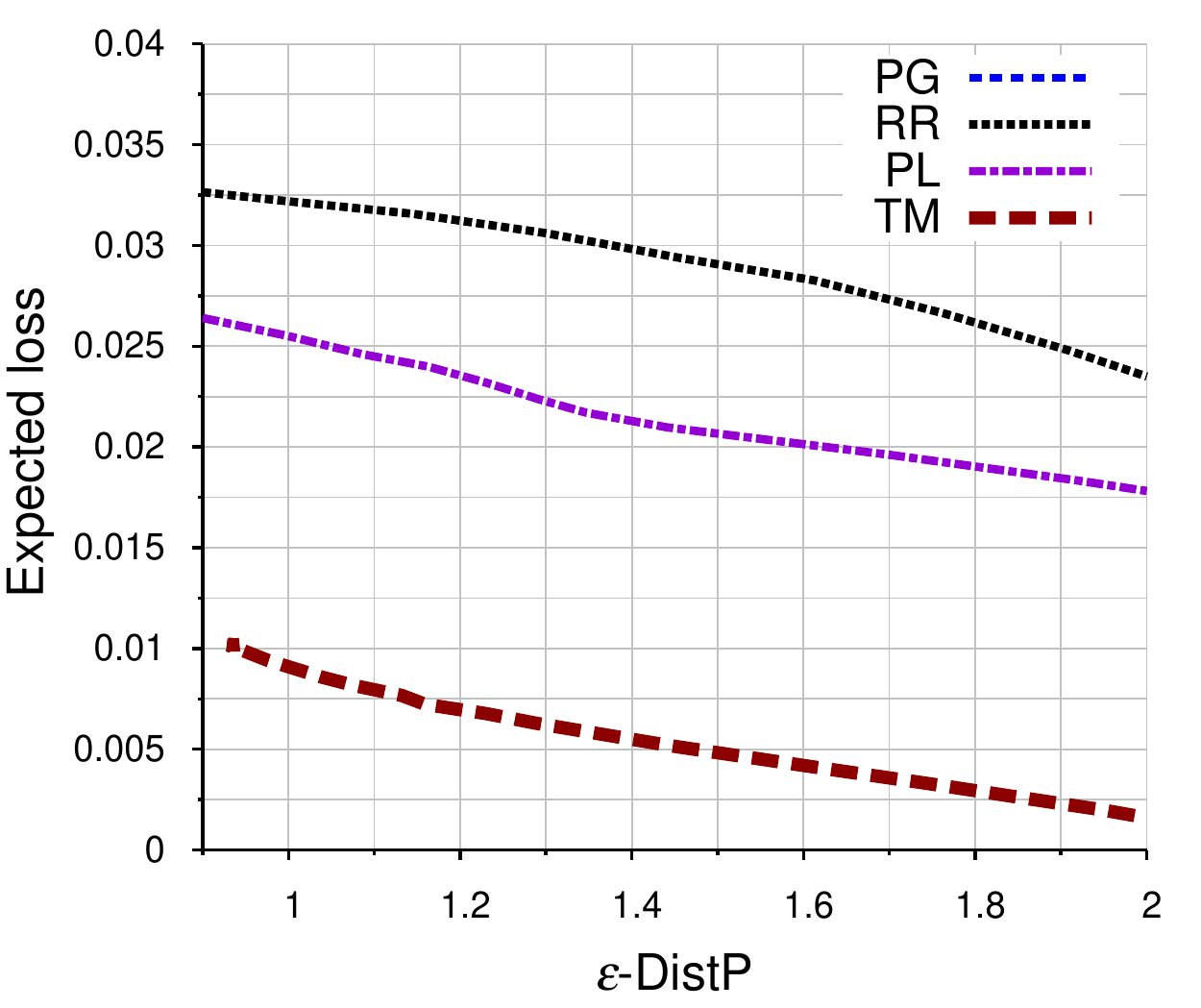}}}
\caption{$(\varepsilon, \allowbreak 0.001)$-\DistP{} and expected loss for \home{}/\outside{} in Manhattan.
\label{fig:compare:delta0.001:HO:esorics}}
\end{subfigure}\hspace{0.5ex}\hfill
\caption{Comparison of the randomized response (RR), the planar Laplace mechanism (PL), the planar Gaussian mechanism (PG), and the tupling mechanism (TM) $\TPM{}$ with $k = 10$ dummies and a radius $r = 0.020$. 
\label{fig:compare-mechanisms-esorics}}
\vspace{5mm}
\end{minipage}
\\
\begin{minipage}{1.0\hsize}
\centering
\begin{subfigure}[t]{0.24\textwidth}
  \mbox{\raisebox{-12pt}{\includegraphics[height=26mm, width=1.00\textwidth]{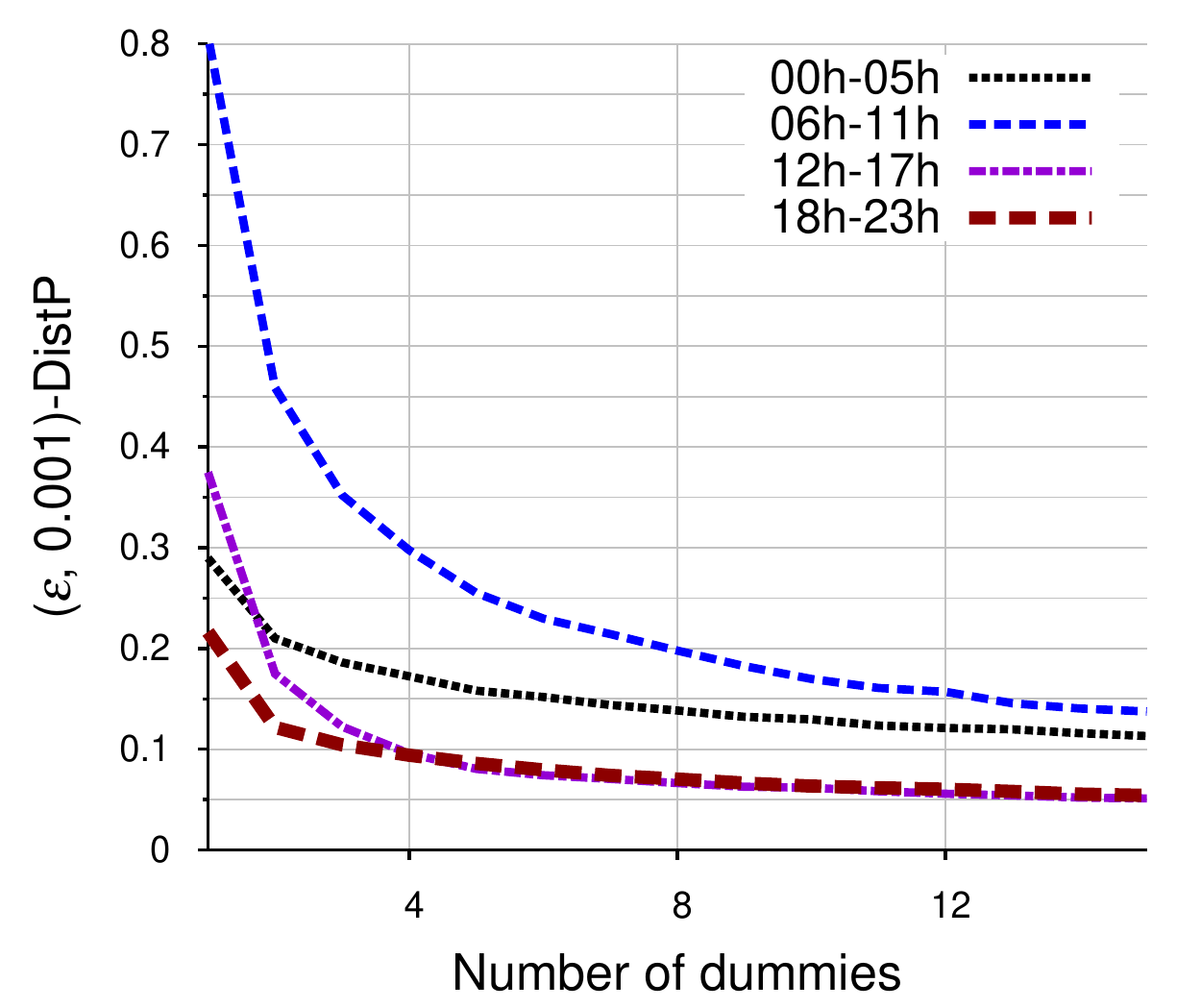}}}
\caption{\#dummies and $\varepsilon$-\DistP{} (when using $(100, 0.020)$-\RL{} mechanism).
\label{fig:times+MH+MF:dummies:DistP:esorics}}
\end{subfigure}\hspace{0.4ex}\hfill
\begin{subfigure}[t]{0.24\textwidth}
  \mbox{\raisebox{-12pt}{\includegraphics[height=26mm, width=1.00\textwidth]{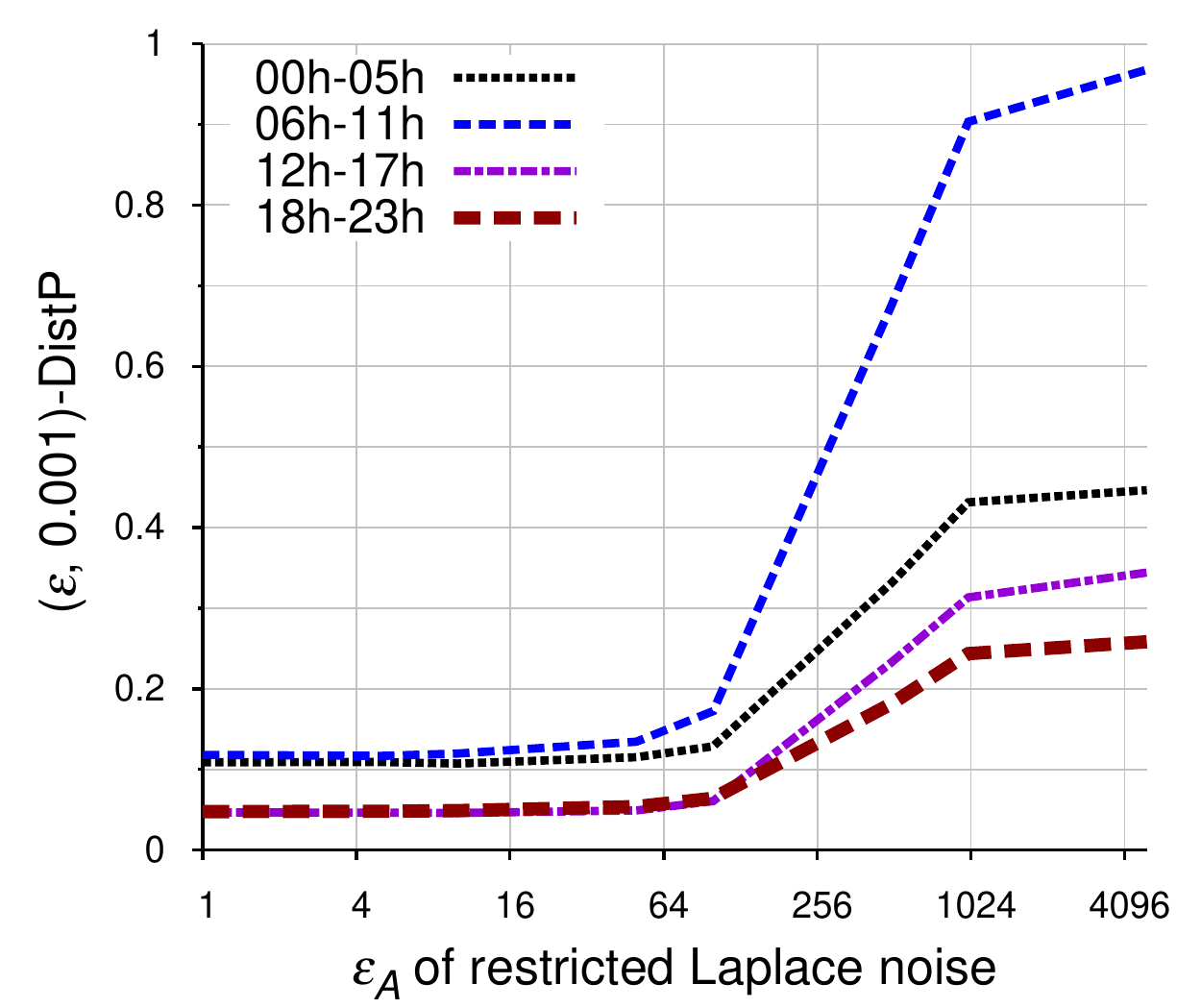}}}
\caption{$\varepsilon_\alg$ of $(\varepsilon_\alg, 0.020)$-\RL{} mechanism and $\varepsilon$-\DistP{} (with $10$ dummies).\label{fig:times+MH+MF:DP-noises:DistP:esorics}}
\end{subfigure}\hspace{0.4ex}\hfill
\begin{subfigure}[t]{0.24\textwidth}
  \mbox{\raisebox{-12pt}{\includegraphics[height=26mm, width=1.00\textwidth]{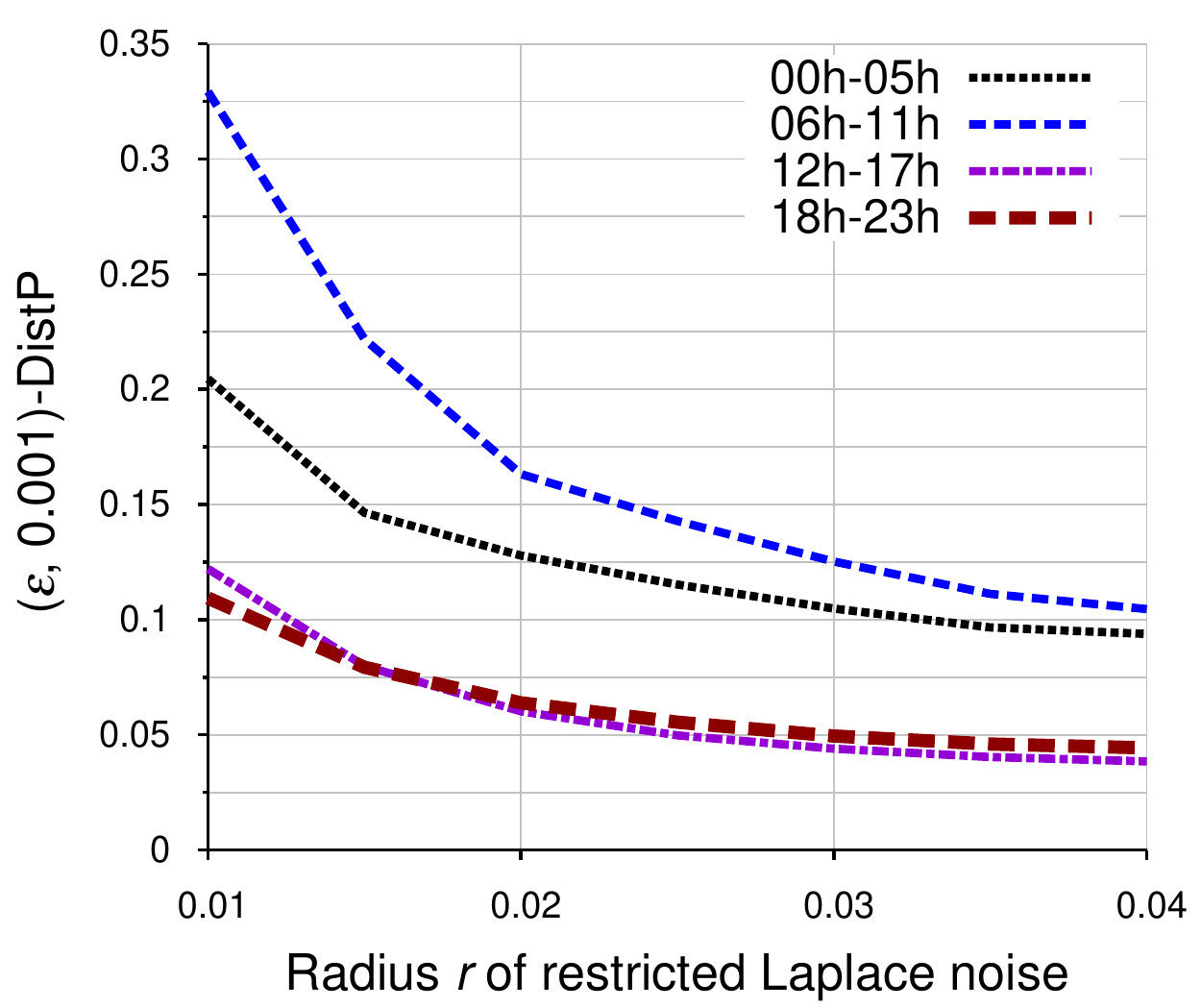}}}
\caption{A radius $r$ of $(100, r)$-\RL{} mechanism and $\varepsilon$-\DistP{} (with $10$ dummies).
\label{fig:times+MH+MF:radius:DistP:esorics}}
\end{subfigure}\hspace{0.4ex}\hfill
\begin{subfigure}[t]{0.24\textwidth}
  \mbox{\raisebox{-12pt}{\includegraphics[height=26mm, width=1.00\textwidth]{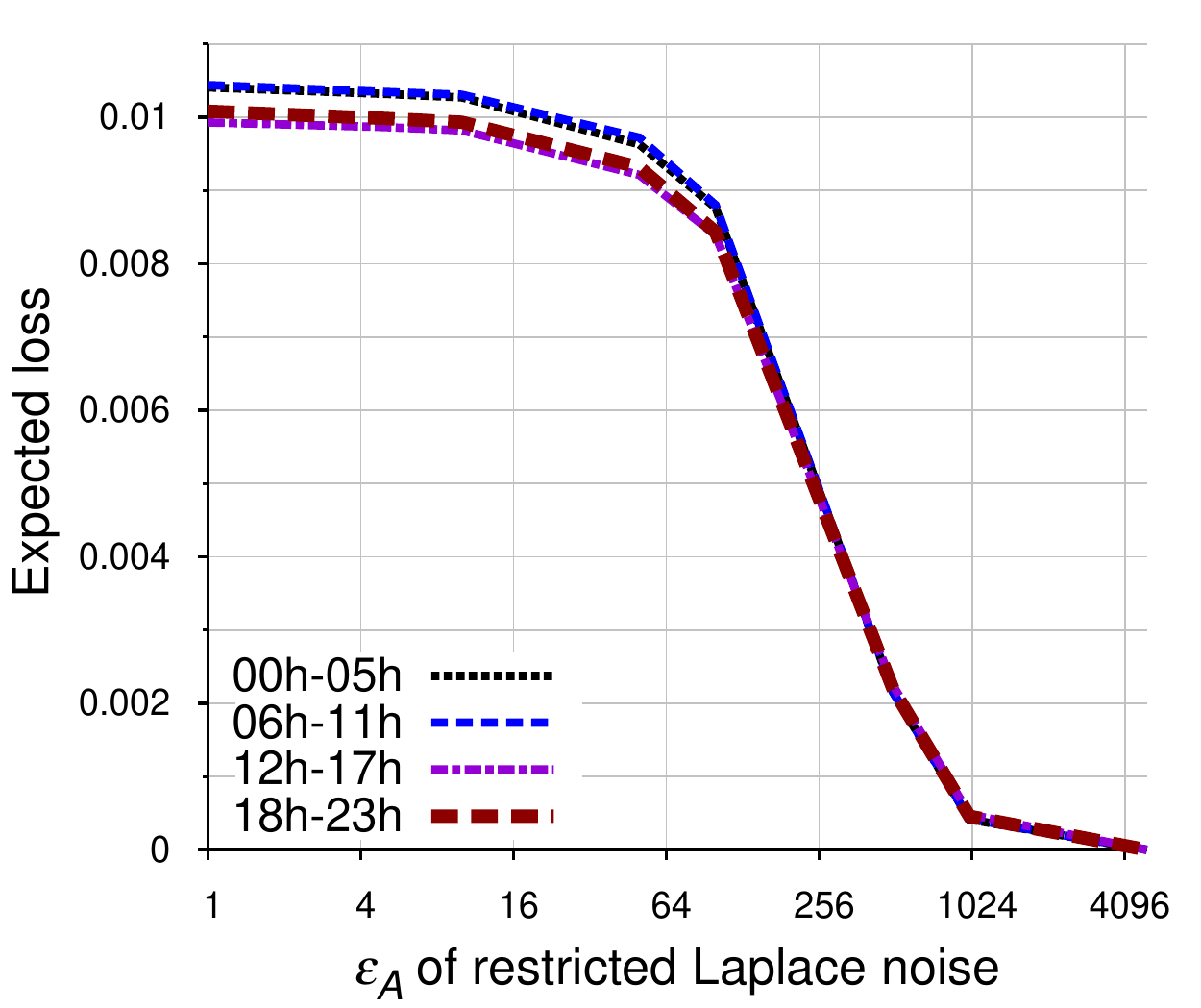}}}
\caption{$\varepsilon_\alg$ of $(\varepsilon_\alg, \allowbreak r)$-\RL{} mechanism and the expected loss (with $5$ dummies).
\label{fig:times+MH+M:tupling:loss:esorics}}
\end{subfigure}
\caption{Empirical \DistP{} and loss for \male{}/\female{} in different hours.
\label{fig:times+MH+MF:tupling:privacy:esorics}}
\vspace{5mm}
\end{minipage}
\\
\begin{minipage}{1.0\hsize}
\centering
\begin{subfigure}[t]{0.23\textwidth}
  \mbox{\raisebox{-10pt}{\includegraphics[height=26.0mm, width=1.00\textwidth]{figs-conf/cities+MF-A.pdf}}}
\caption{\#dummies and $\varepsilon$-\DistP{} (when using $(100, 0.020)$-\RL{} mechanism).\label{fig:cities+MF:dummies:DistP:esorics}}
\end{subfigure}\hspace{0.4ex}\hfill
\begin{subfigure}[t]{0.23\textwidth}
  \mbox{\raisebox{-10pt}{\includegraphics[height=26.0mm, width=1.00\textwidth]{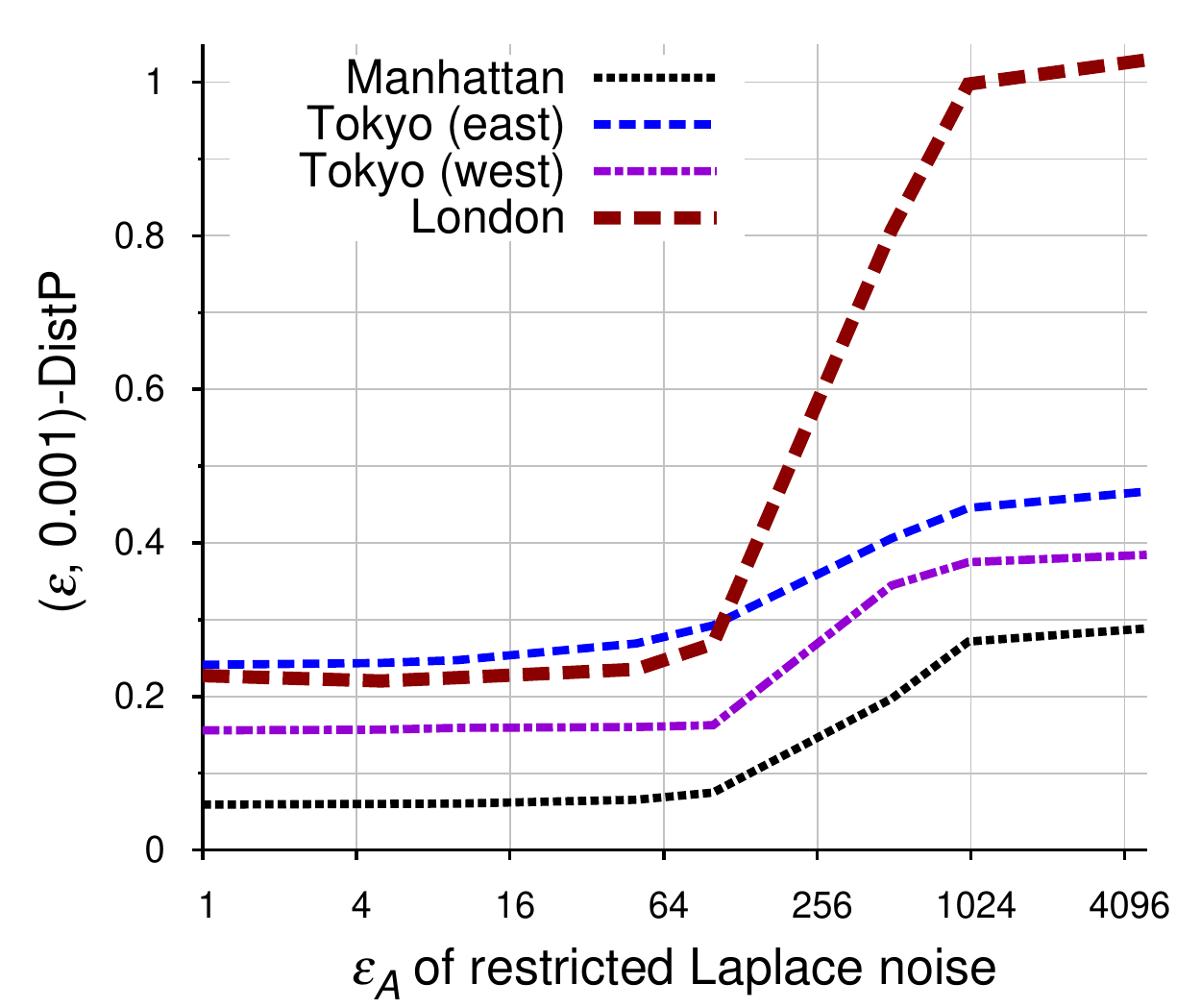}}}
\caption{$\varepsilon_\alg$ of $(\varepsilon_\alg, \allowbreak 0.020)$-\RL{} mechanism and $\varepsilon$-\DistP{} (with $10$ dummies).\label{fig:cities+MF:DP-noises:DistP:esorics}}
\end{subfigure}\hspace{0.4ex}\hfill
\begin{subfigure}[t]{0.23\textwidth}
  \mbox{\raisebox{-10pt}{\includegraphics[height=26.0mm, width=1.00\textwidth]{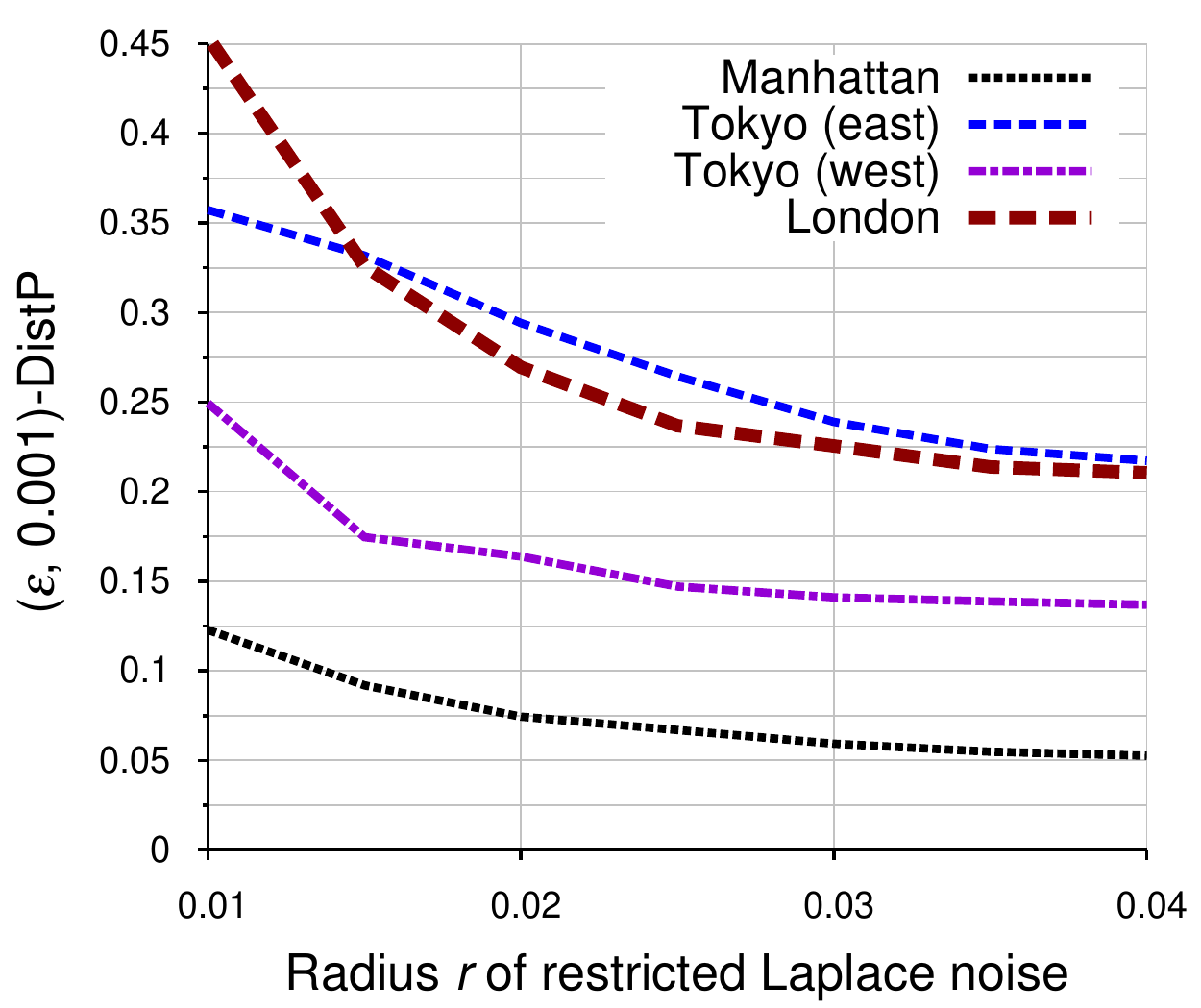}}}
\caption{A radius $r$ of $(100, r)$-\RL{} mechanism and $\varepsilon$-\DistP{} (with $10$ dummies).\label{fig:cities+MF:radius:DistP:esorics}}
\end{subfigure}\hspace{0.4ex}\hfill
\begin{subfigure}[t]{0.23\textwidth}
  \mbox{\raisebox{-10pt}{\includegraphics[height=26.0mm, width=1.00\textwidth]{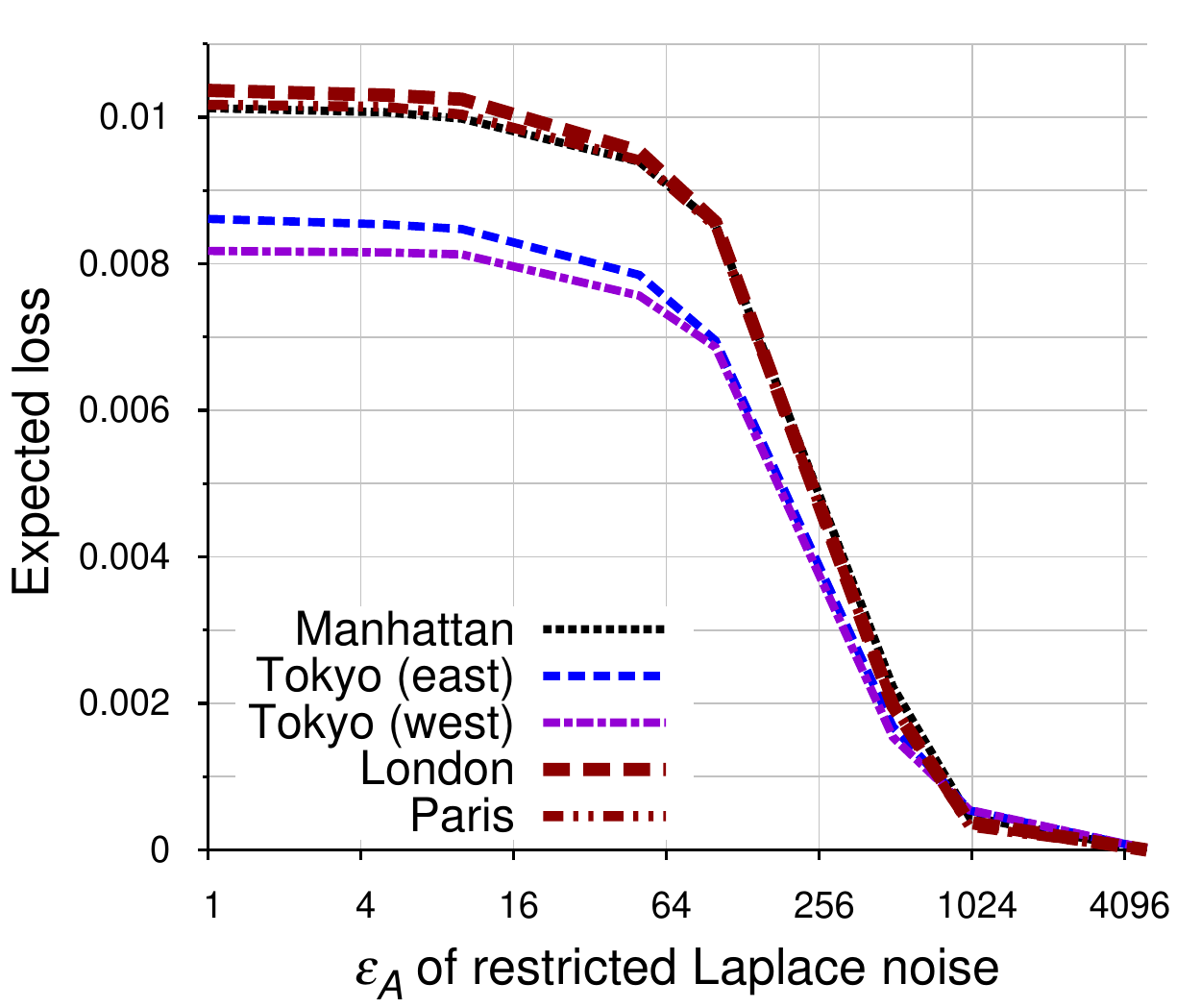}}}
\caption{$\varepsilon_\alg$ of $(\varepsilon_\alg, \allowbreak r)$-\RL{} mechanism and the expected loss (with $5$ dummies).
\label{fig:cities+M:tupling:loss:esorics}}
\end{subfigure}
\caption{Empirical \DistP{} and loss for \male{}/\female{} in different cities.
\label{fig:cities+MF:tupling:privacy:esorics}}
\end{minipage}
\end{tabular}
\end{figure}

}
\arxiv{
\newpage
\section{Experimental Results}
\label{sec:appendix:experimental}

\subsection{Experimental Comparison among Attributes in Manhattan}
\label{sec:experiment-attributes:Manhattan}

In this section we present the experimental results on more attributes of users in the location data in Manhattan.
We investigate the following four attributes:
\begin{itemize}
\item 
\social{}/\lesssocial{} (Fig.~\ref{fig:MH+SI:tupling:privacy}) represent whether a user's social status (defined in~\cite{Cheng:11:ICWSM} as the number of followers divided by the number of followings) is greater than $5$ or not. 
In~\cite{Yang:19:TKDE}, 
the social status is regarded as a private attribute that can be leaked by the user's behaviour and should be protected.
We set $5$ to be a threshold, because mobility patterns of people are substantially different when their social status is greater than $5$~\cite{Cheng:11:ICWSM}.
\item 
\workplace{}/\nonworkplace{} (Fig.~\ref{fig:MH+WP:tupling:privacy}) represent whether a user is at office or not.
This attribute can be sensitive when it implies users are unemployed.
\item 
\home{}/\outside{} (Fig.~\ref{fig:MH+HO:tupling:privacy}) represent whether a user is at home or not.
This attribute should be protected, for instance, from 
robbers \cite{please_rob_me}.
\item 
\north{}/\south{} (Fig.~\ref{fig:NS:tupling:privacy}) represent whether a user's home is located in the northern or southern Manhattan.
This attribute needs to be protected, for instance, from stalkers.
However, the residential area is highly correlated with visited places, hence can be inferred by the current location relatively easily.
\end{itemize}
Note that these figures also show the results of KL-\DistP{}, which is not defined in this paper but is introduced in~\cite{Kawamoto:19:Allerton}.

Compared to the attribute of \male{}/\female{} (Fig.~\ref{fig:tupling:privacy}), the attributes \home{}/\outside{} (Fig.~\ref{fig:MH+HO:tupling:privacy}) and \north{}/\south{} (Fig.~\ref{fig:NS:tupling:privacy}) require more noise for distribution obfuscation.
This is because the distances between the two distributions of users of these attributes are larger than that of $\lambda_{\male}$ and $\lambda_{\female}$.
The histograms for \male{}/\female{} and for \north{}/\south{} are shown in 
Fig.~\ref{fig:plot-MF} and Fig.~\ref{fig:plot-NS} respectively.

We also show the comparison of different obfuscation mechanisms for various attributes in 
Fig.~\ref{fig:compare-mechanisms-MF+SI+WP+HO}.

\subsection{Experimental Comparison among Time Periods in Manhattan}
\label{sec:experiment-attributes:hours}

In this section we present the experimental results on different periods of time: 00h-05h, 06h-11h, 12h-17h, and 18h-23h in Manhattan in Figs.~\ref{fig:times+MH+MF:tupling:privacy} to~\ref{fig:times+MH+HO:tupling:privacy}.
It can be seen that $\varepsilon$ depends on a period of time.
For example, hiding the attribute \male{}/\female{} and \social{}/\lesssocial{} requires more noise in 06h-11h.
This might be related to the fact that more social people tend to attend social events in the morning.
On the other hand, the attribute  \home{}/\outside{} requires similar amount of noise in any period of time.
This implies that residential areas are distant from other areas in each city, and this fact does not change over time.

\subsection{Experimental Comparison among Various Cities}
\label{sec:experiment-attributes:cities}

In this section we compare the experimental results on five cities: Manhattan, eastern Tokyo, western Tokyo, London, and Paris.
In Table~\ref{tab:example-parameters} we show examples of parameters that achieve the same levels of \DistP{} in different cities.
More detailed comparison among those cities is shown 
in Fig.~\ref{fig:cities+MF:tupling:privacy} (\male{}/\female{}),
Fig.~\ref{fig:cities+SI:tupling:privacy} (\social{}/\lesssocial{}),
Fig.~\ref{fig:cities+WP:tupling:privacy} (\workplace{}/\nonworkplace{}), and
Fig.~\ref{fig:cities+HO:tupling:privacy} (\home{}/\allowbreak\outside{}).

When we use the same parameters in the tupling mechanism, the levels $\varepsilon$ of \DistP{} differ among those cities.
In western Tokyo, for instance, the attributes \social{}/\lesssocial{} and \workplace{}/\nonworkplace{} are more difficult to hide.
This implies that areas for social events and workplace might be more separated from the other areas in western Tokyo.
Since each attribute in each city may require different levels of noise for distribution obfuscation, the reference parameters (e.g., those shown in Table~\ref{tab:example-parameters}) would be useful to select appropriate parameters for the tupling mechanism to protect the privacy of attributes in different cities.

\subsection{Theoretical/Empirical Values of $\varepsilon$-\DistP{}}
\label{sub:theory-empirical}

In Table~\ref{table:bounds:theory-experiments}, we show the theoretical values of $\varepsilon$ calculated by Theorem~\ref{thm:TuplingDP} for $\delta = 0.001, \allowbreak 0.01, \allowbreak 0.1$.
Compared to the experimental results, those theoretical values can only give loose upper bounds on $\varepsilon$.
This is because the concentration inequality used to derive Theorem~\ref{thm:TuplingDP} give loose bounds.

\begin{figure}[t]
\begin{tabular}{cc}
\begin{minipage}{1.0\hsize}
\centering
\begin{subfigure}[t]{0.48\textwidth}\centering
  \mbox{\raisebox{-5pt}{\includegraphics[ width=0.5\textwidth]{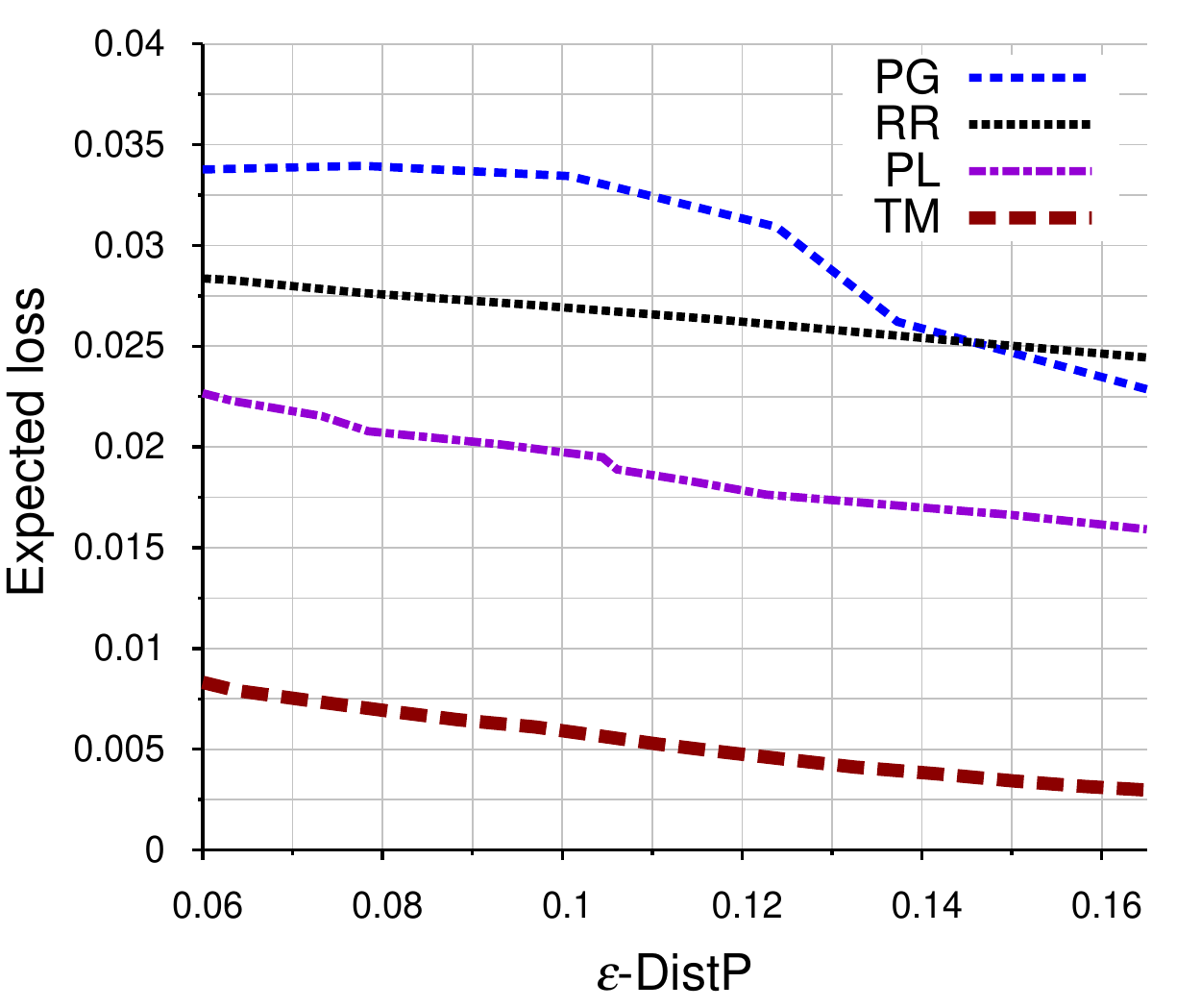}}}
\caption{Relationship between $(\varepsilon, \allowbreak 0.001)$-\DistP{} and expected loss for \male{}/\female{}.\label{fig:compare:delta0.001:MF}}
\end{subfigure}\hspace{0.3ex}\hfill
\begin{subfigure}[t]{0.48\textwidth}\centering
  \mbox{\raisebox{-5pt}{\includegraphics[ width=0.5\textwidth]{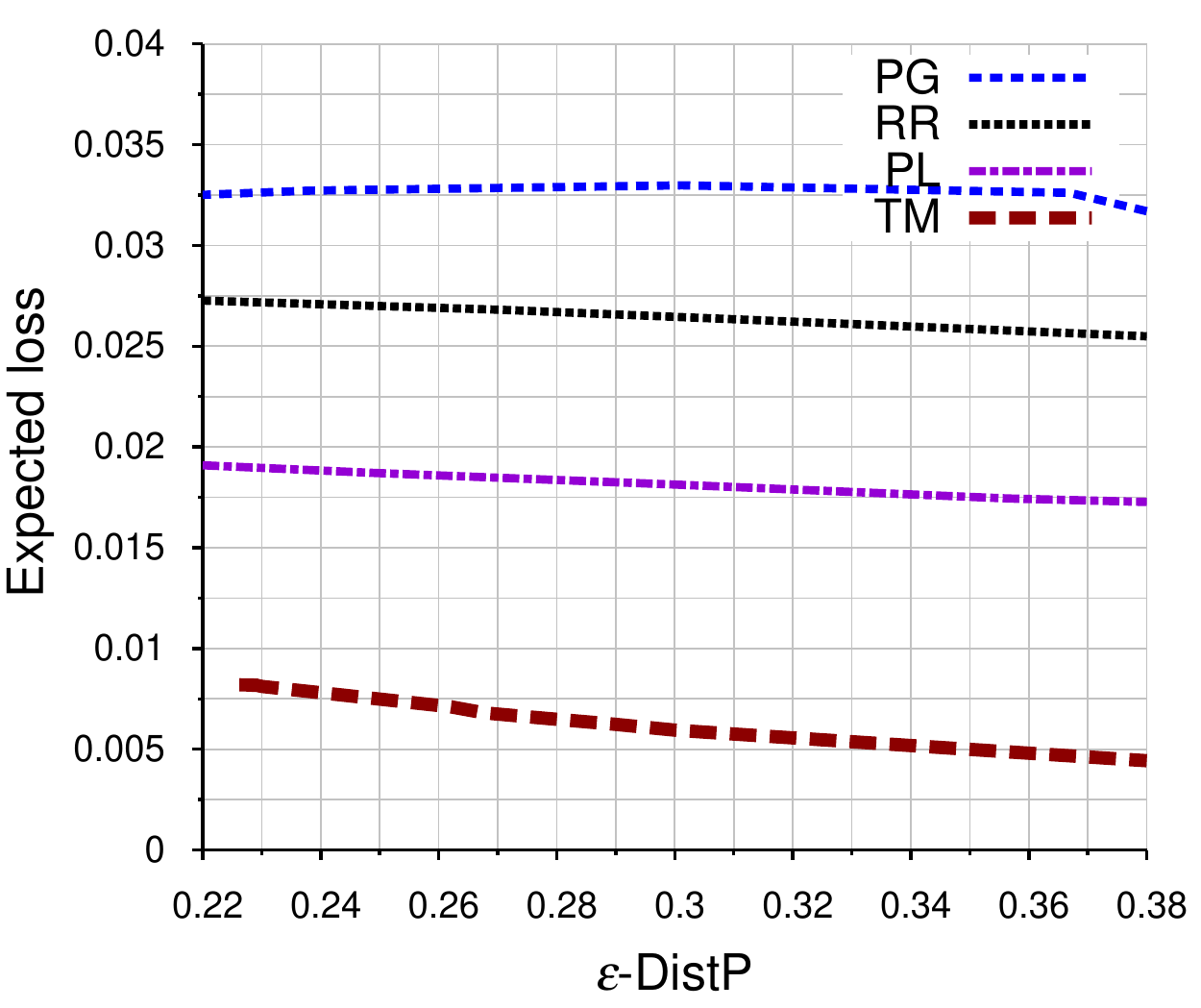}}}
\caption{Relationship between $(\varepsilon, \allowbreak 0.001)$-\DistP{} and expected loss for \social{}/\lesssocial{}.\label{fig:compare:delta0.001:SI}}
\end{subfigure}\hspace{0.3ex}\hfill
\vspace{3mm}
\\
\begin{subfigure}[t]{0.48\textwidth}\centering
  \mbox{\raisebox{-10pt}{\includegraphics[ width=0.5\textwidth]{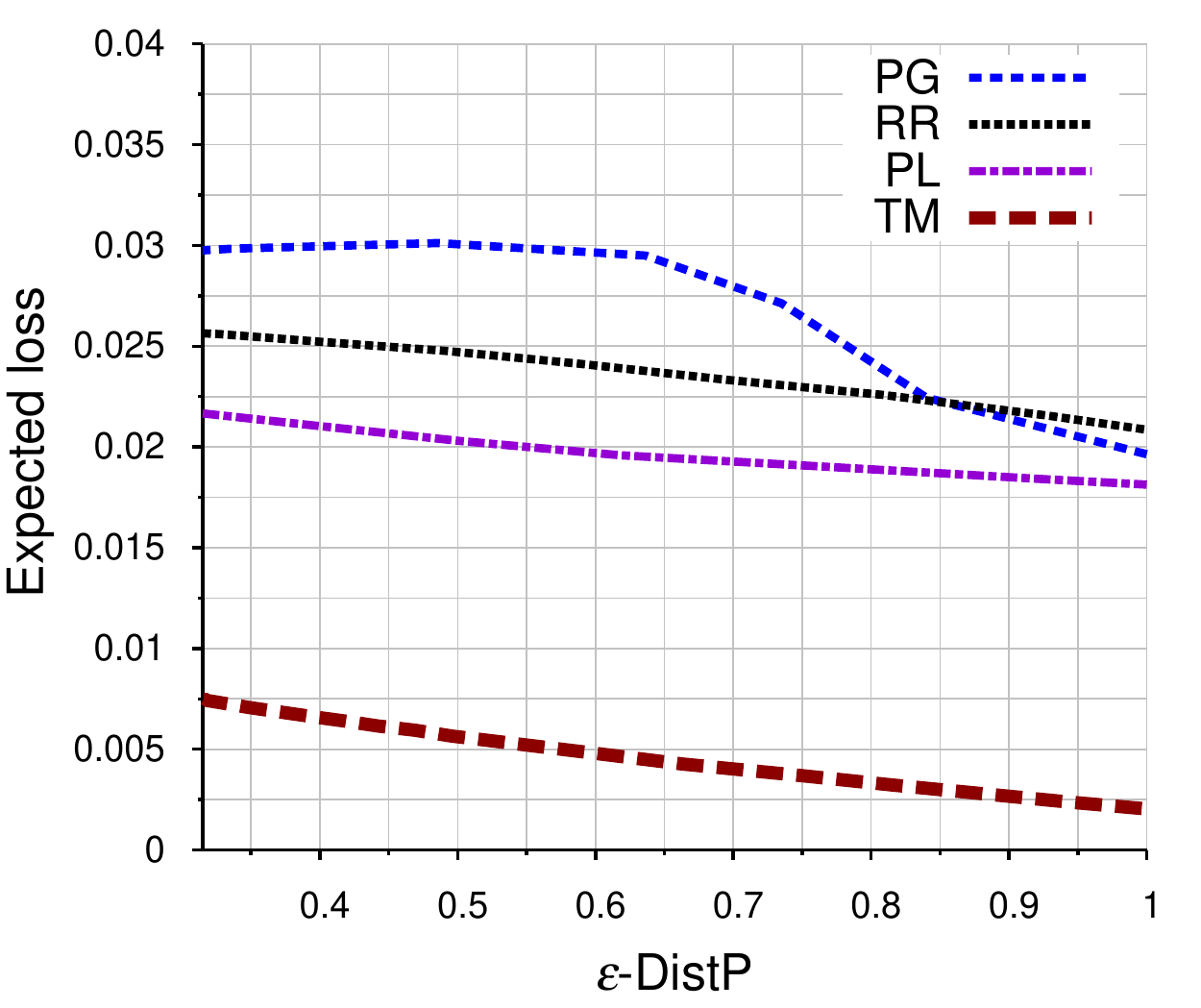}}}
\caption{Relationship between $(\varepsilon, \allowbreak 0.001)$-\DistP{} and loss for \workplace{}/\nonworkplace{}.\label{fig:compare:delta0.001:WP}}
\end{subfigure}\hspace{0.5ex}\hfill
\begin{subfigure}[t]{0.48\textwidth}\centering
  \mbox{\raisebox{-10pt}{\includegraphics[ width=0.5\textwidth]{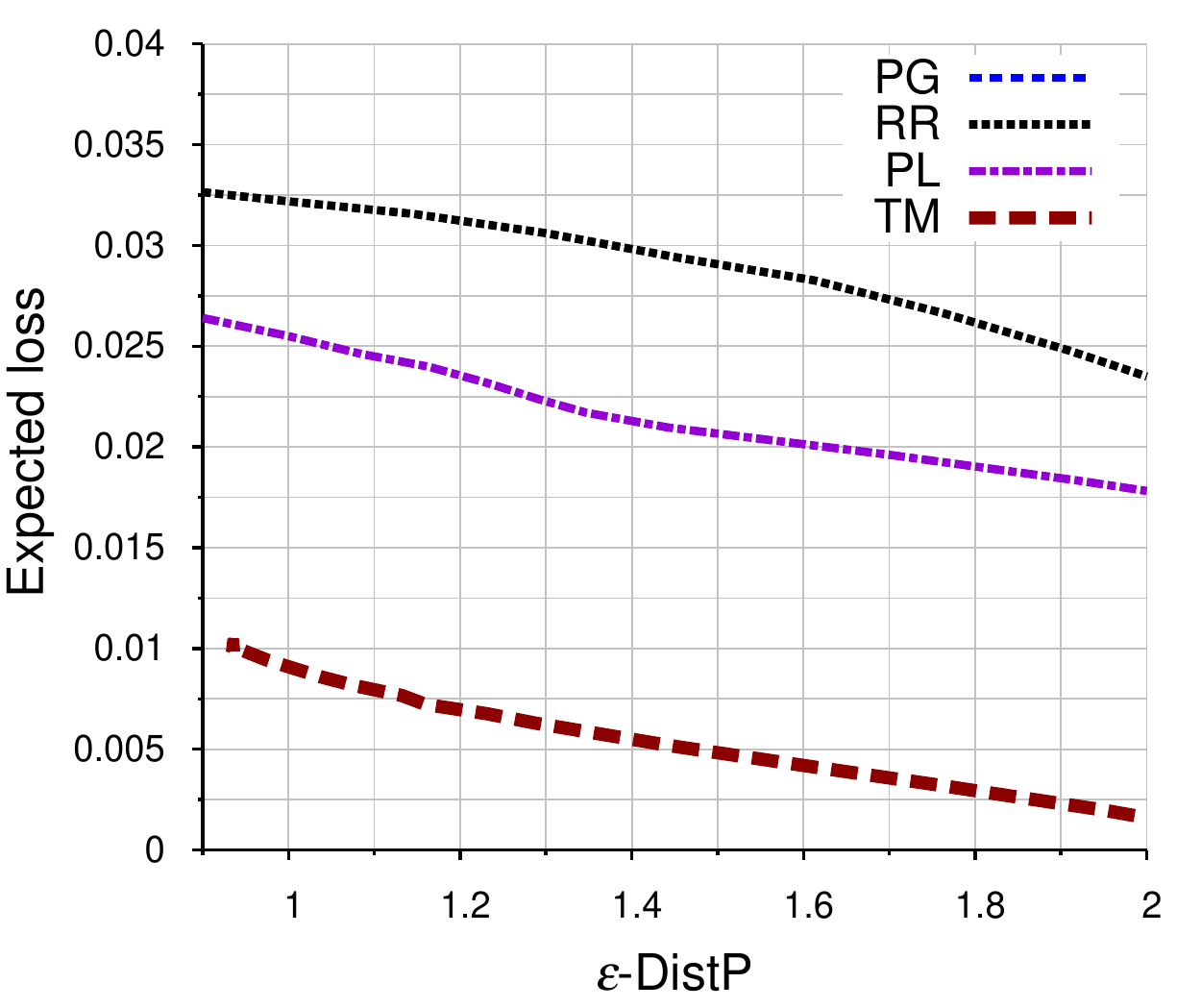}}}
\caption{Relationship between $(\varepsilon, \allowbreak 0.001)$-\DistP{} and expected loss for \home{}/\outside{}.\label{fig:compare:delta0.001:HO}}
\end{subfigure}\hspace{0.5ex}\hfill
\caption{Comparison of the randomized response (RR), the planar Laplace mechanism (PL), the planar Gaussian mechanism (PG), and the tupling mechanism (TM) $\TPM{}$ with $k = 10$ dummies and a radius $r = 0.020$. The experiments are performed for the location data in Manhattan.
\label{fig:compare-mechanisms-MF+SI+WP+HO}}
\vspace{3mm}
\end{minipage}
\end{tabular}
\end{figure}

\begin{figure}[H]
\begin{tabular}{c}
\begin{minipage}{1.0\hsize}
\centering
\makeatletter
\def\@captype{table}
\caption{Theoretical/empirical $\varepsilon$-\DistP{} of $\TPM{}$ ($k = 10$, $\varepsilon_{\!\alg} = 10$, $r = 0.020$). \label{table:bounds:theory-experiments}}
\vspace{0ex}
\renewcommand{\arraystretch}{1}
\makeatother\begin{tabular}{lcccc}
\hline
& $\delta = 0.001$ & $\delta = 0.01$ & $\delta = 0.1$ 
\\ \hline
Theoretical bounds &
2.170 & 1.625 & 1.140 
\\ \hline
Empirical values &
0.04450 & 0.03534 & 0.02295 
\\ \hline
\end{tabular}
\vspace{10mm}
\end{minipage}
\\
\begin{minipage}{1.0\hsize}
\centering
\makeatletter
\def\@captype{table}
\caption{The number $k$ of dummies required for achieving \DistP{} in different cities (MH = Manhattan, TKE = Tokyo (east), TKW = Tokyo (west), LD = London, PR = Paris) when $\varepsilon_{\alg} = 100$ and $r = 0.020$. 
Note that the data of Paris for \male{}/\female{} are excluded because of the insufficient sample size.
\label{tab:example-parameters}}
\makeatother\begin{tabular}{l c c c c c c c}
\cmidrule[0.1em]{2-6}
& {\bf MH} & {\bf TKE} & {\bf TKW} & {\bf LD} & {\bf PR} \\[0.1ex]
\midrule
$(0.25, 0.001)$-\DistP{} for {\bf male\,/\,female} & 2 &  $>$20 & 5 & 10 & ---  \\[0.1ex]
$(0.50, 0.001)$-\DistP{} for {\bf social\,/\,less social}\! & 2  &  3 & $>$20 & 2 & 3 \\[0.1ex]
$(1.00, 0.001)$-\DistP{} for {\bf work\,/\,non-work}\! & 2  &  2  & $>$20 & 1 & 2  \\[0.1ex]
$(1.50, 0.001)$-\DistP{} for {\bf home\,/\,outside} & 3  &  5 & $>$20 & $>$20 & 4  \\[0.1ex]
\bottomrule
\end{tabular}
\vspace{10mm}
\end{minipage}
\\
\begin{tabular}{cc}
\begin{minipage}{0.49\hsize}
\centering
\begin{subfigure}[t]{0.48\textwidth}
  \mbox{\raisebox{-5pt}{\includegraphics[width=1.05\textwidth]{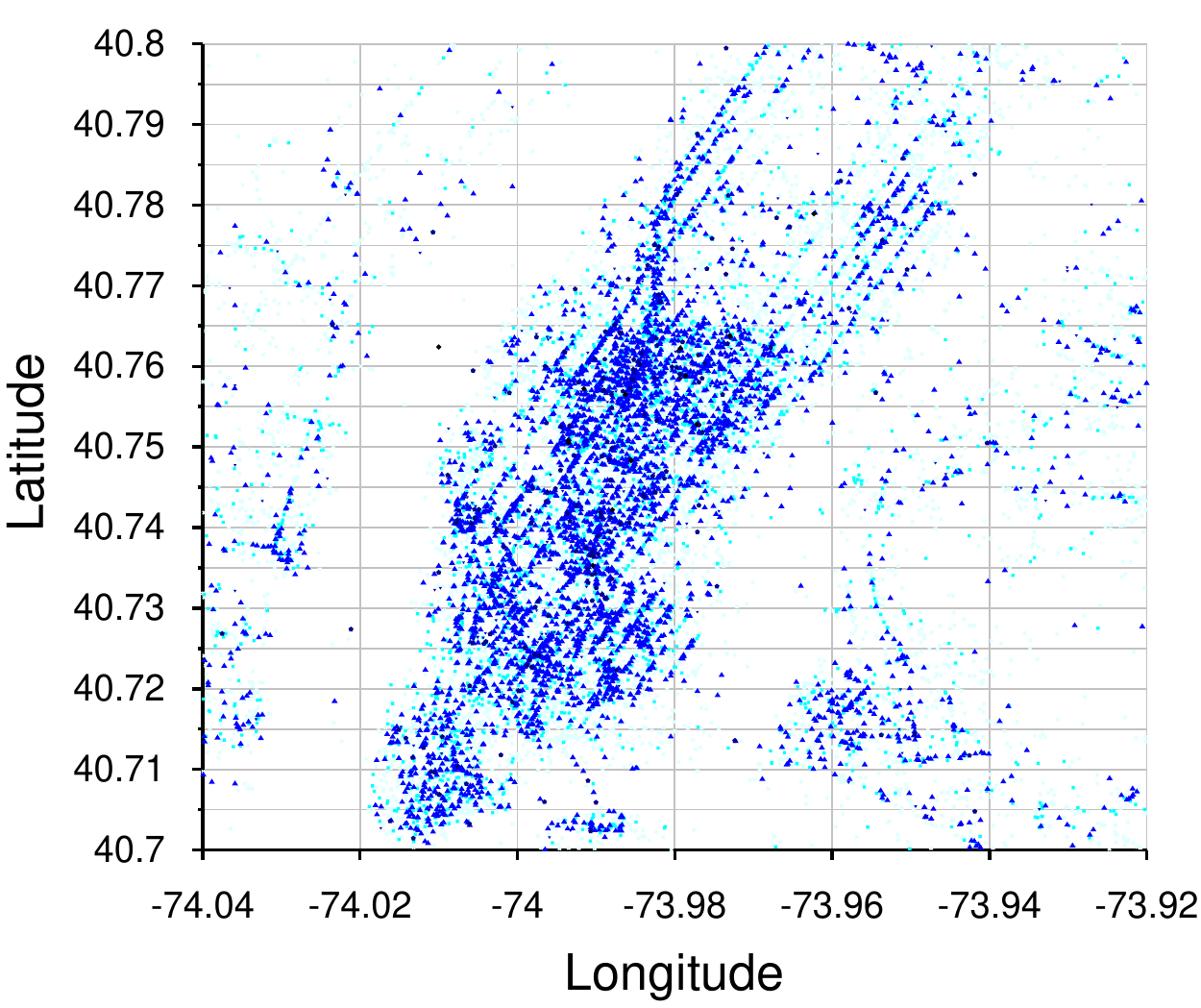}}}
\caption{The histogram of locations of the male users in Manhattan.
\label{fig:plot-M}}
\end{subfigure}\hspace{0.3ex}\hfill
\begin{subfigure}[t]{0.48\textwidth}
  \mbox{\raisebox{-5pt}{\includegraphics[width=1.05\textwidth]{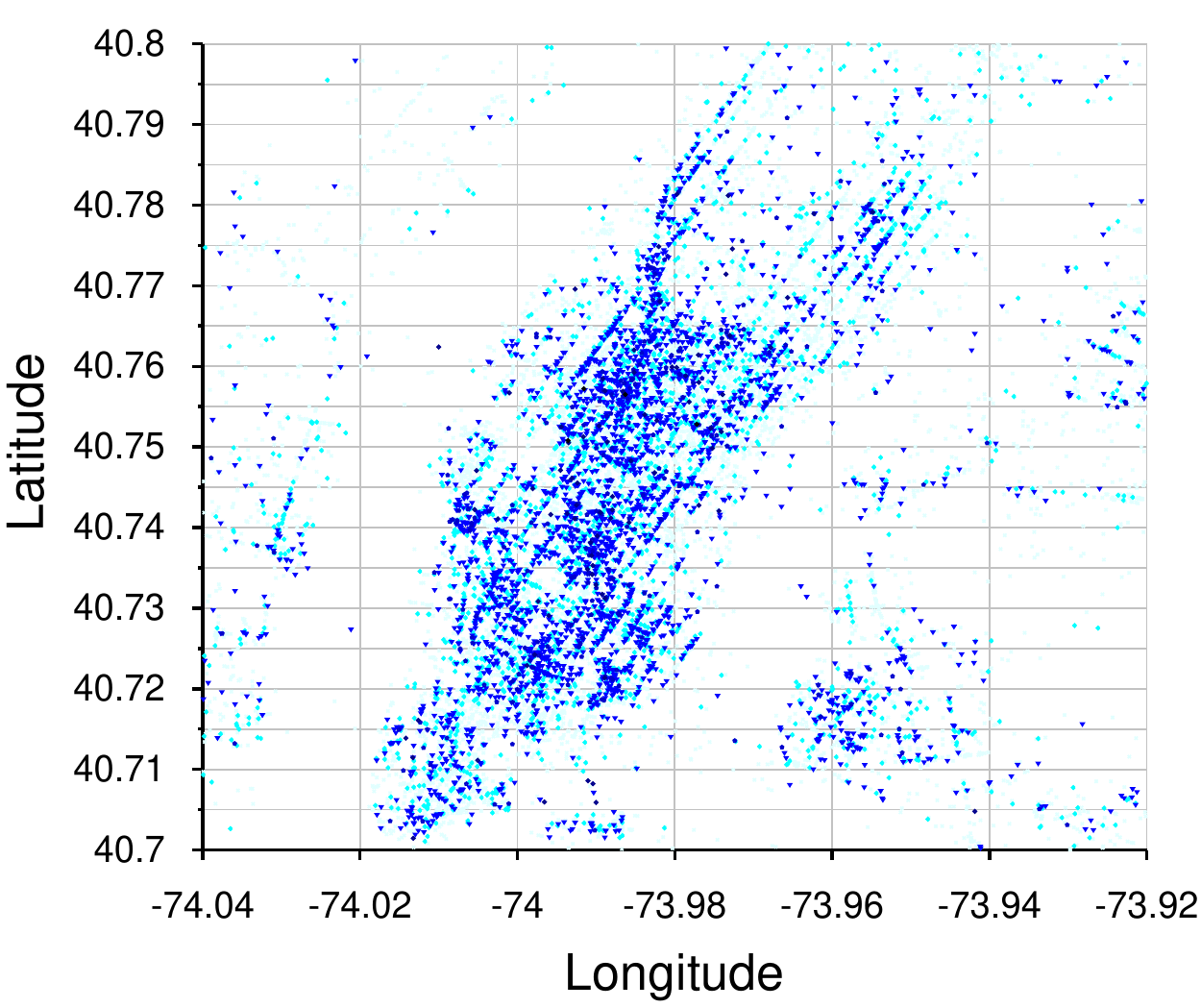}}}
\caption{The histogram of locations of the female users in Manhattan.
\label{fig:plot-F}}
\end{subfigure}\hspace{0.6ex}\hfill
\caption{Histograms of user locations used in experiments. Darker colors represent larger population.
\label{fig:plot-MF}}
\vspace{14mm}
\end{minipage}
~
\begin{minipage}{0.49\hsize}
\centering
\begin{subfigure}[t]{0.48\textwidth}
  \mbox{\raisebox{-5pt}{\includegraphics[width=1.05\textwidth]{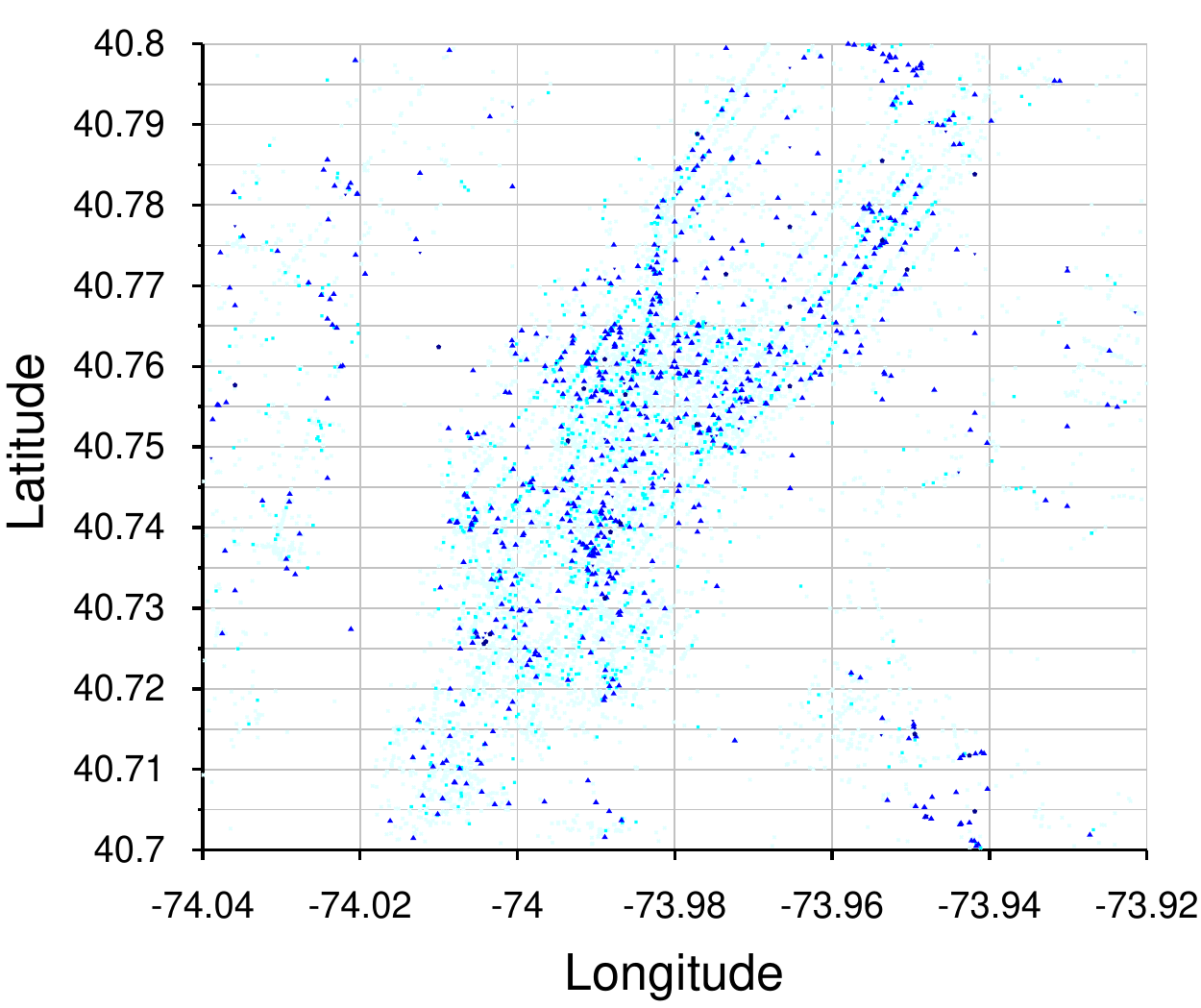}}}
\caption{The histogram of locations of the users living in northern Manhattan.
\label{fig:plot-N}}
\end{subfigure}\hspace{0.3ex}\hfill
\begin{subfigure}[t]{0.48\textwidth}
  \mbox{\raisebox{-5pt}{\includegraphics[width=1.05\textwidth]{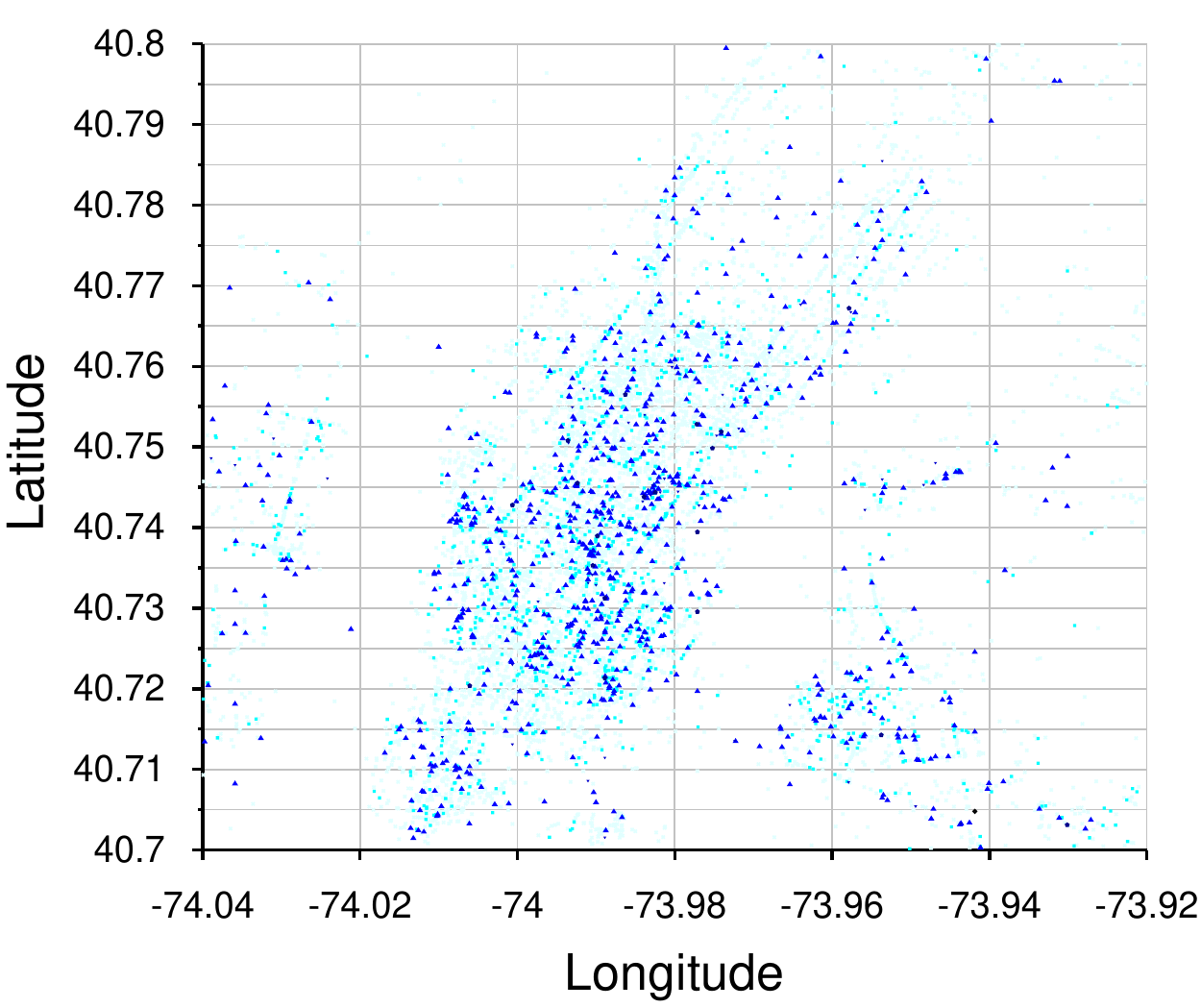}}}
\caption{The histogram of locations of~the users living in southern Manhattan.
\label{fig:plot-S}}
\end{subfigure}\hspace{0.3ex}\hfill
\caption{Histograms of user locations used in experiments. Darker colors represent larger population.\label{fig:plot-NS}}
\vspace{14mm}
\end{minipage}
\end{tabular}
\end{tabular}
\end{figure}

\begin{figure}[H]
\begin{tabular}{c}
\begin{minipage}{1.0\hsize}
\centering
\begin{subfigure}[t]{0.24\textwidth}
  \mbox{\raisebox{-11pt}{\includegraphics[height=24mm, width=1.00\textwidth]{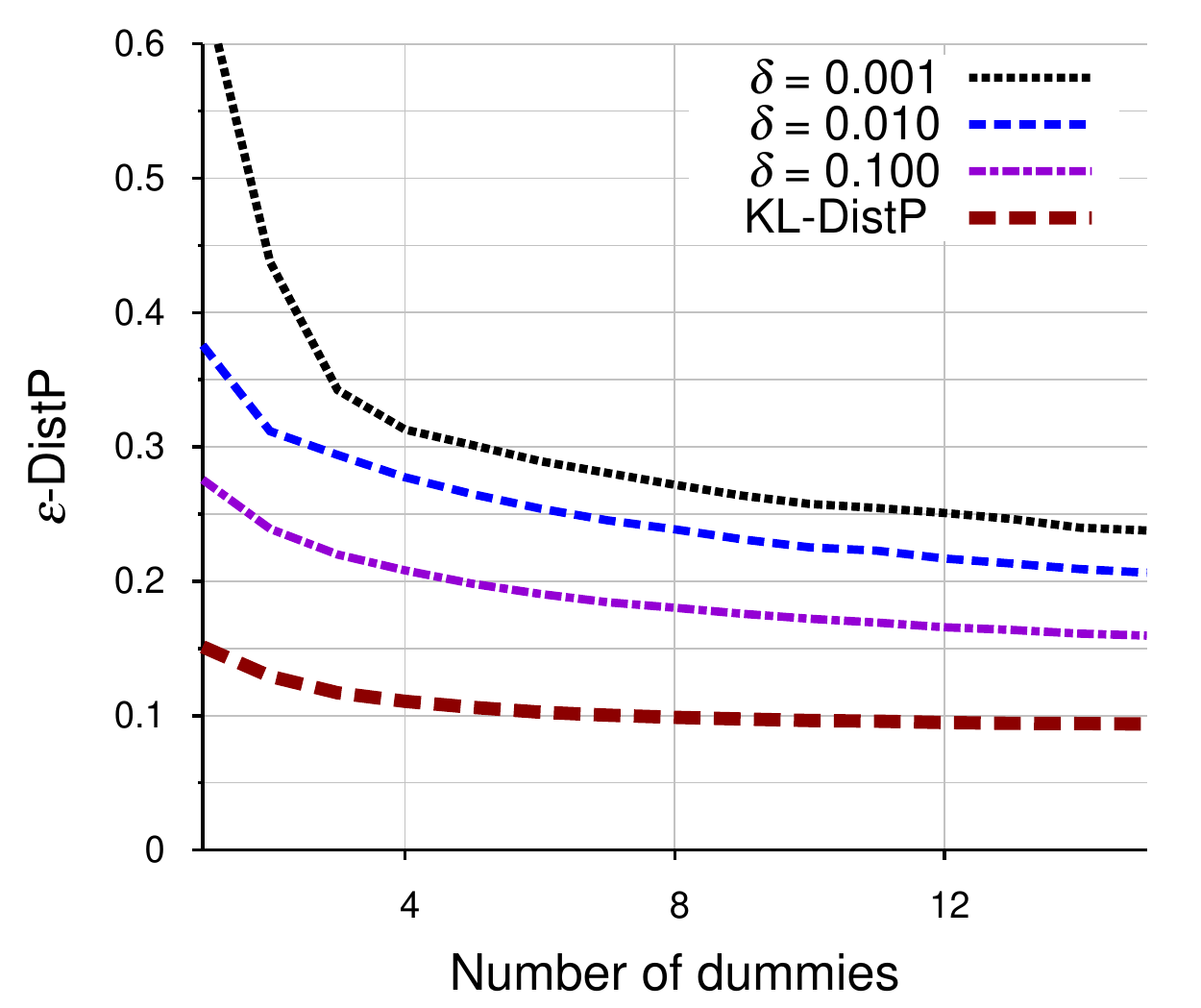}}}
\caption{Relationship between $\varepsilon$-\DistP{} and \#dummies (when using $(100, 0.020)$-\RL{} mechanism).\label{fig:MH+SI:dummies:DistP}}
\end{subfigure}\hspace{0.4ex}\hfill
\begin{subfigure}[t]{0.24\textwidth}
  \mbox{\raisebox{-11pt}{\includegraphics[height=24mm, width=1.00\textwidth]{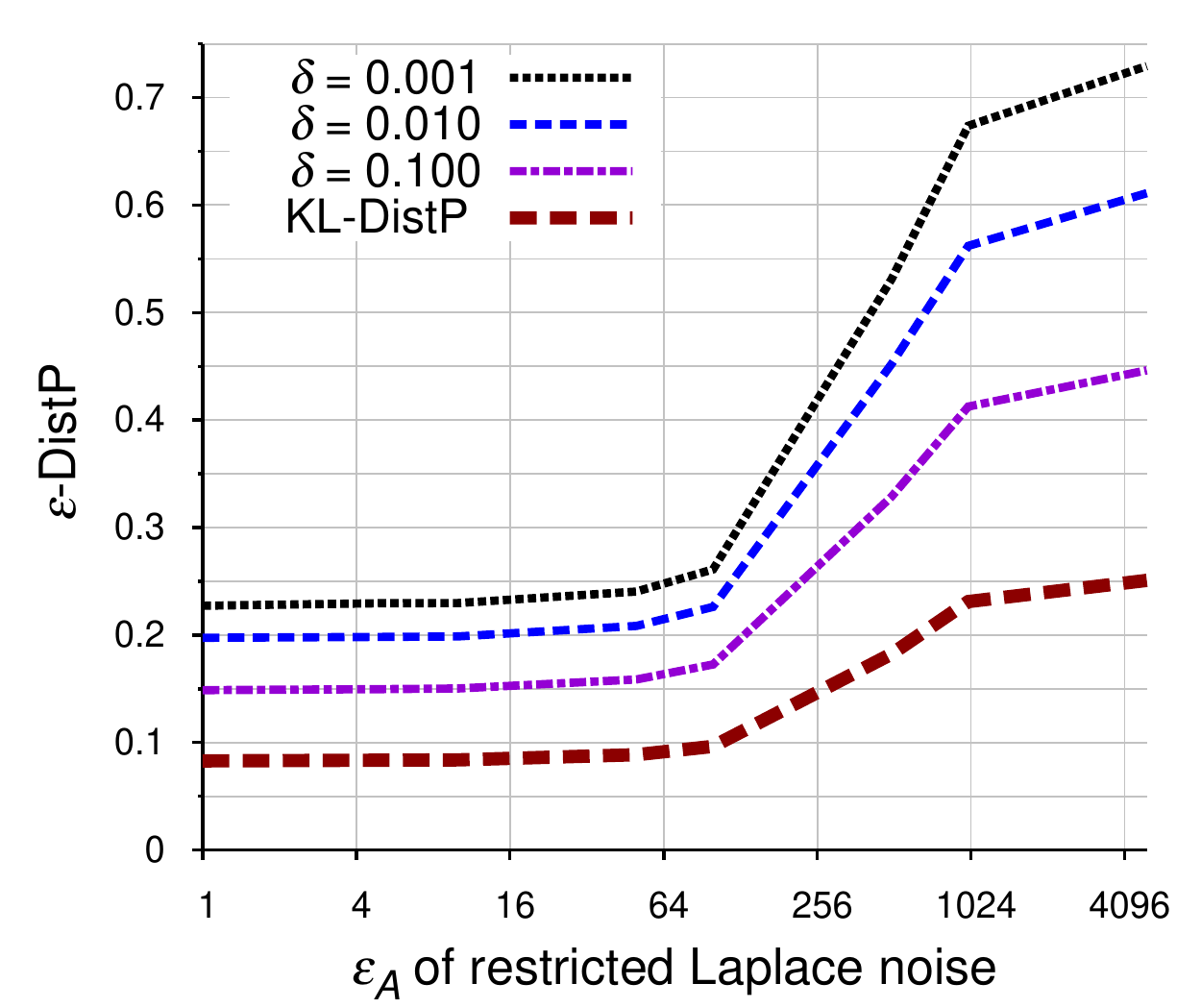}}}
\caption{Relationship between $\varepsilon$-\DistP{} and $\varepsilon_\alg$ of $(\varepsilon_\alg, 0.020)$-\RL{} mechanism (with $10$ dummies).\label{fig:MH+SI:DP-noises:DistP}}
\end{subfigure}\hspace{0.4ex}\hfill
\begin{subfigure}[t]{0.24\textwidth}
  \mbox{\raisebox{-11pt}{\includegraphics[height=24mm, width=1.00\textwidth]{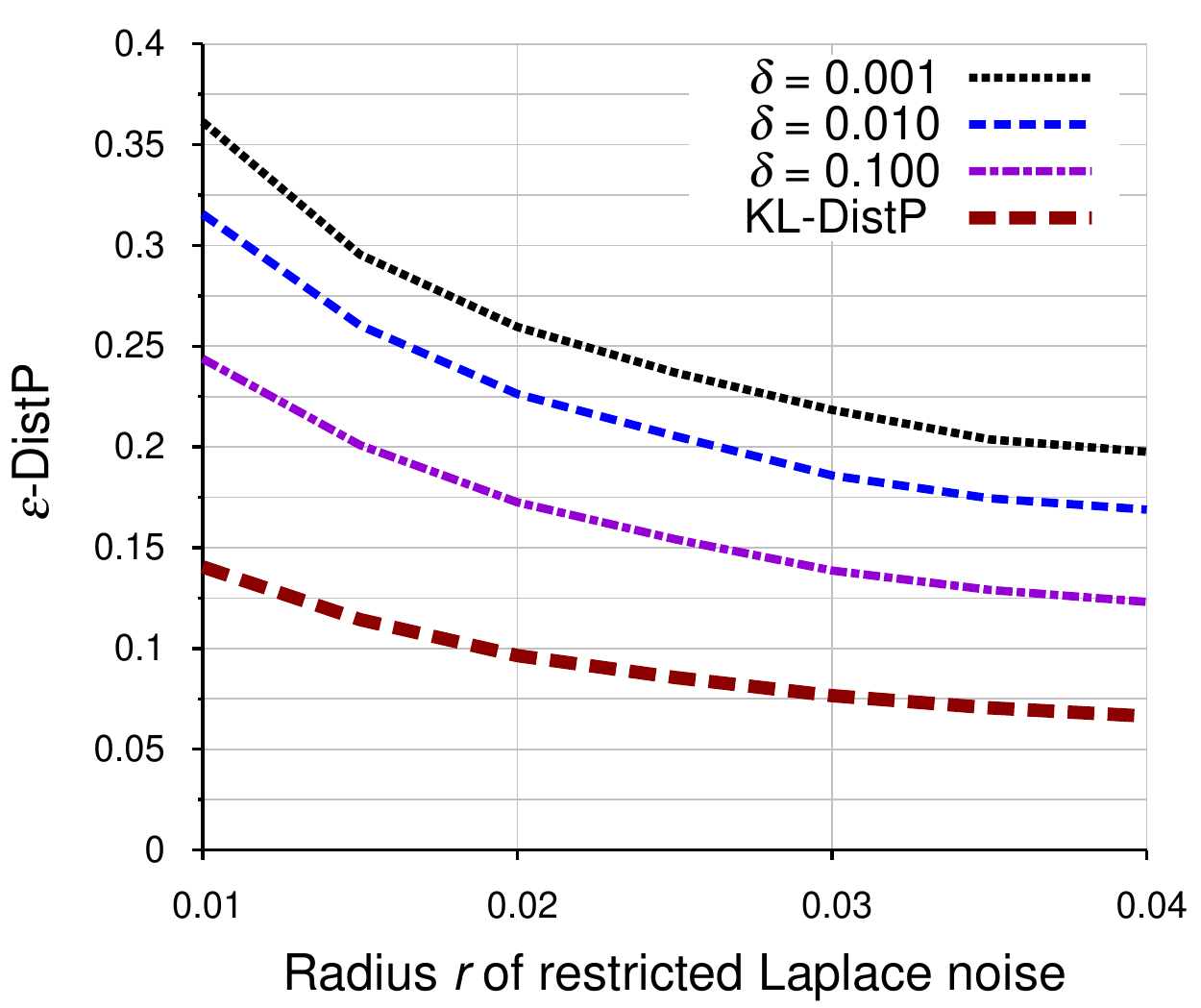}}}
\caption{Relationship between $\varepsilon$-\DistP{} and a radius $r$ of $(100, r)$-\RL{} mechanism (with $10$ dummies).\label{fig:MH+SI:radius:DistP}}
\end{subfigure}\hspace{0.4ex}\hfill
\begin{subfigure}[t]{0.24\textwidth}
  \mbox{\raisebox{-11pt}{\includegraphics[height=24mm, width=1.00\textwidth]{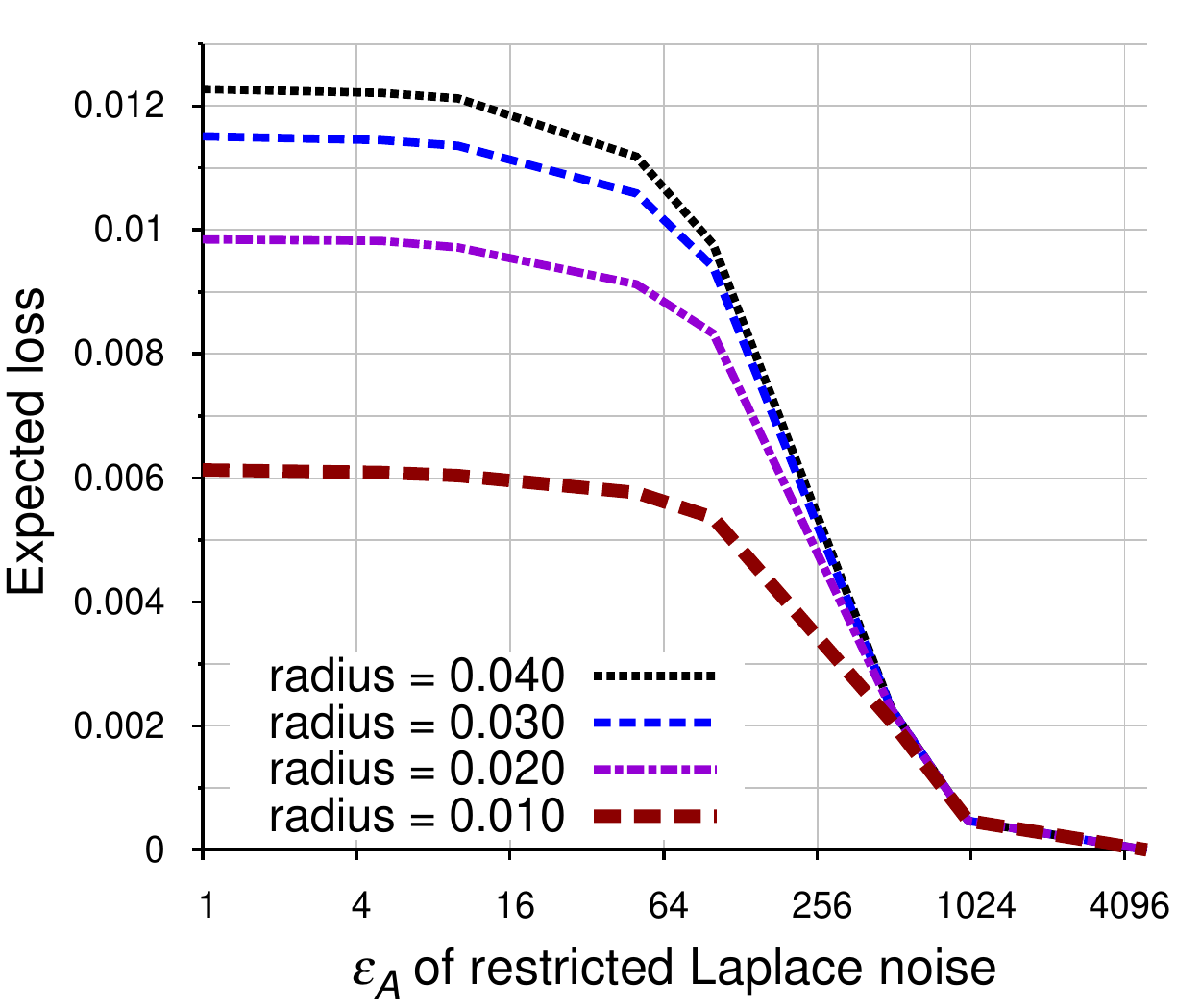}}}
\caption{Relationship between the expected loss and~$\varepsilon_\alg$~of $(\varepsilon_\alg, \allowbreak r)$-\RL{} mechanism (with $5$ dummies).
\label{fig:MH+S:tupling:loss}}
\end{subfigure}
\caption{Empirical \DistP{} and loss for attribute \social{}/\lesssocial{} in Manhattan.
\label{fig:MH+SI:tupling:privacy}}
\vspace{5mm}
\end{minipage}
\\
\begin{minipage}{1.0\hsize}
\centering
\begin{subfigure}[t]{0.24\textwidth}
  \mbox{\raisebox{-11pt}{\includegraphics[height=24mm, width=1.00\textwidth]{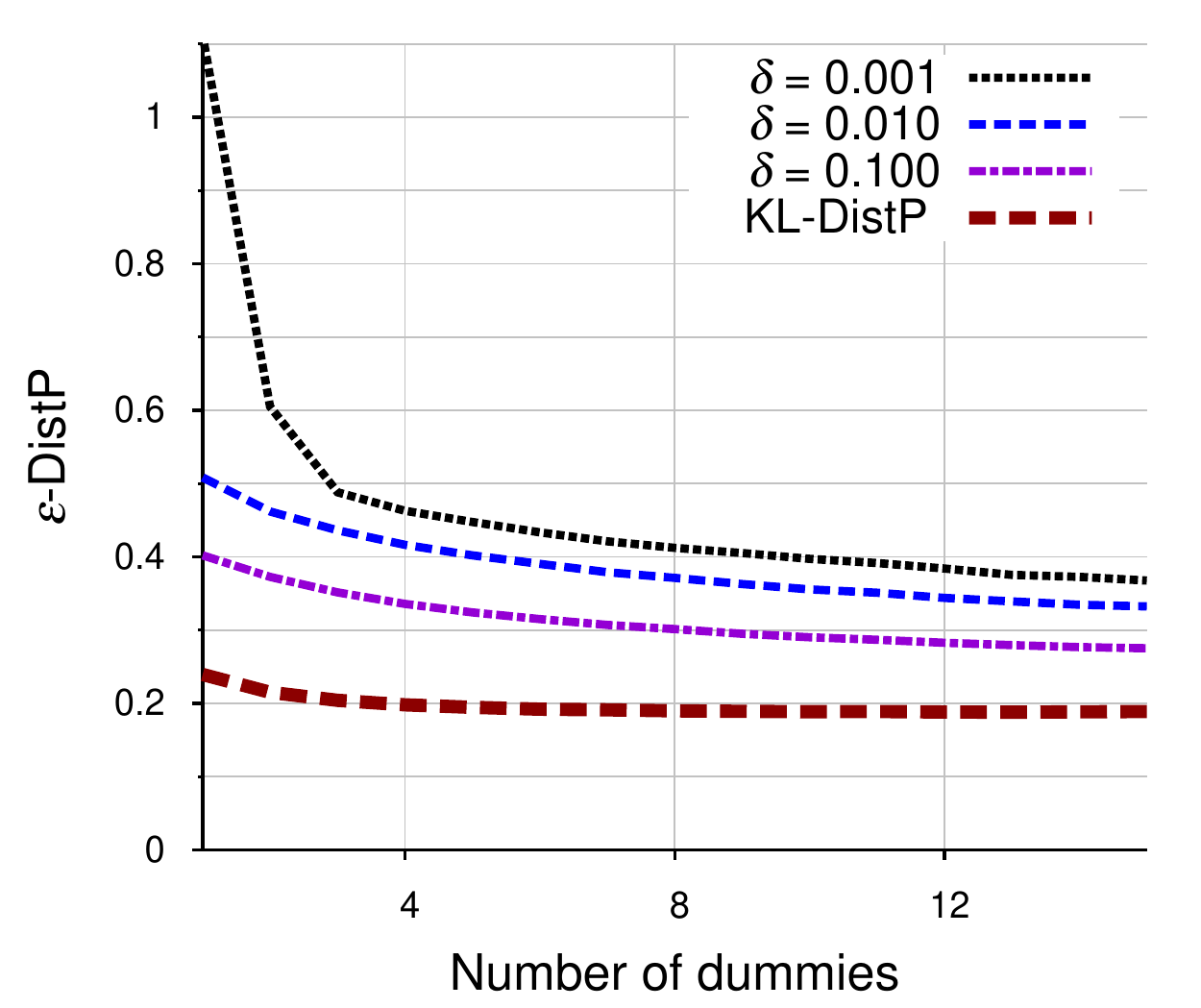}}}
\caption{Relationship between $\varepsilon$-\DistP{} and \#dummies (when using $(100, 0.020)$-\RL{} mechanism).\label{fig:MH+WP:dummies:DistP}}
\end{subfigure}\hspace{0.4ex}\hfill
\begin{subfigure}[t]{0.24\textwidth}
  \mbox{\raisebox{-11pt}{\includegraphics[height=24mm, width=1.00\textwidth]{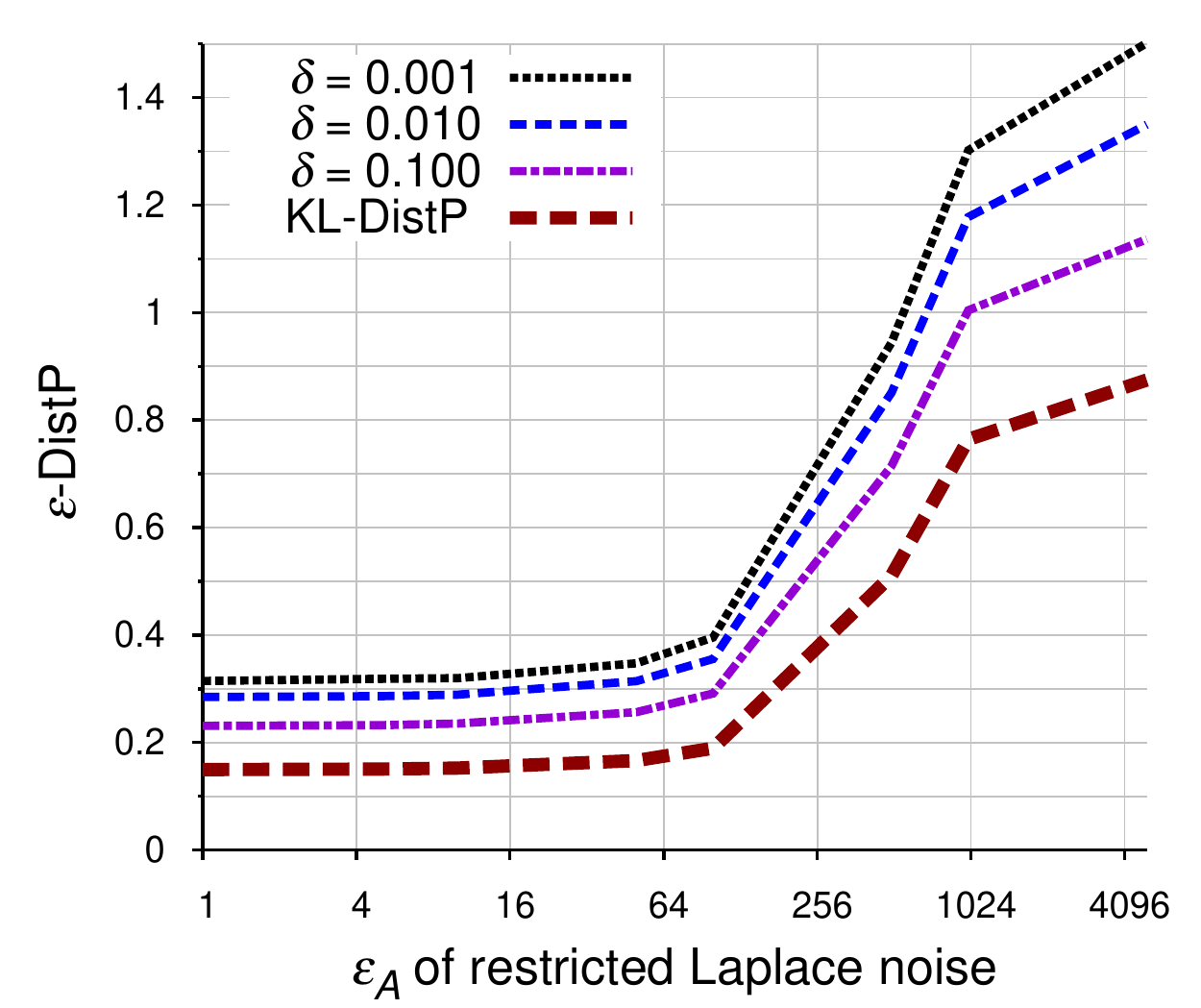}}}
\caption{Relationship between $\varepsilon$-\DistP{} and $\varepsilon_\alg$ of $(\varepsilon_\alg, 0.020)$-\RL{} mechanism (with $10$ dummies).\label{fig:MH+WP:DP-noises:DistP}}
\end{subfigure}\hspace{0.4ex}\hfill
\begin{subfigure}[t]{0.24\textwidth}
  \mbox{\raisebox{-11pt}{\includegraphics[height=24mm, width=1.00\textwidth]{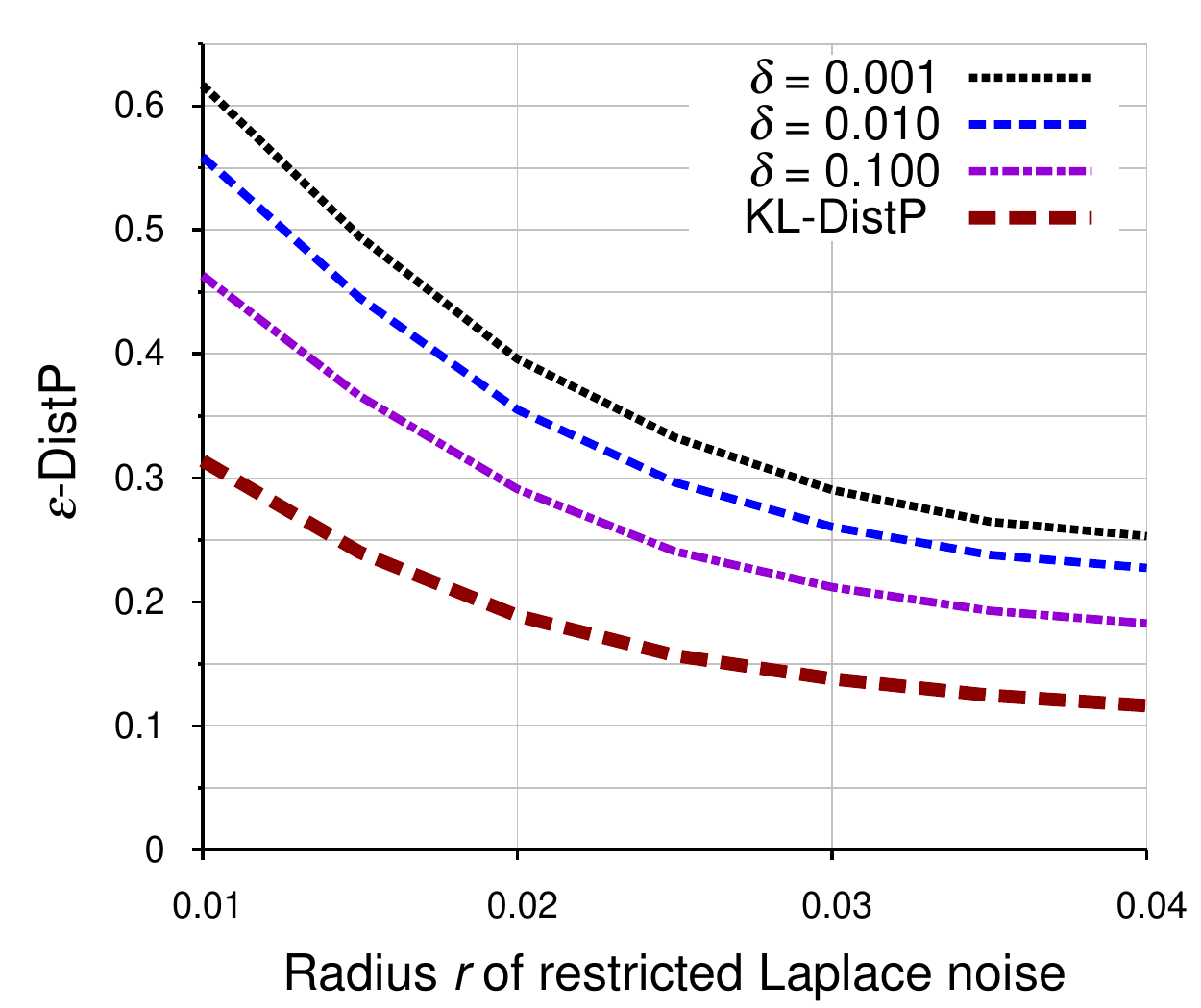}}}
\caption{Relationship between $\varepsilon$-\DistP{} and a radius $r$ of $(100, r)$-\RL{} mechanism (with $10$ dummies).\label{fig:MH+WP:radius:DistP}}
\end{subfigure}\hspace{0.4ex}\hfill
\begin{subfigure}[t]{0.24\textwidth}
  \mbox{\raisebox{-11pt}{\includegraphics[height=24mm, width=1.00\textwidth]{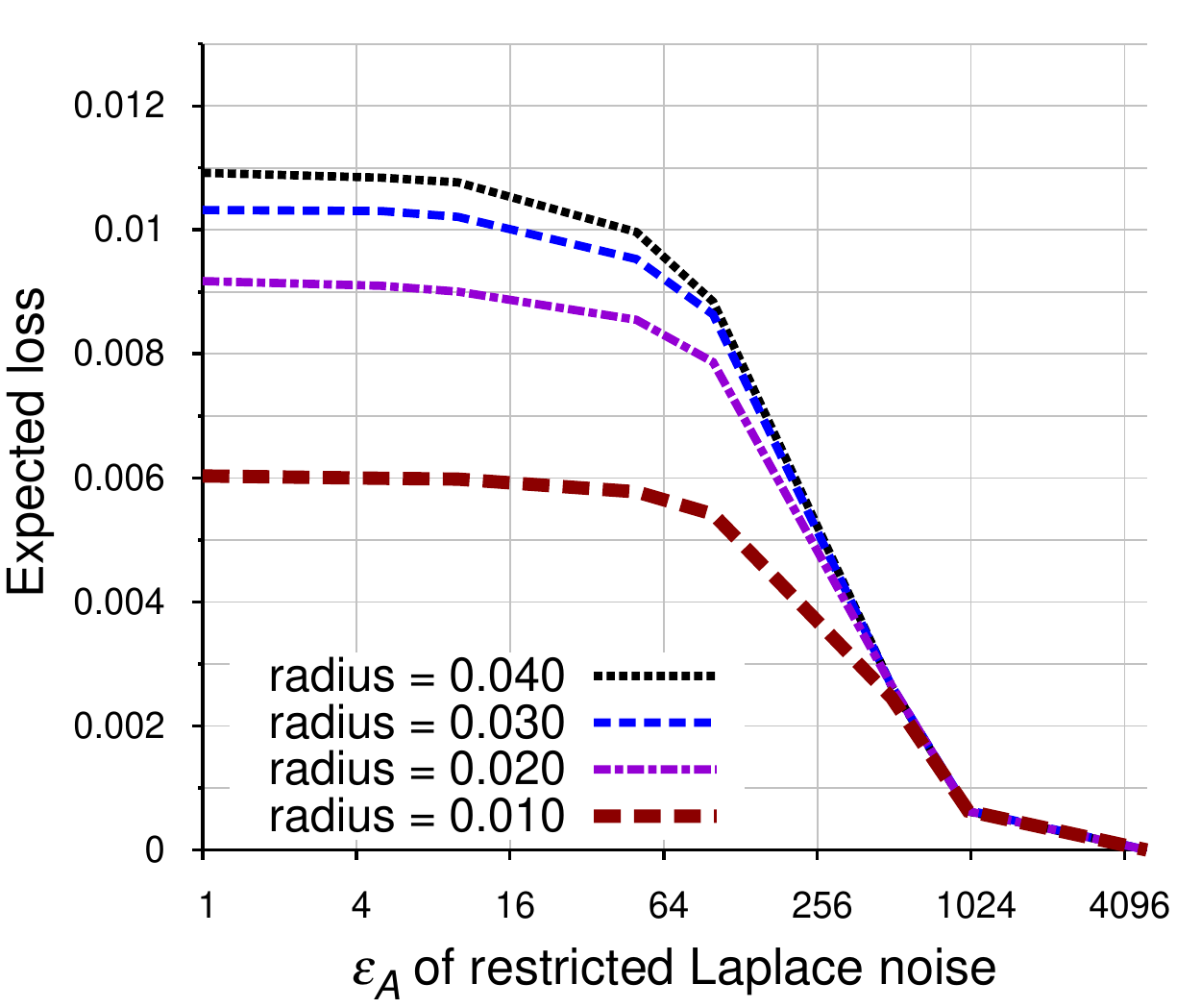}}}
\caption{Relationship between the expected loss and $\varepsilon_\alg$ of $(\varepsilon_\alg, \allowbreak r)$-\RL{} mechanism (with $5$ dummies).
\label{fig:MH+W:tupling:loss}}
\end{subfigure}
\caption{Empirical \DistP{} and loss for \workplace{}/\nonworkplace{} in Manhattan.
\label{fig:MH+WP:tupling:privacy}}
\vspace{5mm}
\end{minipage}
\\
\begin{minipage}{1.0\hsize}
\centering
\begin{subfigure}[t]{0.24\textwidth}
  \mbox{\raisebox{-11pt}{\includegraphics[height=24mm, width=1.00\textwidth]{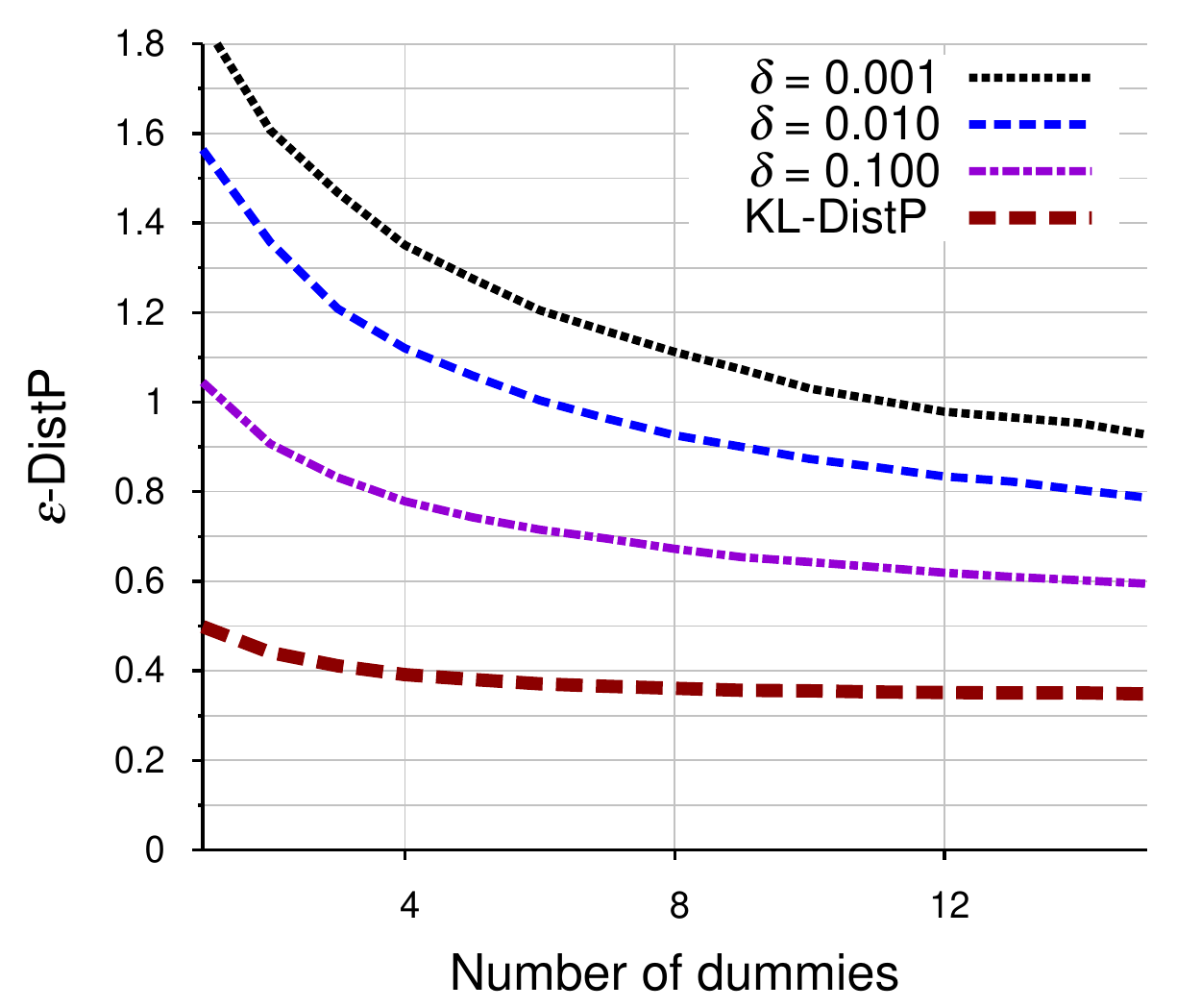}}}
\caption{Relationship between $\varepsilon$-\DistP{} and \#dummies (when using $(100, 0.020)$-\RL{} mechanism).\label{fig:MH+HO:dummies:DistP}}
\end{subfigure}\hspace{0.4ex}\hfill
\begin{subfigure}[t]{0.24\textwidth}
  \mbox{\raisebox{-11pt}{\includegraphics[height=24mm, width=1.00\textwidth]{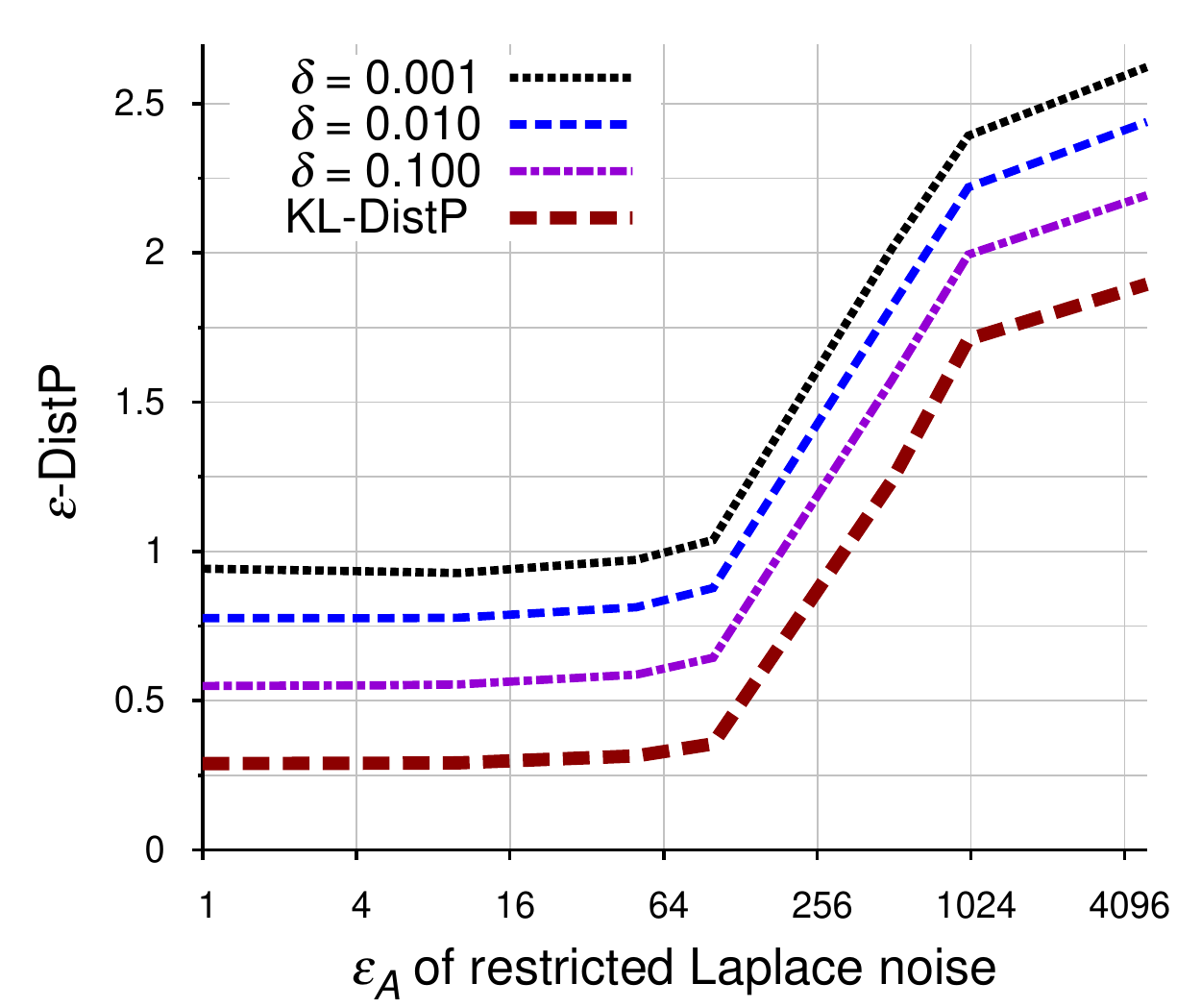}}}
\caption{Relationship between $\varepsilon$-\DistP{} and $\varepsilon_\alg$ of $(\varepsilon_\alg, 0.020)$-\RL{} mechanism (with $10$ dummies).\label{fig:MH+HO:DP-noises:DistP}}
\end{subfigure}\hspace{0.4ex}\hfill
\begin{subfigure}[t]{0.24\textwidth}
  \mbox{\raisebox{-11pt}{\includegraphics[height=24mm, width=1.00\textwidth]{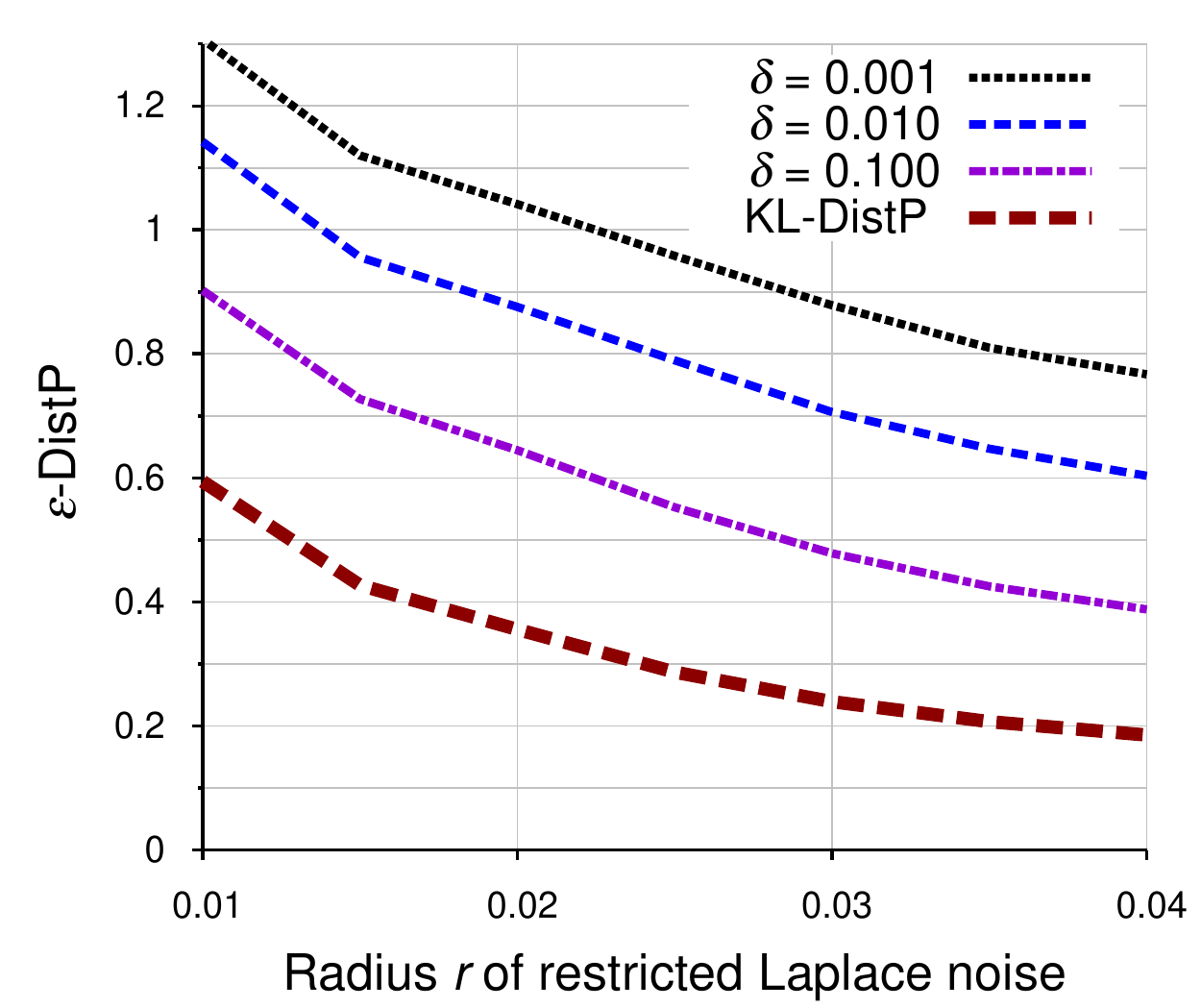}}}
\caption{Relationship between $\varepsilon$-\DistP{} and a radius $r$ of $(100, r)$-\RL{} mechanism (with $10$ dummies).\label{fig:MH+HO:radius:DistP}}
\end{subfigure}\hspace{0.4ex}\hfill
\begin{subfigure}[t]{0.24\textwidth}
  \mbox{\raisebox{-11pt}{\includegraphics[height=24mm, width=1.00\textwidth]{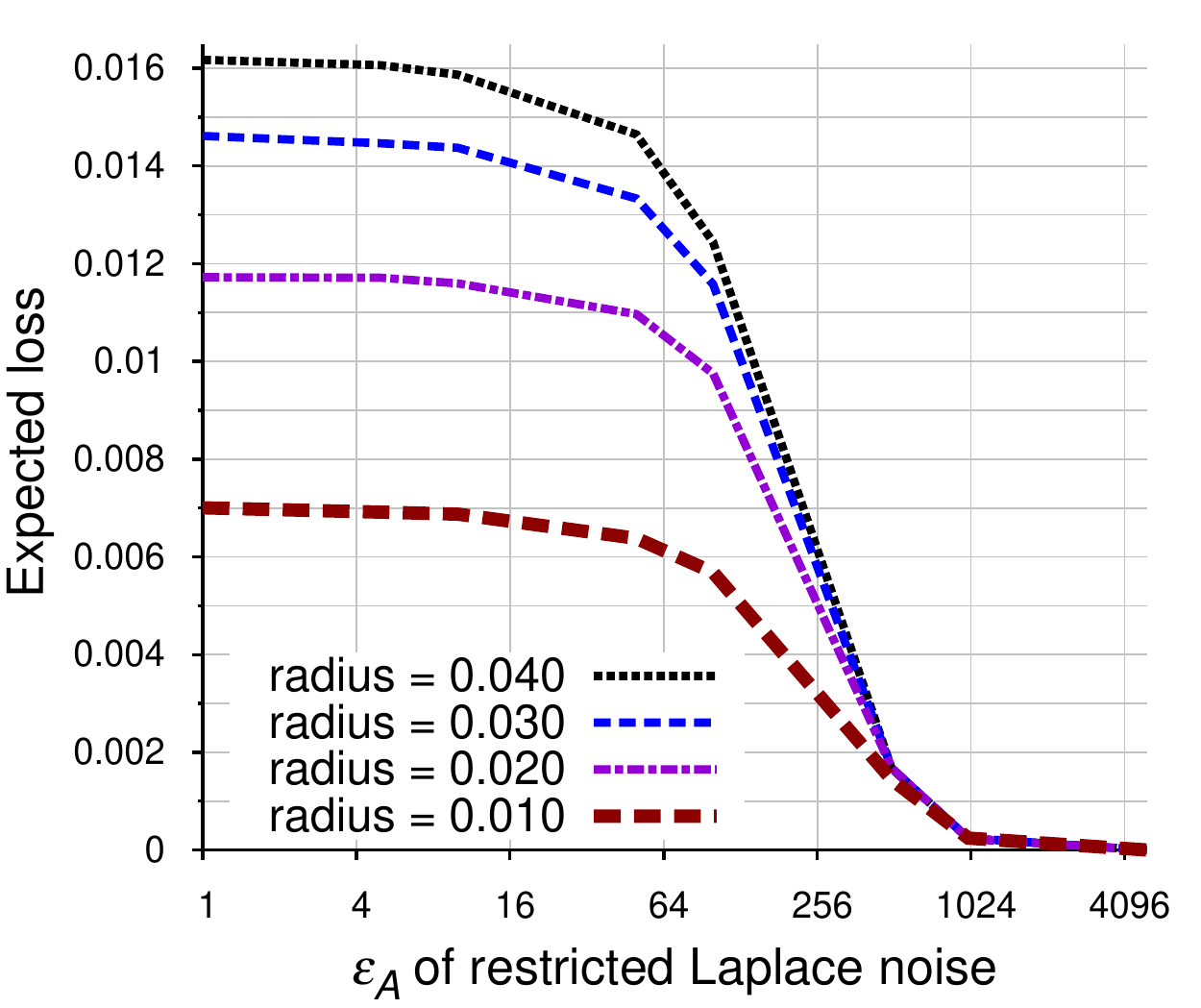}}}
\caption{Relationship between the expected loss and $\varepsilon_\alg$ of $(\varepsilon_\alg, \allowbreak r)$-\RL{} mechanism (with $5$ dummies).
\label{fig:MH+H:tupling:loss}}
\end{subfigure}
\caption{Empirical \DistP{} and loss for the attribute \home{}/\outside{} in Manhattan.
\label{fig:MH+HO:tupling:privacy}}
\vspace{5mm}
\end{minipage}
\\
\begin{minipage}{1.0\hsize}
\centering
\begin{subfigure}[t]{0.24\textwidth}
  \mbox{\raisebox{-11pt}{\includegraphics[height=24mm, width=1.00\textwidth]{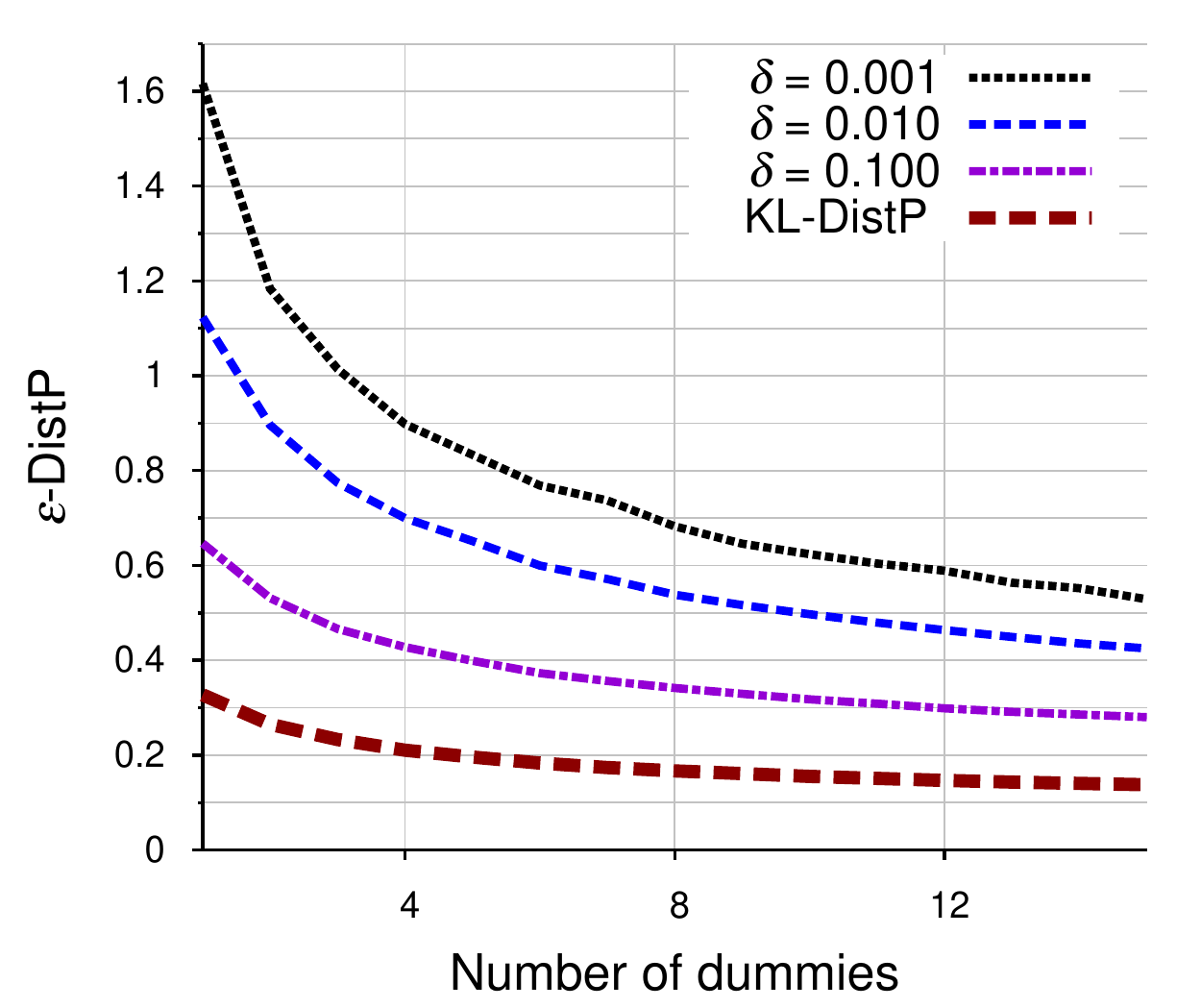}}}
\caption{Relationship between $\varepsilon$-\DistP{} and \#dummies (when using $(100, 0.020)$-\RL{} mechanism).\label{fig:NS:dummies:DistP}}
\end{subfigure}\hspace{0.4ex}\hfill
\begin{subfigure}[t]{0.24\textwidth}
  \mbox{\raisebox{-11pt}{\includegraphics[height=24mm, width=1.00\textwidth]{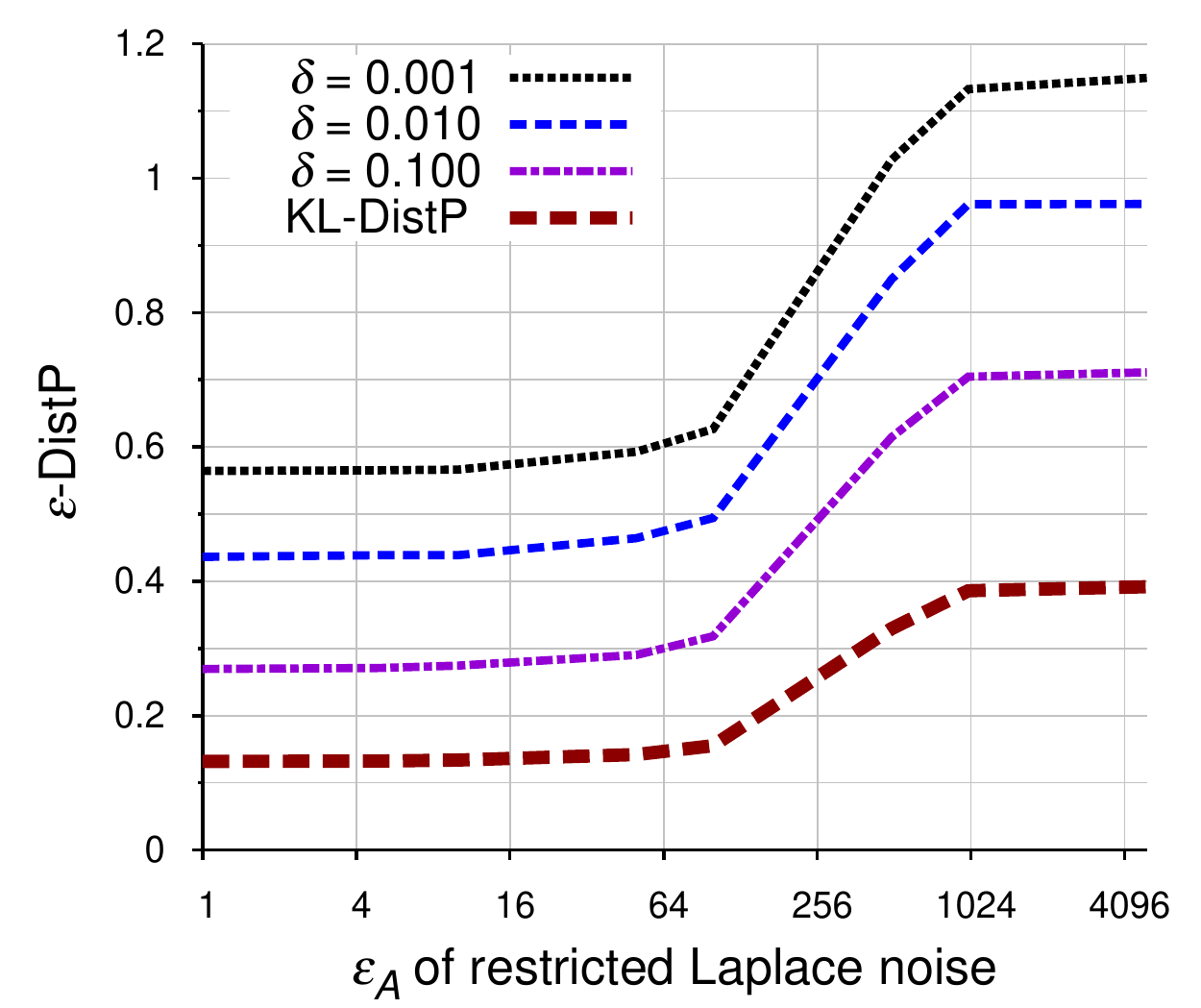}}}
\caption{Relationship between $\varepsilon$-\DistP{} and $\varepsilon_\alg$ of $(\varepsilon_\alg, 0.020)$-\RL{} mechanism (with $10$ dummies).\label{fig:NS:DP-noises:DistP}}
\end{subfigure}\hspace{0.4ex}\hfill
\begin{subfigure}[t]{0.24\textwidth}
  \mbox{\raisebox{-11pt}{\includegraphics[height=24mm, width=1.00\textwidth]{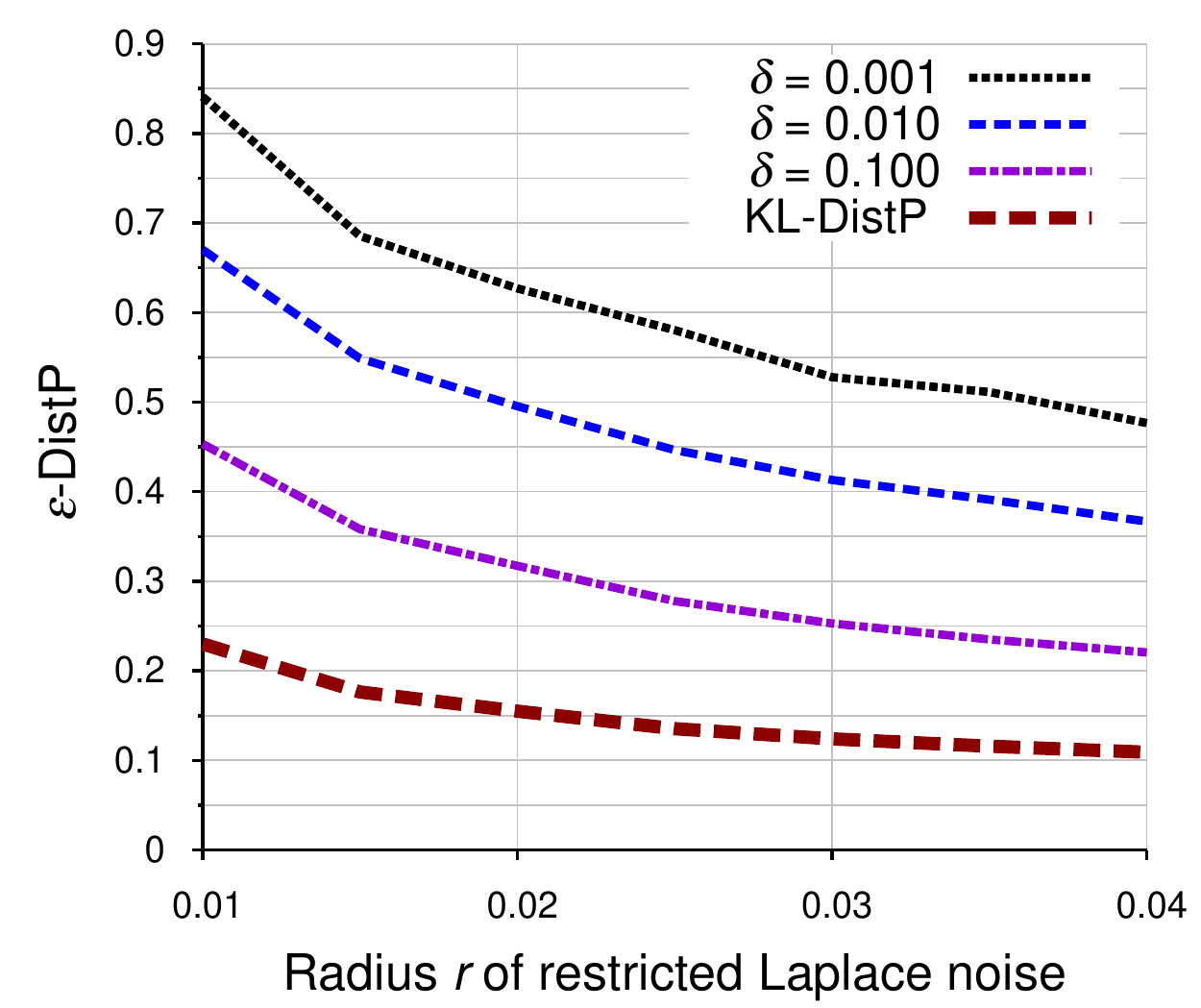}}}
\caption{Relationship between $\varepsilon$-\DistP{} and a radius $r$ of $(100, r)$-\RL{} mechanism (with $10$ dummies).\label{fig:NS:radius:DistP}}
\end{subfigure}\hspace{0.4ex}\hfill
\begin{subfigure}[t]{0.24\textwidth}
  \mbox{\raisebox{-11pt}{\includegraphics[height=24mm, width=1.00\textwidth]{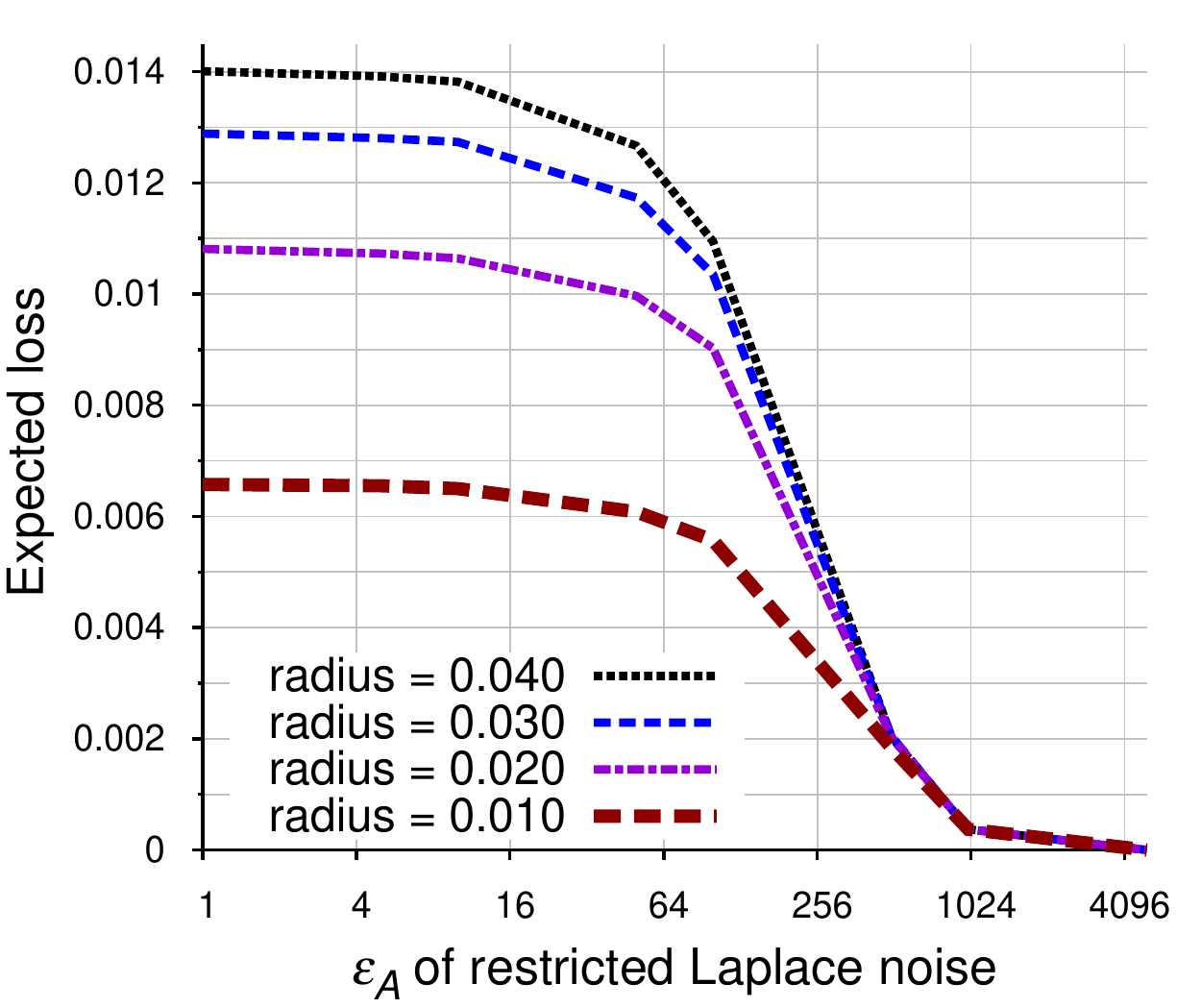}}}
\caption{Relationship between the expected loss and $\varepsilon_\alg$ of $(\varepsilon_\alg, \allowbreak r)$-\RL{} mechanism (with $5$ dummies).
\label{fig:NS:tupling:loss}}
\end{subfigure}
\caption{Empirical \DistP{} and loss for the attribute \north{}/\south{} in Manhattan.
\label{fig:NS:tupling:privacy}}
\end{minipage}
\end{tabular}
\end{figure}

\begin{figure}[H]
\begin{tabular}{c}
\begin{minipage}{1.0\hsize}
\centering
\begin{subfigure}[t]{0.24\textwidth}
  \mbox{\raisebox{-12pt}{\includegraphics[height=26mm, width=1.00\textwidth]{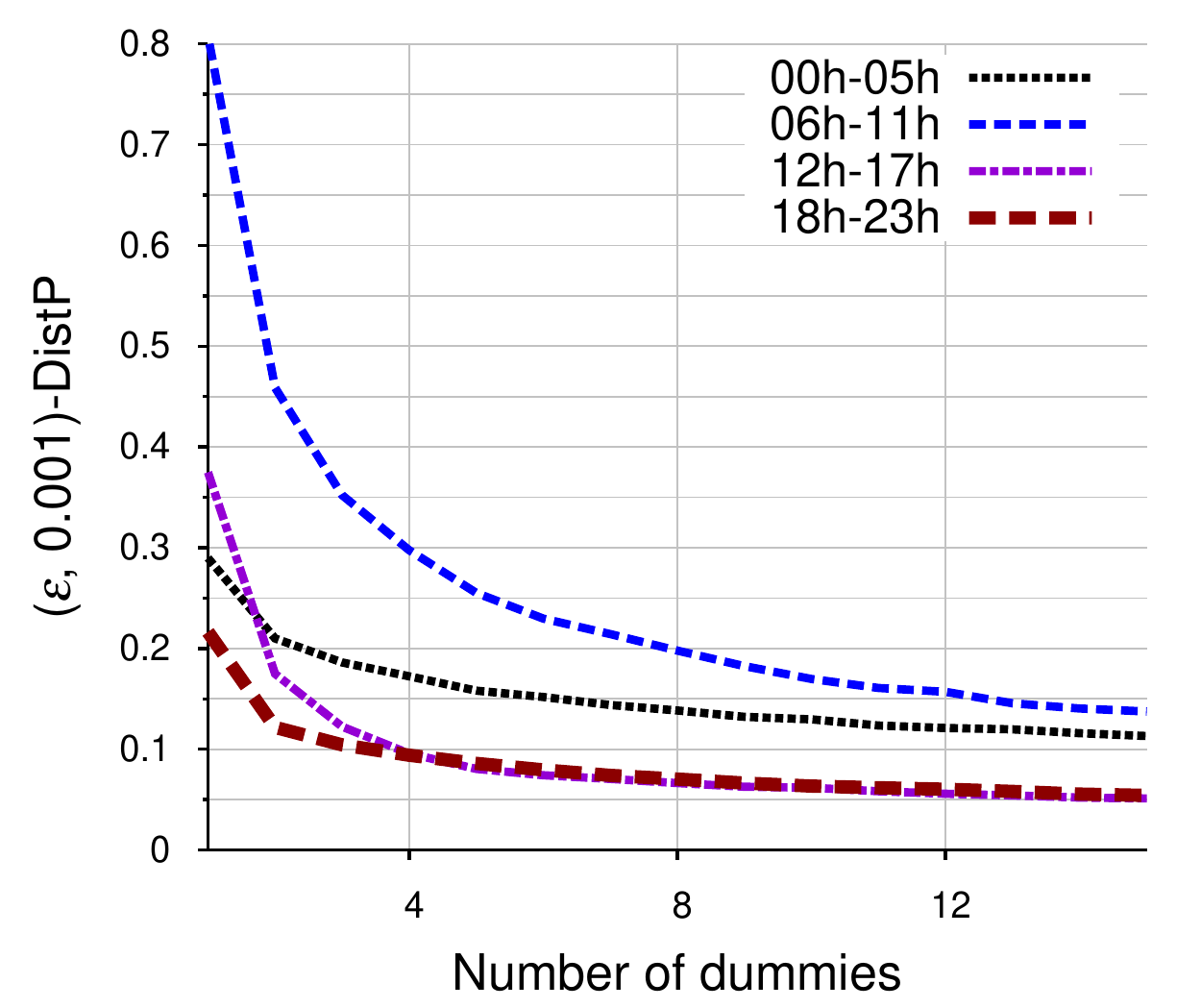}}}
\caption{Relationship between $\varepsilon$-\DistP{} and \#dummies (when using $(100, 0.020)$-\RL{} mechanism).
\label{fig:times+MH+MF:dummies:DistP}}
\end{subfigure}\hspace{0.4ex}\hfill
\begin{subfigure}[t]{0.24\textwidth}
  \mbox{\raisebox{-12pt}{\includegraphics[height=26mm, width=1.00\textwidth]{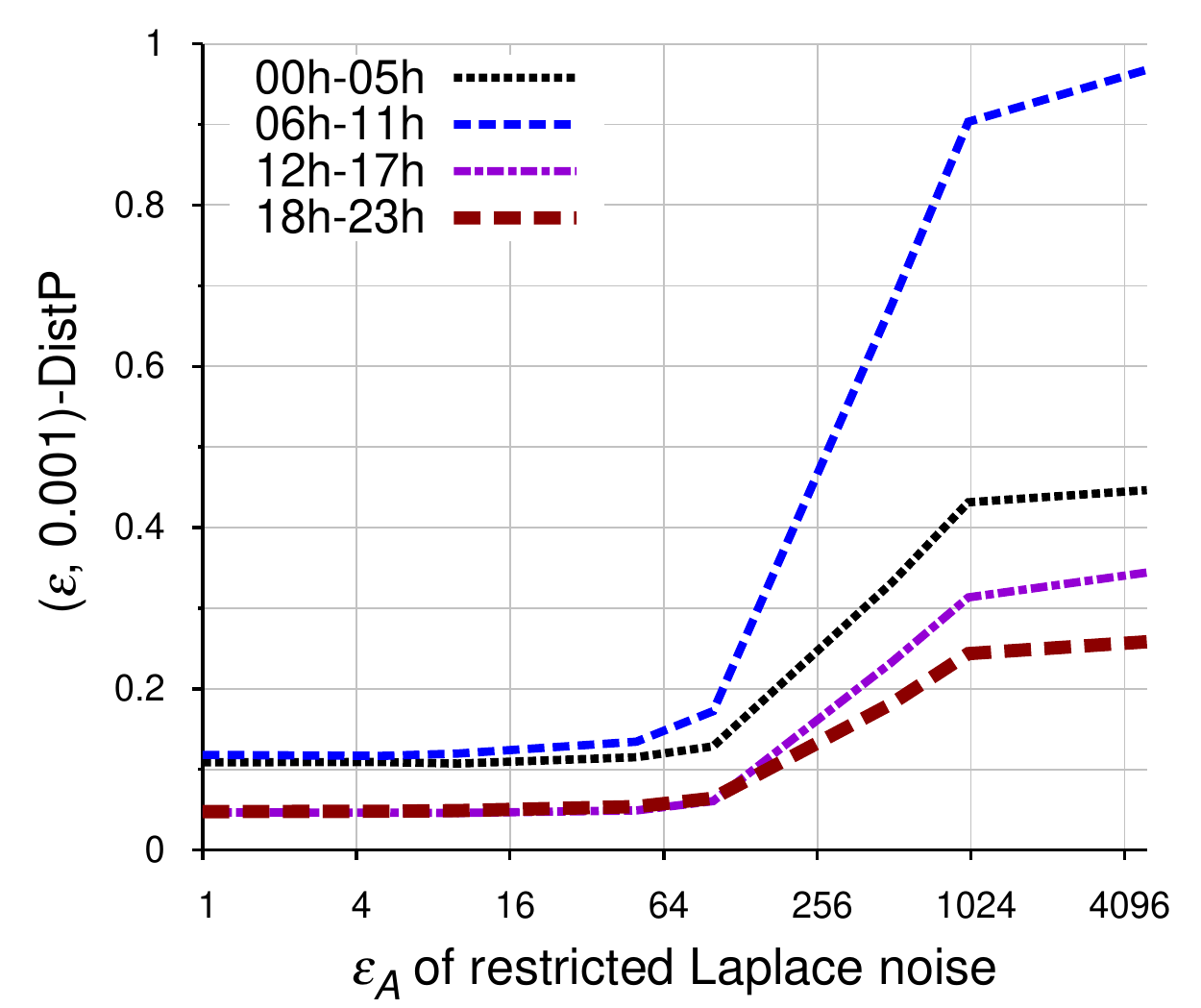}}}
\caption{Relationship between $\varepsilon$-\DistP{} and $\varepsilon_\alg$ of $(\varepsilon_\alg, 0.020)$-\RL{} mechanism (with $10$ dummies).\label{fig:times+MH+MF:DP-noises:DistP}}
\end{subfigure}\hspace{0.4ex}\hfill
\begin{subfigure}[t]{0.24\textwidth}
  \mbox{\raisebox{-12pt}{\includegraphics[height=26mm, width=1.00\textwidth]{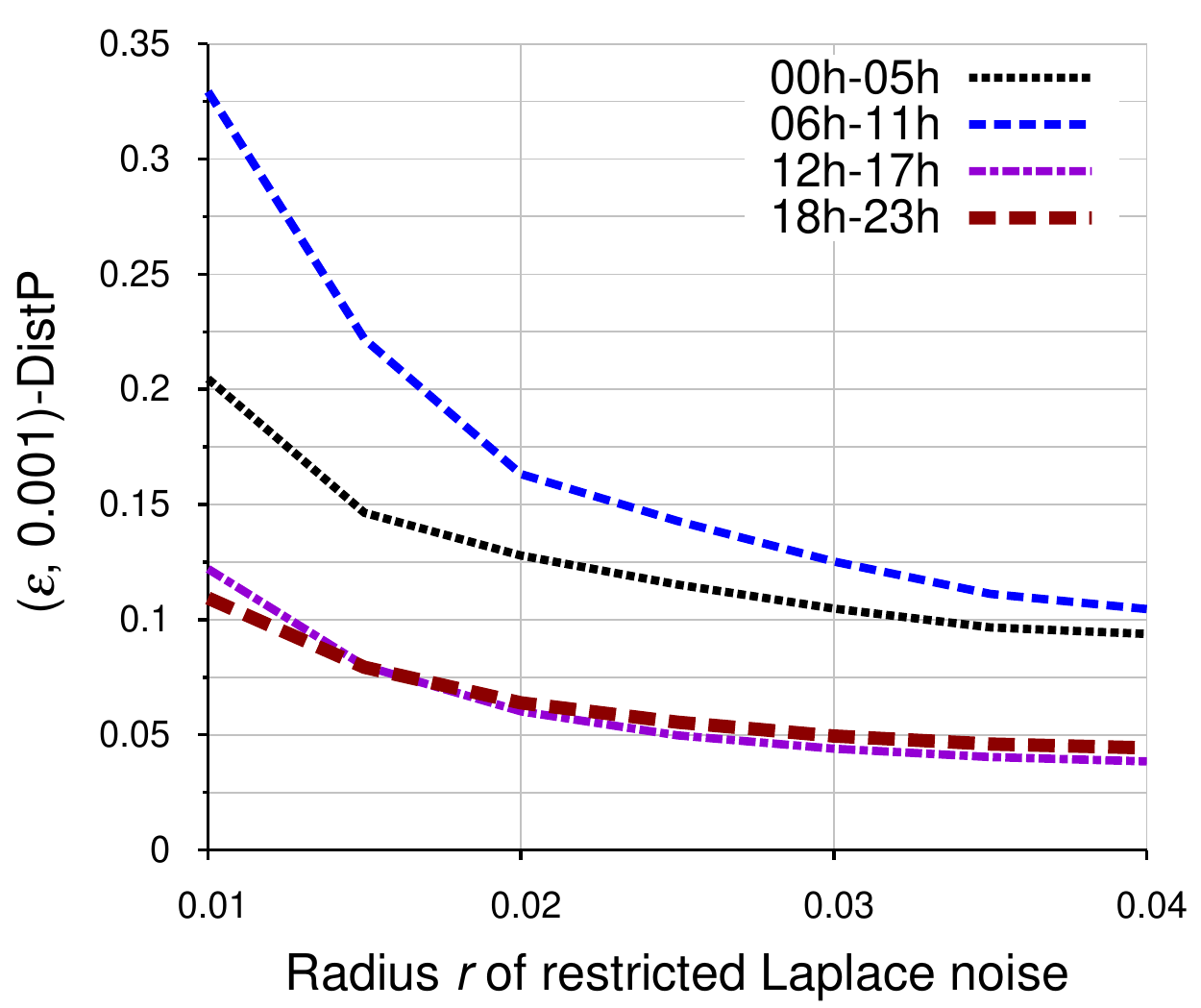}}}
\caption{Relationship between $\varepsilon$-\DistP{} and a radius $r$ of $(100, r)$-\RL{} mechanism (with $10$ dummies).
\label{fig:times+MH+MF:radius:DistP}}
\end{subfigure}\hspace{0.4ex}\hfill
\begin{subfigure}[t]{0.24\textwidth}
  \mbox{\raisebox{-12pt}{\includegraphics[height=26mm, width=1.00\textwidth]{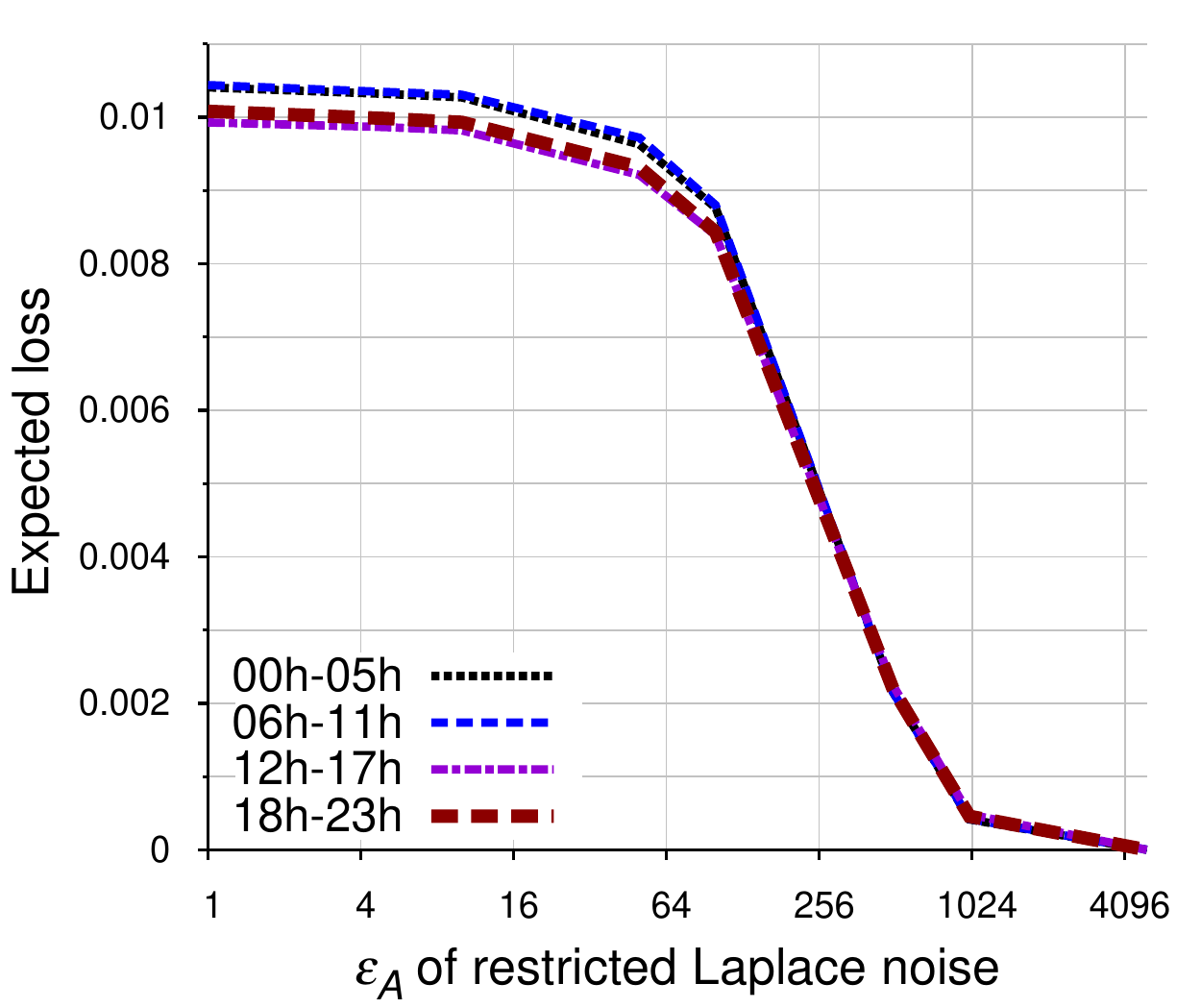}}}
\caption{Relationship between the expected loss and $\varepsilon_\alg$ of $(\varepsilon_\alg, \allowbreak r)$-\RL{} mechanism (with $5$ dummies).
\label{fig:times+MH+M:tupling:loss}}
\end{subfigure}
\vspace{-4mm}
\caption{Empirical \DistP{} and loss for \male{}/\female{} in different hours.
\label{fig:times+MH+MF:tupling:privacy}}
\vspace{1mm}
\end{minipage}
\\
\begin{minipage}{1.0\hsize}
\centering
\begin{subfigure}[t]{0.24\textwidth}
  \mbox{\raisebox{-12pt}{\includegraphics[height=26mm, width=1.00\textwidth]{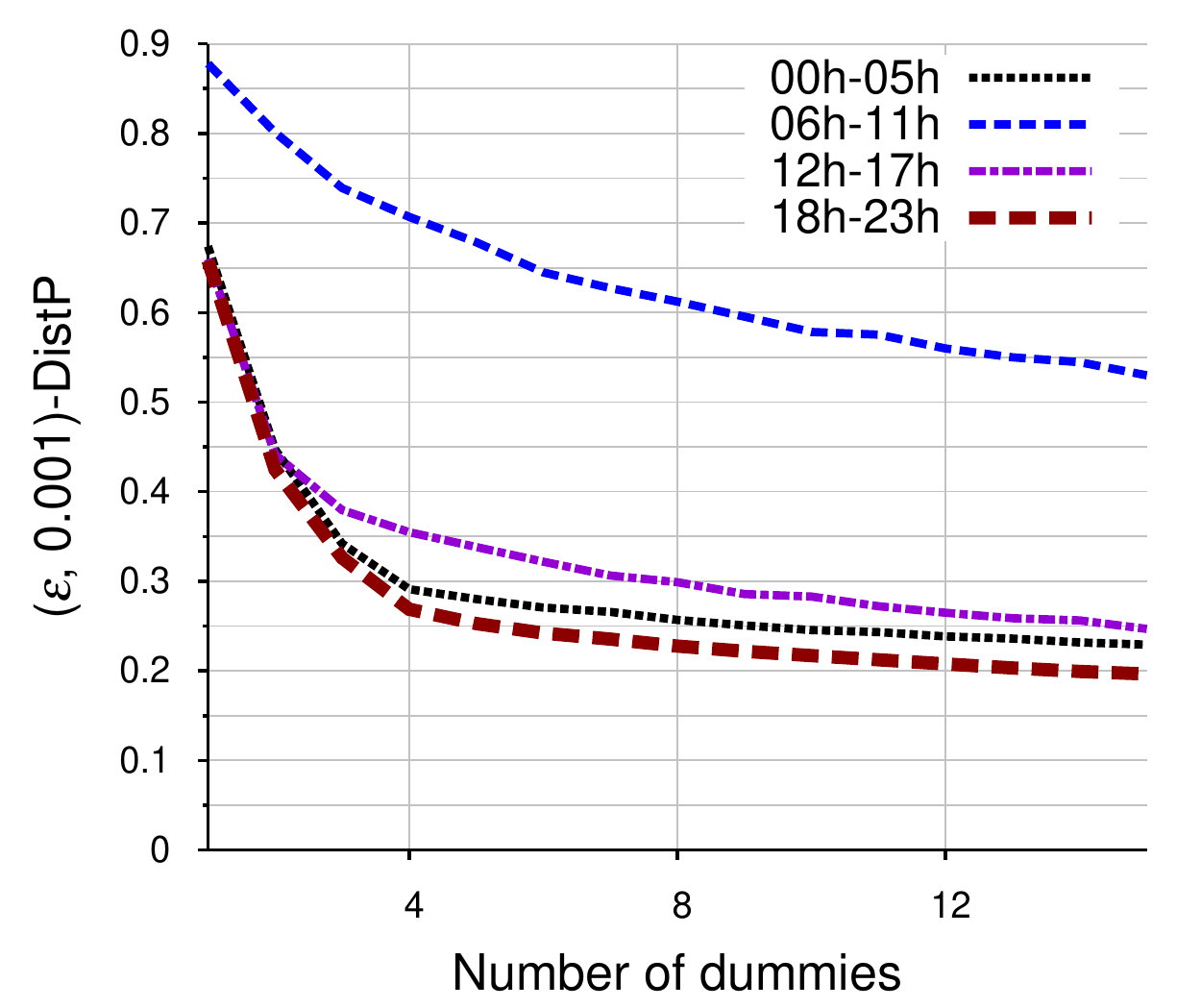}}}
\caption{Relationship between $\varepsilon$-\DistP{} and \#dummies (when using $(100, 0.020)$-\RL{} mechanism).\label{fig:times+MH+SI:dummies:DistP}}
\end{subfigure}\hspace{0.4ex}\hfill
\begin{subfigure}[t]{0.24\textwidth}
  \mbox{\raisebox{-12pt}{\includegraphics[height=26mm, width=1.00\textwidth]{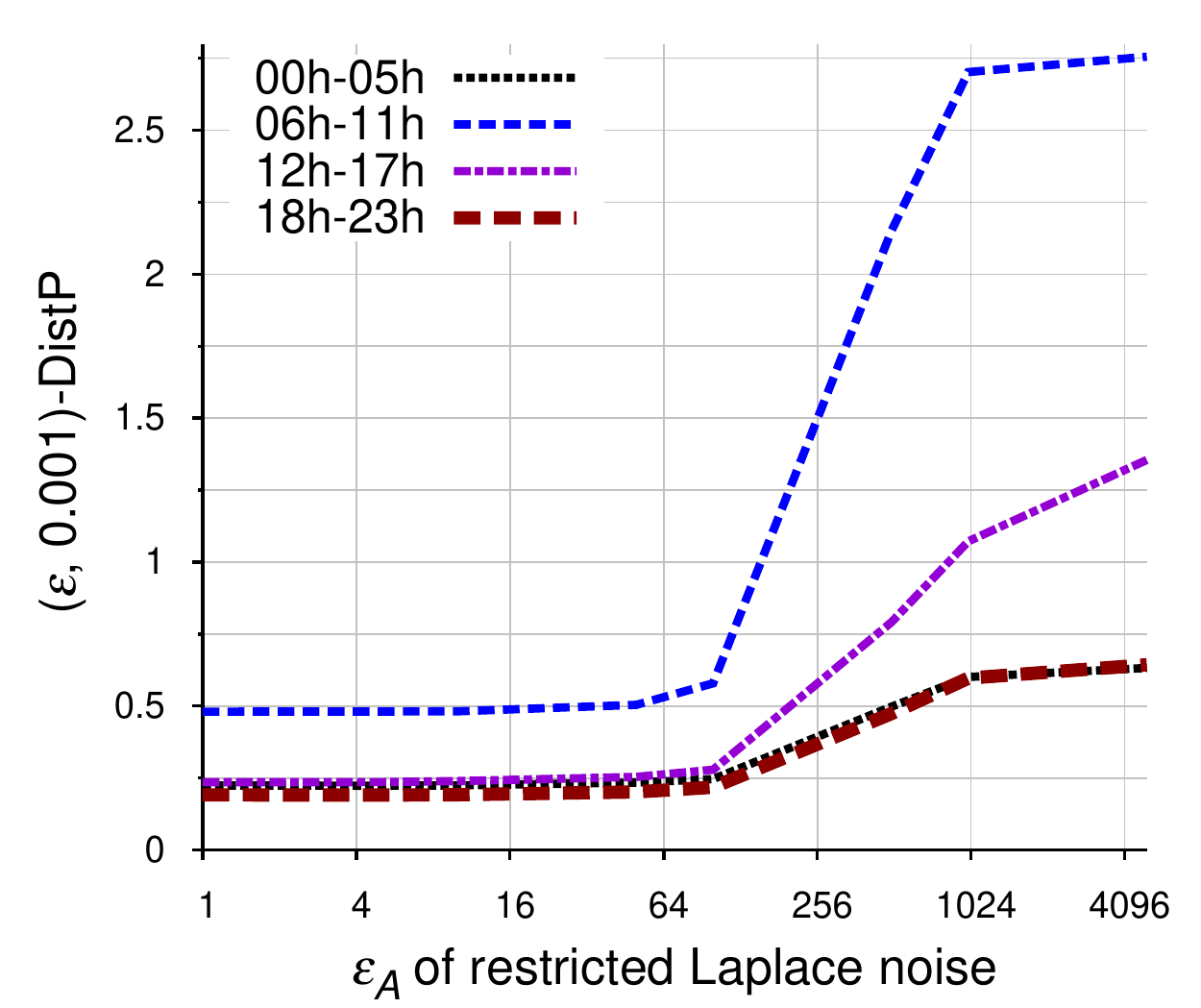}}}
\caption{Relationship between $\varepsilon$-\DistP{} and $\varepsilon_\alg$ of $(\varepsilon_\alg, 0.020)$-\RL{} mechanism (with $10$ dummies).\label{fig:times+MH+SI:DP-noises:DistP}}
\end{subfigure}\hspace{0.4ex}\hfill
\begin{subfigure}[t]{0.24\textwidth}
  \mbox{\raisebox{-12pt}{\includegraphics[height=26mm, width=1.00\textwidth]{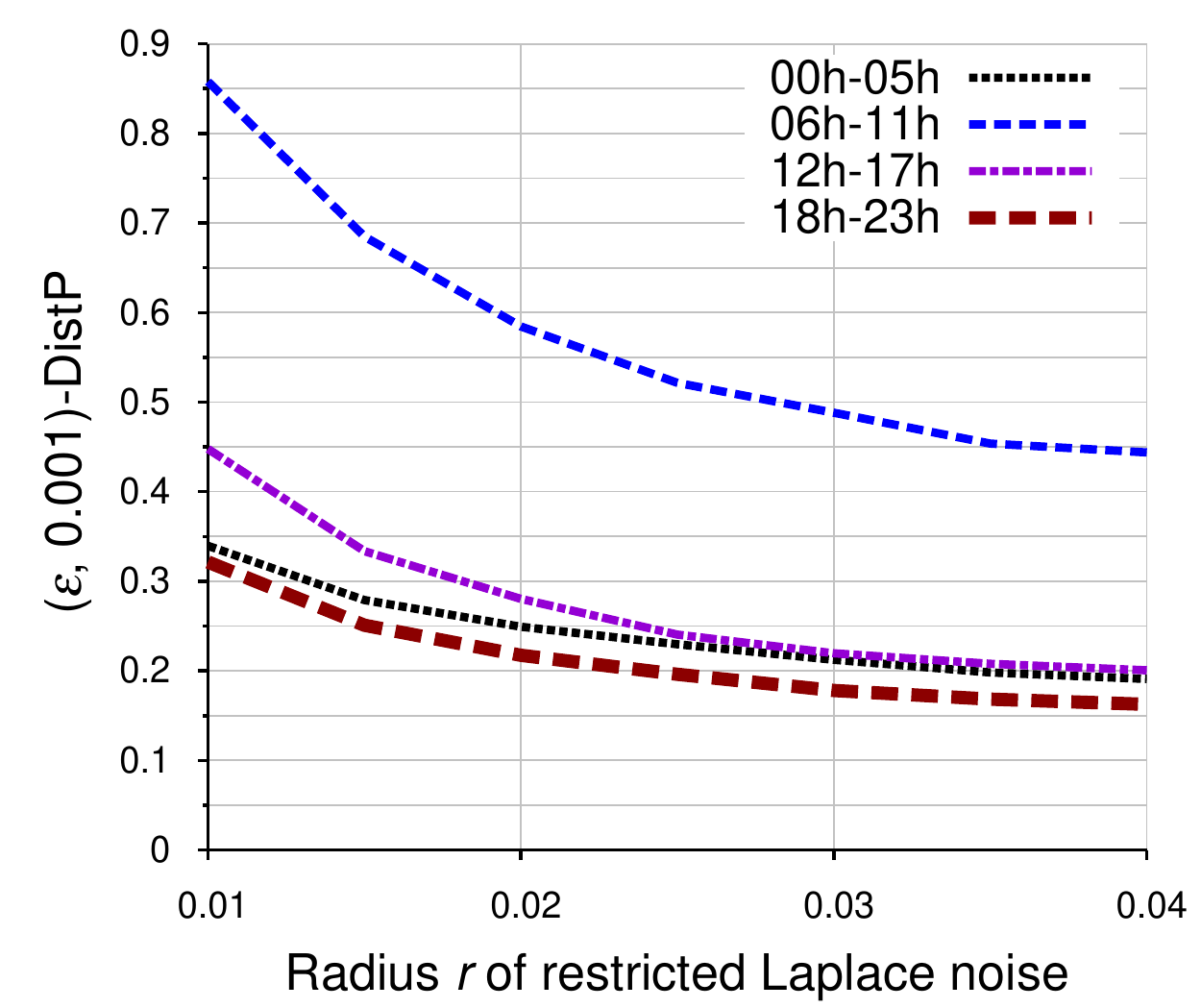}}}
\caption{Relationship between $\varepsilon$-\DistP{} and a radius $r$ of $(100, r)$-\RL{} mechanism (with $10$ dummies).\label{fig:times+MH+SI:radius:DistP}}
\end{subfigure}\hspace{0.4ex}\hfill
\begin{subfigure}[t]{0.24\textwidth}
  \mbox{\raisebox{-12pt}{\includegraphics[height=26mm, width=1.00\textwidth]{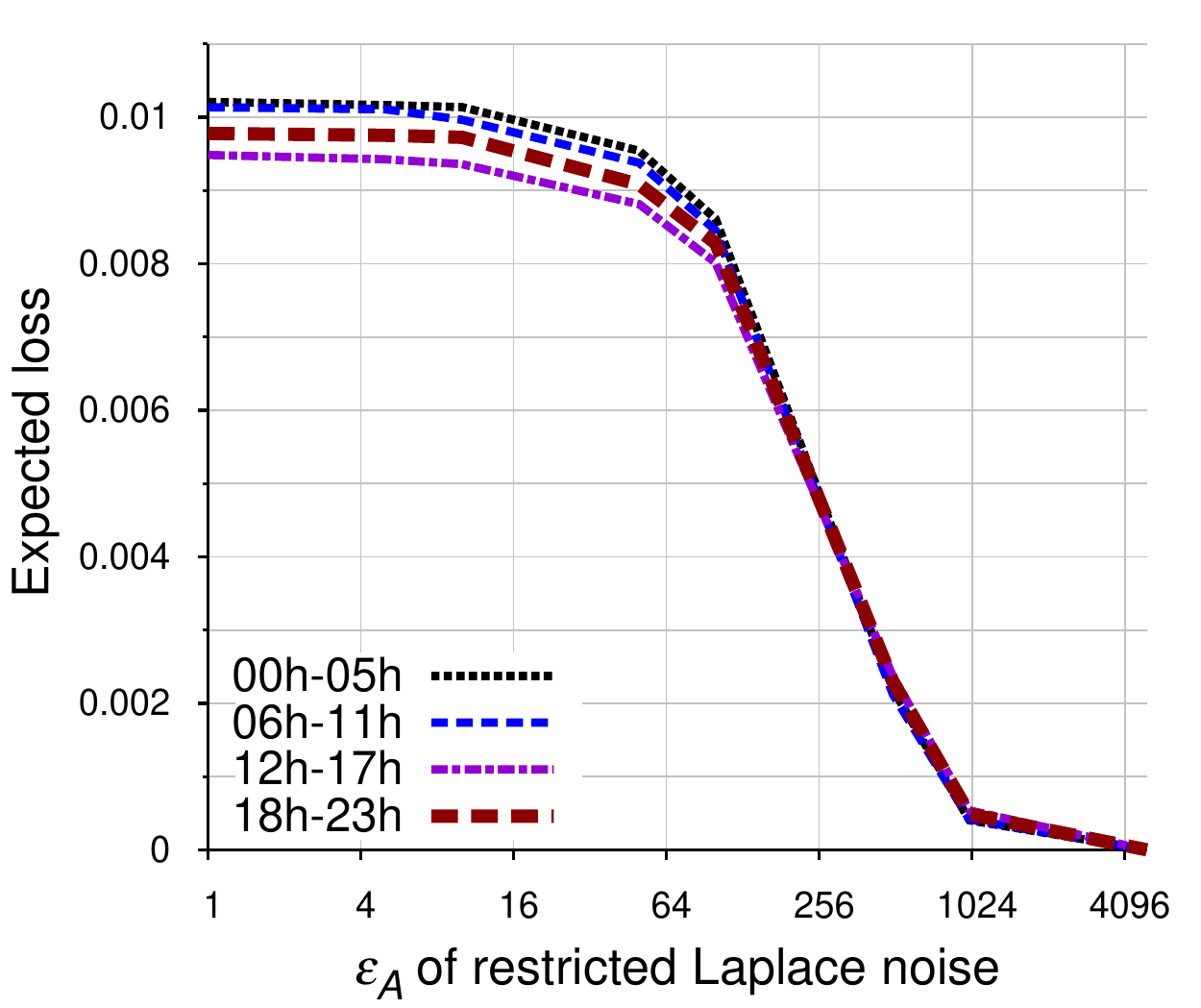}}}
\caption{Relationship between the expected loss and $\varepsilon_\alg$ of $(\varepsilon_\alg, \allowbreak r)$-\RL{} mechanism (with $5$ dummies).
\label{fig:times+MH+S:tupling:loss}}
\end{subfigure}
\caption{Empirical \DistP{} and loss for \social{}/\lesssocial{} in different hours.
\label{fig:times+MH+SI:tupling:privacy}}
\vspace{1mm}
\end{minipage}
\\
\begin{minipage}{1.0\hsize}
\centering
\begin{subfigure}[t]{0.24\textwidth}
  \mbox{\raisebox{-12pt}{\includegraphics[height=26mm, width=1.00\textwidth]{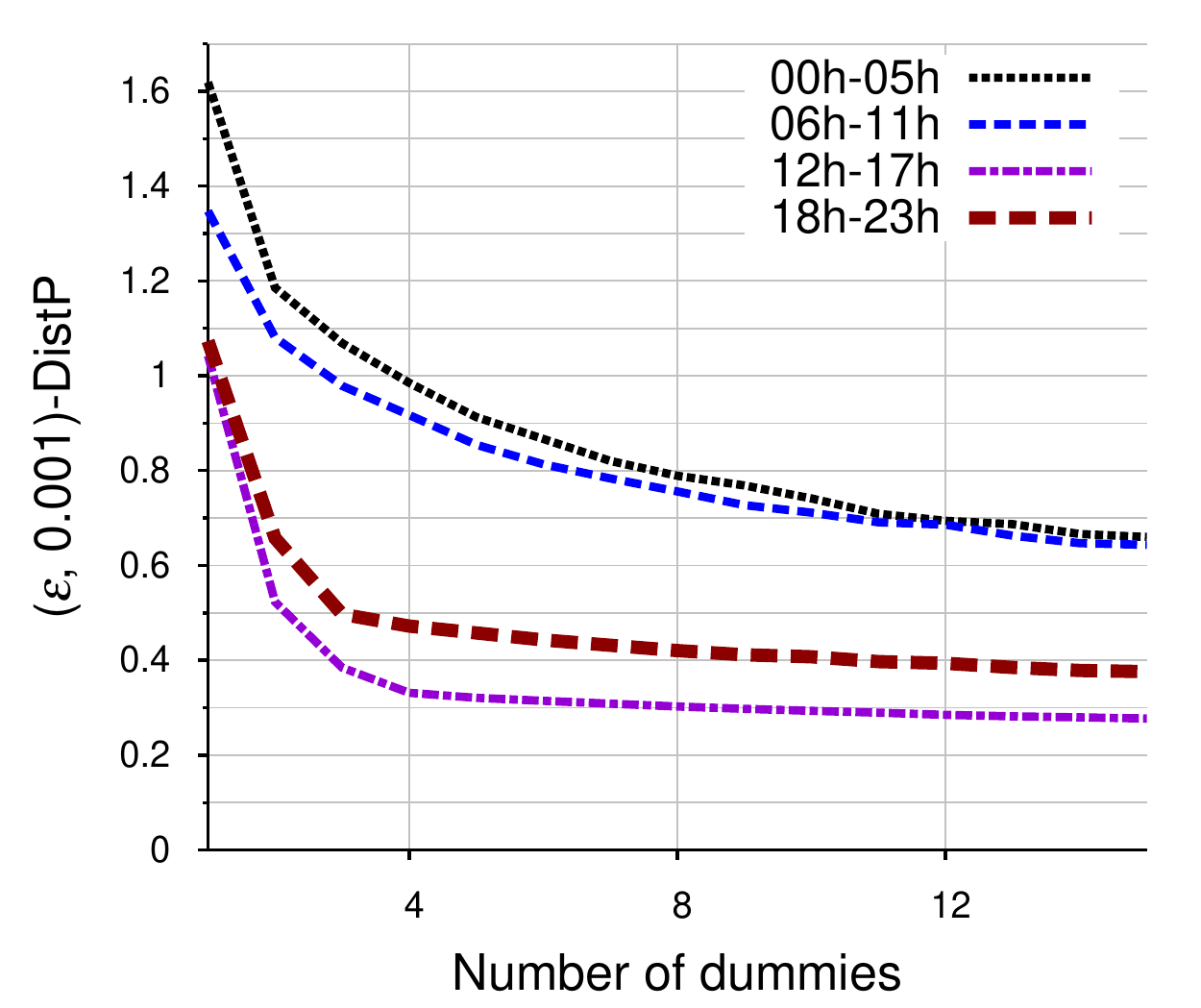}}}
\caption{Relationship between $\varepsilon$-\DistP{} and \#dummies (when using $(100, 0.020)$-\RL{} mechanism).\label{fig:times+MH+WP:dummies:DistP}}
\end{subfigure}\hspace{0.4ex}\hfill
\begin{subfigure}[t]{0.24\textwidth}
  \mbox{\raisebox{-12pt}{\includegraphics[height=26mm, width=1.00\textwidth]{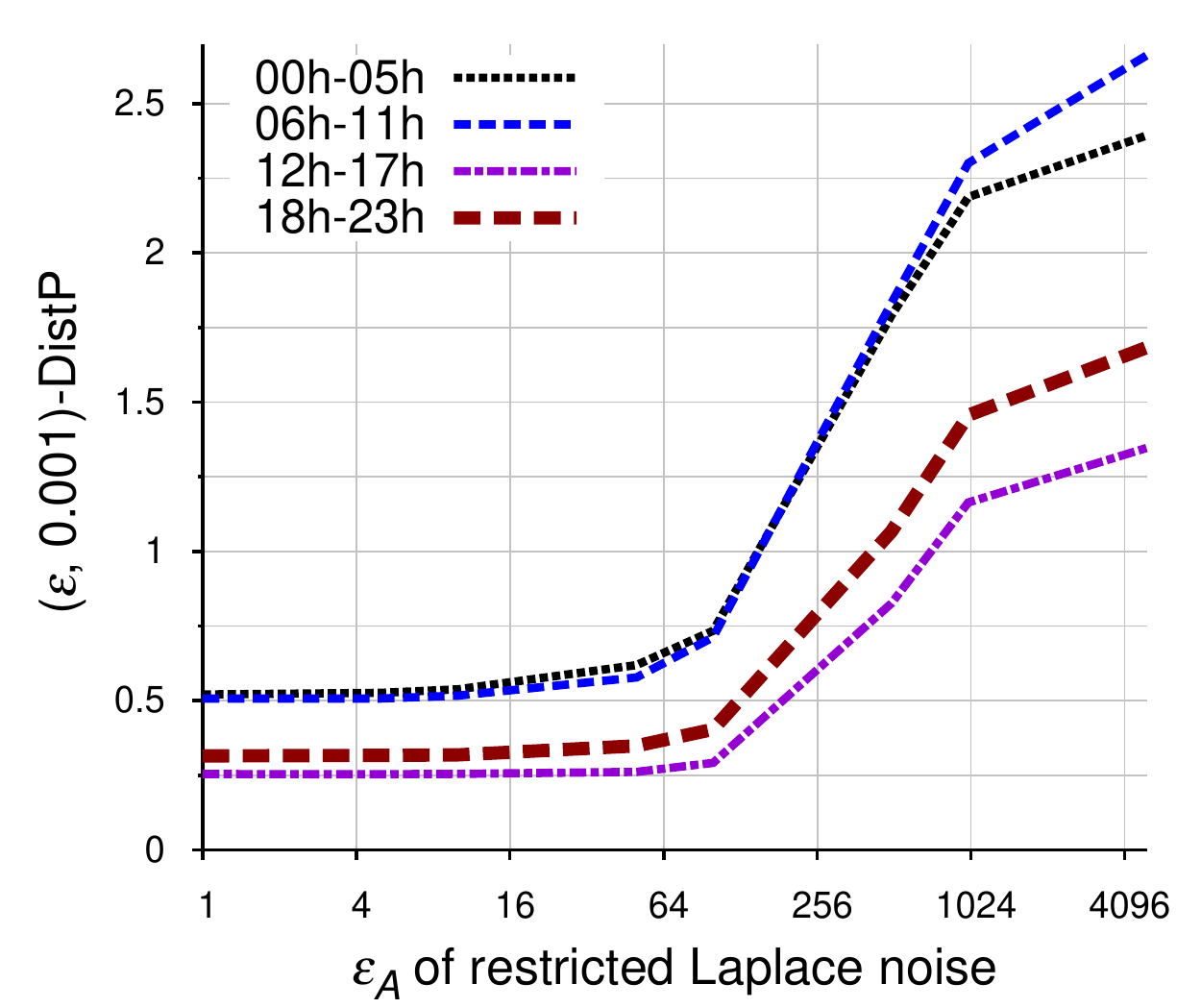}}}
\caption{Relationship between $\varepsilon$-\DistP{} and $\varepsilon_\alg$ of $(\varepsilon_\alg, 0.020)$-\RL{} mechanism (with $10$ dummies).\label{fig:times+MH+WP:DP-noises:DistP}}
\end{subfigure}\hspace{0.4ex}\hfill
\begin{subfigure}[t]{0.24\textwidth}
  \mbox{\raisebox{-12pt}{\includegraphics[height=26mm, width=1.00\textwidth]{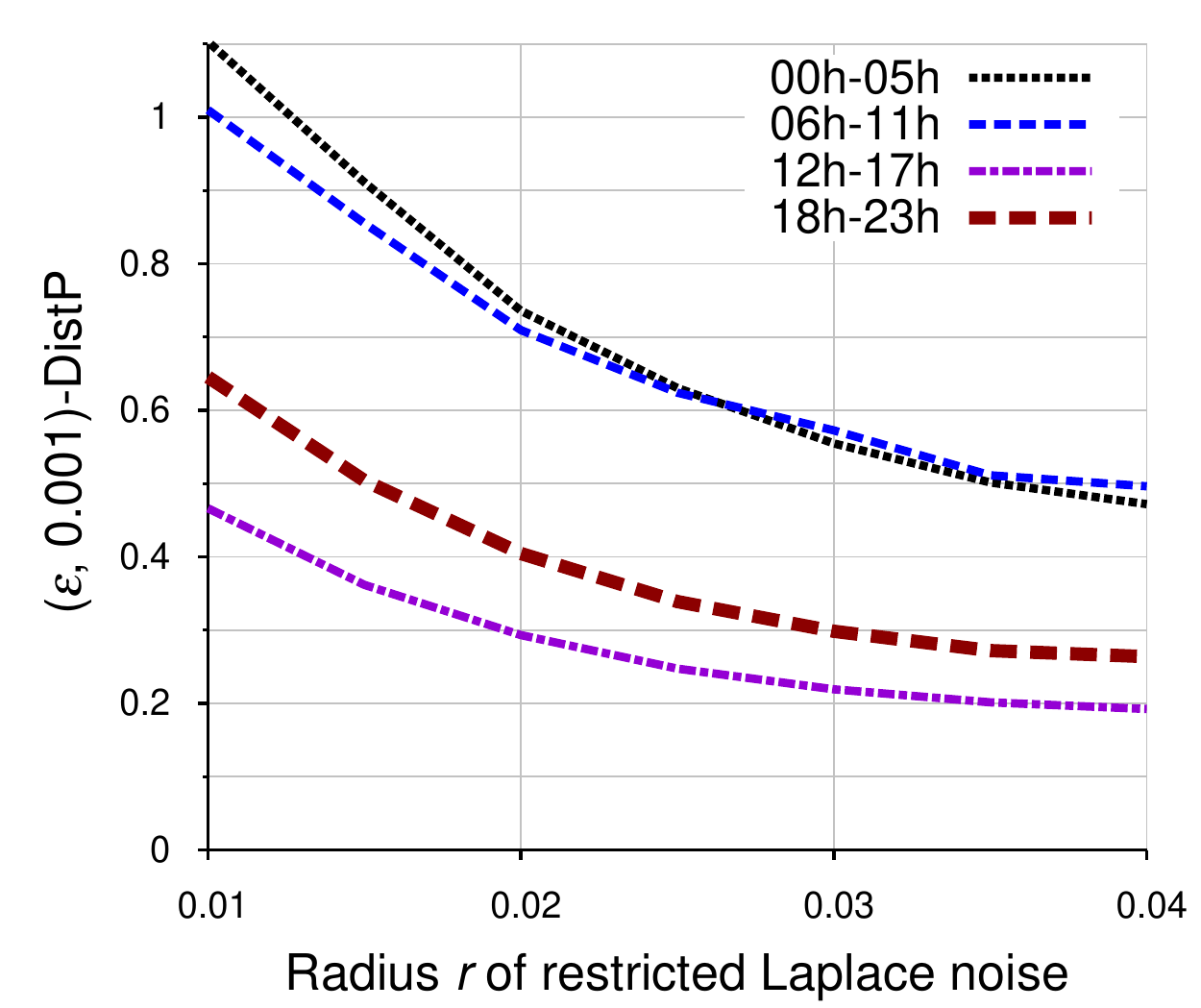}}}
\caption{Relationship between $\varepsilon$-\DistP{} and a radius $r$ of $(100, r)$-\RL{} mechanism (with $10$ dummies).\label{fig:times+MH+WP:radius:DistP}}
\end{subfigure}\hspace{0.4ex}\hfill
\begin{subfigure}[t]{0.24\textwidth}
  \mbox{\raisebox{-12pt}{\includegraphics[height=26mm, width=1.00\textwidth]{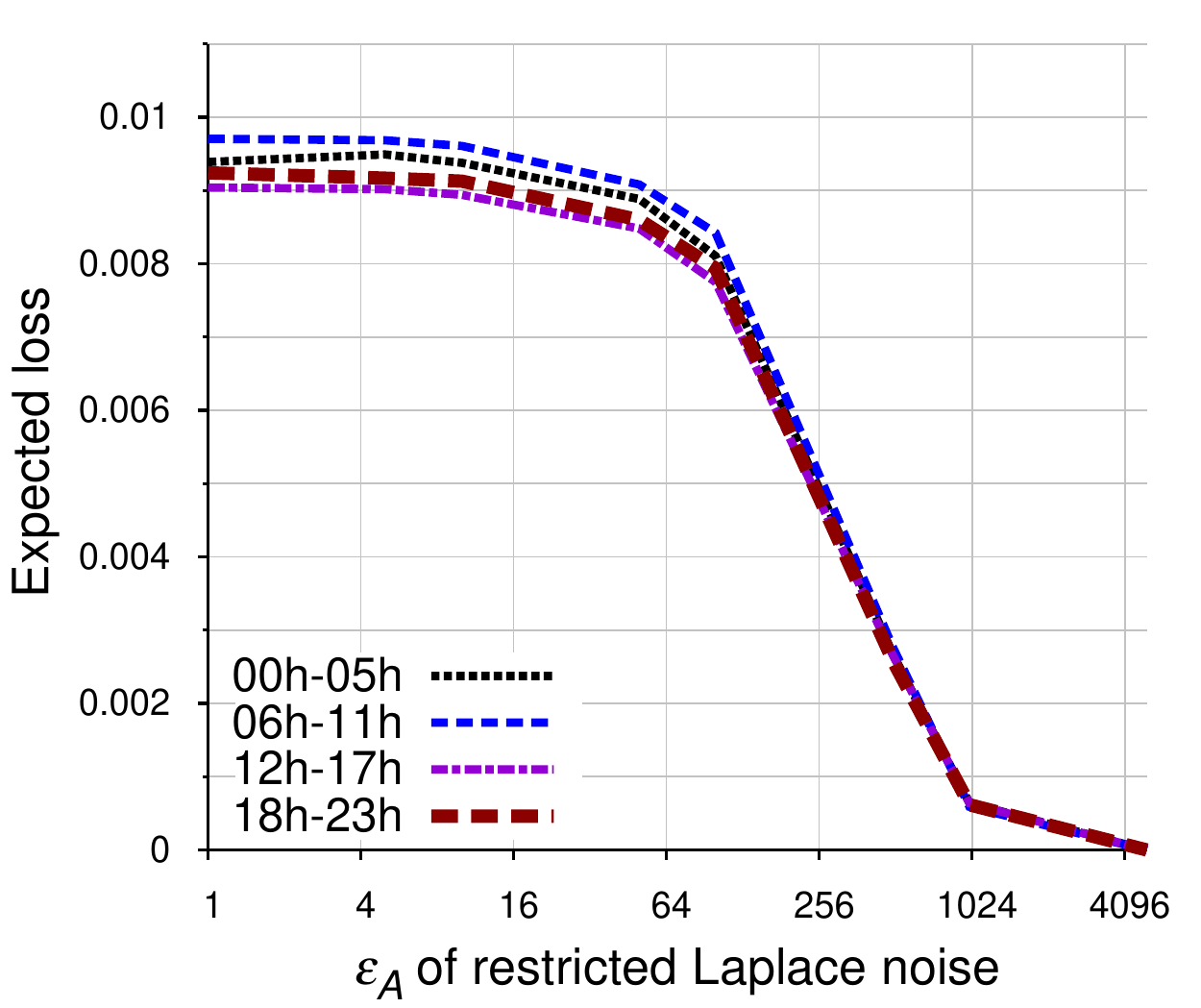}}}
\caption{Relationship between the expected loss and $\varepsilon_\alg$ of $(\varepsilon_\alg, \allowbreak r)$-\RL{} mechanism (with $5$ dummies).
\label{fig:times+MH+W:tupling:loss}}
\end{subfigure}
\vspace{-3mm}
\caption{Empirical \DistP{} and loss for \workplace{}/\nonworkplace{} in different hours.
\label{fig:times+MH+WP:tupling:privacy}}
\vspace{1mm}
\end{minipage}
\\
\begin{minipage}{1.0\hsize}
\centering
\begin{subfigure}[t]{0.24\textwidth}
  \mbox{\raisebox{-12pt}{\includegraphics[height=26mm, width=1.00\textwidth]{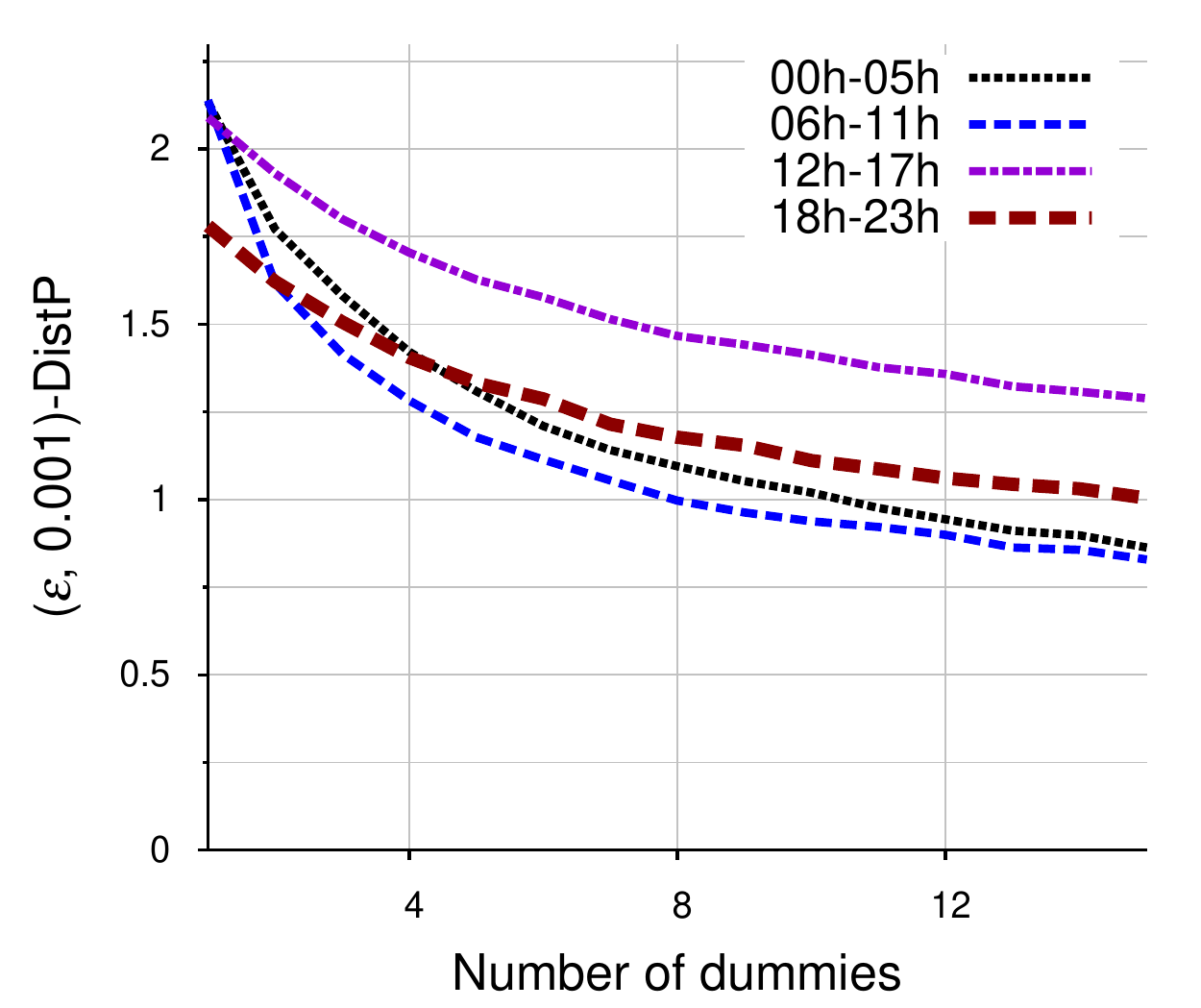}}}
\caption{Relationship between $\varepsilon$-\DistP{} and \#dummies (when using $(100, 0.020)$-\RL{} mechanism).\label{fig:times+MH+HO:dummies:DistP}}
\end{subfigure}\hspace{0.4ex}\hfill
\begin{subfigure}[t]{0.24\textwidth}
  \mbox{\raisebox{-12pt}{\includegraphics[height=26mm, width=1.00\textwidth]{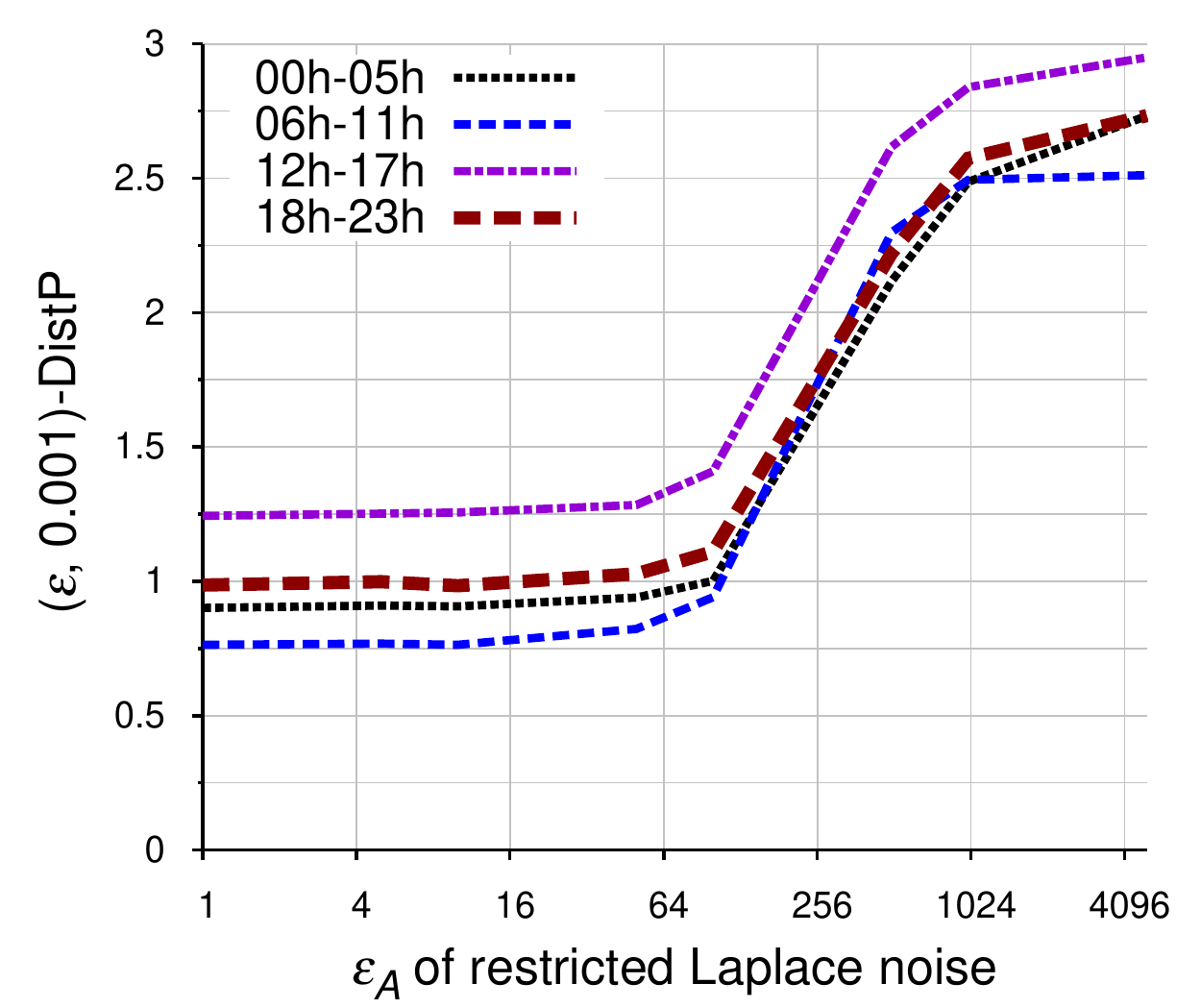}}}
\caption{Relationship between $\varepsilon$-\DistP{} and $\varepsilon_\alg$ of $(\varepsilon_\alg, 0.020)$-\RL{} mechanism (with $10$ dummies).\label{fig:times+MH+HO:DP-noises:DistP}}
\end{subfigure}\hspace{0.4ex}\hfill
\begin{subfigure}[t]{0.24\textwidth}
  \mbox{\raisebox{-12pt}{\includegraphics[height=26mm, width=1.00\textwidth]{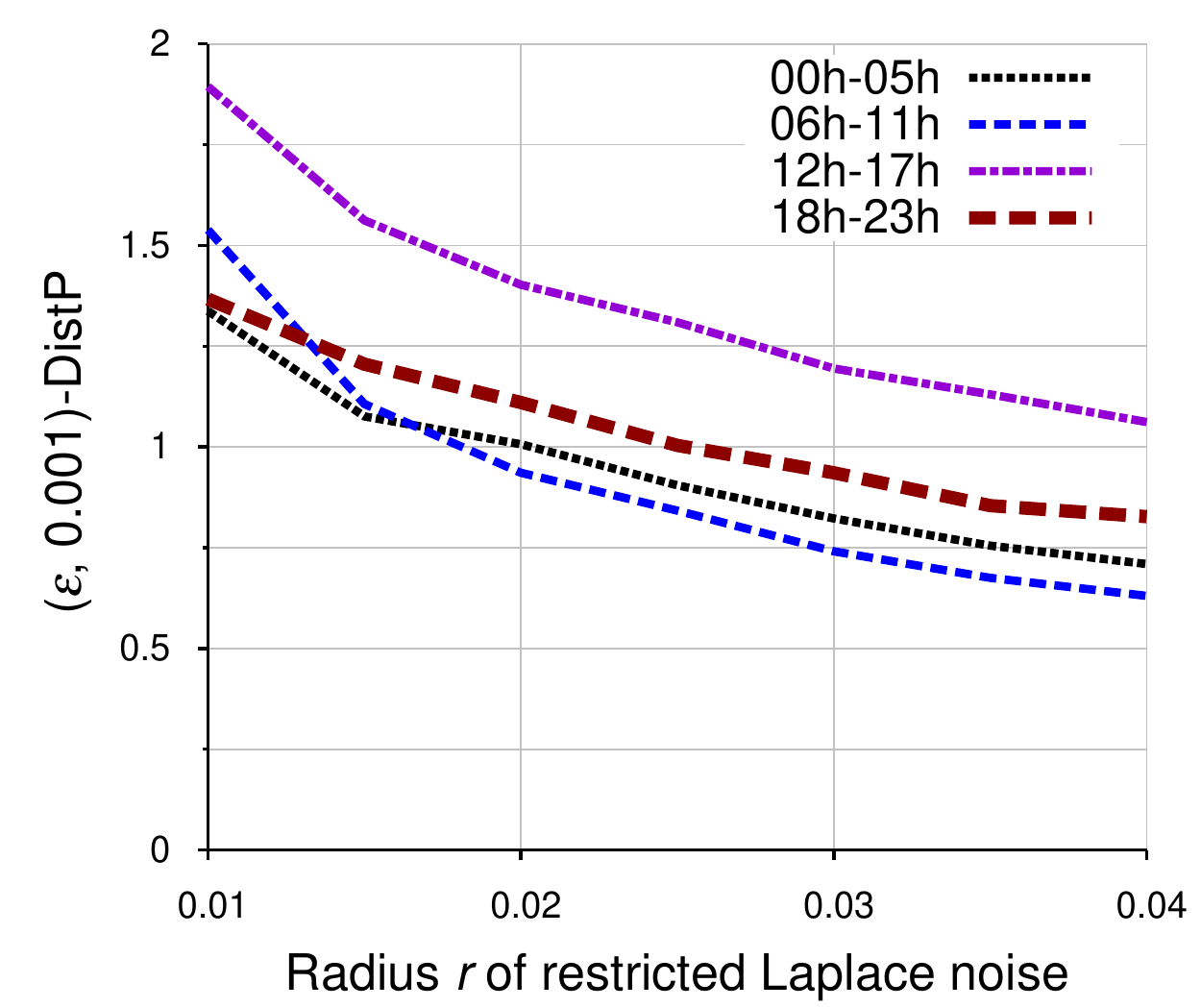}}}
\caption{Relationship between $\varepsilon$-\DistP{} and a radius $r$ of $(100, r)$-\RL{} mechanism (with $10$ dummies).\label{fig:times+MH+HO:radius:DistP}}
\end{subfigure}\hspace{0.4ex}\hfill
\begin{subfigure}[t]{0.24\textwidth}
  \mbox{\raisebox{-12pt}{\includegraphics[height=26mm, width=1.00\textwidth]{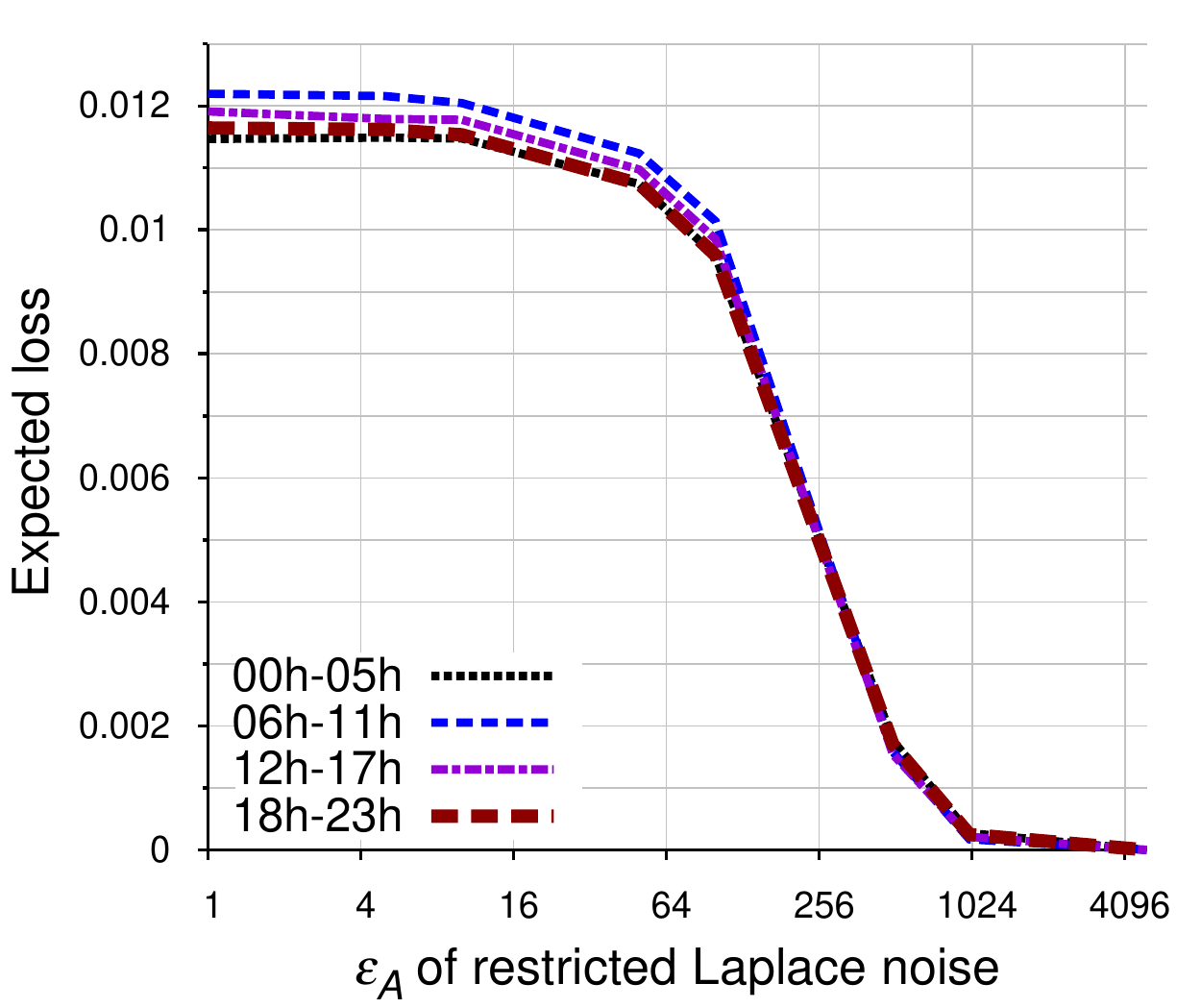}}}
\caption{Relationship between the expected loss and $\varepsilon_\alg$ of $(\varepsilon_\alg, \allowbreak r)$-\RL{} mechanism (with $5$ dummies).
\label{fig:times+MH+H:tupling:loss}}
\end{subfigure}
\vspace{-3mm}
\caption{Empirical \DistP{} and loss for \home{}/\outside{} in different hours.
\label{fig:times+MH+HO:tupling:privacy}}
\vspace{1mm}
\end{minipage}
\end{tabular}
\end{figure}

\begin{figure}[H]
\begin{tabular}{c}
\begin{minipage}{1.0\hsize}
\centering
\begin{subfigure}[t]{0.24\textwidth}
  \mbox{\raisebox{-10pt}{\includegraphics[height=26.0mm, width=1.00\textwidth]{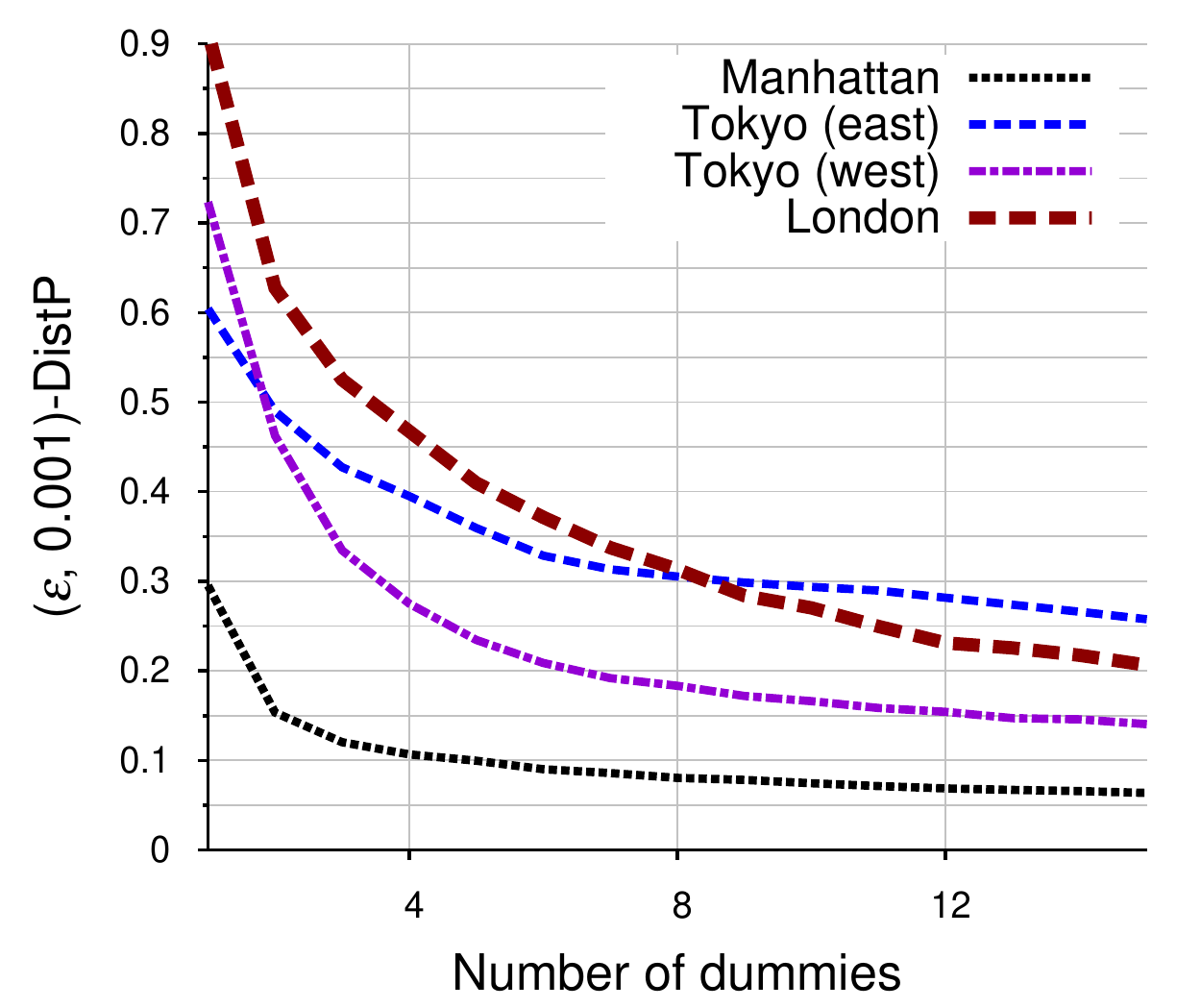}}}
\caption{Relationship between $\varepsilon$-\DistP{} and \#dummies (when using $(100, 0.020)$-\RL{} mechanism).\label{fig:cities+MF:dummies:DistP}}
\end{subfigure}\hspace{0.4ex}\hfill
\begin{subfigure}[t]{0.24\textwidth}
  \mbox{\raisebox{-10pt}{\includegraphics[height=26.0mm, width=1.00\textwidth]{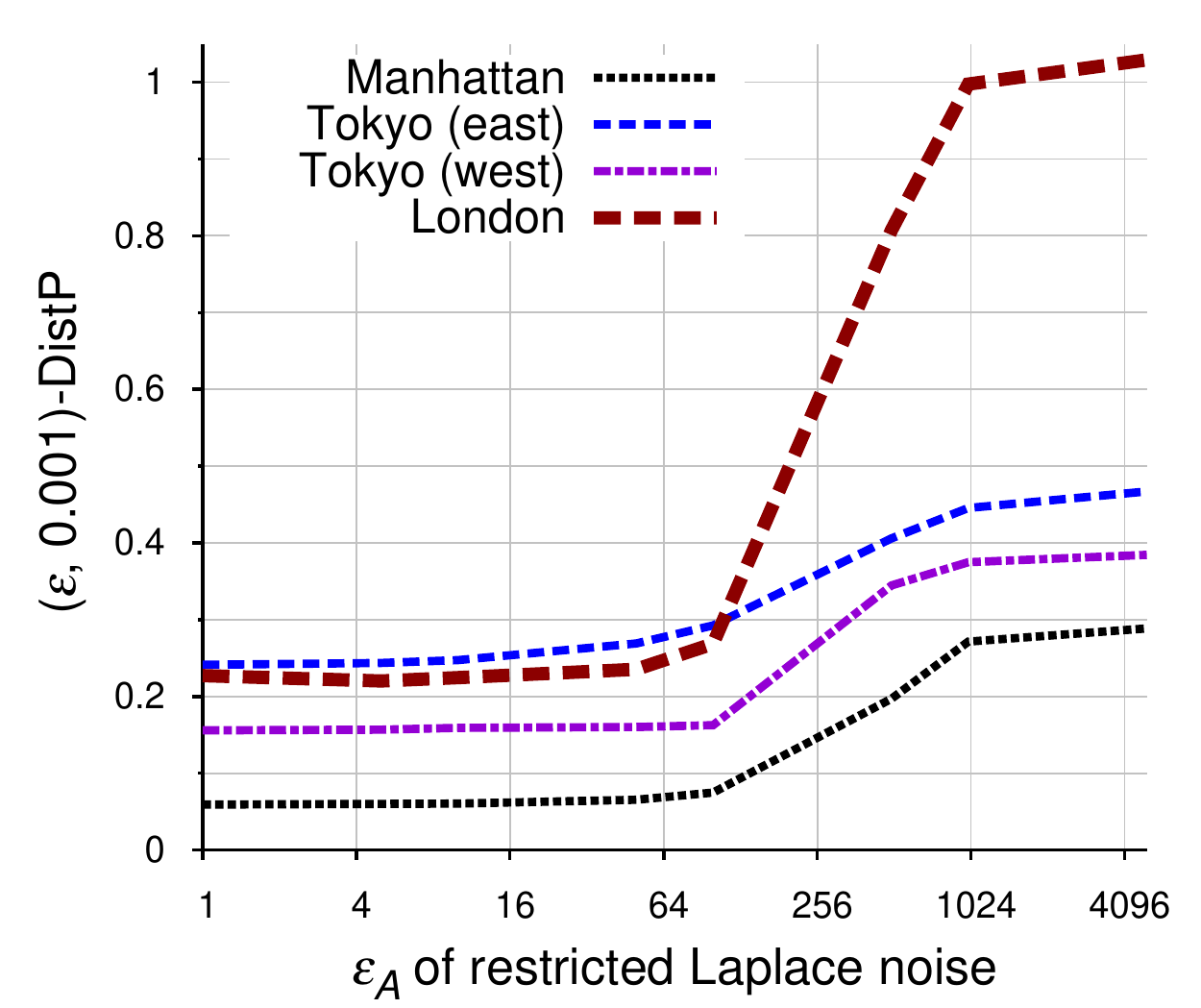}}}
\caption{Relationship between $\varepsilon$-\DistP{} and $\varepsilon_\alg$ of $(\varepsilon_\alg, 0.020)$-\RL{} mechanism (with $10$ dummies).\label{fig:cities+MF:DP-noises:DistP}}
\end{subfigure}\hspace{0.4ex}\hfill
\begin{subfigure}[t]{0.24\textwidth}
  \mbox{\raisebox{-10pt}{\includegraphics[height=26.0mm, width=1.00\textwidth]{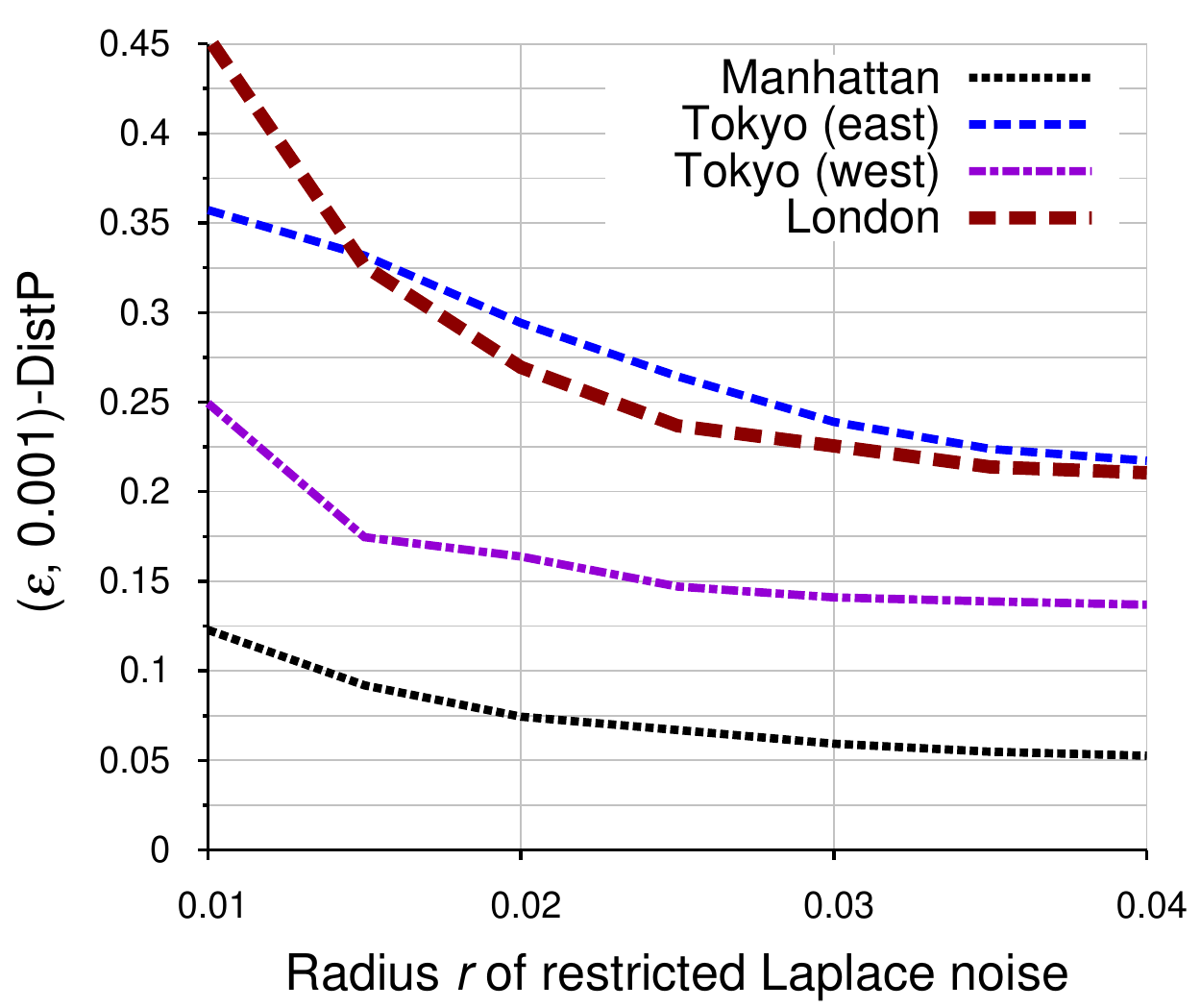}}}
\caption{Relationship between $\varepsilon$-\DistP{} and a radius $r$ of $(100, r)$-\RL{} mechanism (with $10$ dummies).\label{fig:cities+MF:radius:DistP}}
\end{subfigure}\hspace{0.4ex}\hfill
\begin{subfigure}[t]{0.24\textwidth}
  \mbox{\raisebox{-10pt}{\includegraphics[height=26.0mm, width=1.00\textwidth]{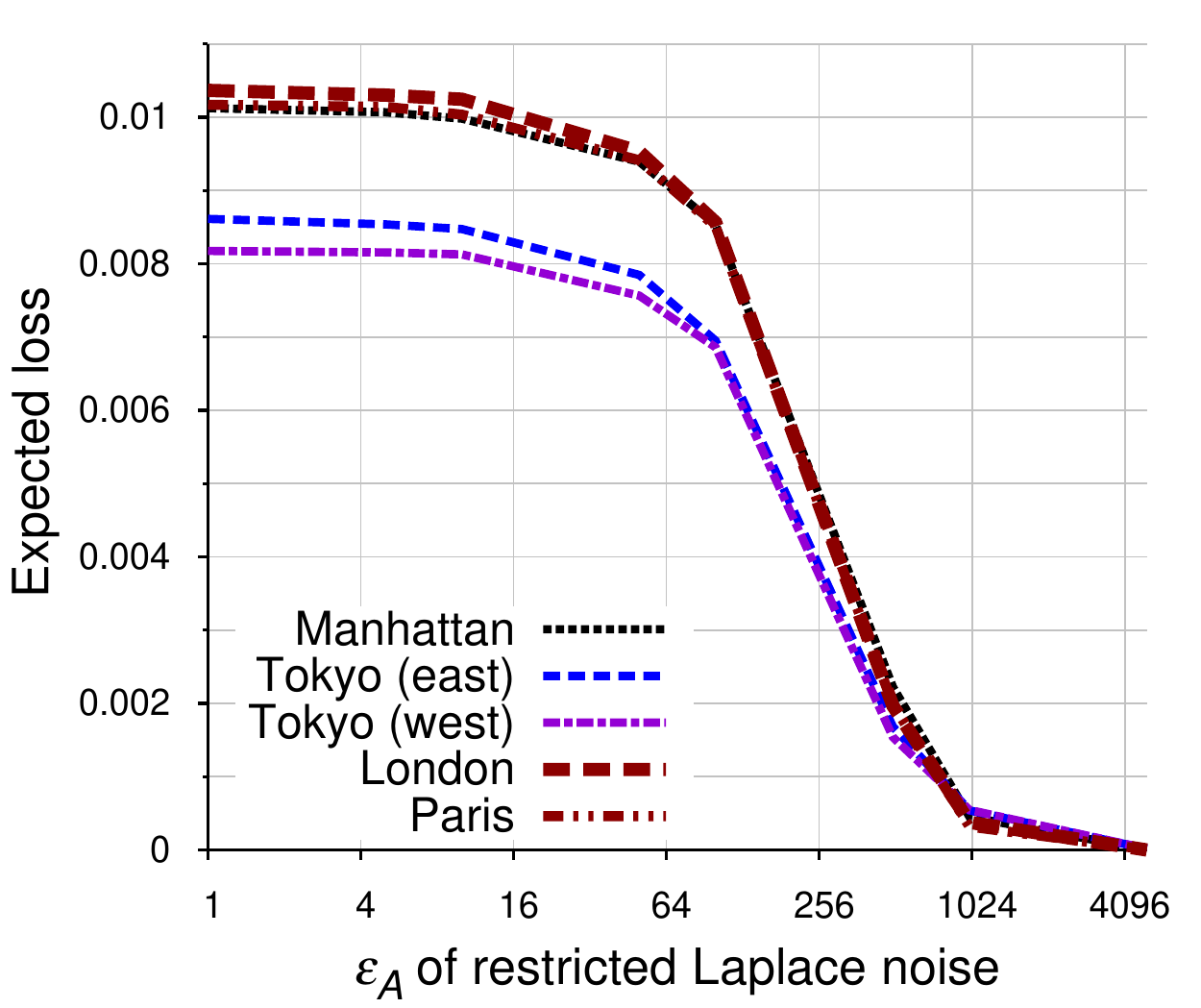}}}
\caption{Relationship between the expected loss and $\varepsilon_\alg$ of $(\varepsilon_\alg, \allowbreak r)$-\RL{} mechanism (with $5$ dummies).
\label{fig:cities+M:tupling:loss}}
\end{subfigure}
\vspace{-3mm}
\caption{Empirical \DistP{} and loss for \male{}/\female{} in different cities.
\label{fig:cities+MF:tupling:privacy}}
\vspace{1mm}
\end{minipage}
\\
\begin{minipage}{1.0\hsize}
\centering
\begin{subfigure}[t]{0.24\textwidth}
  \mbox{\raisebox{-10pt}{\includegraphics[height=26.0mm, width=1.00\textwidth]{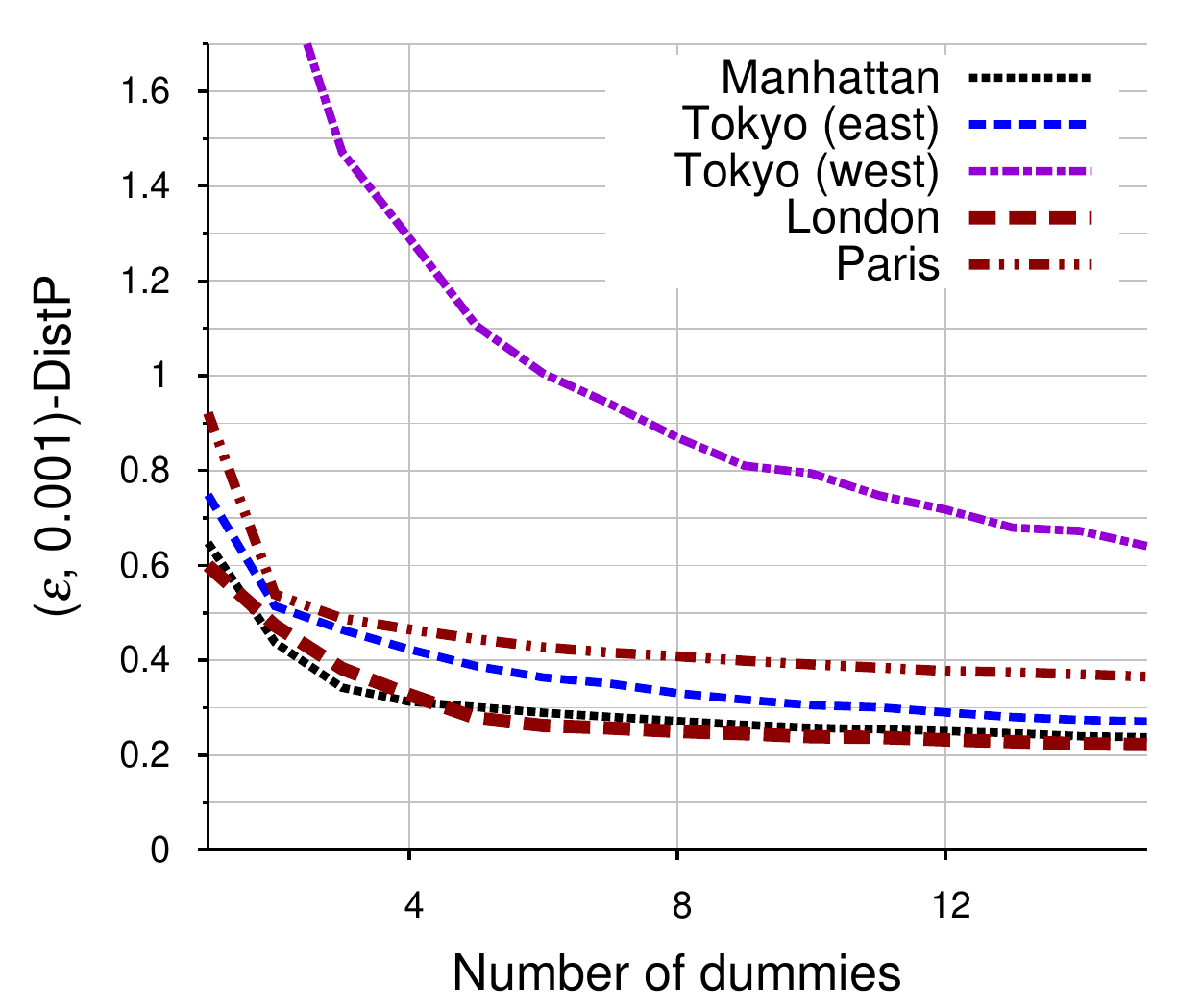}}}
\caption{Relationship between $\varepsilon$-\DistP{} and \#dummies (when using $(100, 0.020)$-\RL{} mechanism).\label{fig:cities+SI:dummies:DistP}}
\end{subfigure}\hspace{0.4ex}\hfill
\begin{subfigure}[t]{0.24\textwidth}
  \mbox{\raisebox{-10pt}{\includegraphics[height=26.0mm, width=1.00\textwidth]{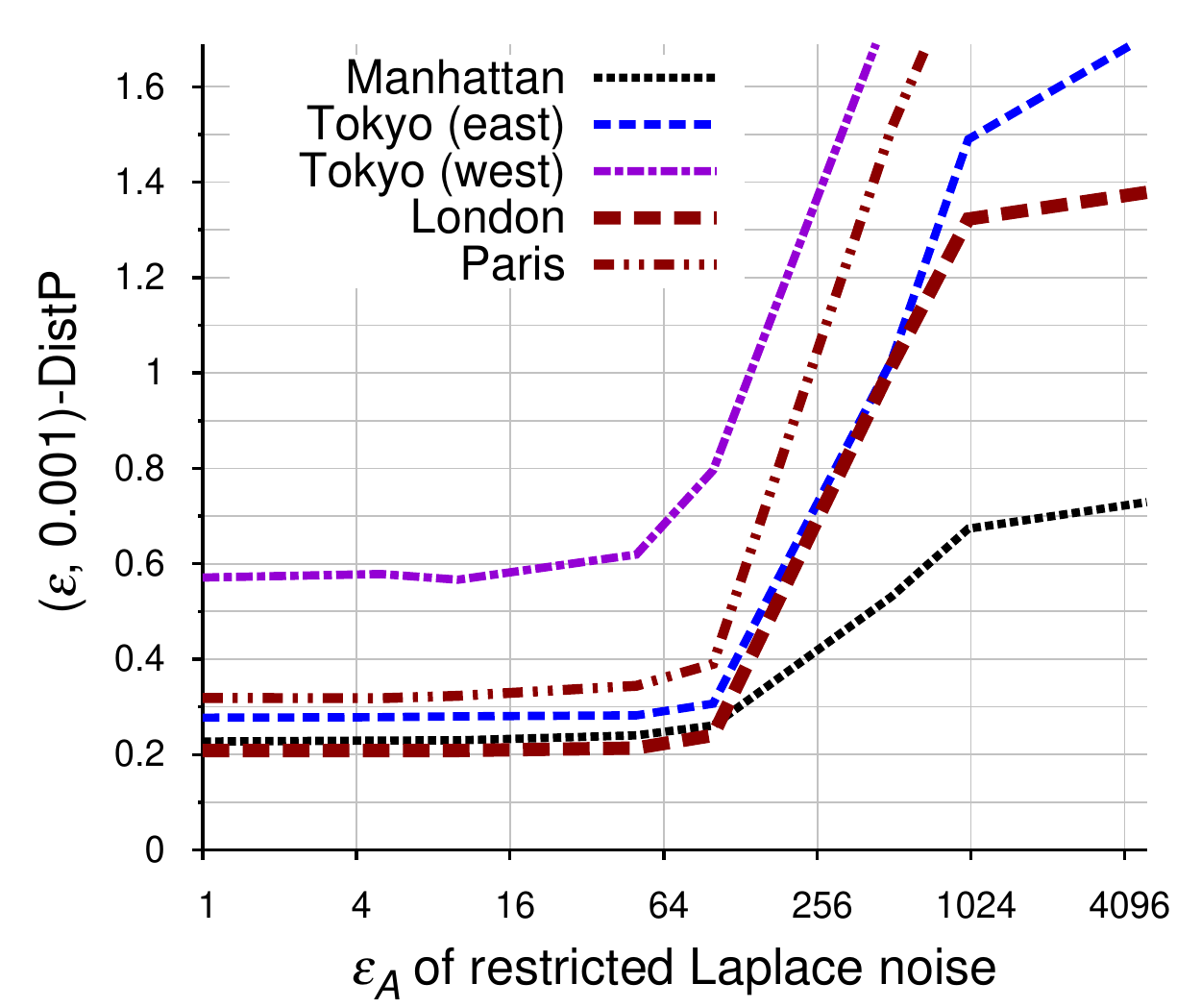}}}
\caption{Relationship between $\varepsilon$-\DistP{} and $\varepsilon_\alg$ of $(\varepsilon_\alg, 0.020)$-\RL{} mechanism (with $10$ dummies).\label{fig:cities+SI:DP-noises:DistP}}
\end{subfigure}\hspace{0.4ex}\hfill
\begin{subfigure}[t]{0.24\textwidth}
  \mbox{\raisebox{-10pt}{\includegraphics[height=26.0mm, width=1.00\textwidth]{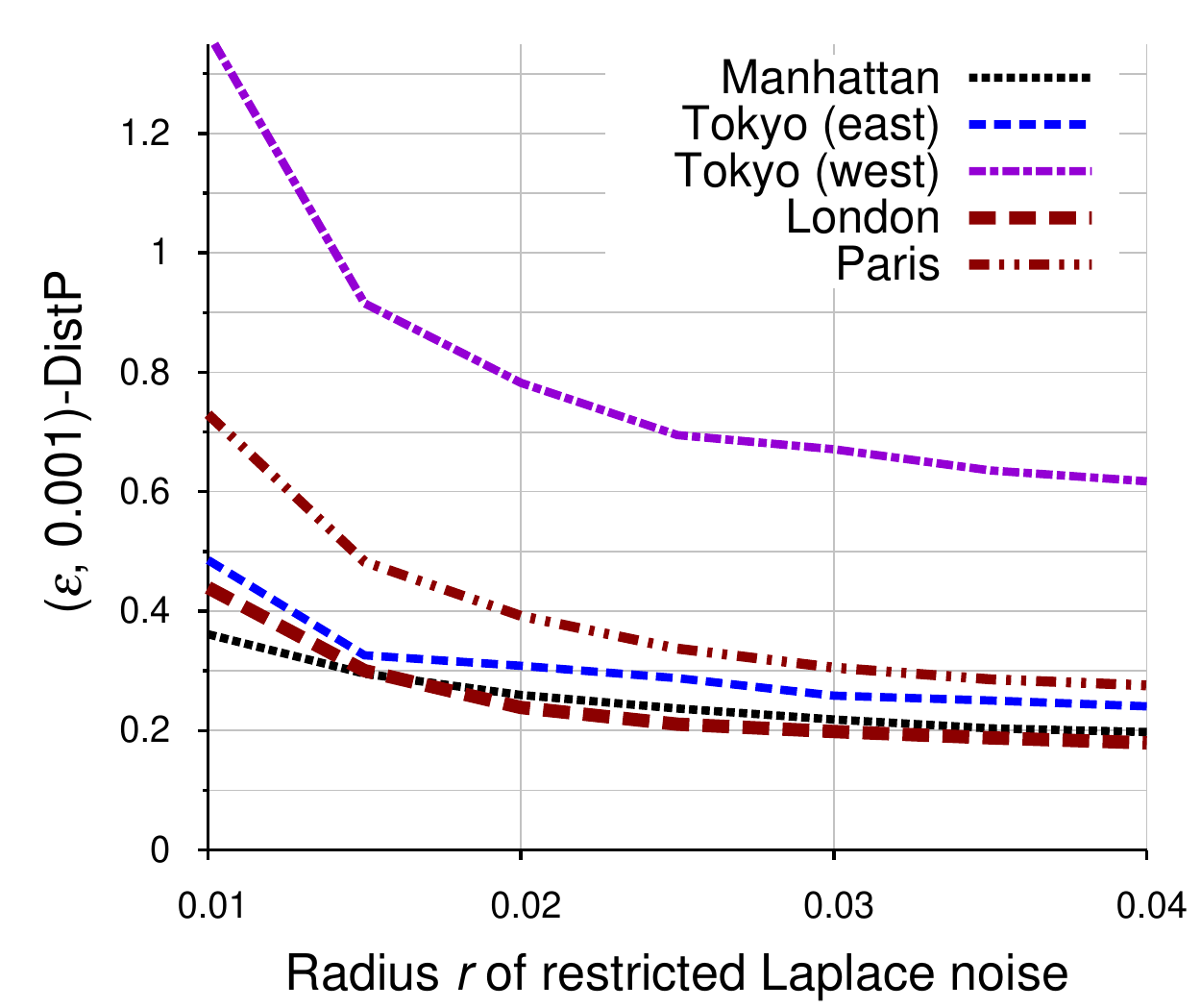}}}
\caption{Relationship between $\varepsilon$-\DistP{} and a radius $r$ of $(100, r)$-\RL{} mechanism (with $10$ dummies).\label{fig:cities+SI:radius:DistP}}
\end{subfigure}\hspace{0.4ex}\hfill
\begin{subfigure}[t]{0.24\textwidth}
  \mbox{\raisebox{-10pt}{\includegraphics[height=26.0mm, width=1.00\textwidth]{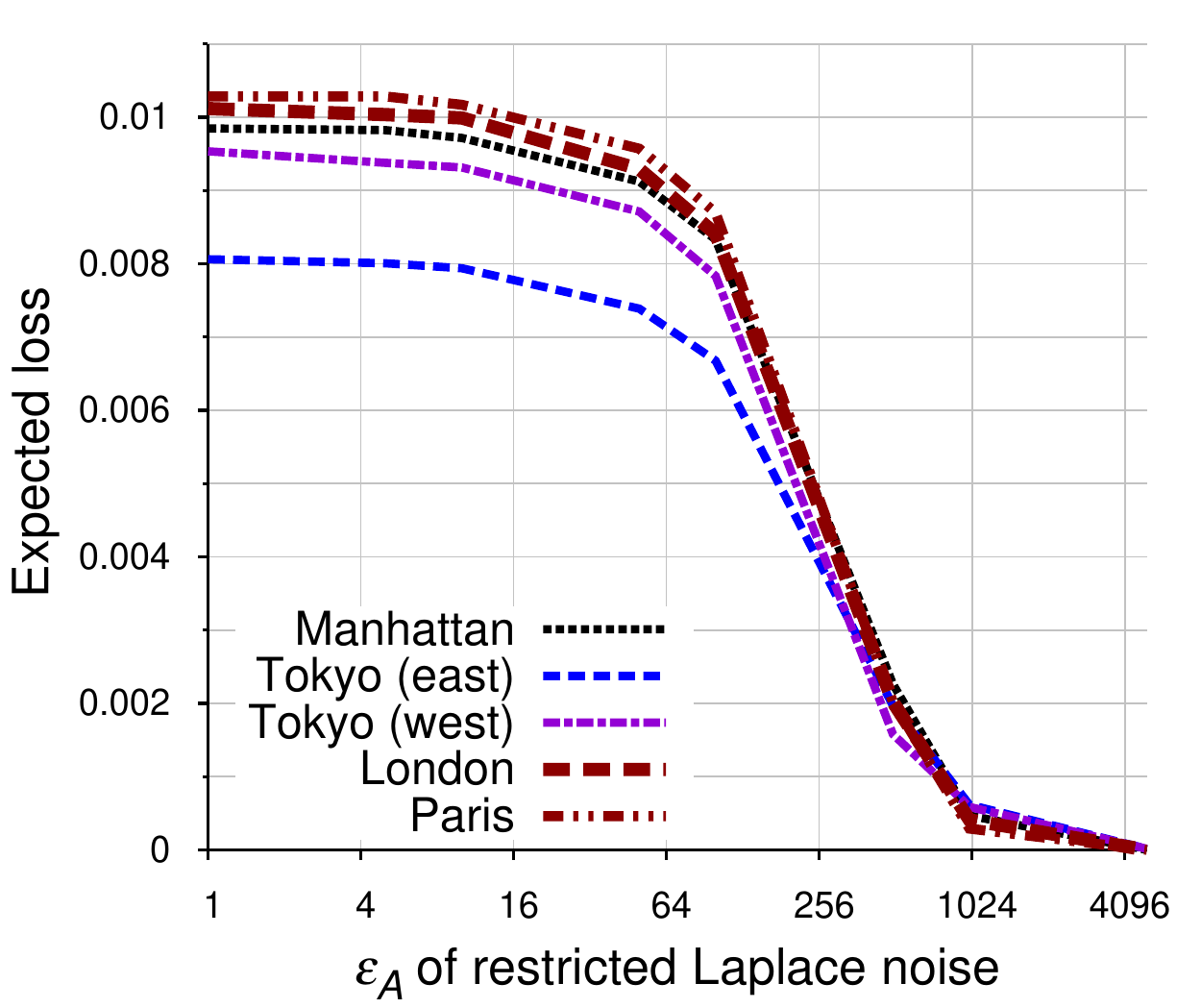}}}
\caption{Relationship between the expected loss and $\varepsilon_\alg$ of $(\varepsilon_\alg, \allowbreak r)$-\RL{} mechanism (with $5$ dummies).
\label{fig:cities+S:tupling:loss}}
\end{subfigure}
\vspace{-3mm}
\caption{Empirical \DistP{} and loss for \social{}/\lesssocial{} in different cities.
\label{fig:cities+SI:tupling:privacy}}
\vspace{1mm}
\end{minipage}
\\
\begin{minipage}{1.0\hsize}
\centering
\begin{subfigure}[t]{0.24\textwidth}
  \mbox{\raisebox{-10pt}{\includegraphics[height=26.0mm, width=1.00\textwidth]{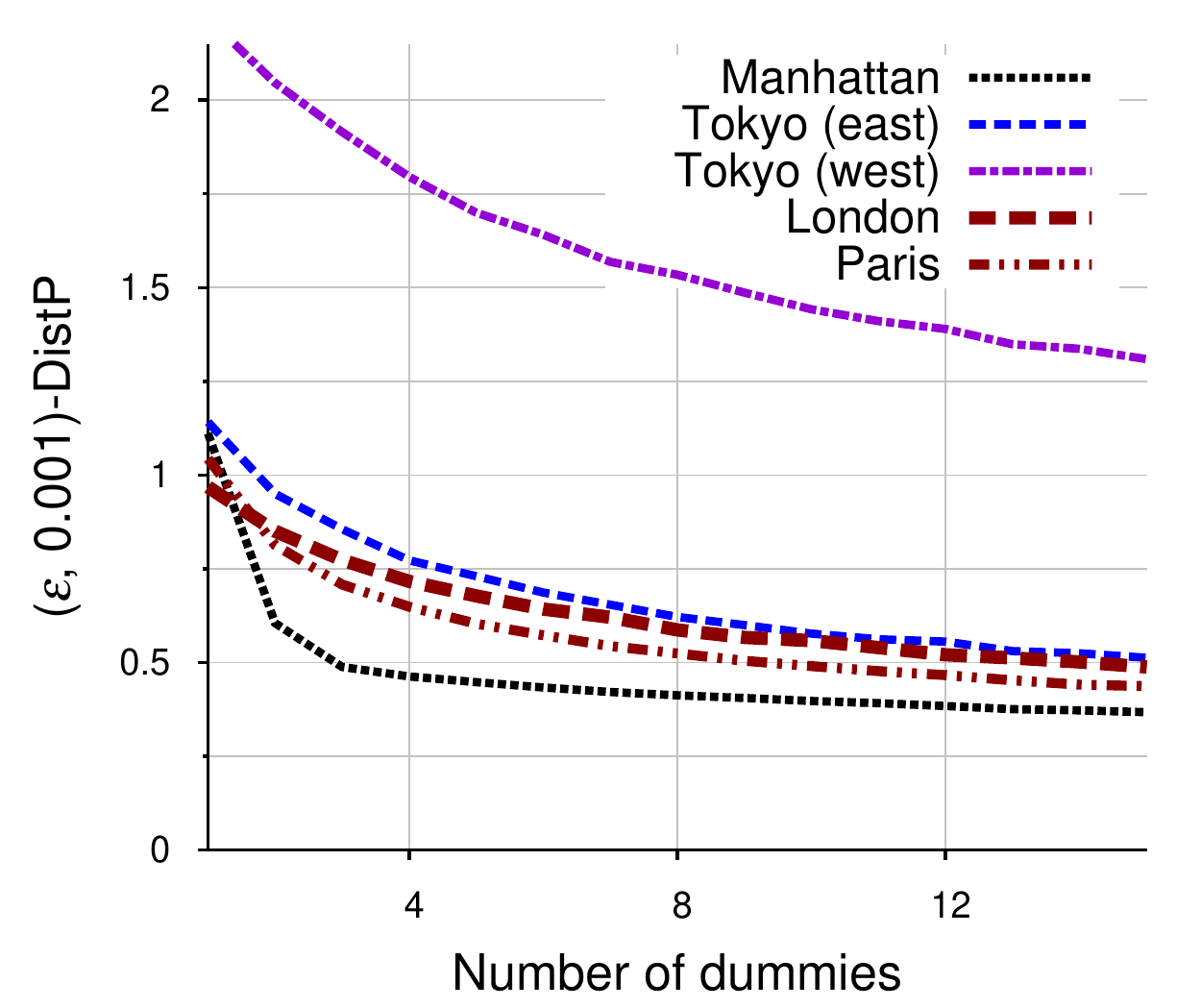}}}
\caption{Relationship between $\varepsilon$-\DistP{} and \#dummies (when using $(100, 0.020)$-\RL{} mechanism).\label{fig:cities+WP:dummies:DistP}}
\end{subfigure}\hspace{0.4ex}\hfill
\begin{subfigure}[t]{0.24\textwidth}
  \mbox{\raisebox{-10pt}{\includegraphics[height=26.0mm, width=1.00\textwidth]{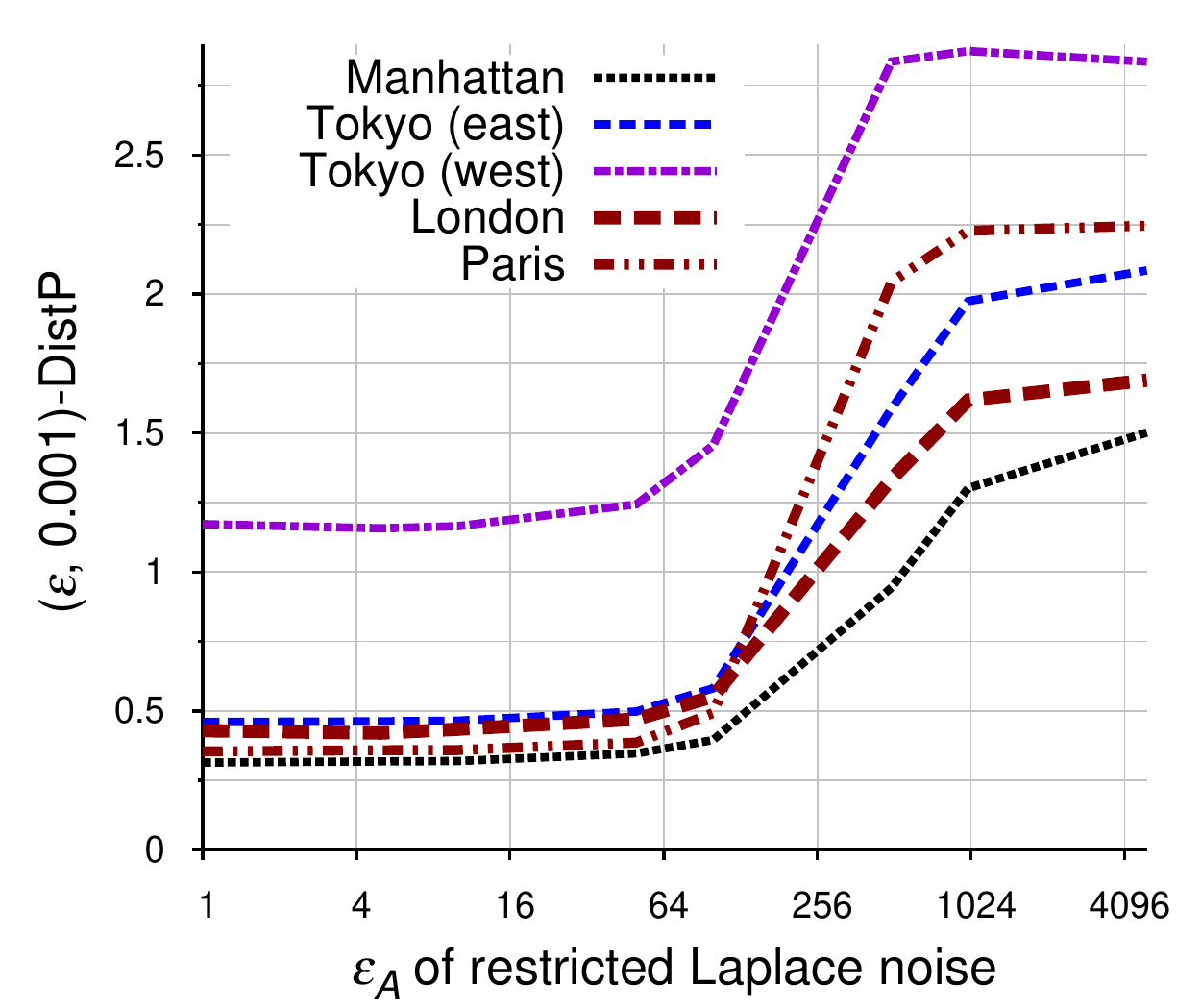}}}
\caption{Relationship between $\varepsilon$-\DistP{} and $\varepsilon_\alg$ of $(\varepsilon_\alg, 0.020)$-\RL{} mechanism (with $10$ dummies).\label{fig:cities+WP:DP-noises:DistP}}
\end{subfigure}\hspace{0.4ex}\hfill
\begin{subfigure}[t]{0.24\textwidth}
  \mbox{\raisebox{-10pt}{\includegraphics[height=26.0mm, width=1.00\textwidth]{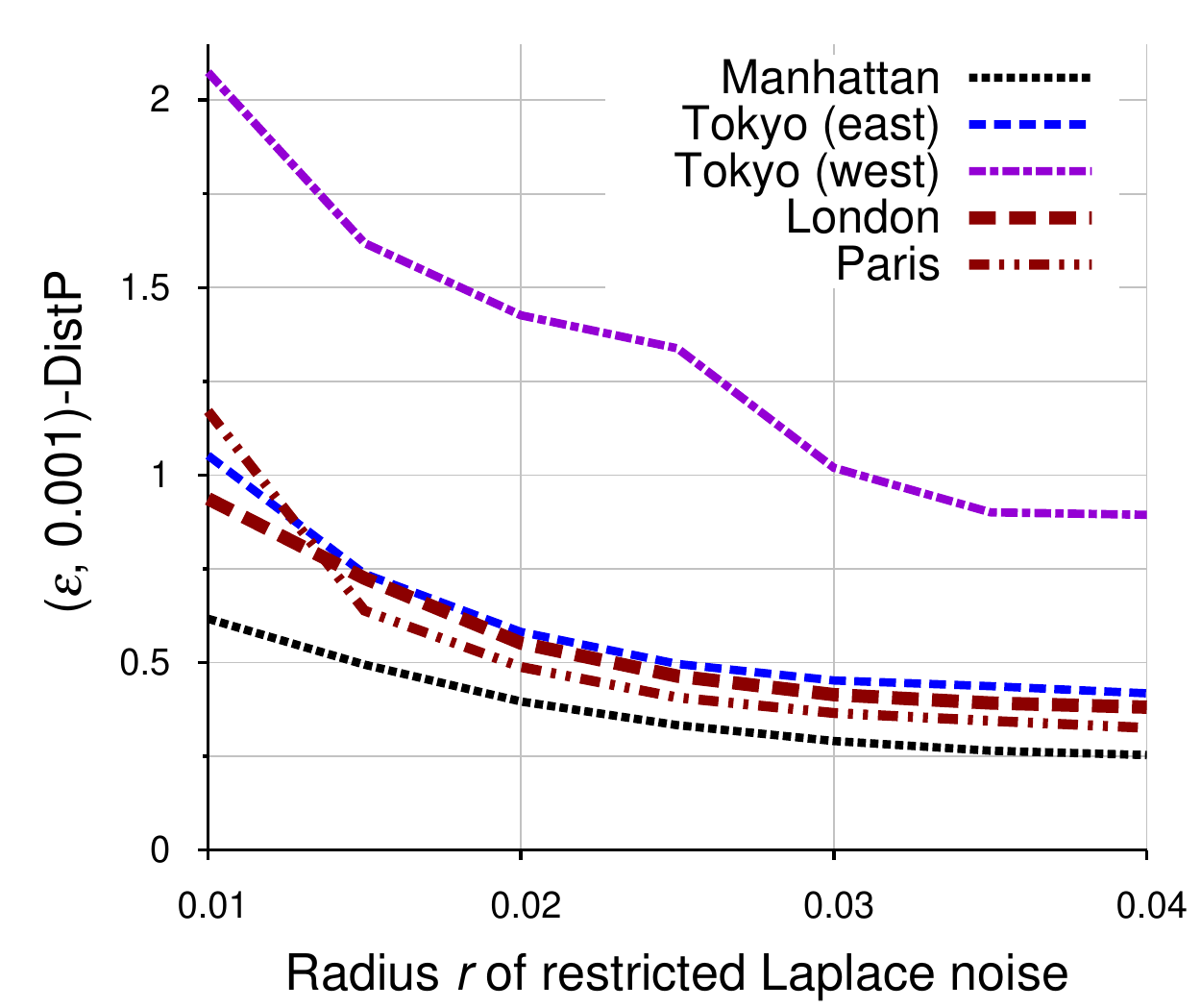}}}
\caption{Relationship between $\varepsilon$-\DistP{} and a radius $r$ of $(100, r)$-\RL{} mechanism (with $10$ dummies).\label{fig:cities+WP:radius:DistP}}
\end{subfigure}\hspace{0.4ex}\hfill
\begin{subfigure}[t]{0.24\textwidth}
  \mbox{\raisebox{-10pt}{\includegraphics[height=26.0mm, width=1.00\textwidth]{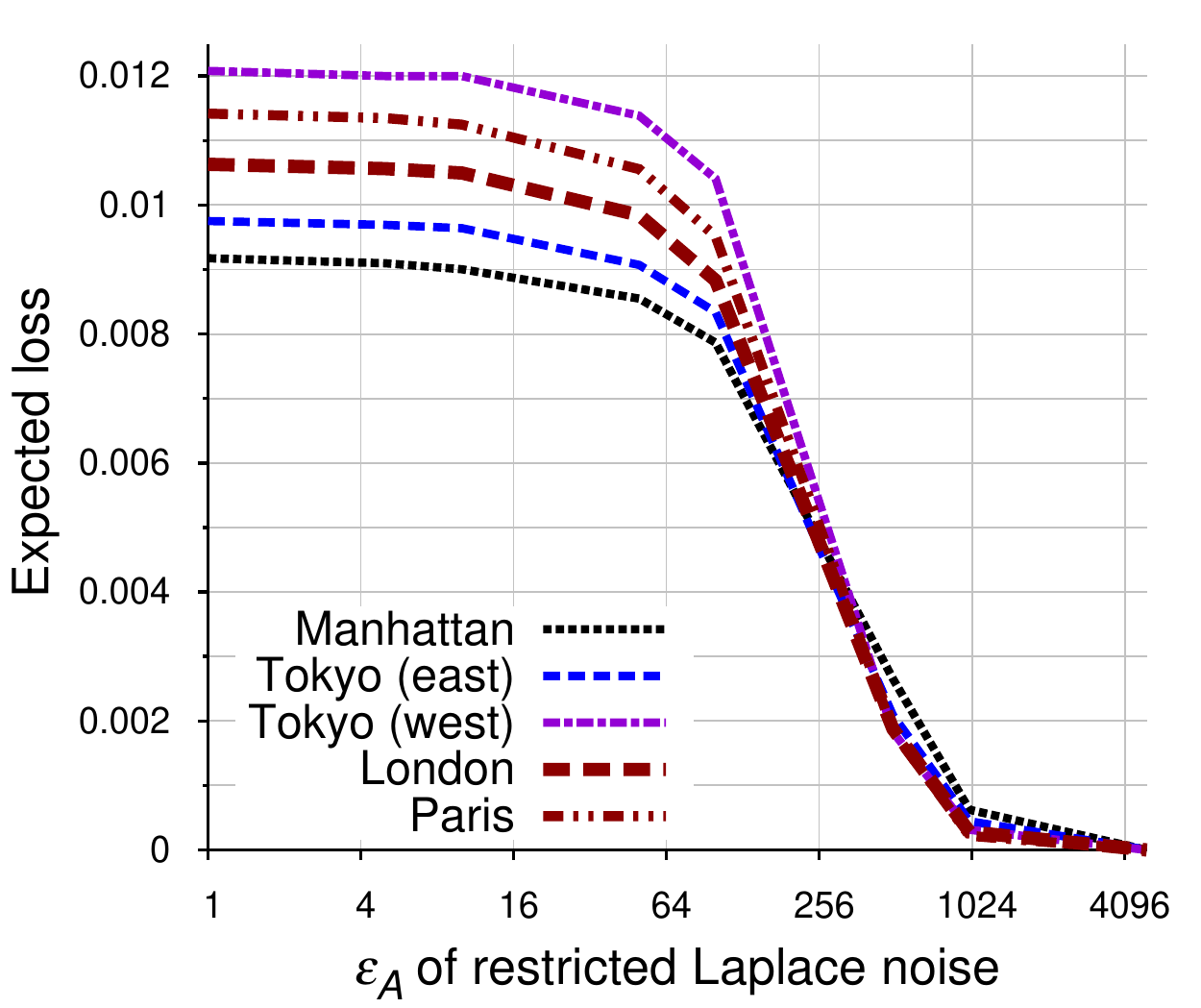}}}
\caption{Relationship between the expected loss and $\varepsilon_\alg$ of $(\varepsilon_\alg, \allowbreak r)$-\RL{} mechanism (with $5$ dummies).
\label{fig:cities+W:tupling:loss}}
\end{subfigure}
\vspace{-3mm}
\caption{Empirical \DistP{} and loss for \workplace{}/\nonworkplace{} in different cities.
\label{fig:cities+WP:tupling:privacy}}
\vspace{1mm}
\end{minipage}
\\
\begin{minipage}{1.0\hsize}
\centering
\begin{subfigure}[t]{0.24\textwidth}
  \mbox{\raisebox{-10pt}{\includegraphics[height=26.0mm, width=1.00\textwidth]{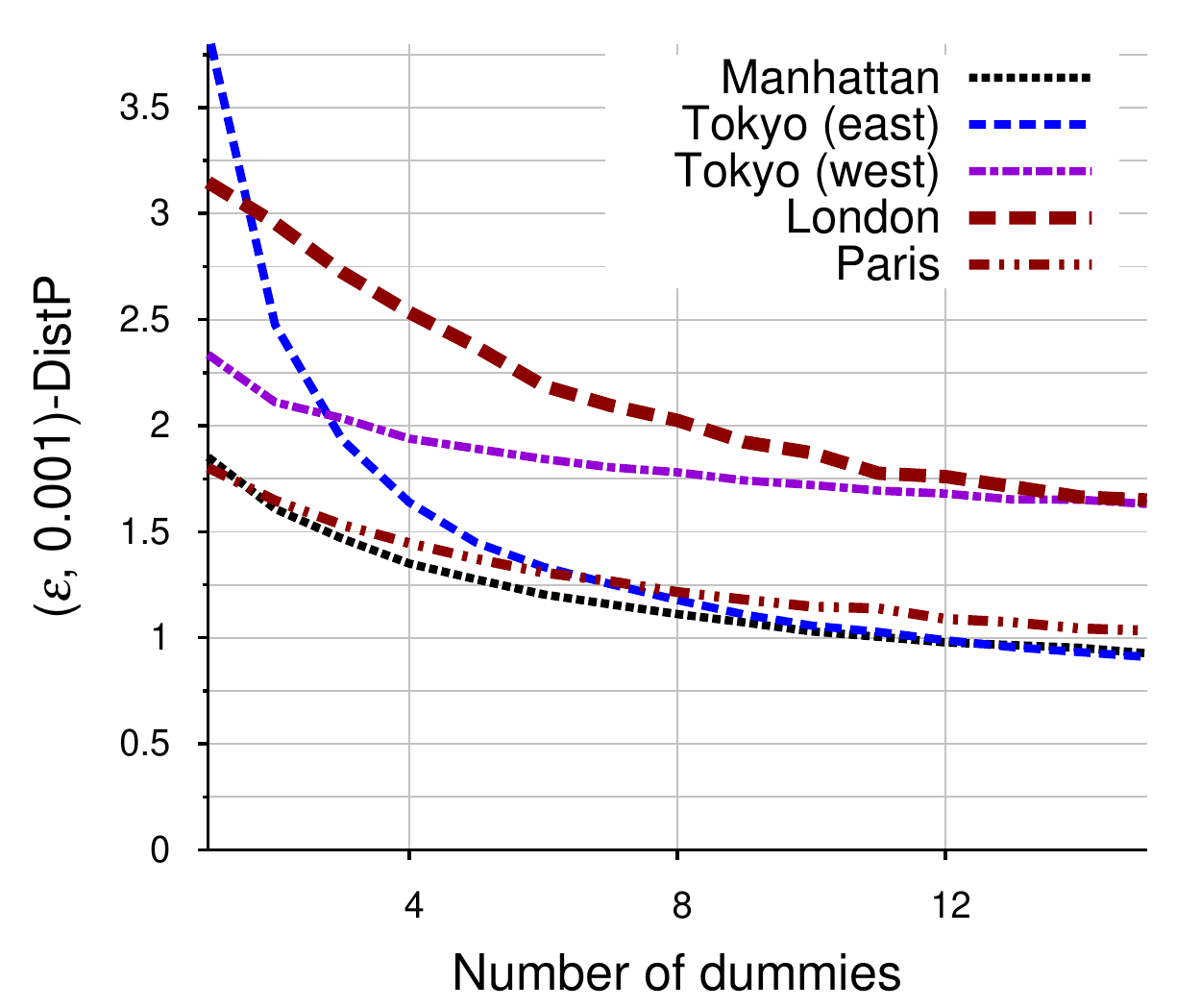}}}
\caption{Relationship between $\varepsilon$-\DistP{} and \#dummies (when using $(100, 0.020)$-\RL{} mechanism).\label{fig:cities+HO:dummies:DistP}}
\end{subfigure}\hspace{0.4ex}\hfill
\begin{subfigure}[t]{0.24\textwidth}
  \mbox{\raisebox{-10pt}{\includegraphics[height=26.0mm, width=1.00\textwidth]{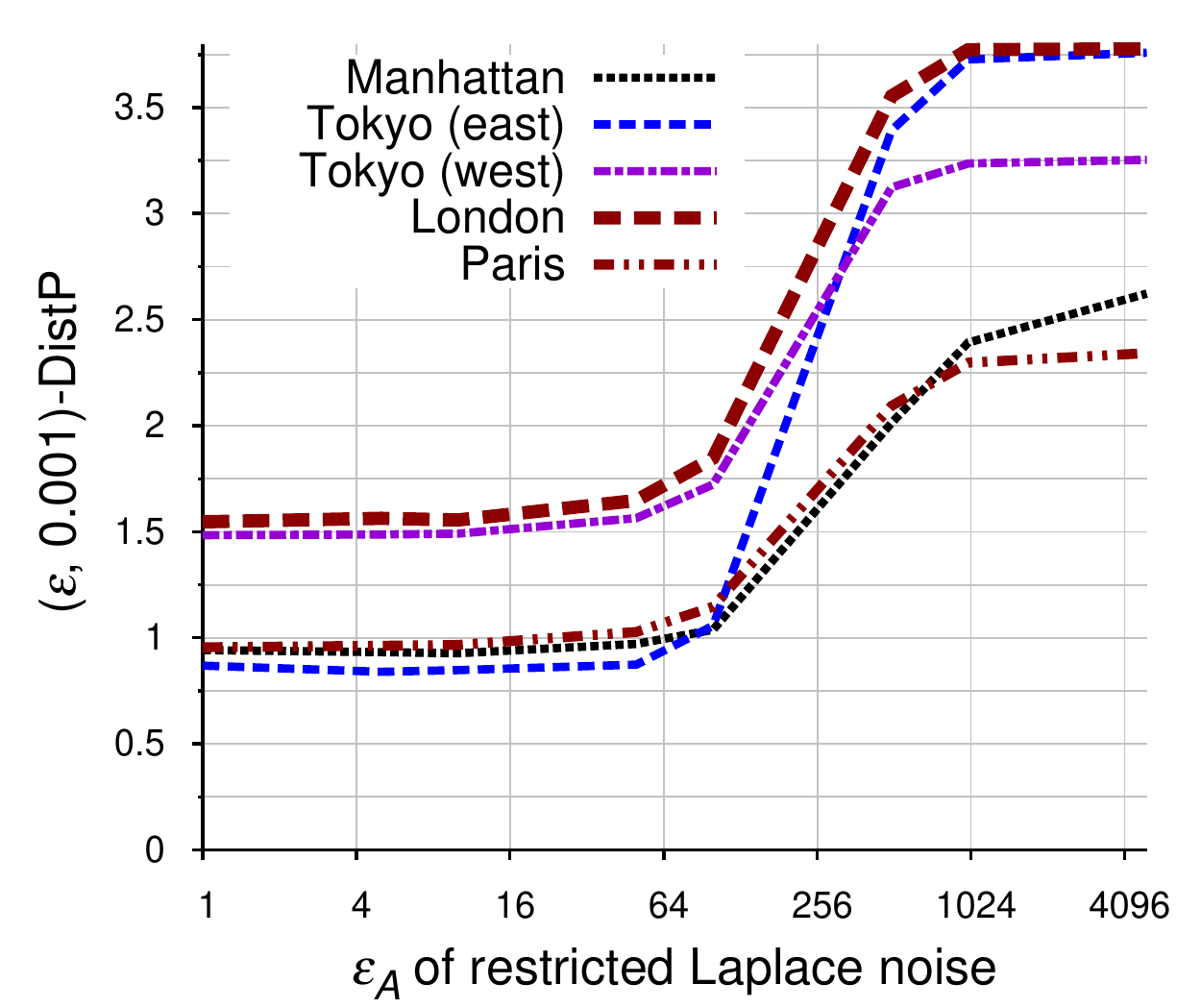}}}
\caption{Relationship between $\varepsilon$-\DistP{} and $\varepsilon_\alg$ of $(\varepsilon_\alg, 0.020)$-\RL{} mechanism (with $10$ dummies).\label{fig:cities+HO:DP-noises:DistP}}
\end{subfigure}\hspace{0.4ex}\hfill
\begin{subfigure}[t]{0.24\textwidth}
  \mbox{\raisebox{-10pt}{\includegraphics[height=26.0mm, width=1.00\textwidth]{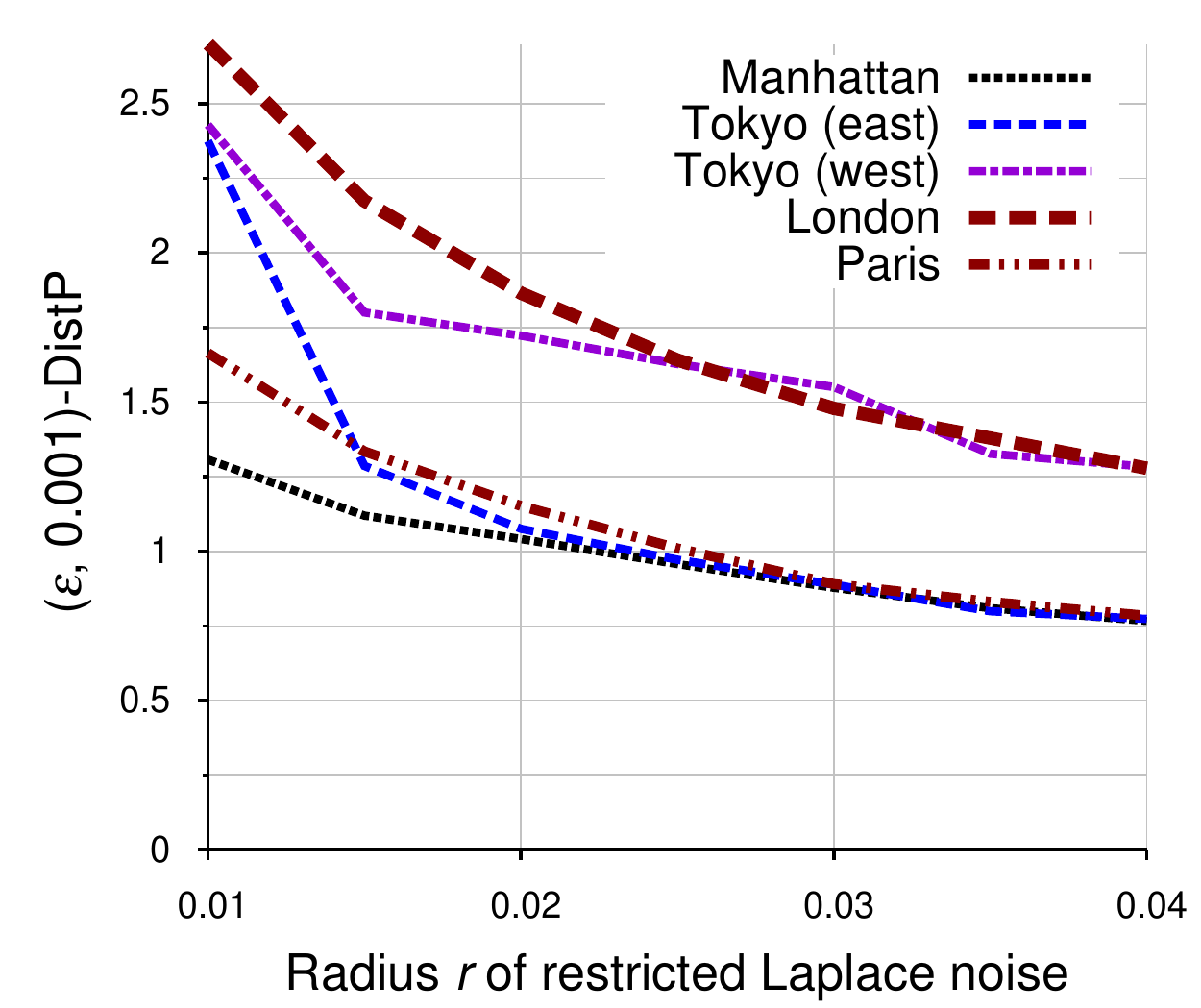}}}
\caption{Relationship between $\varepsilon$-\DistP{} and a radius $r$ of $(100, r)$-\RL{} mechanism (with $10$ dummies).\label{fig:cities+HO:radius:DistP}}
\end{subfigure}\hspace{0.4ex}\hfill
\begin{subfigure}[t]{0.24\textwidth}
  \mbox{\raisebox{-10pt}{\includegraphics[height=26.0mm, width=1.00\textwidth]{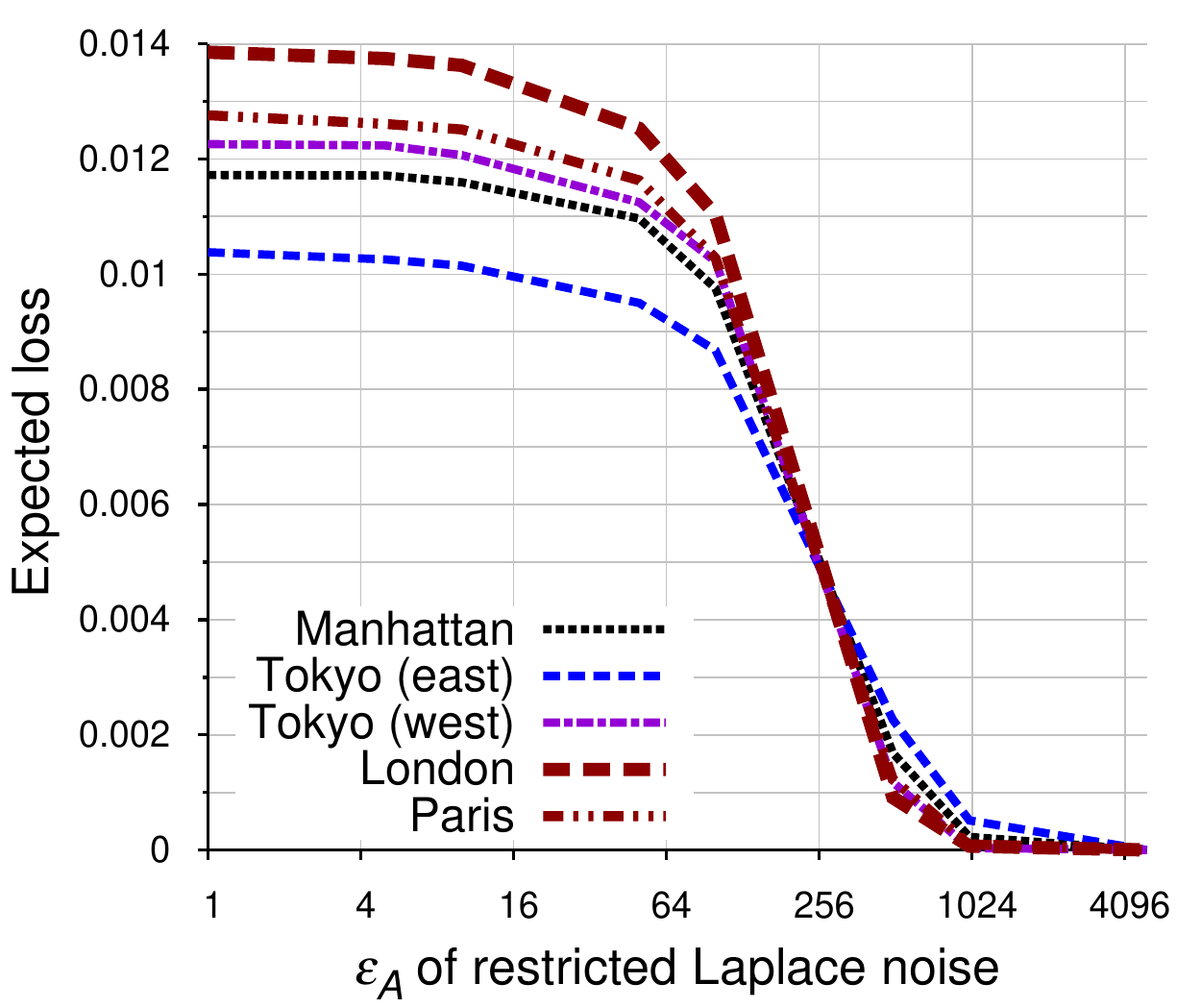}}}
\caption{Relationship between the expected loss and $\varepsilon_\alg$ of $(\varepsilon_\alg, \allowbreak r)$-\RL{} mechanism (with $5$ dummies).
\label{fig:cities+H:tupling:loss}}
\end{subfigure}
\vspace{-3mm}
\caption{Empirical \DistP{} and loss for \home{}/\outside{} in different cities.
\label{fig:cities+HO:tupling:privacy}}
\end{minipage}
\end{tabular}
\end{figure}

\newpage
\section{Proofs of Technical Results}
\label{sec:proofs}

\allowdisplaybreaks[1]

\subsection{Distribution Obfuscation by Point Obfuscation}
\label{sub:DistP-by-PointP}

In this section we show the results on how point obfuscation mechanisms provide distribution privacy.

\DPimpliesMaxDistP*

\begin{proof}
Assume $\alg$ provides $(\varepsilon, \delta)$-\DP{} w.r.t. $\varPhi$.
Let $\delta' = \delta\cdot|\varPhi|$, and
$(\lambda_0, \lambda_1)\in\lift{\varPhi}$.
By Definition~\ref{def:lifting-relations}, there is a coupling $\gamma\in\cp{\lambda_0}{\lambda_1}$ satisfying $\supp(\gamma)\subseteq\varPhi$.
By Definition~\ref{def:coupling}, we have:
\begin{align}
\text{for each $x_0\in\calx$, }~~  \lambda_0[x_0] &= {\textstyle\sum_{x_1\in\calx}}\, \gamma[x_0,x_1] \label{eq:coupling-project1} \\
\text{for each $x_1\in\calx$, }~~  \lambda_1[x_1] &= {\textstyle\sum_{x_0\in\calx}}\, \gamma[x_0,x_1] \label{eq:coupling-project2}
{.}
\end{align}

Let $R\subseteq\caly$ such that $\lift{\alg}(\lambda_0)[R] > \delta'$.
We first show that $\lift{\alg}(\lambda_1)[R] > 0$ as follows.
To derive a contradiction, we assume $\lift{\alg}(\lambda_1)[R] = 0$.
Then for any $x_1\in\supp(\lambda_1)$, $\alg(x_1)[R] = 0$.
Let $S_{x_1} = \{ x_0 \in\calx \mid (x_0, x_1)\in\varPhi \}$.
For any $x_0 \in S_{x_1}$, we have:
\begin{align*}
\alg(x_0)[R] 
&\le
e^{\varepsilon} \alg(x_1)[R] + \delta
&\! \text{(by the $(\varepsilon, \delta)$-\DP{} of $\!\alg$ w.r.t. $\!\varPhi$)}
\\ &\le \delta
&\! \text{(by $\alg(x_1)[R] = 0$)}
\end{align*}
Thus we obtain:
\begin{align*}
\lift{\alg}(\lambda_0)[R] 
&= 
\hspace{-1ex}\sum_{x_0\in\supp(\lambda_0)}\hspace{-2ex} \lambda_0[x_0] \alg(x_0)[R]
\\ &\le
|\supp(\lambda_0)| \cdot\delta
\\ &\le
|\varPhi| \cdot\delta
\hspace{5ex}\text{(by $|\supp(\lambda_0)| \le |\supp(\gamma)| \le |\varPhi|$)}
\\ &=
\delta'
\end{align*}
This contradicts the definition of $R$.
Hence $\lift{\alg}(\lambda_1)[R] > 0$.

Then we calculate the ratio of the probability that $\lift{\alg}$ outputs an element of $R$ given input $\lambda_0$ to that given input~$\lambda_1$:
\begin{align*}
&
\frac%
{\small \mbox{$\lift{\alg}(\lambda_0)[R] - \delta'$}}
{\small \mbox{$\lift{\alg}(\lambda_1)[R]$}}
\\ =&
\frac%
{ \bigl( \sum_{x_0\in\calx}
  \lambda_0[x_0] \cdot \alg(x_0)[R] \bigr) - \delta'~~
}
{ \sum_{x_1\in\calx}
  \lambda_1[x_1] \cdot \alg(x_1)[R]~~
}
\\ =&
\frac%
{ \sum_{(x_0, x_1) \in \supp(\gamma)}
  \hspace{-0ex}
  \bigl( \gamma[x_0, x_1] \alg(x_0)[R]
   - {\scriptstyle\frac{\delta'}{|\supp(\gamma)|}} \bigr)~
}
{ \sum_{(x_0, x_1) \in \supp(\gamma)}
  \gamma[x_0, x_1] \alg(x_1)[R]~~
}
~
\text{(by \eqref{eq:coupling-project1}, \eqref{eq:coupling-project2})}
\\[0.8ex] \le &
\max_{(x_0, x_1) \in \supp(\gamma)}\hspace{-1ex}
\frac%
{ \displaystyle \gamma[x_0, x_1] \cdot \alg(x_0)[R]
   - {\scriptstyle\frac{\delta'}{|\supp(\gamma)|}}~
}
{ \displaystyle \gamma[x_0, x_1] \cdot \alg(x_1)[R]~~
}
\\[0.8ex] \le &
\max_{(x_0, x_1) \in \supp(\gamma)}\hspace{-1ex}
\frac%
{ \displaystyle \alg(x_0)[R]
   - \delta~~
}
{ \displaystyle \alg(x_1)[R]~~
}
\hspace{6ex}
\text{(by $\scriptstyle-\frac{\delta\cdot|\varPhi|}{\gamma[x_0, x_1]\cdot|\supp(\gamma)|}\le-\delta$)}
\\[0.8ex] \le &~
e^{\varepsilon}
{.}
\hspace{5ex}
\text{(by $(x_0,x_1)\in\supp(\gamma) \subseteq \varPhi$ and $(\varepsilon, \delta)$-\DP{} w.r.t. $\varPhi$)}
\end{align*}
Therefore $\alg$ provides $(\varepsilon,\, \delta\cdot|\varPhi|)$-\DistP{} w.r.t. $\lift{\varPhi}$.
\qed
\end{proof}

\maxDistPFromDP*

\begin{proof}
Assume that $\alg$ provides $(\varepsilon, \utmetric, \delta)$-\XDP{} w.r.t. $\varPhi$.
Let $(\lambda_0, \lambda_1)\in\liftinf{\varPhi}$.
By definition,
there exists a coupling $\gamma\in\cp{\lambda_0}{\lambda_1}$ that satisfies $\supp(\gamma)\subseteq\varPhi$ and $\gamma\in\GammaInf(\lambda_0, \lambda_1)$.
Then $\gamma$ minimizes the sensitivity $\sens_{\supp(\gamma)}$; i.e.,  
$\displaystyle\gamma\in\!\hspace{-1.3ex}\argmin_{\gamma'\in\cp{\lambda_0}{\lambda_1}}\hspace{-1ex}\!\sens_{\supp(\gamma')}$.

Let $\delta' = \delta\cdot|\varPhi|$.
Let $R\subseteq\caly$ such that $\lift{\alg}(\lambda_0)[R] > \delta'$.
Analogously to the proof for Theorem~\ref{thm:DPimpliesMaxDistP}, 
we obtain $\lift{\alg}(\lambda_1)[R] > 0$.
Then it follows from $\supp(\gamma) \subseteq \varPhi$ and $(\varepsilon, \utmetric, \delta)$-\XDP{} w.r.t. $\varPhi$ that:
\begin{align*}
\frac%
{\small \mbox{$\lift{\alg}(\lambda_0)[R] - \delta'$}}
{\small \mbox{$\lift{\alg}(\lambda_1)[R]$}}
\le &
\max_{(x_0, x_1) \in \supp(\gamma)}\hspace{-1ex}
\frac%
{ \displaystyle \alg(x_0)[R]
   - \delta~~
}
{ \displaystyle \alg(x_1)[R]~~
}
\\[0.8ex] \le &
\max_{\substack{~\\[0.5ex](x_0, x_1) \in \supp(\gamma)}}
e^{\varepsilon \utmetric(x_0, x_1)}
\\ = 
&
e^{\varepsilon \sensinf(\lambda_0, \lambda_1)}
{.}
\hspace{3ex}
\text{(by $\gamma\in\GammaInf(\lambda_0, \lambda_1)$)}
\end{align*}
Therefore $\alg$ provides $(\varepsilon, \sensinf, \delta')$-\XDistP{} w.r.t. $\liftinf{\varPhi}$.
\qed
\end{proof}

\subsection{Distribution Privacy with Attacker's Close Beliefs}
\label{sub:proofs:DistP:beliefs}

Next we show the propositions on distribution privacy with an attacker's beliefs.

\beliefsXDistP*

\begin{proof}
We obtain the proposition by:
\begin{align*}
\lift{\alg}(\tilde{\lambda_0})[R] &\leq
e^{\varepsilon c} \cdot \lift{\alg}(\lambda_0)[R]
& \text{(by $\utmetric(\lambda_0, \tilde{\lambda_0}) \le c$)}
\\ &\leq
e^{\varepsilon \left( \utmetric(\lambda_0, \lambda_1) + c \right)}
\cdot \lift{\alg}(\lambda_1)[R]
\\ &\leq
e^{\varepsilon \left( \utmetric(\lambda_0, \lambda_1) + 2c \right)}
\cdot\lift{\alg}(\tilde{\lambda_1})[R]
& \text{~(by $\utmetric(\lambda_1, \tilde{\lambda_1}) \le c$)}
{.}
\end{align*}
\hfill\qed
\end{proof}

Similarly, we say that the attacker has \emph{close beliefs} w.r.t an adjacency relation $\varPsi$ if for each $(\lambda_0, \lambda_1)\in\varPsi$, we have $(\tilde{\lambda_0}, \lambda_0) \in \varPsi$ and $(\lambda_1, \tilde{\lambda_1}) \in \varPsi$.
Then we obtain:

\begin{restatable}[\DistP{} with close beliefs]{prop}{beliefsDistP}
\label{prop:beliefsDistP}
Let $\alg: \calx \rightarrow \Dists\caly$ provide $(\varepsilon,\delta)$-\DistP{} w.r.t. an adjacency relation $\varPsi \subseteq \Dists\calx\times\Dists\calx$.
Let $\delta' = (1+e^{\varepsilon}+e^{2\varepsilon})\delta$.
If an attacker has close beliefs w.r.t. $\varPsi$, 
then we obtain $(3\varepsilon, \delta')$-\DistP{} against the attacker, i.e., 
for all $(\lambda_0, \lambda_1)\in\varPsi$ and all $R \subseteq\caly$,\, we have:
\begin{align*}
\lift{\alg}(\tilde{\lambda_0})[R] \leq
e^{3\varepsilon}\cdot\lift{\alg}(\tilde{\lambda_1})[R] + \delta'
{.}
\end{align*}
\end{restatable}

\begin{proof}
We obtain the proposition by:
\begin{align*}
\lift{\alg}(\tilde{\lambda_0})[R] &\leq
e^{\varepsilon} \cdot \lift{\alg}(\lambda_0)[R] + \delta
& \text{(by $(\tilde{\lambda_0}, \lambda_0)\in\varPsi$)}
\\ &\leq
e^{\varepsilon} 
(e^{\varepsilon} \cdot \lift{\alg}(\lambda_1)[R] + \delta) + \delta
& \text{(by $(\lambda_0, \lambda_1)\in\varPsi$)}
\\ &\leq
e^{\varepsilon} 
(e^{2\varepsilon} \cdot \lift{\alg}(\tilde{\lambda_1})[R] + (1+e^\varepsilon)\delta) + \delta
& \text{(by $(\lambda_1, \tilde{\lambda_1})\in\varPsi$)}
\\ &\leq
e^{3\varepsilon} \cdot \lift{\alg}(\tilde{\lambda_1})[R] + 
(1+e^{\varepsilon}+e^{2\varepsilon})\delta
{.}
&\hfill\qed
\end{align*}
\end{proof}

\subsection{Point Obfuscation by Distribution Obfuscation}
\label{sub:proofs:DistP-implies-DP}

Next we show that \DP{} is an instance of \DistP{} if an adjacency relation includes pairs of point distributions.
\begin{definition}[Point distribution]\label{def:point-distribution}\rm
For each $x\in\calx$, the \emph{point distribution} $\eta_x$ of $x$ is the distribution over $\calx$ such that:
$\eta_x[x'] = 1$ if $x' = x$, and $\eta_x[x'] = 0$ otherwise.
\end{definition}

\begin{lemma}\label{lem:lifting-includes-point-dist}
Let $\varPhi\subseteq\calx\times\calx$.
For any $(x_0, x_1)\in\varPhi$, we have $(\eta_{x_0}, \eta_{x_1})\in\liftinf{\varPhi}$.
\end{lemma}

\begin{proof}
Let $(x_0, x_1)\in\varPhi$.
We define $\gamma\in\Dists(\calx\times\calx)$ by:
$\gamma[x, x'] = 1$ if $x = x_0$ and $x' = x_1$, and
$\gamma[x, x'] = 0$ otherwise.
Then $\gamma$ is the only coupling between $\eta_{x_0}$ and $\eta_{x_1}$, hence $\gamma\in\GammaInf(\eta_{x_0}, \eta_{x_1})$.
Also $\supp(\gamma) = \{(x_0, x_1)\} \subseteq \varPhi$.
Therefore 
by definition,
we obtain $(\eta_{x_0}, \eta_{x_1})\in\liftinf{\varPhi}$.
\qed
\end{proof}

\begin{restatable}[\DistP{} $\Rightarrow$ \DP{} and \XDistP{} $\Rightarrow$ \XDP{}]{thm}{DPfromDistP}
\label{thm:DPfromDistP}
Let $\varepsilon\in\realsnng$,\, $\varPhi\subseteq\calx\times\calx$, and $\alg:\calx\rightarrow\Dists\caly$ be a randomized algorithm.
\begin{enumerate}
\item
If $\alg$ provides $(\varepsilon, \delta)$-\DistP{} w.r.t. $\lift{\varPhi}$, 
it provides $(\varepsilon, \delta)$-\DP{} w.r.t. $\varPhi$.
\item
If $\alg$ provides $(\varepsilon, \Winfu, \delta)$-\XDistP{} w.r.t. $\liftinf{\varPhi}$, it provides $(\varepsilon, \utmetric, \delta)$-\XDP{} w.r.t.~$\varPhi$.
\end{enumerate}
\end{restatable}

\begin{proof}
We prove the first claim as follows.
Assume that $\alg$ provides $(\varepsilon, \delta)$-\DistP{} w.r.t. $\lift{\varPhi}$.
Let $(x_0, x_1)\in\varPhi$, and
$\eta_{x_0}$ and $\eta_{x_1}$ be the point distributions, 
defined in Definition~\ref{def:point-distribution}.
By Lemma~\ref{lem:lifting-includes-point-dist} and $\liftinf{\varPhi}\subseteq\lift{\varPhi}$, we have $(\eta_{x_0}, \eta_{x_1})\in\lift{\varPhi}$.
It follows from $(\varepsilon, \delta)$-\DistP{} that for any $R\subseteq\caly$, we obtain:
\begin{align*}
\alg(x_0)[R] =
\lift{\alg}(\eta_{x_0})[R] \leq
e^{\varepsilon}\cdot\lift{\alg}(\eta_{x_1})[R] + \delta =
e^{\varepsilon}\cdot\alg(x_1)[R] + \delta
.
\end{align*}
Hence $\alg$ provides $(\varepsilon, \delta)$-\DP{} w.r.t.~$\varPhi$.

Next we show the second claim.
Assume that $\alg$ provides $(\varepsilon, \Winfu, \delta)$-\XDistP{} w.r.t. $\liftinf{\varPhi}$.
Let $(x_0, x_1)\in\varPhi$, and
$\eta_{x_0}$ and $\eta_{x_1}$ be the point distributions.
By Lemma~\ref{lem:lifting-includes-point-dist}, we have $(\eta_{x_0}, \eta_{x_1})\in\liftinf{\varPhi}$.
Then for any $R\subseteq\caly$, we obtain:
\begin{align*}
\alg(x_0)[R] &= \lift{\alg}(\eta_{x_0})[R] 
\\ &\leq
e^{\varepsilon \Winfu(\eta_{x_0}, \eta_{x_1})}\cdot\lift{\alg}(\eta_{x_1})[R] + \delta 
~~~\text{(by \XDistP{} of $\alg$)}
\\ &=
e^{\varepsilon \utmetric(x_0, x_1)}\cdot\alg(x_1)[R] + \delta
\end{align*}
where the last equality follows from the definition of $\Winfu$.
Hence $\alg$ provides $(\varepsilon, \utmetric, \delta)$-\XDP{} w.r.t.~$\varPhi$.
\qed
\end{proof}

\begin{figure}[t]\label{fig:compositions}
\centering
\begin{subfigure}[t]{0.45\textwidth}
\centering
\begin{picture}(105, 75)
 \put( 57, 69){$\alg_1 \Seq \alg_0$}
 \thicklines
 \put( 53, 40){\framebox(40,20){$\alg_0$}}
 \put( 53,  0){\framebox(40,20){$\alg_1$}}

 \put(   0,  30){\vector(  1,  0){36}}
 \put(  36,  10){\line(  0,  1){40}}
 \put(  36,  50){\vector(  1,  0){16}}
 \put(  36,  10){\vector(  1,  0){16}}
 %\put( 95, 50){\vector(  1,  0){30}}
 \put(  73,  38){\vector(  0, -1){16}}
 \put(  95,  10){\vector(  1,  0){28}}
 \thicklines
 \put(   1,  36){$x$ {\footnotesize ($\sim\lambda$)}}
 \put(  78,  29){$y_0$}
 \put( 105,  16){$y_1$}
 \put(  47, -3){\dashbox{1.0}(51,67){}}
\end{picture}
\caption{Composition $\Seq$ with shared input.\label{fig:sequential-shared}}
\end{subfigure}\hspace{0.1ex}\hfill
\begin{subfigure}[t]{0.51\textwidth}
\centering
\begin{picture}(105, 75)
 \put( 51, 69){$\alg_1 \liftSeq \alg_0$}
 \thicklines
 \put( 48, 40){\framebox(40,20){$\alg_0$}}
 \put( 48,  0){\framebox(40,20){$\alg_1$}}

 \put(  5,  50){\vector(  1,  0){41}}
 \put(  5,  10){\vector(  1,  0){41}}
 %\put( 90, 50){\vector(  1,  0){30}}
 \put( 68,  38){\vector(  0, -1){16}}
 \put( 90,  10){\vector(  1,  0){28}}
 \thicklines
 \put(  4,  56){$x_0$ {\footnotesize ($\sim\lambda$)}}
 \put(  4,  16){$x_1$ {\footnotesize ($\sim\lambda$)}}
 \put(  73,  29){$y_0$}
 \put( 100,  16){$y_1$}
 \put( 42, -3){\dashbox{1.0}(51,67){}}
\end{picture}
\caption{Composition $\liftSeq$ with independent inputs.\label{fig:sequential-independent}}
\end{subfigure}\hspace{0.4ex}\hfill
\caption{Two kinds of sequential compositions $\Seq$ and $\liftSeq$.\label{fig:sequentials}}
\end{figure}

\subsection{Sequential Compositions $\Seq$ and $\liftSeq$}
\label{subsec:compositionality}

In this section we show two kinds of compositionality.

\subsubsection{Sequential Composition $\Seq$ with Shared Input}
\label{sub:compositionality:Seq}
We first present the definition of the sequential composition with shared input (Fig.~\ref{fig:sequential-shared}).

\begin{definition}[Sequential composition $\Seq$]\label{def:seq}\rm
Given two randomized algorithms $\alg_0:\calx\rightarrow\Dists\caly_0$ and $\alg_1:\caly_0\times\calx\rightarrow\Dists\caly_1$, we define the \emph{sequential composition} of $\alg_0$ and $\alg_1$ as the randomized algorithm $\alg_1 \Seq \alg_0: \calx\rightarrow\Dists\caly_1$ such that: for any $x\in\calx$,\,
$(\alg_1 \Seq \alg_0)(x) = \alg_1(\alg_0(x), x))$.
\end{definition}

Then we show that it is harder to obfuscate distributions when an identical input is applied to the mechanism multiple times.

\begin{restatable}[Sequential composition $\Seq$ of $(\varepsilon, \delta)$-\DistP{}]{prop}{CompositionInf}\label{prop:Composition:Inf}
Let 
$\varPhi\subseteq\calx\times\calx$.
If $\alg_0:\calx\rightarrow\Dists\caly_0$ provides $(\varepsilon_0, \delta_0)$-\DistP{} w.r.t. $\lift{\varPhi}$ 
and for each $y_0\in\caly_0$, $\alg_1(y_0): \calx\rightarrow\Dists\caly_1$ provides $(\varepsilon_1, \delta_1)$-\DistP{} w.r.t. $\lift{\varPhi}$ then the sequential composition $\alg_1 \Seq \alg_0$ provides $(\varepsilon_0+\varepsilon_1, (\delta_0+\delta_1)\cdot|\varPhi|)$-\DistP{} w.r.t. $\lift{\varPhi}$.
\end{restatable}

\begin{proof}
By Theorem~\ref{thm:DPfromDistP}, $\alg_0$ provides $(\varepsilon_0, \delta_0)$-\DP{} w.r.t. $\varPhi$, 
and for each $y_0\in\caly_0$, $\alg_1(y_0)$ provides $(\varepsilon_1, \delta_1)$-\DP{} w.r.t. $\varPhi$.
By the sequential composition theorem for \DP{} mechanisms, $\alg_1 \Seq \alg_0$ provides $(\varepsilon_0+\varepsilon_1, \delta_0+\delta_1)$-\DP{} w.r.t. $\varPhi$.
By Theorem~\ref{thm:DPimpliesMaxDistP}, $\alg_1 \Seq \alg_0$ provides $(\varepsilon_0+\varepsilon_1, (\delta_0+\delta_1)\cdot|\varPhi|)$-\DistP{} w.r.t. $\lift{\varPhi}$.
\end{proof}

The compositionality for \XDistP{} can be shown analogously to Proposition~\ref{prop:Composition:Inf}.

\begin{restatable}[Sequential composition $\Seq$ of $(\varepsilon, \delta)$-\XDistP{}]{prop}{CompositionInfXDistP}\label{prop:Composition:Inf:XDistP}
Let 
$\varPhi\subseteq\calx\times\calx$.
If $\alg_0:\calx\rightarrow\Dists\caly_0$ provides $(\varepsilon_0,  \sensinf, \delta_0)$-\XDistP{} w.r.t. $\liftinf{\varPhi}$ 
and for each $y_0\in\caly_0$, $\alg_1(y_0): \calx\rightarrow\Dists\caly_1$ provides $(\varepsilon_1, \sensinf, \delta_1)$-\XDistP{} w.r.t. $\liftinf{\varPhi}$ then the sequential composition $\alg_1 \Seq \alg_0$ provides $(\varepsilon_0+\varepsilon_1, \sensinf, \allowbreak (\delta_0+\delta_1)\cdot|\varPhi|)$-\XDistP{} w.r.t. $\liftinf{\varPhi}$.
\end{restatable}

\begin{proof}
By Theorem~\ref{thm:DPfromDistP}, $\alg_0$ provides $(\varepsilon_0, \utmetric, \delta_0)$-\XDP{} w.r.t. $\varPhi$, 
and for each $y_0\in\caly_0$, $\alg_1(y_0)$ provides $(\varepsilon_1, \utmetric, \delta_1)$-\XDP{} w.r.t. $\varPhi$.
By the sequential composition theorem for \XDP{} mechanisms, $\alg_1 \Seq \alg_0$ provides $(\varepsilon_0+\varepsilon_1, \utmetric, \delta_0+\delta_1)$-\XDP{} w.r.t. $\varPhi$.
By Theorem~\ref{thm:max-DistPFromDP}, $\alg_1 \Seq \alg_0$ provides $(\varepsilon_0+\varepsilon_1, \sensinf, (\delta_0+\delta_1)\cdot|\varPhi|)$-\XDistP{} w.r.t. $\liftinf{\varPhi}$.
\end{proof}

\subsubsection{Sequential Composition $\liftSeq$ with Independent Sampling}
\label{subsec:compositionality:lift}

We first note that the lifting of the sequential composition $\lift{(\alg_1 \Seq \alg_0)}$ does not coincide with the sequential composition of the liftings $\lift{\alg_1} \Seq \lift{\alg_1}$.
\arxiv{Then}
\conference{Then}
the latter is a randomized algorithm $\lift{\alg_1} \Seq \lift{\alg_0}: \Dists\calx\rightarrow\Dists\caly_1$ such that for any $\lambda\in\Dists\calx$,\,
$(\lift{\alg_1} \Seq \lift{\alg_0})(\lambda) = \lift{\alg_1}(\lift{\alg_0}(\lambda), \lambda)$.
Then each of $\alg_0$ and $\alg_1$ receives an input independently sampled from $\lambda$.
This is different from $\alg_1 \Seq \alg_0$, in which $\alg_0$ and $\alg_1$ share an identical input drawn from $\lambda$, as shown in Figure~\ref{fig:sequential-shared}.

To see this difference in detail, we deal with another definition of sequential composition.
\begin{definition}[Sequential composition $\liftSeq$]\label{def:seq:lift}\rm
Given two randomized algorithms $\alg_0:\calx\rightarrow\Dists\caly_0$ and $\alg_1:\caly_0\times\calx\rightarrow\Dists\caly_1$, we define the \emph{sequential composition} of $\alg_0$ and $\alg_1$ as the randomized algorithm $\alg_1 \liftSeq \alg_0: \calx\times\calx\rightarrow\Dists\caly_1$ such that: for any $x_0, x_1\in\calx$,\,
$(\alg_1 \liftSeq \alg_0)(x_0, x_1) = \alg_1(\alg_0(x_0), x_1))$.
\end{definition}

Then the lifting of the sequential composition $\lift{(\alg_1 \liftSeq \alg_0)}$ coincides with the sequential composition of the liftings $\lift{\alg_1\!} \liftSeq \lift{\alg_1\!}$ in the sense that
for any $\lambda_0, \lambda_1 \in\Dists\calx$,\,
\begin{align}\label{eq:relation:product}
\lift{(\alg_1 \liftSeq \alg_0)}(\lambda_0 {\times} \lambda_1) 
&= 
\lift{\alg_1\!}(\lift{\alg_0\!}(\lambda_0), \lambda_1) 
=
(\lift{\alg_1\!} \liftSeq \lift{\alg_1\!})(\lambda_0, \lambda_1),
\end{align}
where $\lambda_0\times\lambda_1$ is the probability distribution over $\calx\times\calx$ such that for all $x_0, x_1\in\calx$,\, $(\lambda_0\times\lambda_1)[x_0, x_1] = \lambda_0[x_0] \lambda_1[x_1]$.

To show the compositionality for distribution privacy, we introduce an operator $\diamond$ between binary relations $\varPsi_0$ and $\varPsi_1$ by:
\[
\varPsi_0 \diamond \varPsi_1 =
\{ (\lambda_0\times\lambda_1,\, \lambda'_0\times\lambda'_1) \mid
   (\lambda_0,\lambda'_0)\in\varPsi_0, (\lambda_1,\lambda'_1)\in\varPsi_1 \}.
\]

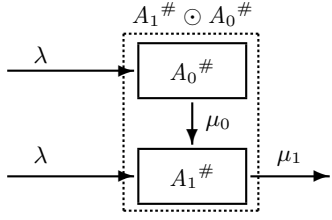
\begin{figure}[t]
\centering
\begin{picture}(169, 68)
 \put( 55, 68){$\lift{\alg_1} \Seq \lift{\alg_0}$}
 \thicklines
 \put( 58, 40){\framebox(40,20){$\lift{\alg_0}$}}
 \put( 58,  0){\framebox(40,20){$\lift{\alg_1}$}}

 \put(   8,  50){\vector(  1,  0){48}}
 \put(   8,  10){\vector(  1,  0){48}}
 %\put( 100, 50){\vector(  1,  0){30}}
 \put(  78,  38){\vector(  0, -1){16}}
 \put( 100,  10){\vector(  1,  0){30}}
 \thicklines
 \put(  18,  53){$\lambda$}
 \put(  18,  15){$\lambda$}
 \put(  83,  29){$\mu_0$}
 \put( 110,  16){$\mu_1$}
 \put(52, -3){\dashbox{1.0}(51,67){}}
\end{picture}
\caption{Sequential composition $\lift{\alg_1} \Seq \lift{\alg_0}$ involves two independent samplings from~$\lambda$.\label{fig:sequential-lift}}
\end{figure}

\begin{restatable}[Sequential composition $\liftSeq$ of $(\varepsilon, \delta)$-\DistP{}]{prop}{CompositionInfLift}\label{prop:Composition:InfLift}
Let 
$\varPsi\subseteq\Dists\calx\times\Dists\calx$.
If $\alg_0:\calx\rightarrow\Dists\caly_0$ provides $(\varepsilon_0, \delta_0)$-\DistP{} w.r.t. $\varPsi$ 
and for each $y_0\in\caly_0$, $\alg_1(y_0): \calx\rightarrow\Dists\caly_1$ provides $(\varepsilon_1, \delta_1)$-\DistP{} w.r.t. $\varPsi$ then the sequential composition $\alg_1 \liftSeq \alg_0$ provides $(\varepsilon_0+\varepsilon_1, \delta_0+\delta_1)$-\DistP{} w.r.t. $\varPsi \diamond \varPsi$.
\end{restatable}

\begin{proof}
By the definition of \DistP{}, we have that $\lift{\alg_0}$ provides $(\varepsilon_0, \delta_0)$-\DP{} w.r.t. $\varPsi$ 
and for each $y_0\in\caly_0$, $\lift{\alg_1(y_0)}$ provides $(\varepsilon_1, \delta_1)$-\DP{} w.r.t. $\varPsi$.
By the sequential composition theorem for \DP{} mechanisms and Equation~\eqref{eq:relation:product},
$\lift{(\alg_1 \liftSeq \alg_0)}$ provides $(\varepsilon_0+\varepsilon_1, \delta_0+\delta_1)$-\DP{} w.r.t. $\varPsi \diamond \varPsi$.
Therefore $\alg_1 \liftSeq \alg_0$ provides $(\varepsilon_0+\varepsilon_1, \delta_0+\delta_1)$-\DistP{} w.r.t. $\varPsi \diamond \varPsi$.
\end{proof}

For brevity, we omit the case of \XDistP{}.

Finally, we remark that the comparison of the compositions with shared input and with independent input is also discussed from the viewpoint of quantitative information flow in~\cite{Kawamoto:17:LMCS}.

\subsection{Post-processing and Pre-processing}
\label{subsec:post-pre-processing}

Next we show that distribution privacy is immune to the post-processing.
For $\alg_0:\calx\rightarrow\Dists\caly$ and $\alg_1:\caly\rightarrow\Dists\calz$, we define $\alg_1 \circ \alg_0$ by: $(\alg_1 \circ \alg_0)(x) = \alg_1(\alg_0(x))$.

\begin{restatable}[Post-processing]{prop}{PostProcess}\label{prop:PostProcess}
Let 
$\varPsi\subseteq\Dists\calx\times\Dists\calx$, and $\sensfunc:\Dists\calx\times\Dists\calx\rightarrow\realsnng$ be a metric.
Let $\alg_0:\calx\rightarrow\Dists\caly$, and $\alg_1:\caly\rightarrow\Dists\calz$.
\begin{enumerate}
\item
If $\alg_0$ provides $(\varepsilon, \delta)$-\DistP{} w.r.t. $\varPsi$ then so does the composite function $\alg_1\circ\alg_0$.
\item
If $\alg_0$ provides $(\varepsilon, \sensfunc, \delta)$-\XDistP{} w.r.t. $\varPsi$ then so does the composite function $\alg_1\circ\alg_0$.
\end{enumerate}
\end{restatable}

\begin{proof}
Let $(\lambda_0, \lambda_1)\in\varPsi$.
Since every randomized algorithm is a convex combination of deterministic algorithms, there are a distribution $\mu$ over an index set $\cali$ and deterministic algorithms $\alg_{1,i}: U_i\rightarrow R$ such that $\alg_1 = \sum_{i\in\cali} \mu[i] \alg_{1,i}$.
Then we obtain:
\begin{align}
&
  \lift{(\alg_1\circ\alg_0)}(\lambda_0)[R]
\nonumber
\\ 
=&\,
  {\textstyle \sum_{x_0\in\calx}}\,\lambda_0[x_0]
  \cdot (\alg_1\circ\alg_0)(x_0)[R]
\nonumber
\\ =&\,
  {\textstyle \sum_{i\in\cali}}\, \mu[i]\,
  {\textstyle \sum_{x_0\in\calx}}\,\lambda_0[x_0]
  \cdot (\alg_{1,i}\circ\alg_0)(x_0)[R]
\nonumber
\\ =&\,
  {\textstyle \sum_{i\in\cali}}\, \mu[i]\,
  {\textstyle \sum_{x_0\in\calx}}\, \lambda_0[x_0]
  \cdot \alg_0(x_0)[U_i]
\nonumber
\\ =&\,
  {\textstyle \sum_{i\in\cali}}\, \mu[i]
  \cdot \lift{\alg_0}(\lambda_0)[U_i]
\label{eq:post-process-inf}
{.}
\end{align}

Then we show the first claim as follows.
Assume that $\alg_0$ provides $(\varepsilon, \delta)$-\DistP{} w.r.t. $\varPsi$.
\begin{align*}
&
  \lift{(\alg_1\circ\alg_0)}(\lambda_0)[R]
\\ =&\,
  {\textstyle \sum_{i\in\cali}}\, \mu[i]
  \cdot \lift{\alg_0}(\lambda_0)[U_i]
\hspace{20ex}
\text{(by \eqref{eq:post-process-inf})}
\\ \le&\,
  {\textstyle \sum_{i\in\cali}}\, \mu[i]
  \cdot \bigl( e^{\varepsilon}
  \cdot\lift{\alg_0}(\lambda_1)[U_i]
  + \delta \bigr)
\\[-1ex]&\hspace{24ex}
  \text{(by $(\varepsilon, \delta)$-\DistP{} of $\lift{\alg_0}$)}
\\ =&\,
  {\textstyle \sum_{i\in\cali}}\, \mu[i]
  \cdot \bigl( e^{\varepsilon}
  \cdot{\textstyle \sum_{x_1}}\, \lambda_1[x_1]\alg_0(x_1)[U_i]
  + \delta \bigr)
\\ =&\,
  {\textstyle \sum_{i\in\cali}}\, \mu[i]
  \cdot \bigl( e^{\varepsilon}
  \cdot{\textstyle \sum_{x_1}}\, \lambda_1[x_1](\alg_{1,i}\circ\alg_0)(x_1)[R]
  + \delta \bigr)
\\ =&\,
  e^{\varepsilon}\cdot
  \Bigl( {\textstyle \sum_{i\in\cali}}\, \mu[i]
  \cdot \lift{(\alg_{1,i}\circ\alg_0)}(\lambda_1)[R] \Bigr)
  + \delta
\\ =&\,
  e^{\varepsilon}\cdot
  \lift{(\alg_1\circ\alg_0)}(\lambda_1))[R]
  + \delta
{.}
\end{align*}
Therefore $\alg_1\circ\alg_0$ provides $(\varepsilon, \delta)$-\DistP{} w.r.t. $\varPsi$.

\arxiv{Analogously, we show the second claim.
Assume that $\alg_0$ provides $(\varepsilon, \sensfunc, \delta)$-\XDistP{} w.r.t. $\varPsi$.
\begin{align*}
&
  \lift{(\alg_1\circ\alg_0)}(\lambda_0)[R]
\\ =&\,
  \sum_{i\in\cali} \mu[i]
  \cdot \lift{\alg_0}(\lambda_0)[U_i]
\hspace{21ex}
\text{(by \eqref{eq:post-process-inf})}
\\ \le&\,
  \sum_{i\in\cali} \mu[i]
  \cdot \bigl( e^{\varepsilon \cdot \sensfunc(\lambda_0, \lambda_1)}
  \cdot\lift{\alg_0}(\lambda_1)[U_i]
  + \delta \bigr)
\\[-1ex]&\hspace{20ex}
  \text{(by $(\varepsilon, \sensfunc, \delta)$-\XDistP{} of $\lift{\alg_0}$)}
\\ =&\,
  e^{\varepsilon {\cdot} \sensfunc(\lambda_0, \lambda_1)} \cdot
  \Bigl( \sum_{i\in\cali} \mu[i]
  {\cdot} \lift{(\alg_{1,i}\circ\alg_0)}(\lambda_1)[R] \Bigr)
  + \delta
\\ =&\,
  e^{\varepsilon {\cdot} \sensfunc(\lambda_0, \lambda_1)} {\cdot}
  \lift{(\alg_1\circ\alg_0)}(\lambda_1))[R]
  + \delta
{.}
\end{align*}
Therefore $\alg_1\circ\alg_0$ provides $(\varepsilon, \sensfunc, \delta)$-\XDistP{} w.r.t. $\varPsi$.
}
\qed
\end{proof}

We then show a property on pre-processing.
\begin{restatable}[Pre-processing]{prop}{PreProcessMAX}\label{prop:PreProcessMAX}
Let $c\in\realsnng$, 
$\varPsi\subseteq\Dists\calx\times\Dists\calx$, 
and $\sensfunc:\Dists\calx\times\Dists\calx\rightarrow\realsnng$ be a metric.
\begin{enumerate}
\item
If $T:\Dists\calx\rightarrow\Dists\calx$ is a \emph{$(c,\varPsi)$-stable} transformation and 
$\alg:\calx\rightarrow\Dists\caly$ provides $(\varepsilon,\delta)$-\DistP{} w.r.t $\varPsi$, then $\alg\circ T$ provides $(c\,\varepsilon, \delta)$-\DistP{} w.r.t $\varPsi$.
\item
If $T:\Dists\calx\rightarrow\Dists\calx$ is a \emph{$(c,\sensfunc)$-stable} transformation and 
$\alg:\calx\rightarrow\Dists\caly$ provides $(\varepsilon,\sensfunc, \delta)$-\XDistP{}, then $\alg\circ T$ provides $(c\,\varepsilon, \sensfunc, \delta)$-\XDistP{}.
\end{enumerate}
\end{restatable}

\begin{proof}
We show the first claim as follows.
Assume that $\alg$ provides $(\varepsilon,\delta)$-\DistP{} w.r.t. $\varPsi$.
Let $(\lambda,\lambda')\in\varPsi$,
and $R \subseteq \caly$.
Then we have:
\begin{align*}
  \lift{(\alg\circ T)}(\lambda)[R]
=&\,
  \lift{\alg}(\lift{T}(\lambda))[R]
\\ \le&\,
  e^{c \varepsilon} \lift{\alg}(\lift{T}(\lambda'))[R]
\\ =&\,
  e^{c \varepsilon} \lift{(\alg\circ T)}(\lambda')[R]
{.}
\end{align*}
Therefore $\alg\circ T$ provides $(c\,\varepsilon, \delta)$-\DistP{}.

\conference{The second claim can be shown analogously.}
\arxiv{
Next we show the second claim.
Assume that $\alg$ provides $(\varepsilon,\sensfunc,\delta)$-\XDistP{}.
Let $\lambda,\lambda'\in\Dists\calx$,
and $R \subseteq \caly$.
Then we obtain:
\begin{align*}
  \lift{(\alg\circ T)}(\lambda)[R]
=&\,
  \lift{\alg}(\lift{T}(\lambda))[R]
\\ \le&\,
  e^{\varepsilon \sensfunc(\lift{T}(\lambda),\lift{T}(\lambda'))}
  \lift{\alg}(\lift{T}(\lambda'))[R]
\\ \le&\,
  e^{c \varepsilon \sensfunc(\lambda,\lambda')} \lift{(\alg\circ T)}(\lambda')[R]
{.}
\end{align*}
Therefore $\alg\circ T$ provides $(c\,\varepsilon, \sensfunc, \delta)$-\XDistP{}.
}
\qed
\end{proof}

\subsection{Probabilistic Distribution Privacy (\pDistP{})}
\label{subsec:intro:pDistP}

We next introduce an approximate notion of distribution privacy analogously to the notion of probabilistic differential privacy (\pDP{})~\cite{Dwork:10:FOCS}.
Intuitively, a randomized algorithm provides \emph{$(\varepsilon,\delta)$-probabilistic distribution privacy} if it provides $\varepsilon$-distribution privacy with probability at least $(1-\delta)$.

\begin{definition}[Probabilistic distribution privacy]\label{def:prob-max-DistP}\rm
Let 
$\varPsi\subseteq\Dists\calx\times\Dists\calx$.
We say that a randomized algorithm $\alg:\calx\rightarrow\Dists\caly$ provides \emph{$(\varepsilon,\delta)$-probabilistic distribution privacy (\pDistP{}) w.r.t.}\,$\varPsi$ if the lifting $\lift{\alg}$ provides $(\varepsilon,\delta)$-probabilistic differential privacy w.r.t. $\varPsi$ ,
i.e., for all $(\lambda, \lambda')\in\varPsi$, there exists an $R'\subseteq\caly$ such that $\lift{\alg}(\lambda)[R'] \leq \delta$, $\lift{\alg}(\lambda')[R'] \leq \delta$, and that for all $R\subseteq\caly$,\, we have:
\begin{align*}
\lift{\alg}(\lambda)[R\setminus R'] & \leq
e^{\varepsilon} \cdot \lift{\alg}(\lambda')[R\setminus R']
\\
\lift{\alg}(\lambda')[R\setminus R'] & \leq
e^{\varepsilon} \cdot \lift{\alg}(\lambda)[R\setminus R']
{,}
\end{align*}
where the probability space is taken over the choices of randomness in $\alg$.
\end{definition}

\subsection{Relationships between \pDistP{} and \DistP{}}
\label{sub:proofs-PDistP-DistP}

In this section we show the relationships between \pDistP{} and \DistP{}.
By definition, $(\varepsilon,0)$-\DistP{} is equivalent to $(\varepsilon,0)$-\pDistP{}.
In general, however, $(\varepsilon,\delta)$-\DistP{} does not imply $(\varepsilon,\delta)$-\pDistP{}, while $(\varepsilon,\delta)$-\pDistP{} implies $(\varepsilon,\delta)$-\DistP{}.

\begin{restatable}[$(\varepsilon, 0)$-\pDistP{} $\Leftrightarrow$ $(\varepsilon, 0)$-\DistP{}]{prop}{PDistPequlasDistP}
\label{prop:PDistPequlasDistP}
A randomized algorithm $\alg$ provides $(\varepsilon, 0)$-\pDistP{} w.r.t. an adjacency relation $\varPsi$ iff it provides $(\varepsilon, 0)$-\DistP{} w.r.t.~$\varPsi$.
\end{restatable}

\begin{proof}
Immediate from the definitions.
\qed
\end{proof}

\begin{restatable}[$(\varepsilon, \delta)$-\pDistP{} $\Rightarrow$ $(\varepsilon, \delta)$-\DistP{}]{prop}{PDistPimpliesDistP}
\label{prop:PDistPimpliesDistP}
If a randomized algorithm $\alg$ provides $(\varepsilon, \delta)$-\pDistP{} w.r.t. an adjacency relation $\varPsi$, then it provides $(\varepsilon, \delta)$-\DistP{} w.r.t. $\varPsi$.
\end{restatable}

\begin{proof}
Let $\varPsi\subseteq\Dists\calx\times\Dists\calx$, $(\lambda, \lambda')\in\varPsi$, $R\subseteq\caly$,
and $\alg:\calx\rightarrow\Dists\caly$ be any randomized algorithm that provides $(\varepsilon, \delta)$-\pDistP{} w.r.t.~$\varPsi$.
Then there exists an $R'\subseteq\caly$ such that 
$\lift{\alg}(\lambda)[R'] \leq \delta$, and 
$\lift{\alg}(\lambda)[R\setminus R'] \leq
e^{\varepsilon}\cdot \lift{\alg}(\lambda')[R\setminus R']$.
Hence we obtain:
\begin{align*}
\lift{\alg}(\lambda)[R] 
&=
\lift{\alg}(\lambda)[R\setminus R'] + \lift{\alg}(\lambda)[R']
\\ &\leq
e^{\varepsilon}\cdot\lift{\alg}(\lambda')[R\setminus R'] + \delta
\\ &\leq
e^{\varepsilon}\cdot\lift{\alg}(\lambda')[R] + \delta
{.}
\end{align*}
Therefore $\alg$ provides $(\varepsilon, \delta)$-\DistP{} w.r.t.~$\varPsi$.
\qed
\end{proof}

\subsection{Properties of the Tupling Mechanism}
\label{sub:tupling:properties}

In this section we show properties of the tupling mechanism.

We first present a minor result with $\delta = 0$.

\begin{restatable}[$(\varepsilon_{\alg}, 0)$-\DistP{} of the tupling mechanism]{prop}{TuplingDistP}
\label{prop:TuplingDistP}
If $\alg$ provides $(\varepsilon_{\alg}, 0)$-\DP{} w.r.t. $\varPhi$, then  the $(k,\nu,\alg)$-tupling mechanism $\TPM$ provides $(\varepsilon_{\alg}, 0)$-\DistP{} w.r.t.~$\lift{\varPhi}$.
\end{restatable}

\begin{proof}
Let $\varPhi\subseteq\calx\times\calx$.
Since $\alg$ provides $(\varepsilon_{\alg}, 0)$-\DP{} w.r.t. $\varPhi$, Theorem~\ref{thm:DPimpliesMaxDistP} implies that $\alg$ provides $(\varepsilon_{\alg}, 0)$-\DistP{} w.r.t.~$\lift{\varPhi}$.

Let $(\lambda_0, \lambda_1)\in\lift{\varPhi}$ and $R\subseteq \caly^{k+1}$ such that $\TPM(\lambda_0)[R] > 0$.
Since the dummies are uniformly distributed over $\caly$, we have $\TPM(\lambda_1)[R] > 0$.
Let $\bar{y} = (y_1, y_2, \ldots, y_{k+1}) \in R$ be an output of $\TPM$.
Since $i$ is uniformly drawn from $\{1, 2, \ldots, k+1\}$ in the mechanism $\TPM$, the output of $\alg$ appears as the $i$-th element $y_i$ of the tuple $\bar{y}$ with probability $\frac{1}{k+1}$.
For each $b = 0, 1$, when an input $x$ is drawn from $\lambda_b$, the probability that $\alg$ outputs $y_i$ is:
\begin{align*}
\lift{\alg}(\lambda_b)[y_i] =
{\textstyle \sum_{x\in\calx}}\, \lambda_b[x] \alg(x)[y_i]
{.}
\end{align*}
On the other hand, for each $j\neq i$, the probability that $y_j$ is drawn from the dummy distribution $\nu$ is given by $\nu[y_j]$.
Therefore, the probability that the mechanism $\TPM$ outputs the tuple $\bar{y}$ is:
\begin{align*}
\TPM(\lambda_b)[\bar{y}] =
{\textstyle 
\frac{1}{k+1}
\sum_{i=1}^{k+1}\, 
\lift{\alg}(\lambda_b)[y_i]\prod_{j\neq i} \nu[y_j]
}
{.}
\end{align*}
Hence we obtain:
\begin{align*}
&\,
\frac{ \TPM(\lambda_0)[R] }
       { \TPM(\lambda_1)[R] } 
\\ =&\,
\frac{ \sum_{\bar{y}\in R} \TPM(\lambda_0)[\bar{y}] }
       { \sum_{\bar{y}\in R} \TPM(\lambda_1)[\bar{y}] } 
\\ =&\,
\frac{ \sum_{\bar{y}\in R} \frac{1}{k+1} \sum_{i=1}^{k+1} \lift{\alg}(\lambda_0)[y_i]\prod_{j\neq i} \nu[y_j] }
       { \sum_{\bar{y}\in R} \frac{1}{k+1} \sum_{i=1}^{k+1} \lift{\alg}(\lambda_1)[y_i]\hspace{-0ex}\prod_{j\neq i} \nu[y_j] }
\\ =&\,
\frac{ \sum_{\bar{y}\in R} \sum_{i=1}^{k+1} \lift{\alg}(\lambda_0)[y_i] }
       { \sum_{\bar{y}\in R} \sum_{i=1}^{k+1} \lift{\alg}(\lambda_1)[y_i] }
&\hspace{-16ex}\text{(since $\nu$ is uniform)}
\\ \le&\,
\max_{\bar{y}\in R}\, \max_{i\in\{1, 2, \ldots , k+1\} } 
\frac{ \lift{\alg}(\lambda_0)[y_i] }{ \lift{\alg}(\lambda_1)[y_i] }
\\ \le&\,
e^{\varepsilon_{\alg}}
&\hspace{-18ex}\text{(by $(\varepsilon_{\alg}, 0)$-\DistP{} of $\alg$ w.r.t. $\lift{\varPhi}$)}
\end{align*}
Therefore $\TPM$ provides $(\varepsilon_{\alg}, 0)$-\DistP{} w.r.t. $\lift{\varPhi}$.
\end{proof}

Next we show that the tupling mechanism provides \pDistP{}
without any restriction on $\alg$, i.e., $\alg$ does not have to provide \DP{} in order for the tupling mechanism to provides \pDistP{}.

\begin{restatable}[$(\varepsilon_{\alpha}, \delta_{\alpha})$-\pDistP{} of the tupling mechanism]{prop}{TuplingPDistP}
\label{prop:TuplingPDistP}
Let 
$k\in\natspos$, $\nu$ be the uniform distribution over $\caly$, $\alg: \calx\rightarrow\Dists\caly$,
and $\beta, \eta\in[0, 1]$.
For an $\alpha\in\realspos$, let 
$\varepsilon_{\alpha} = \ln{\textstyle\frac{ k + (\alpha + \beta)\cdot|\caly| }{ k - \alpha\cdot|\caly| }}$
and $\delta_{\alpha} = 2\exp\bigl(-\frac{2\alpha^2}{k\beta^2}\bigr) + \eta$.
Then for any $0 < \alpha < \frac{k}{|\caly|}$, the tupling mechanism $\TPM$ provides $(\varepsilon_{\alpha}, \delta_{\alpha})$-\pDistP{}
w.r.t $\LambdaBeta^2$.
\end{restatable}

\begin{proof}
Let $\lambda_0, \lambda_1\in \LambdaBeta$,
and $R\subseteq \caly^{k+1}$ such that $\TPM(\lambda_0)[R] > 0$.
Since the dummies are uniformly distributed over $\caly$, we have $\TPM(\lambda_1)[R] >~0$.
As in the proof for Proposition~\ref{prop:TuplingDistP}, for any $\bar{y} = (y_1, y_2, \allowbreak \ldots, y_{k+1}) \in R$:
\begin{align}
\frac{ \TPM(\lambda_0)[\bar{y}] }
       { \TPM(\lambda_1)[\bar{y}] } 
&=
\frac{ \sum_{i=1}^{k+1} \lift{\alg}(\lambda_0)[y_i] }
       { \sum_{i=1}^{k+1} \lift{\alg}(\lambda_1)[y_i] }
{.}
\label{eq:TPM-equality}
\end{align}
Recall that in the definition of the tupling mechanism, $y_i$ is the output of $\alg$ while for each $j\neq i$, $y_j$ is generated from the uniform distribution $\nu$ over $\caly$.
Hence the expected value of $\sum_{j\neq i} \lift{\alg}(\lambda_0)[y_j]$ is given by $\frac{k}{|\caly|}$.
Therefore it follows from the Hoeffding's inequality that for any $\alpha>0$, we have:
\begin{align*}
\Pr\bigl[\, {\textstyle\sum_{j\neq i}}\, \lift{\alg}(\lambda_0)[y_j] \ge {\textstyle\frac{k}{|\caly|}} + \alpha \,\bigr]
&\le
\exp\bigl({\textstyle -\frac{2\alpha^2}{k\beta^2} }\bigr)
\\[-0.0ex]
\Pr\bigl[\, {\textstyle\sum_{j\neq i}}\, \lift{\alg}(\lambda_1)[y_j] \le {\textstyle\frac{k}{|\caly|}} - \alpha \,\bigr]
&\le
\exp\bigl({\textstyle -\frac{2\alpha^2}{k\beta^2} }\bigr)
\end{align*}
where each $y_j$ is independently drawn from $\nu$, hence each of $\lift{\alg}(\lambda_0)[y_j]$ and $\lift{\alg}(\lambda_1)[y_j]$ is independent.
By $\delta_{\alpha} = 2\exp\bigl(-\frac{2\alpha^2}{k\beta^2}\bigr)$,
we obtain:
\begin{align}
\!\Pr\Bigl[ \sum_{j\neq i}\!\lift{\alg}(\lambda_0)[y_j]
  {\ge} {\textstyle\frac{k}{|\caly|}} + \alpha
  \text{ \footnotesize and} 
  \sum_{j\neq i}\!\lift{\alg}(\lambda_1)[y_j]
  {\le} {\textstyle\frac{k}{|\caly|}} - \alpha
  \Bigr]
{\le} \delta_{\alpha}
\label{eq:two-prob-less-than-delta}
{.}
\end{align}

When $\bar{y}$ is drawn from $\TPM(\lambda_0)$, let $y_i$ be the output of $\lift{\alg}(\lambda_0)$.
Then we obtain:
\begin{align}
&\phantom{=}~
\Pr\biggl[\,
\frac{ \TPM(\lambda_0)[\bar{y}] }
       { \TPM(\lambda_1)[\bar{y}] } 
\ge e^\varepsilon
\,\biggr]
\nonumber \\ &=
\Pr\biggl[\,
\frac{ \sum_{j=1}^{k+1} \lift{\alg}(\lambda_0)[y_j] }
     { \sum_{j=1}^{k+1} \lift{\alg}(\lambda_1)[y_j] }
\ge e^\varepsilon
\,\biggr]
&\hspace{-12ex}
\text{(by \eqref{eq:TPM-equality})}
\nonumber \\ &=
\Pr\biggl[\,
\frac{ \sum_{j=1}^{k+1} \lift{\alg}(\lambda_0)[y_j] }
     { \sum_{j=1}^{k+1} \lift{\alg}(\lambda_1)[y_j] }
\ge
\frac{ {\textstyle\frac{k}{|\caly|}} + \alpha + \beta }
     { {\textstyle\frac{k}{|\caly|}} - \alpha }
\,\biggr]
&\hspace{-9ex}
\text{(by def. of $\varepsilon$)}
\nonumber \\ &\le
\Pr\biggl[\,
\frac{ \sum_{j=1}^{k+1} \lift{\alg}(\lambda_0)[y_j] }
     { \sum_{j=1}^{k+1} \lift{\alg}(\lambda_1)[y_j] }
\ge
\frac{ {\textstyle\frac{k}{|\caly|}} + \alpha + \lift{\alg}(\lambda_0)[y_i] }
     { {\textstyle\frac{k}{|\caly|}} - \alpha + \lift{\alg}(\lambda_1)[y_i]}
\,\biggr]
+ \eta
\hspace{-4ex}~
\nonumber \\ &&\hspace{-50ex}
\bigl(
\text{by
$\Pr\bigl[ \lift{\alg}(\lambda_0)[y_i] \le \beta \bigr]
\ge 1 - \eta$ 
~and~$\lift{\alg}(\lambda_1)[y_i] \ge 0$}
\bigr)
\nonumber \\ &=
\Pr\biggl[\,
\frac{ \sum_{j\neq i} \lift{\alg}(\lambda_0)[y_j] }
     { \sum_{j\neq i} \lift{\alg}(\lambda_1)[y_j] }
\ge
\frac{ {\textstyle\frac{k}{|\caly|}} + \alpha }
     { {\textstyle\frac{k}{|\caly|}} - \alpha }
\,\biggr]
+ \eta
\nonumber \\ &\le \delta_{\alpha}
{.}
&\hspace{-25ex}
\text{(by \eqref{eq:two-prob-less-than-delta})}
\nonumber
\end{align}
Then there is an $R'\subseteq\caly^{k+1}$ such that 
$\TPM(\lambda_0)[R'] \le \delta_{\alpha}$, and that 
for any $\bar{y}\in R$,\, $\bar{y} \in R'$ iff
$\frac{ \TPM(\lambda_0)[\bar{y}] }
      { \TPM(\lambda_1)[\bar{y}] } 
\ge e^\varepsilon$.
Then: 
\begin{align*}
\TPM(\lambda_1)[R']
&= 
\sum_{\bar{y} \in R'} \TPM(\lambda_1)[\bar{y}]
\\ &\le 
\sum_{\bar{y} \in R'} e^{-\varepsilon}\cdot \TPM(\lambda_0)[\bar{y}] 
\\ &=
e^{-\varepsilon}\cdot \TPM(\lambda_0)[R'] 
\\ &\le
\delta_{\alpha}
{.}
\end{align*}
Therefore $\TPM$ provides $(\varepsilon_{\alpha}, \delta_{\alpha})$-\pDistP{}
w.r.t. $\LambdaBeta^2$.

\end{proof}

Then, we obtain the \DistP{} of the tupling mechanism from Propositions~\ref{prop:TuplingPDistP} and 
\ref{prop:PDistPimpliesDistP} as follows.
\TuplingDP*

\begin{proof}
By Proposition~\ref{prop:TuplingPDistP}, $\TPM$ provides $(\varepsilon_{\alpha}, \delta_{\alpha})$-\pDistP{} w.r.t. 
$\LambdaBeta^2$.
Hence by Proposition~\ref{prop:PDistPimpliesDistP}, $\TPM$ provides $(\varepsilon_{\alpha}, \delta_{\alpha})$-\DistP{} w.r.t. 
$\LambdaBeta^2$.
\end{proof}

Finally, we note that if $\alg$ spreads the input distribution (e.g., $\alg$ provides $\varepsilon$-\DP{} for a smaller $\varepsilon$), then $\beta$ becomes smaller (by Proposition~\ref{prop:dist-prob-bounds}), hence smaller values of $\varepsilon_{\alpha}$ and $\delta_{\alpha}$ as shown in experimental results in Section~\ref{sec:location:dist-obf}.

\begin{restatable}{prop}{distProbBounds}
\label{prop:dist-prob-bounds}
Let $\lambda\in\Dists\calx$, $\alg: \calx\rightarrow\Dists\caly$, and $y\in\caly$.
Then we have $\min_{x\in\calx} \lambda[x] \le \lift{\alg}(\lambda)[y] \le \max_{x\in\calx} \lambda[x]$.
\end{restatable}

\begin{proof}
By 
$\lift{\alg}(\lambda)[y] = \sum_{x'\in\calx} \lambda[x'] \alg(x')[y]$,
we obtain:
\begin{align*}
\lift{\alg}(\lambda)[y]
&\le
\bigl( \max_{x\in\calx} \lambda[x] \bigr) \cdot {\textstyle\sum_{x'\in\calx}}\,\alg(x')[y]
=
\max_{x\in\calx} \lambda[x]
\\[-0.5ex]
\lift{\alg}(\lambda)[y]
&\ge
\bigl( \min_{x\in\calx} \lambda[x] \bigr) \cdot {\textstyle\sum_{x'\in\calx}}\,\alg(x')[y]
=
\min_{x\in\calx} \lambda[x]
{.}
\vspace{-5ex}
\end{align*}
\end{proof}
\vspace{-3ex}

}

\end{document}